\newif\iflncs \lncstrue
\newif\iffull \fulltrue
\newif\iforcid \orcidfalse
\newif\ifrestate \restatetrue

\documentclass[runningheads]{llncs}

\newcommand{\ottdrule}[4][]{{\displaystyle\frac{\begin{array}{l}#2\end{array}}{#3}\quad\ottdrulename{#4}}}

\newcommand{\ottpremise}[1]{ #1 \\}
\newenvironment{ottdefnblock}[3][]{ \framebox{\mbox{#2}} \quad #3 \\[0pt]}{}

\newcommand{\ottnt}[1]{\mathit{#1}}
\newcommand{\ottmv}[1]{\mathit{#1}}

\newcommand{\ottsym}[1]{#1}

\newcommand{\ottdrulename}[1]{\textsc{#1}}

\usepackage{amssymb}
\usepackage{amsmath,bm}
\usepackage{centernot}

\DeclareFontEncoding{LS1}{}{}
\DeclareFontSubstitution{LS1}{stix}{m}{n}
\DeclareSymbolFont{symbols2}{LS1}{stixfrak}{m}{n}
\DeclareMathSymbol{\typecolon}{\mathbin}{symbols2}{"25} % "

\newif\ifvector \vectorfalse

\newcommand{\algeffseqover}[1]{\ifvector\overrightarrow{#1}\else\bm{#1}\fi}
\newcommand{\algefftransarrow}[1]{ \vartriangleright^{#1} }

\newcommand{\algeffseqoverindex}[2]{ \algeffseqover{#1}^{#2} }

\newcommand{\ottdruleElabXXVar}[1]{\ottdrule[#1]{%
\ottpremise{ \vdash  \Gamma  \quad   \mathit{x} \,  \mathord{:}  \,  \text{\unboldmath$\forall$}  \,  \algeffseqover{ \alpha }   \ottsym{.}  \ottnt{A} \,  \in  \, \Gamma  \quad  \Gamma  \vdash   \algeffseqover{ \ottnt{B} }   }%
}{
 \Gamma ;  \ottnt{R}   \vdash   \mathit{x}  :   \ottnt{A}    [   \algeffseqover{ \ottnt{B} }   \ottsym{/}   \algeffseqover{ \alpha }   ]    \,  |   \, \epsilon  \mathrel{\algefftransarrow{ \ottnt{S} } }  \ottnt{S}  \ottsym{(}  \mathit{x}  \ottsym{)} \,  \algeffseqover{ \ottnt{B} }  }{%
{\ottdrulename{Elab\_Var}}{}%
}}

\newcommand{\ottdruleElabXXConst}[1]{\ottdrule[#1]{%
\ottpremise{\vdash  \Gamma}%
}{
 \Gamma ;  \ottnt{R}   \vdash   \ottnt{c}  :   \mathit{ty}  (  \ottnt{c}  )   \,  |   \, \epsilon  \mathrel{\algefftransarrow{ \ottnt{S} } }  \ottnt{c} }{%
{\ottdrulename{Elab\_Const}}{}%
}}

\newcommand{\ottdruleElabXXAbs}[1]{\ottdrule[#1]{%
\ottpremise{ \Gamma  \ottsym{,}  \mathit{x} \,  \mathord{:}  \, \ottnt{A} ;  \ottnt{R}   \vdash   \ottnt{M}  :  \ottnt{B}  \,  |   \, \epsilon'  \mathrel{\algefftransarrow{  \ottnt{S}  \,\circ\, \{  \mathit{x}  \, {\mapsto} \,  \mathit{x}  \}  } }  \ottnt{e} }%
}{
 \Gamma ;  \ottnt{R}   \vdash    \lambda\!  \, \mathit{x}  \ottsym{.}  \ottnt{M}  :   \ottnt{A}   \rightarrow  \!  \epsilon'  \;  \ottnt{B}   \,  |   \, \epsilon  \mathrel{\algefftransarrow{ \ottnt{S} } }   \lambda\!  \, \mathit{x}  \ottsym{.}  \ottnt{e} }{%
{\ottdrulename{Elab\_Abs}}{}%
}}

\newcommand{\ottdruleElabXXApp}[1]{\ottdrule[#1]{%
\ottpremise{  \Gamma ;  \ottnt{R}   \vdash   \ottnt{M_{{\mathrm{1}}}}  :   \ottnt{A}   \rightarrow  \!  \epsilon'  \;  \ottnt{B}   \,  |   \, \epsilon  \mathrel{\algefftransarrow{ \ottnt{S} } }  \ottnt{e_{{\mathrm{1}}}}   \quad    \Gamma ;  \ottnt{R}   \vdash   \ottnt{M_{{\mathrm{2}}}}  :  \ottnt{A}  \,  |   \, \epsilon  \mathrel{\algefftransarrow{ \ottnt{S} } }  \ottnt{e_{{\mathrm{2}}}}   \quad  \epsilon' \,  \subseteq  \, \epsilon  }%
}{
 \Gamma ;  \ottnt{R}   \vdash   \ottnt{M_{{\mathrm{1}}}} \, \ottnt{M_{{\mathrm{2}}}}  :  \ottnt{B}  \,  |   \, \epsilon  \mathrel{\algefftransarrow{ \ottnt{S} } }  \ottnt{e_{{\mathrm{1}}}} \, \ottnt{e_{{\mathrm{2}}}} }{%
{\ottdrulename{Elab\_App}}{}%
}}

\newcommand{\ottdruleElabXXOp}[1]{\ottdrule[#1]{%
\ottpremise{ \mathit{ty} \, \ottsym{(}  \mathsf{op}  \ottsym{)} \,  =  \,   \text{\unboldmath$\forall$}     \algeffseqover{ \alpha }   .  \ottnt{A}  \hookrightarrow  \ottnt{B}   \quad   \mathsf{op} \,  \in  \, \epsilon  \quad    \Gamma ;  \ottnt{R}   \vdash   \ottnt{M}  :   \ottnt{A}    [   \algeffseqover{ \ottnt{C} }   \ottsym{/}   \algeffseqover{ \alpha }   ]    \,  |   \, \epsilon  \mathrel{\algefftransarrow{ \ottnt{S} } }  \ottnt{e}   \quad  \Gamma  \vdash   \algeffseqover{ \ottnt{C} }    }%
}{
 \Gamma ;  \ottnt{R}   \vdash    \textup{\texttt{\#}\relax}  \mathsf{op}   \ottsym{(}   \ottnt{M}   \ottsym{)}   :   \ottnt{B}    [   \algeffseqover{ \ottnt{C} }   \ottsym{/}   \algeffseqover{ \alpha }   ]    \,  |   \, \epsilon  \mathrel{\algefftransarrow{ \ottnt{S} } }   \textup{\texttt{\#}\relax}  \mathsf{op}   \ottsym{(}    \algeffseqover{ \ottnt{C} }    \ottsym{,}   \ottnt{e}   \ottsym{)}  }{%
{\ottdrulename{Elab\_Op}}{}%
}}

\newcommand{\ottdruleElabXXHandle}[1]{\ottdrule[#1]{%
\ottpremise{  \Gamma ;  \ottnt{R}   \vdash   \ottnt{M}  :  \ottnt{A}  \,  |   \, \epsilon  \mathrel{\algefftransarrow{ \ottnt{S} } }  \ottnt{e}   \quad   \Gamma ;  \ottnt{R}   \vdash   \ottnt{H}  :  \ottnt{A}  \,  |   \, \epsilon  \, \Rightarrow   \ottnt{B}  \,  |   \, \epsilon'  \mathrel{\algefftransarrow{ \ottnt{S} } }  \ottnt{h}  }%
}{
 \Gamma ;  \ottnt{R}   \vdash   \mathsf{handle} \, \ottnt{M} \, \mathsf{with} \, \ottnt{H}  :  \ottnt{B}  \,  |   \, \epsilon'  \mathrel{\algefftransarrow{ \ottnt{S} } }  \mathsf{handle} \, \ottnt{e} \, \mathsf{with} \, \ottnt{h} }{%
{\ottdrulename{Elab\_Handle}}{}%
}}

\newcommand{\ottdruleElabXXLet}[1]{\ottdrule[#1]{%
\ottpremise{  \Gamma  \ottsym{,}   \algeffseqover{ \alpha }  ;  \ottnt{R}   \vdash   \ottnt{M_{{\mathrm{1}}}}  :  \ottnt{A}  \,  |   \, \epsilon  \mathrel{\algefftransarrow{ \ottnt{S} } }  \ottnt{e_{{\mathrm{1}}}}   \quad  \ottnt{S'} \,  =  \,  \ottnt{S}  \,\circ\, \{  \mathit{x}  \, {\mapsto} \,  \mathit{x}  \}  }%
\ottpremise{ \Gamma  \ottsym{,}  \mathit{x} \,  \mathord{:}  \,  \text{\unboldmath$\forall$}  \,  \algeffseqover{ \alpha }   \ottsym{.}  \ottnt{A} ;  \ottnt{R}   \vdash   \ottnt{M_{{\mathrm{2}}}}  :  \ottnt{B}  \,  |   \, \epsilon  \mathrel{\algefftransarrow{ \ottnt{S'} } }  \ottnt{e_{{\mathrm{2}}}} }%
}{
 \Gamma ;  \ottnt{R}   \vdash   \mathsf{let} \, \mathit{x}  \ottsym{=}  \ottnt{M_{{\mathrm{1}}}} \,  \mathsf{in}  \, \ottnt{M_{{\mathrm{2}}}}  :  \ottnt{B}  \,  |   \, \epsilon  \mathrel{\algefftransarrow{ \ottnt{S} } }  \mathsf{let} \, \mathit{x}  \ottsym{=}   \Lambda\!  \,  \algeffseqover{ \alpha }   \ottsym{.}  \ottnt{e_{{\mathrm{1}}}} \,  \mathsf{in}  \, \ottnt{e_{{\mathrm{2}}}} }{%
{\ottdrulename{Elab\_Let}}{}%
}}

\newcommand{\ottdruleElabXXResume}[1]{\ottdrule[#1]{%
\ottpremise{ \ottnt{R} \,  =  \, \ottsym{(}   \algeffseqover{ \alpha }   \ottsym{,}  \mathit{x} \,  \mathord{:}  \, \ottnt{A}  \ottsym{,}   \ottnt{B}   \rightarrow  \!  \epsilon  \;  \ottnt{C}   \ottsym{)}  \quad   \vdash  \Gamma_{{\mathrm{1}}}  \ottsym{,}  \mathit{x} \,  \mathord{:}  \, \ottnt{D}  \ottsym{,}  \Gamma_{{\mathrm{2}}}  \quad    \algeffseqover{ \alpha }  \,  \in  \, \Gamma_{{\mathrm{1}}}  \quad  \epsilon \,  \subseteq  \, \epsilon'   }%
\ottpremise{  \mathit{y}  \text{ is fresh}   \quad  \ottnt{S'} \,  =  \,  \ottnt{S}  \,\circ\, \{  \mathit{x}  \, {\mapsto} \,  \mathit{y}  \}  }%
\ottpremise{ \Gamma_{{\mathrm{1}}}  \ottsym{,}  \Gamma_{{\mathrm{2}}}  \ottsym{,}    \algeffseqover{ \beta }   \ottsym{,}   \mathit{x} \,  \mathord{:}  \,    \ottnt{A}    [   \algeffseqover{ \beta }   \ottsym{/}   \algeffseqover{ \alpha }   ]       ;  \ottnt{R}   \vdash   \ottnt{M}  :   \ottnt{B}    [   \algeffseqover{ \beta }   \ottsym{/}   \algeffseqover{ \alpha }   ]    \,  |   \, \epsilon'  \mathrel{\algefftransarrow{ \ottnt{S'} } }  \ottnt{e} }%
}{
 \Gamma_{{\mathrm{1}}}  \ottsym{,}  \mathit{x} \,  \mathord{:}  \, \ottnt{D}  \ottsym{,}  \Gamma_{{\mathrm{2}}} ;  \ottnt{R}   \vdash   \mathsf{resume} \, \ottnt{M}  :  \ottnt{C}  \,  |   \, \epsilon'  \mathrel{\algefftransarrow{ \ottnt{S} } }  \mathsf{resume} \,  \algeffseqover{ \beta }  \, \mathit{y}  \ottsym{.}  \ottnt{e} }{%
{\ottdrulename{Elab\_Resume}}{}%
}}

\newcommand{\ottdruleElabXXWeak}[1]{\ottdrule[#1]{%
\ottpremise{  \Gamma ;  \ottnt{R}   \vdash   \ottnt{M}  :  \ottnt{A}  \,  |   \, \epsilon'  \mathrel{\algefftransarrow{ \ottnt{S} } }  \ottnt{e}   \quad  \epsilon' \,  \subseteq  \, \epsilon }%
}{
 \Gamma ;  \ottnt{R}   \vdash   \ottnt{M}  :  \ottnt{A}  \,  |   \, \epsilon  \mathrel{\algefftransarrow{ \ottnt{S} } }  \ottnt{e} }{%
{\ottdrulename{Elab\_Weak}}{}%
}}

\newcommand{\ottdruleElabHXXReturn}[1]{\ottdrule[#1]{%
\ottpremise{  \Gamma  \ottsym{,}  \mathit{x} \,  \mathord{:}  \, \ottnt{A} ;  \ottnt{R}   \vdash   \ottnt{M}  :  \ottnt{B}  \,  |   \, \epsilon'  \mathrel{\algefftransarrow{  \ottnt{S}  \,\circ\, \{  \mathit{x}  \, {\mapsto} \,  \mathit{x}  \}  } }  \ottnt{e}   \quad  \epsilon \,  \subseteq  \, \epsilon' }%
}{
 \Gamma ;  \ottnt{R}   \vdash   \mathsf{return} \, \mathit{x}  \rightarrow  \ottnt{M}  :  \ottnt{A}  \,  |   \, \epsilon  \, \Rightarrow   \ottnt{B}  \,  |   \, \epsilon'  \mathrel{\algefftransarrow{ \ottnt{S} } }  \mathsf{return} \, \mathit{x}  \rightarrow  \ottnt{e} }{%
{\ottdrulename{ElabH\_Return}}{}%
}}

\newcommand{\ottdruleElabHXXOp}[1]{\ottdrule[#1]{%
\ottpremise{ \mathit{ty} \, \ottsym{(}  \mathsf{op}  \ottsym{)} \,  =  \,   \text{\unboldmath$\forall$}     \algeffseqover{ \alpha }   .  \ottnt{C}  \hookrightarrow  \ottnt{D}   \quad   \Gamma ;  \ottnt{R}   \vdash   \ottnt{H}  :  \ottnt{A}  \,  |   \, \epsilon  \, \Rightarrow   \ottnt{B}  \,  |   \, \epsilon'  \mathrel{\algefftransarrow{ \ottnt{S} } }  \ottnt{h}  }%
\ottpremise{ \Gamma  \ottsym{,}   \algeffseqover{ \alpha }   \ottsym{,}  \mathit{x} \,  \mathord{:}  \, \ottnt{C} ;  \ottsym{(}   \algeffseqover{ \alpha }   \ottsym{,}  \mathit{x} \,  \mathord{:}  \, \ottnt{C}  \ottsym{,}   \ottnt{D}   \rightarrow  \!  \epsilon'  \;  \ottnt{B}   \ottsym{)}   \vdash   \ottnt{M}  :  \ottnt{B}  \,  |   \, \epsilon'  \mathrel{\algefftransarrow{  \ottnt{S}  \,\circ\, \{  \mathit{x}  \, {\mapsto} \,  \mathit{x}  \}  } }  \ottnt{e} }%
}{
 \Gamma ;  \ottnt{R}   \vdash   \ottnt{H}  \ottsym{;}  \mathsf{op}  \ottsym{(}  \mathit{x}  \ottsym{)}  \rightarrow  \ottnt{M}  :  \ottnt{A}  \,  |   \,  \epsilon    \mathbin{\uplus}    \ottsym{\{}  \mathsf{op}  \ottsym{\}}   \, \Rightarrow   \ottnt{B}  \,  |   \, \epsilon'  \mathrel{\algefftransarrow{ \ottnt{S} } }  \ottnt{h}  \ottsym{;}   \Lambda\!  \,  \algeffseqover{ \alpha }   \ottsym{.}  \mathsf{op}  \ottsym{(}  \mathit{x}  \ottsym{)}  \rightarrow  \ottnt{e} }{%
{\ottdrulename{ElabH\_Op}}{}%
}}

\newcommand{\ottdruleElabGXXEmpty}[1]{\ottdrule[#1]{%
}{
  \emptyset   \mathrel{\algefftransarrow{ \ottnt{S} } }   \emptyset  }{%
{\ottdrulename{ElabG\_Empty}}{}%
}}

\newcommand{\ottdruleElabGXXVar}[1]{\ottdrule[#1]{%
\ottpremise{ \Gamma  \mathrel{\algefftransarrow{ \ottnt{S} } }  \Gamma' }%
}{
 \Gamma  \ottsym{,}  \mathit{x} \,  \mathord{:}  \, \sigma  \mathrel{\algefftransarrow{ \ottnt{S} } }  \Gamma'  \ottsym{,}  \ottnt{S}  \ottsym{(}  \mathit{x}  \ottsym{)} \,  \mathord{:}  \, \sigma }{%
{\ottdrulename{ElabG\_Var}}{}%
}}

\newcommand{\ottdruleElabGXXTyVar}[1]{\ottdrule[#1]{%
\ottpremise{ \Gamma  \mathrel{\algefftransarrow{ \ottnt{S} } }  \Gamma' }%
}{
 \Gamma  \ottsym{,}  \alpha  \mathrel{\algefftransarrow{ \ottnt{S} } }  \Gamma'  \ottsym{,}  \alpha }{%
{\ottdrulename{ElabG\_TyVar}}{}%
}}

\newcommand{\ottdruleEXXEval}[1]{\ottdrule[#1]{%
\ottpremise{\ottnt{e_{{\mathrm{1}}}}  \rightsquigarrow  \ottnt{e_{{\mathrm{2}}}}}%
}{
 \ottnt{E}  [  \ottnt{e_{{\mathrm{1}}}}  ]   \longrightarrow   \ottnt{E}  [  \ottnt{e_{{\mathrm{2}}}}  ] }{%
{\ottdrulename{E\_Eval}}{}%
}}

\newcommand{\ottdruleWFXXEmpty}[1]{\ottdrule[#1]{%
}{
\vdash   \emptyset }{%
{\ottdrulename{WF\_Empty}}{}%
}}

\newcommand{\ottdruleWFXXVar}[1]{\ottdrule[#1]{%
\ottpremise{ \vdash  \Gamma  \quad   \mathit{x} \,  \not\in  \,  \mathit{dom}  (  \Gamma  )   \quad  \Gamma  \vdash  \sigma  }%
}{
\vdash  \Gamma  \ottsym{,}  \mathit{x} \,  \mathord{:}  \, \sigma}{%
{\ottdrulename{WF\_Var}}{}%
}}

\newcommand{\ottdruleWFXXTyVar}[1]{\ottdrule[#1]{%
\ottpremise{ \vdash  \Gamma  \quad  \alpha \,  \not\in  \,  \mathit{dom}  (  \Gamma  )  }%
}{
\vdash  \Gamma  \ottsym{,}  \alpha}{%
{\ottdrulename{WF\_TyVar}}{}%
}}

\newcommand{\ottdruleTSXXVar}[1]{\ottdrule[#1]{%
\ottpremise{ \vdash  \Gamma  \quad   \mathit{x} \,  \mathord{:}  \,  \text{\unboldmath$\forall$}  \,  \algeffseqover{ \alpha }   \ottsym{.}  \ottnt{A} \,  \in  \, \Gamma  \quad  \Gamma  \vdash   \algeffseqover{ \ottnt{B} }   }%
}{
\Gamma  \ottsym{;}  \ottnt{R}  \vdash  \mathit{x}  \ottsym{:}   \ottnt{A}    [   \algeffseqover{ \ottnt{B} }   \ottsym{/}   \algeffseqover{ \alpha }   ]   \,  |  \, \epsilon}{%
{\ottdrulename{TS\_Var}}{}%
}}

\newcommand{\ottdruleTSXXConst}[1]{\ottdrule[#1]{%
\ottpremise{\vdash  \Gamma}%
}{
\Gamma  \ottsym{;}  \ottnt{R}  \vdash  \ottnt{c}  \ottsym{:}   \mathit{ty}  (  \ottnt{c}  )  \,  |  \, \epsilon}{%
{\ottdrulename{TS\_Const}}{}%
}}

\newcommand{\ottdruleTSXXAbs}[1]{\ottdrule[#1]{%
\ottpremise{\Gamma  \ottsym{,}  \mathit{x} \,  \mathord{:}  \, \ottnt{A}  \ottsym{;}  \ottnt{R}  \vdash  \ottnt{M}  \ottsym{:}  \ottnt{B} \,  |  \, \epsilon'}%
}{
\Gamma  \ottsym{;}  \ottnt{R}  \vdash   \lambda\!  \, \mathit{x}  \ottsym{.}  \ottnt{M}  \ottsym{:}   \ottnt{A}   \rightarrow  \!  \epsilon'  \;  \ottnt{B}  \,  |  \, \epsilon}{%
{\ottdrulename{TS\_Abs}}{}%
}}

\newcommand{\ottdruleTSXXApp}[1]{\ottdrule[#1]{%
\ottpremise{ \Gamma  \ottsym{;}  \ottnt{R}  \vdash  \ottnt{M_{{\mathrm{1}}}}  \ottsym{:}   \ottnt{A}   \rightarrow  \!  \epsilon'  \;  \ottnt{B}  \,  |  \, \epsilon  \quad   \Gamma  \ottsym{;}  \ottnt{R}  \vdash  \ottnt{M_{{\mathrm{2}}}}  \ottsym{:}  \ottnt{A} \,  |  \, \epsilon  \quad  \epsilon' \,  \subseteq  \, \epsilon  }%
}{
\Gamma  \ottsym{;}  \ottnt{R}  \vdash  \ottnt{M_{{\mathrm{1}}}} \, \ottnt{M_{{\mathrm{2}}}}  \ottsym{:}  \ottnt{B} \,  |  \, \epsilon}{%
{\ottdrulename{TS\_App}}{}%
}}

\newcommand{\ottdruleTSXXOp}[1]{\ottdrule[#1]{%
\ottpremise{ \mathit{ty} \, \ottsym{(}  \mathsf{op}  \ottsym{)} \,  =  \,   \text{\unboldmath$\forall$}     \algeffseqover{ \alpha }   .  \ottnt{A}  \hookrightarrow  \ottnt{B}   \quad   \mathsf{op} \,  \in  \, \epsilon  \quad   \Gamma  \ottsym{;}  \ottnt{R}  \vdash  \ottnt{M}  \ottsym{:}   \ottnt{A}    [   \algeffseqover{ \ottnt{C} }   \ottsym{/}   \algeffseqover{ \alpha }   ]   \,  |  \, \epsilon  \quad  \Gamma  \vdash   \algeffseqover{ \ottnt{C} }    }%
}{
\Gamma  \ottsym{;}  \ottnt{R}  \vdash   \textup{\texttt{\#}\relax}  \mathsf{op}   \ottsym{(}   \ottnt{M}   \ottsym{)}   \ottsym{:}   \ottnt{B}    [   \algeffseqover{ \ottnt{C} }   \ottsym{/}   \algeffseqover{ \alpha }   ]   \,  |  \, \epsilon}{%
{\ottdrulename{TS\_Op}}{}%
}}

\newcommand{\ottdruleTSXXHandle}[1]{\ottdrule[#1]{%
\ottpremise{ \Gamma  \ottsym{;}  \ottnt{R}  \vdash  \ottnt{M}  \ottsym{:}  \ottnt{A} \,  |  \, \epsilon  \quad  \Gamma  \ottsym{;}  \ottnt{R}  \vdash  \ottnt{H}  \ottsym{:}  \ottnt{A} \,  |  \, \epsilon  \Rightarrow  \ottnt{B} \,  |  \, \epsilon' }%
}{
\Gamma  \ottsym{;}  \ottnt{R}  \vdash  \mathsf{handle} \, \ottnt{M} \, \mathsf{with} \, \ottnt{H}  \ottsym{:}  \ottnt{B} \,  |  \, \epsilon'}{%
{\ottdrulename{TS\_Handle}}{}%
}}

\newcommand{\ottdruleTSXXLet}[1]{\ottdrule[#1]{%
\ottpremise{ \Gamma  \ottsym{,}   \algeffseqover{ \alpha }   \ottsym{;}  \ottnt{R}  \vdash  \ottnt{M_{{\mathrm{1}}}}  \ottsym{:}  \ottnt{A} \,  |  \, \epsilon  \quad  \Gamma  \ottsym{,}  \mathit{x} \,  \mathord{:}  \,  \text{\unboldmath$\forall$}  \,  \algeffseqover{ \alpha }   \ottsym{.}  \ottnt{A}  \ottsym{;}  \ottnt{R}  \vdash  \ottnt{M_{{\mathrm{2}}}}  \ottsym{:}  \ottnt{B} \,  |  \, \epsilon }%
}{
\Gamma  \ottsym{;}  \ottnt{R}  \vdash  \mathsf{let} \, \mathit{x}  \ottsym{=}  \ottnt{M_{{\mathrm{1}}}} \,  \mathsf{in}  \, \ottnt{M_{{\mathrm{2}}}}  \ottsym{:}  \ottnt{B} \,  |  \, \epsilon}{%
{\ottdrulename{TS\_Let}}{}%
}}

\newcommand{\ottdruleTSXXResume}[1]{\ottdrule[#1]{%
\ottpremise{ \vdash  \Gamma_{{\mathrm{1}}}  \ottsym{,}  \mathit{x} \,  \mathord{:}  \, \ottnt{D}  \ottsym{,}  \Gamma_{{\mathrm{2}}}  \quad    \algeffseqover{ \alpha }  \,  \in  \, \Gamma_{{\mathrm{1}}}  \quad  \epsilon \,  \subseteq  \, \epsilon'  }%
\ottpremise{\Gamma_{{\mathrm{1}}}  \ottsym{,}  \Gamma_{{\mathrm{2}}}  \ottsym{,}    \algeffseqover{ \beta }   \ottsym{,}   \mathit{x} \,  \mathord{:}  \,    \ottnt{A}    [   \algeffseqover{ \beta }   \ottsym{/}   \algeffseqover{ \alpha }   ]        \ottsym{;}  \ottsym{(}   \algeffseqover{ \alpha }   \ottsym{,}  \mathit{x} \,  \mathord{:}  \, \ottnt{A}  \ottsym{,}   \ottnt{B}   \rightarrow  \!  \epsilon  \;  \ottnt{C}   \ottsym{)}  \vdash  \ottnt{M}  \ottsym{:}   \ottnt{B}    [   \algeffseqover{ \beta }   \ottsym{/}   \algeffseqover{ \alpha }   ]   \,  |  \, \epsilon'}%
}{
\Gamma_{{\mathrm{1}}}  \ottsym{,}  \mathit{x} \,  \mathord{:}  \, \ottnt{D}  \ottsym{,}  \Gamma_{{\mathrm{2}}}  \ottsym{;}  \ottsym{(}   \algeffseqover{ \alpha }   \ottsym{,}  \mathit{x} \,  \mathord{:}  \, \ottnt{A}  \ottsym{,}   \ottnt{B}   \rightarrow  \!  \epsilon  \;  \ottnt{C}   \ottsym{)}  \vdash  \mathsf{resume} \, \ottnt{M}  \ottsym{:}  \ottnt{C} \,  |  \, \epsilon'}{%
{\ottdrulename{TS\_Resume}}{}%
}}

\newcommand{\ottdruleTSXXWeak}[1]{\ottdrule[#1]{%
\ottpremise{ \Gamma  \ottsym{;}  \ottnt{R}  \vdash  \ottnt{M}  \ottsym{:}  \ottnt{A} \,  |  \, \epsilon'  \quad  \epsilon' \,  \subseteq  \, \epsilon }%
}{
\Gamma  \ottsym{;}  \ottnt{R}  \vdash  \ottnt{M}  \ottsym{:}  \ottnt{A} \,  |  \, \epsilon}{%
{\ottdrulename{TS\_Weak}}{}%
}}

\newcommand{\ottdruleTHSXXReturn}[1]{\ottdrule[#1]{%
\ottpremise{ \Gamma  \ottsym{,}  \mathit{x} \,  \mathord{:}  \, \ottnt{A}  \ottsym{;}  \ottnt{R}  \vdash  \ottnt{M}  \ottsym{:}  \ottnt{B} \,  |  \, \epsilon'  \quad  \epsilon \,  \subseteq  \, \epsilon' }%
}{
\Gamma  \ottsym{;}  \ottnt{R}  \vdash  \mathsf{return} \, \mathit{x}  \rightarrow  \ottnt{M}  \ottsym{:}  \ottnt{A} \,  |  \, \epsilon  \Rightarrow  \ottnt{B} \,  |  \, \epsilon'}{%
{\ottdrulename{THS\_Return}}{}%
}}

\newcommand{\ottdruleTHSXXOp}[1]{\ottdrule[#1]{%
\ottpremise{\Gamma  \ottsym{;}  \ottnt{R}  \vdash  \ottnt{H}  \ottsym{:}  \ottnt{A} \,  |  \, \epsilon  \Rightarrow  \ottnt{B} \,  |  \, \epsilon'}%
\ottpremise{ \mathit{ty} \, \ottsym{(}  \mathsf{op}  \ottsym{)} \,  =  \,   \text{\unboldmath$\forall$}     \algeffseqover{ \alpha }   .  \ottnt{C}  \hookrightarrow  \ottnt{D}   \quad  \Gamma  \ottsym{,}   \algeffseqover{ \alpha }   \ottsym{,}  \mathit{x} \,  \mathord{:}  \, \ottnt{C}  \ottsym{;}  \ottsym{(}   \algeffseqover{ \alpha }   \ottsym{,}  \mathit{x} \,  \mathord{:}  \, \ottnt{C}  \ottsym{,}   \ottnt{D}   \rightarrow  \!  \epsilon'  \;  \ottnt{B}   \ottsym{)}  \vdash  \ottnt{M}  \ottsym{:}  \ottnt{B} \,  |  \, \epsilon' }%
}{
\Gamma  \ottsym{;}  \ottnt{R}  \vdash  \ottnt{H}  \ottsym{;}  \mathsf{op}  \ottsym{(}  \mathit{x}  \ottsym{)}  \rightarrow  \ottnt{M}  \ottsym{:}  \ottnt{A} \,  |  \,  \epsilon    \mathbin{\uplus}    \ottsym{\{}  \mathsf{op}  \ottsym{\}}   \Rightarrow  \ottnt{B} \,  |  \, \epsilon'}{%
{\ottdrulename{THS\_Op}}{}%
}}

\newcommand{\ottdruleTXXVar}[1]{\ottdrule[#1]{%
\ottpremise{ \vdash  \Gamma  \quad   \mathit{x} \,  \mathord{:}  \,  \text{\unboldmath$\forall$}  \,  \algeffseqover{ \alpha }   \ottsym{.}  \ottnt{A} \,  \in  \, \Gamma  \quad  \Gamma  \vdash   \algeffseqover{ \ottnt{B} }   }%
}{
\Gamma  \ottsym{;}  \ottnt{r} \,   \vdash  \mathit{x} \,  \algeffseqover{ \ottnt{B} }  \,   \ottsym{:}   \ottnt{A}    [   \algeffseqover{ \ottnt{B} }   \ottsym{/}   \algeffseqover{ \alpha }   ]   \,  |  \, \epsilon}{%
{\ottdrulename{T\_Var}}{}%
}}

\newcommand{\ottdruleTXXConst}[1]{\ottdrule[#1]{%
\ottpremise{\vdash  \Gamma}%
}{
\Gamma  \ottsym{;}  \ottnt{r} \,   \vdash  \ottnt{c} \,   \ottsym{:}   \mathit{ty}  (  \ottnt{c}  )  \,  |  \, \epsilon}{%
{\ottdrulename{T\_Const}}{}%
}}

\newcommand{\ottdruleTXXAbs}[1]{\ottdrule[#1]{%
\ottpremise{\Gamma  \ottsym{,}  \mathit{x} \,  \mathord{:}  \, \ottnt{A}  \ottsym{;}  \ottnt{r} \,   \vdash  \ottnt{e} \,   \ottsym{:}  \ottnt{B} \,  |  \, \epsilon'}%
}{
\Gamma  \ottsym{;}  \ottnt{r} \,   \vdash   \lambda\!  \, \mathit{x}  \ottsym{.}  \ottnt{e} \,   \ottsym{:}   \ottnt{A}   \rightarrow  \!  \epsilon'  \;  \ottnt{B}  \,  |  \, \epsilon}{%
{\ottdrulename{T\_Abs}}{}%
}}

\newcommand{\ottdruleTXXApp}[1]{\ottdrule[#1]{%
\ottpremise{ \Gamma  \ottsym{;}  \ottnt{r} \,   \vdash  \ottnt{e_{{\mathrm{1}}}} \,   \ottsym{:}   \ottnt{A}   \rightarrow  \!  \epsilon'  \;  \ottnt{B}  \,  |  \, \epsilon  \quad   \Gamma  \ottsym{;}  \ottnt{r} \,   \vdash  \ottnt{e_{{\mathrm{2}}}} \,   \ottsym{:}  \ottnt{A} \,  |  \, \epsilon  \quad  \epsilon' \,  \subseteq  \, \epsilon  }%
}{
\Gamma  \ottsym{;}  \ottnt{r} \,   \vdash  \ottnt{e_{{\mathrm{1}}}} \, \ottnt{e_{{\mathrm{2}}}} \,   \ottsym{:}  \ottnt{B} \,  |  \, \epsilon}{%
{\ottdrulename{T\_App}}{}%
}}

\newcommand{\ottdruleTXXOp}[1]{\ottdrule[#1]{%
\ottpremise{ \mathit{ty} \, \ottsym{(}  \mathsf{op}  \ottsym{)} \,  =  \,   \text{\unboldmath$\forall$}     \algeffseqover{ \alpha }   .  \ottnt{A}  \hookrightarrow  \ottnt{B}   \quad   \mathsf{op} \,  \in  \, \epsilon  \quad   \Gamma  \ottsym{;}  \ottnt{r} \,   \vdash  \ottnt{e} \,   \ottsym{:}   \ottnt{A}    [   \algeffseqover{ \ottnt{C} }   \ottsym{/}   \algeffseqover{ \alpha }   ]   \,  |  \, \epsilon  \quad  \Gamma  \vdash   \algeffseqover{ \ottnt{C} }    }%
}{
\Gamma  \ottsym{;}  \ottnt{r} \,   \vdash   \textup{\texttt{\#}\relax}  \mathsf{op}   \ottsym{(}    \algeffseqover{ \ottnt{C} }    \ottsym{,}   \ottnt{e}   \ottsym{)}  \,   \ottsym{:}   \ottnt{B}    [   \algeffseqover{ \ottnt{C} }   \ottsym{/}   \algeffseqover{ \alpha }   ]   \,  |  \, \epsilon}{%
{\ottdrulename{T\_Op}}{}%
}}

\newcommand{\ottdruleTXXOpCont}[1]{\ottdrule[#1]{%
\ottpremise{ \mathit{ty} \, \ottsym{(}  \mathsf{op}  \ottsym{)} \,  =  \,   \text{\unboldmath$\forall$}     \algeffseqoverindex{ \alpha }{ \text{\unboldmath$\mathit{I}$} }   .  \ottnt{A}  \hookrightarrow  \ottnt{B}   \quad   \mathsf{op} \,  \in  \, \epsilon  \quad  \Gamma  \vdash   \algeffseqoverindex{  \text{\unboldmath$\forall$}  \,  \algeffseqoverindex{ \beta }{ \text{\unboldmath$\mathit{J}$} }   \ottsym{.}  \ottnt{C} }{ \text{\unboldmath$\mathit{I}$} }   }%
\ottpremise{ \Gamma  \ottsym{,}   \algeffseqoverindex{ \beta }{ \text{\unboldmath$\mathit{J}$} }   \ottsym{;}  \ottnt{r} \,   \vdash  \ottnt{v} \,   \ottsym{:}   \ottnt{A}    [   \algeffseqoverindex{ \ottnt{C} }{ \text{\unboldmath$\mathit{I}$} }   \ottsym{/}   \algeffseqoverindex{ \alpha }{ \text{\unboldmath$\mathit{I}$} }   ]   \,  |  \, \epsilon  \quad   \Gamma   \vdash    \ottnt{E} ^{  \algeffseqoverindex{ \beta }{ \text{\unboldmath$\mathit{J}$} }  }    \ottsym{:}     \text{\unboldmath$\forall$}  \,  \algeffseqoverindex{ \beta }{ \text{\unboldmath$\mathit{J}$} }   \ottsym{.}  \ottsym{(}   \ottnt{B}    [   \algeffseqoverindex{ \ottnt{C} }{ \text{\unboldmath$\mathit{I}$} }   \ottsym{/}   \algeffseqoverindex{ \alpha }{ \text{\unboldmath$\mathit{I}$} }   ]    \ottsym{)}  \multimap  \ottnt{D}   \,  |  \, \epsilon  }%
}{
\Gamma  \ottsym{;}  \ottnt{r} \,   \vdash   \textup{\texttt{\#}\relax}  \mathsf{op}   \ottsym{(}    \algeffseqoverindex{  \text{\unboldmath$\forall$}  \,  \algeffseqoverindex{ \beta }{ \text{\unboldmath$\mathit{J}$} }   \ottsym{.}  \ottnt{C} }{ \text{\unboldmath$\mathit{I}$} }    \ottsym{,}    \Lambda\!  \,  \algeffseqoverindex{ \beta }{ \text{\unboldmath$\mathit{J}$} }   \ottsym{.}  \ottnt{v}   \ottsym{,}    \ottnt{E} ^{  \algeffseqoverindex{ \beta }{ \text{\unboldmath$\mathit{J}$} }  }    \ottsym{)}  \,   \ottsym{:}  \ottnt{D} \,  |  \, \epsilon}{%
{\ottdrulename{T\_OpCont}}{}%
}}

\newcommand{\ottdruleTXXHandle}[1]{\ottdrule[#1]{%
\ottpremise{ \Gamma  \ottsym{;}  \ottnt{r} \,   \vdash  \ottnt{e} \,   \ottsym{:}  \ottnt{A} \,  |  \, \epsilon  \quad  \Gamma  \ottsym{;}  \ottnt{r} \,   \vdash  \ottnt{h}  \ottsym{:}  \ottnt{A} \,  |  \, \epsilon  \Rightarrow  \ottnt{B} \,  |  \, \epsilon' }%
}{
\Gamma  \ottsym{;}  \ottnt{r} \,   \vdash  \mathsf{handle} \, \ottnt{e} \, \mathsf{with} \, \ottnt{h} \,   \ottsym{:}  \ottnt{B} \,  |  \, \epsilon'}{%
{\ottdrulename{T\_Handle}}{}%
}}

\newcommand{\ottdruleTXXLet}[1]{\ottdrule[#1]{%
\ottpremise{ \Gamma  \ottsym{,}   \algeffseqover{ \alpha }   \ottsym{;}  \ottnt{r} \,   \vdash  \ottnt{e_{{\mathrm{1}}}} \,   \ottsym{:}  \ottnt{A} \,  |  \, \epsilon  \quad  \Gamma  \ottsym{,}  \mathit{x} \,  \mathord{:}  \,  \text{\unboldmath$\forall$}  \,  \algeffseqover{ \alpha }   \ottsym{.}  \ottnt{A}  \ottsym{;}  \ottnt{r} \,   \vdash  \ottnt{e_{{\mathrm{2}}}} \,   \ottsym{:}  \ottnt{B} \,  |  \, \epsilon }%
}{
\Gamma  \ottsym{;}  \ottnt{r} \,   \vdash  \mathsf{let} \, \mathit{x}  \ottsym{=}   \Lambda\!  \,  \algeffseqover{ \alpha }   \ottsym{.}  \ottnt{e_{{\mathrm{1}}}} \,  \mathsf{in}  \, \ottnt{e_{{\mathrm{2}}}} \,   \ottsym{:}  \ottnt{B} \,  |  \, \epsilon}{%
{\ottdrulename{T\_Let}}{}%
}}

\newcommand{\ottdruleTXXResume}[1]{\ottdrule[#1]{%
\ottpremise{  \algeffseqover{ \alpha }  \,  \in  \, \Gamma  \quad   \Gamma  \ottsym{,}    \algeffseqover{ \beta }   \ottsym{,}   \mathit{x} \,  \mathord{:}  \,    \ottnt{A}    [   \algeffseqover{ \beta }   \ottsym{/}   \algeffseqover{ \alpha }   ]        \ottsym{;}  \ottsym{(}   \algeffseqover{ \alpha }   \ottsym{,}  \ottnt{A}  \ottsym{,}   \ottnt{B}   \rightarrow  \!  \epsilon  \;  \ottnt{C}   \ottsym{)} \,   \vdash  \ottnt{e} \,   \ottsym{:}   \ottnt{B}    [   \algeffseqover{ \beta }   \ottsym{/}   \algeffseqover{ \alpha }   ]   \,  |  \, \epsilon'  \quad  \epsilon \,  \subseteq  \, \epsilon'  }%
}{
\Gamma  \ottsym{;}  \ottsym{(}   \algeffseqover{ \alpha }   \ottsym{,}  \ottnt{A}  \ottsym{,}   \ottnt{B}   \rightarrow  \!  \epsilon  \;  \ottnt{C}   \ottsym{)} \,   \vdash  \mathsf{resume} \,  \algeffseqover{ \beta }  \, \mathit{x}  \ottsym{.}  \ottnt{e} \,   \ottsym{:}  \ottnt{C} \,  |  \, \epsilon'}{%
{\ottdrulename{T\_Resume}}{}%
}}

\newcommand{\ottdruleTXXWeak}[1]{\ottdrule[#1]{%
\ottpremise{ \Gamma  \ottsym{;}  \ottnt{r} \,   \vdash  \ottnt{e} \,   \ottsym{:}  \ottnt{A} \,  |  \, \epsilon'  \quad  \epsilon' \,  \subseteq  \, \epsilon }%
}{
\Gamma  \ottsym{;}  \ottnt{r} \,   \vdash  \ottnt{e} \,   \ottsym{:}  \ottnt{A} \,  |  \, \epsilon}{%
{\ottdrulename{T\_Weak}}{}%
}}

\newcommand{\ottdruleTEXXHole}[1]{\ottdrule[#1]{%
}{
 \Gamma   \vdash    [\,]    \ottsym{:}    \ottnt{A}  \multimap  \ottnt{A}   \,  |  \, \epsilon }{%
{\ottdrulename{TE\_Hole}}{}%
}}

\newcommand{\ottdruleTEXXAppOne}[1]{\ottdrule[#1]{%
\ottpremise{  \Gamma   \vdash   \ottnt{E}   \ottsym{:}    \sigma  \multimap  \ottsym{(}   \ottnt{A}   \rightarrow  \!  \epsilon'  \;  \ottnt{B}   \ottsym{)}   \,  |  \, \epsilon   \quad   \Gamma  \ottsym{;}   \mathsf{none}  \,   \vdash  \ottnt{e_{{\mathrm{2}}}} \,   \ottsym{:}  \ottnt{A} \,  |  \, \epsilon  \quad  \epsilon' \,  \subseteq  \, \epsilon  }%
}{
 \Gamma   \vdash   \ottnt{E} \, \ottnt{e_{{\mathrm{2}}}}   \ottsym{:}    \sigma  \multimap  \ottnt{B}   \,  |  \, \epsilon }{%
{\ottdrulename{TE\_App1}}{}%
}}

\newcommand{\ottdruleTEXXAppTwo}[1]{\ottdrule[#1]{%
\ottpremise{ \Gamma  \ottsym{;}   \mathsf{none}  \,   \vdash  \ottnt{v_{{\mathrm{1}}}} \,   \ottsym{:}  \ottsym{(}   \ottnt{A}   \rightarrow  \!  \epsilon'  \;  \ottnt{B}   \ottsym{)} \,  |  \, \epsilon  \quad    \Gamma   \vdash   \ottnt{E}   \ottsym{:}    \sigma  \multimap  \ottnt{A}   \,  |  \, \epsilon   \quad  \epsilon' \,  \subseteq  \, \epsilon  }%
}{
 \Gamma   \vdash   \ottnt{v_{{\mathrm{1}}}} \, \ottnt{E}   \ottsym{:}    \sigma  \multimap  \ottnt{B}   \,  |  \, \epsilon }{%
{\ottdrulename{TE\_App2}}{}%
}}

\newcommand{\ottdruleTEXXOp}[1]{\ottdrule[#1]{%
\ottpremise{ \mathit{ty} \, \ottsym{(}  \mathsf{op}  \ottsym{)} \,  =  \,   \text{\unboldmath$\forall$}     \algeffseqover{ \alpha }   .  \ottnt{A}  \hookrightarrow  \ottnt{B}   \quad   \mathsf{op} \,  \in  \, \epsilon  \quad    \Gamma   \vdash   \ottnt{E}   \ottsym{:}    \sigma  \multimap   \ottnt{A}    [   \algeffseqover{ \ottnt{C} }   \ottsym{/}   \algeffseqover{ \alpha }   ]     \,  |  \, \epsilon   \quad  \Gamma  \vdash   \algeffseqover{ \ottnt{C} }    }%
}{
 \Gamma   \vdash    \textup{\texttt{\#}\relax}  \mathsf{op}   \ottsym{(}    \algeffseqover{ \ottnt{C} }    \ottsym{,}   \ottnt{E}   \ottsym{)}    \ottsym{:}    \sigma  \multimap   \ottnt{B}    [   \algeffseqover{ \ottnt{C} }   \ottsym{/}   \algeffseqover{ \alpha }   ]     \,  |  \, \epsilon }{%
{\ottdrulename{TE\_Op}}{}%
}}

\newcommand{\ottdruleTEXXHandle}[1]{\ottdrule[#1]{%
\ottpremise{  \Gamma   \vdash   \ottnt{E}   \ottsym{:}    \sigma  \multimap  \ottnt{A}   \,  |  \, \epsilon   \quad  \Gamma  \ottsym{;}   \mathsf{none}  \,   \vdash  \ottnt{h}  \ottsym{:}  \ottnt{A} \,  |  \, \epsilon  \Rightarrow  \ottnt{B} \,  |  \, \epsilon' }%
}{
 \Gamma   \vdash   \mathsf{handle} \, \ottnt{E} \, \mathsf{with} \, \ottnt{h}   \ottsym{:}    \sigma  \multimap  \ottnt{B}   \,  |  \, \epsilon' }{%
{\ottdrulename{TE\_Handle}}{}%
}}

\newcommand{\ottdruleTEXXLet}[1]{\ottdrule[#1]{%
\ottpremise{  \Gamma  \ottsym{,}   \algeffseqover{ \alpha }    \vdash   \ottnt{E}   \ottsym{:}    \sigma  \multimap  \ottnt{A}   \,  |  \, \epsilon   \quad  \Gamma  \ottsym{,}  \mathit{x} \,  \mathord{:}  \,  \text{\unboldmath$\forall$}  \,  \algeffseqover{ \alpha }   \ottsym{.}  \ottnt{A}  \ottsym{;}   \mathsf{none}  \,   \vdash  \ottnt{e} \,   \ottsym{:}  \ottnt{B} \,  |  \, \epsilon }%
}{
 \Gamma   \vdash   \mathsf{let} \, \mathit{x}  \ottsym{=}   \Lambda\!  \,  \algeffseqover{ \alpha }   \ottsym{.}  \ottnt{E} \,  \mathsf{in}  \, \ottnt{e}   \ottsym{:}     \text{\unboldmath$\forall$}  \,  \algeffseqover{ \alpha }   \ottsym{.}  \sigma  \multimap  \ottnt{B}   \,  |  \, \epsilon }{%
{\ottdrulename{TE\_Let}}{}%
}}

\newcommand{\ottdruleTEXXWeak}[1]{\ottdrule[#1]{%
\ottpremise{  \Gamma   \vdash   \ottnt{E}   \ottsym{:}    \sigma  \multimap  \ottnt{A}   \,  |  \, \epsilon'   \quad  \epsilon' \,  \subseteq  \, \epsilon }%
}{
 \Gamma   \vdash   \ottnt{E}   \ottsym{:}    \sigma  \multimap  \ottnt{A}   \,  |  \, \epsilon }{%
{\ottdrulename{TE\_Weak}}{}%
}}

\newcommand{\ottdruleTHXXReturn}[1]{\ottdrule[#1]{%
\ottpremise{ \Gamma  \ottsym{,}  \mathit{x} \,  \mathord{:}  \, \ottnt{A}  \ottsym{;}  \ottnt{r} \,   \vdash  \ottnt{e} \,   \ottsym{:}  \ottnt{B} \,  |  \, \epsilon'  \quad  \epsilon \,  \subseteq  \, \epsilon' }%
}{
\Gamma  \ottsym{;}  \ottnt{r} \,   \vdash  \mathsf{return} \, \mathit{x}  \rightarrow  \ottnt{e}  \ottsym{:}  \ottnt{A} \,  |  \, \epsilon  \Rightarrow  \ottnt{B} \,  |  \, \epsilon'}{%
{\ottdrulename{TH\_Return}}{}%
}}

\newcommand{\ottdruleTHXXOp}[1]{\ottdrule[#1]{%
\ottpremise{\Gamma  \ottsym{;}  \ottnt{r} \,   \vdash  \ottnt{h}  \ottsym{:}  \ottnt{A} \,  |  \, \epsilon  \Rightarrow  \ottnt{B} \,  |  \, \epsilon'}%
\ottpremise{ \mathit{ty} \, \ottsym{(}  \mathsf{op}  \ottsym{)} \,  =  \,   \text{\unboldmath$\forall$}     \algeffseqover{ \alpha }   .  \ottnt{C}  \hookrightarrow  \ottnt{D}   \quad  \Gamma  \ottsym{,}   \algeffseqover{ \alpha }   \ottsym{,}  \mathit{x} \,  \mathord{:}  \, \ottnt{C}  \ottsym{;}  \ottsym{(}   \algeffseqover{ \alpha }   \ottsym{,}  \ottnt{C}  \ottsym{,}   \ottnt{D}   \rightarrow  \!  \epsilon'  \;  \ottnt{B}   \ottsym{)} \,   \vdash  \ottnt{e} \,   \ottsym{:}  \ottnt{B} \,  |  \, \epsilon' }%
}{
\Gamma  \ottsym{;}  \ottnt{r} \,   \vdash  \ottnt{h}  \ottsym{;}   \Lambda\!  \,  \algeffseqover{ \alpha }   \ottsym{.}  \mathsf{op}  \ottsym{(}  \mathit{x}  \ottsym{)}  \rightarrow  \ottnt{e}  \ottsym{:}  \ottnt{A} \,  |  \,  \epsilon    \mathbin{\uplus}    \ottsym{\{}  \mathsf{op}  \ottsym{\}}   \Rightarrow  \ottnt{B} \,  |  \, \epsilon'}{%
{\ottdrulename{TH\_Op}}{}%
}}

\renewcommand{\ottdrule}[4][]{{\displaystyle\frac{\begin{array}{c}#2\end{array}}{#3}\ \ottdrulename{#4}}}

\usepackage{amsmath}
\usepackage{multirow}
\usepackage{color}
\usepackage{thmtools}
\usepackage{thm-restate}

\usepackage{listings}
\usepackage{enumitem}
\newenvironment{caseanalysis}
{\begin{description}[leftmargin=1.2em]}
{\end{description}}
\def\case#1:{\item[\textmd{Case {#1}:}]}

\newcommand{\reffig}[1]{Figure~\ref{fig:#1}}
\newcommand{\refsec}[1]{Section~\ref{sec:#1}}

\newcommand{\surfacelang}{$\lambda_\text{eff}^\text{let}$}
\newcommand{\interlang}{$\lambda_\text{eff}^{\Lambda}$}

\newcommand{\failp}{\texttt{fail}$^\forall$}
\newcommand{\divp}{\texttt{div100}$^\forall$}
\newcommand{\choosep}{\texttt{choose}$^\forall$}

\definecolor{gray96}{gray}{.9}
\newcommand{\graybox}[1]{\!\colorbox{gray96}{\(\!#1\!\)}\!}

\iflncs
\spnewtheorem{defn}{Definition}{\bfseries}{\itshape}
\spnewtheorem{assum}{Assumption}{\bfseries}{\rmfamily}
\spnewtheorem{conv}{Convention}{\bfseries}{\rmfamily}

\makeatletter
\spn@wtheorem{lemm}{Lemma}{\bfseries}{\itshape}
\renewenvironment{lemma}[1]
{\begin{lemm} \label{lem:#1} \noindent}
{\end{lemm}}
\newenvironment{lemmap}[2]
{\begin{lemm}[#1] \label{lem:#2} \noindent}
{\end{lemm}}
\newcommand{\reflem}[1]{Lemma~\ref{lem:#1}}
\makeatother

\spnewtheorem{thm}{Theorem}{\bfseries}{\itshape}
\renewenvironment{theorem}[2][]
{\begin{thm}[#1] \label{thm:#2} \noindent}
 {\end{thm}}
 
\else
\newtheorem{defn}{Definition}
\newtheorem{assum}{Assumption}
\newtheorem{conv}{Convention}

\newtheorem{lemm}{Lemma}
\newenvironment{lemmap}[2]
{\begin{lemm}[#1] \label{lem:#2} \noindent}
{\end{lemm}}
\newenvironment{lemma}[1]
{\begin{lemm} \label{lem:#1} \noindent}
{\end{lemm}}
\newcommand{\reflem}[1]{Lemma~\ref{lem:#1}}

\newtheorem{thm}{Theorem}
\newenvironment{theorem}[2][]
{\begin{thm}[#1] \label{thm:#2} \noindent}
 {\end{thm}}
 
\fi

\newcommand{\Rule}[2]{\ensuremath{\text{({\sc{{#1}\_{#2}}})}}}
\newcommand{\WF}[1]{\Rule{WF}{#1}}
\newcommand{\TSrule}[1]{\Rule{TS}{#1}}
\newcommand{\THSrule}[1]{\Rule{THS}{#1}}
\newcommand{\R}[1]{\Rule{R}{#1}}
\newcommand{\E}[1]{\Rule{E}{#1}}
\newcommand{\T}[1]{\Rule{T}{#1}}
\newcommand{\THrule}[1]{\Rule{TH}{#1}}
\newcommand{\TE}[1]{\Rule{TE}{#1}}
\newcommand{\Elab}[1]{\Rule{Elab}{#1}}

\newcommand{\ElabG}[1]{\Rule{ElabG}{#1}}

\newcommand{\defeq}{\stackrel{\rm \tiny def}{=}}

\begin{document}
\title{Handling polymorphic algebraic effects}
\author{
Taro Sekiyama\inst{1}\iforcid\orcidID{0000-0001-9286-230X}\fi
\and
Atsushi Igarashi\inst{2}\iforcid\orcidID{0000-0002-5143-9764}\fi
}
\authorrunning{T.\ Sekiyama and A.\ Igarashi}
\institute{
National Institute of Informatics, Tokyo, Japan
\email{tsekiyama@acm.org}
\\ \and
Kyoto University, Kyoto, Japan
\email{igarashi@kuis.kyoto-u.ac.jp}
}
\maketitle              % typeset the header of the contribution
\begin{abstract}
 Algebraic effects and handlers are a powerful abstraction mechanism to
 represent and implement control effects.
 In this work, we study their extension with parametric polymorphism that allows
 abstracting not only expressions but also effects and handlers.
 Although polymorphism makes it possible to reuse and reason about effect
 implementations more effectively, it has long been known that a naive combination
 of polymorphic effects and let-polymorphism breaks type safety.
 Although type safety can often be gained by restricting let-bound
 expressions---e.g., by adopting value restriction or weak polymorphism---we
 propose a complementary approach that restricts handlers instead of let-bound
 expressions.
 Our key observation is that, informally speaking, a handler is safe
 if resumptions from the handler do not interfere with each other.
 To formalize our idea, we define a call-by-value lambda calculus
 {\surfacelang} that supports let-polymorphism and polymorphic
 algebraic effects and handlers, design a type system that rejects
 interfering handlers, and prove type safety of our calculus.
\end{abstract}

\section{Introduction}

Algebraic effects~\cite{Plotkin/Power_2003_ACS} and handlers~\cite{Plotkin/Pretnar_2009_ESOP} are a powerful abstraction
mechanism to represent and implement control effects, such as
exceptions, interactive I/O, mutable states, and nondeterminism.
They are growing in popularity, thanks to their success in achieving modularity of
effects, especially the clear separation between their interfaces and
their implementations.
An interface of effects is given as a set of \emph{operations}---e.g., an
interface of mutable states consists of two operations, namely, \textsf{put} and
\textsf{get}---with their signatures.
An implementation is given by a \emph{handler} $\ottnt{H}$, which provides
a set of interpretations of the operations (called \emph{operation
  clauses}), and a $ \mathsf{handle} $--$ \mathsf{with} $ expression
$\mathsf{handle} \, \ottnt{M} \, \mathsf{with} \, \ottnt{H}$ associates effects invoked during the computation of
$\ottnt{M}$ with handler $\ottnt{H}$.
Algebraic effects and handlers work as \emph{resumable exceptions}: when an
effect operation is invoked, the run-time system tries to find the nearest
handler that handles the invoked operation; if it is found, the corresponding
operation clause is evaluated by using the argument to the operation invocation
and the continuation up to the handler.
The continuation gives the ability to resume the computation from the
point where the operation was invoked, using the result from the
operation clause.
Another modularity that algebraic effects provide is flexible
composition: multiple algebraic effects can be combined
freely~\cite{Kammar/Lindley/Oury_2013_ICFP}.

In this work, we study an extension of algebraic effects and handlers with
another type-based abstraction mechanism---parametric polymorphism~\cite{Reynolds_IFIP_1983}.
In general, parametric polymorphism is a basis of generic programming and
enhance code reusability by abstracting expressions over types.
This work allows abstracting not only expressions but also effect operations
and handlers, which makes it possible to reuse and reason about effect
implementations that are independent of concrete type representations.
Like in many functional languages, we introduce polymorphism in the form of
\emph{let-polymorphism} for its practically desirable properties such as
decidable typechecking and type inference.
%
% \begin{figure}[t]
%  \lstinputlisting[numbers=left]{source/example.ml}
%  \caption{Example of polymorphic algebraic effects and handlers.}
%  \label{fig:example}
% \end{figure}
% %
% \reffig{example} shows an example in ML-like language where function \texttt{f}
% invokes effect \texttt{choose} and a handler gives its implementation.
% %
% Line 1 declares \texttt{choose} with its signature $\forall \alpha.\, \alpha
% \times \alpha \rightarrow \alpha$, where polymorphism ensures that
% \texttt{choose} takes two arguments and returns one of them (if control is
% returned).
% %
% Lines 3--6 defines function \texttt{f} that invokes \texttt{choose} twice.
% %
% If the first invocation returns the first argument \texttt{true}, \texttt{f}
% returns the result of the second invocation; otherwise, if \texttt{false} is
% chosen, it returns zero.

As is well known, however, a naive combination of polymorphic effects and
let-polymorphism breaks type safety~\cite{Tofte_1990_IC,Harper/Lillibridge_1993_LSC}.
Many researchers have attacked this classical problem~\cite{Tofte_1990_IC,Leroy/Weis_1991_POPL,Appel/MacQueen_1991_PLILP,Hoang/Mitchell/Viswanathan_1993_LICS,Wright_1995_LSC,Garrigue_2004_FLOPS,Asai/Kameyama_2007_APLAS,Kammar/Pretnar_2017_JFP}, and their common
idea is to restrict the form of let-bound expressions.
For example, value restriction~\cite{Tofte_1990_IC,Wright_1995_LSC}, which is the standard way to make
ML-like languages with imperative features and let-polymorphism type safe,
allows only syntactic values to be polymorphic.
%
% The restrictions in the literature however would prevent generalization of
% expressions with effects whose implementations have no conflicts with
% let-polymorphism.

In this work, we propose a new approach to achieving type safety in a
language with let-polymorphic and polymorphic effects and handlers:
the idea is to restrict handlers instead of let-bound expressions.
Since a handler gives an implementation of an effect, our work can be viewed as
giving a criterion that suggests what effects can cooperate safely with
(unrestricted) let-polymorphism and what effects cannot.
Our key observation for type safety is that, informally speaking, an
invocation of a polymorphic effect in a let-bound expression is safe
if resumptions in the corresponding operation clause do not interfere
with each other.
We formalize this discipline into a type system and show that typeable programs
do not get stuck.

Our contributions are summarized as follows.
\begin{itemize}
 \item We introduce a call-by-value, statically typed lambda calculus
       {\surfacelang} that supports let-polymorphism and polymorphic algebraic
       effects and handlers.
       The type system of {\surfacelang} allows any let-bound expressions involving
       effects to be polymorphic, but, instead, disallows handlers where resumptions interfere with each other.
       %
       % As far as we know, this is the first work to deal with polymorphic effect
       % operations and handlers formally, while they are implemented in
       % Koka~\cite{Leijen_2017_POPL} with the value restriction.

 \item To give the semantics of {\surfacelang}, we formalize an intermediate
       language {\interlang} wherein type information is made explicit and define
       a formal elaboration from {\surfacelang} to {\interlang}.

 \item We prove type safety of {\surfacelang} by type preservation of the
       elaboration and type soundness of {\interlang}.
\end{itemize}
We believe that our approach is complementary to the usual approach of
restricting let-bound expressions: for handlers that are considered
unsafe by our criterion, the value restriction can still be used.

The rest of this paper is organized as follows.
\refsec{overview} provides an overview of our work, giving motivating examples
of polymorphic effects and handlers, a problem in naive combination of
polymorphic effects and let-polymorphism, and our solution to gain type safety
with those features.
\refsec{surfacelang} defines the surface language {\surfacelang}, and \refsec{interlang}
defines the intermediate language {\interlang} and the elaboration from {\surfacelang} to
{\interlang}.
%
% The semantics of {\surfacelang} is given by elaboration to intermediate langauge
% {\interlang}, which is defined in \refsec{intermediate}.
%
We also state that the elaboration is type-preserving and that {\interlang}
is type sound in \refsec{interlang}.
Finally, we discuss related work in \refsec{relwork} and conclude in
\refsec{conclusion}.
The proofs of the stated properties and the full definition of the elaboration
are given in
\iffull
Appendix.
\else
the full version at \url{https://arxiv.org/abs/1811.07332}.
\fi

\section{Overview}
\label{sec:overview}

We start with reviewing how monomorphic algebraic effects and handlers work
through examples and then extend them to a polymorphic version.
We also explain why polymorphic effects are inconsistent with let-polymorphism,
if naively combined, and how we resolve it.

\subsection{Monomorphic algebraic effects and handlers}
\label{sec:overview-mono}
\paragraph{Exception.}
Our first example is exception handling, shown in an ML-like language below.
\begin{lstlisting}[numbers=left,xleftmargin=1.3\parindent]
effect fail : unit (/$\hookrightarrow$/) unit

let div100 (x:int) : int =
  if x = 0 then (#fail(); -1)
  else 100 / x

let f (y:int) : int option =
  handle (div_100 y) with
    return z (/$\rightarrow$/) Some z
    fail    z (/$\rightarrow$/) None
\end{lstlisting}
\texttt{Some} and \texttt{None} are constructors of datatype
$\alpha\;$\texttt{option}.
Line 1 declares an effect operation \texttt{fail}, which signals that an
anomaly happens, with its signature
$\texttt{unit} \hookrightarrow \texttt{unit}$, which means that the operation is invoked
with the unit value \texttt{()}, causes some effect, and may return the unit value.
The function \texttt{div100}, defined in Lines 3--5, is an example that
uses \texttt{fail}; it returns the number obtained by dividing 100 by
argument \texttt{x} if \texttt{x} is not zero; otherwise, if
\texttt{x} is zero, it raises an exception by calling effect operation
\texttt{fail}.\footnote{%
  Here, ``\texttt{; -1}'' is necessary to make the types of both branches the
same; it becomes unnecessary when we introduce polymorphic effects.}
In general, we write \texttt{\#op($\ottnt{M}$)} for
invoking effect operation \texttt{op} with argument $\ottnt{M}$.  
The function \texttt{f} (Lines 7--10) calls \texttt{div\_100} inside a
\texttt{handle}--\texttt{with} expression, which returns \texttt{Some}
$n$ if \texttt{div\_100} returns integer $n$ normally and returns
\texttt{None} if it invokes \texttt{fail}.

An expression of the form \texttt{handle} $\ottnt{M}$ \texttt{with} $\ottnt{H}$ handles effect
operations invoked in $\ottnt{M}$ (which we call \emph{handled expression}) according to the effect interpretations
given by handler $\ottnt{H}$.
A handler $\ottnt{H}$ consists of two parts: a single \emph{return clause} and zero or more
\emph{operation clauses}.
A return clause \texttt{return x} $\rightarrow$ $\ottnt{M'}$ will be executed
if the evaluation of $\ottnt{M}$ results in a value $\ottnt{v}$.  Then, the value
of $\ottnt{M'}$ (where \texttt{x} is bound to $\ottnt{v}$) will be the value
of the entire \texttt{handle}--\texttt{with} expression.
For example, in the program above, if a nonzero number $n$ is passed
to \texttt{f}, the \texttt{handle}--\texttt{with} expression would
return \texttt{Some $(100/n)$} because \texttt{div100 $n$} returns
$100 / n$.
An operation clause $\mathsf{op}$ \texttt{x} $\rightarrow$ $\ottnt{M'}$ defines an
implementation of effect \texttt{op}: if the evaluation of handled expression
$\ottnt{M}$ invokes effect \texttt{op} with argument $\ottnt{v}$, expression $\ottnt{M'}$
will be evaluated after substituting $\ottnt{v}$ for \texttt{x} and the value of
$\ottnt{M'}$ will be the value of the entire \texttt{handle}--\texttt{with} expression.
In the program example above, if zero is given to \texttt{f}, then
\texttt{None} will be returned because \texttt{div100 0} invokes
\texttt{fail}.

As shown above, algebraic effect handling is similar to exception handling.
However, a distinctive feature of algebraic effect handling is that it allows
\emph{resumption} of the computation from the point where
an effect operation was invoked.
The next example demonstrates such an ability of algebraic effect handlers.

\paragraph{Choice.}
The next example is effect \texttt{choose}, which returns one of the
given two arguments.
\begin{lstlisting}[numbers=left,xleftmargin=1.3\parindent]
effect choose : int (/$\times$/) int (/$\hookrightarrow$/) int

handle (#choose(1,2) + #choose(10,20)) with
    return x (/$\rightarrow$/) x
    choose x (/$\rightarrow$/) resume (fst x)
\end{lstlisting}
As usual, $\ottnt{A_{{\mathrm{1}}}} \times \ottnt{A_{{\mathrm{2}}}}$ is a product type, \texttt{($\ottnt{M_{{\mathrm{1}}}}$,$\ottnt{M_{{\mathrm{2}}}}$)} is
a pair expression, and \texttt{fst} is the first projection function.
The first line declares that effect \texttt{choose} is for choosing integers.
The handled expression \texttt{\#choose(1,2) +
\#choose(10,20)} intuitively suggests that there would be four possible results---11, 21, 12, and 22---depending on which value each invocation of
\texttt{choose} returns.
The handler in this example always chooses the first element of a
given pair\footnote{%
  We can think of more practical implementations, which choose one of
  the two arguments by other means, say, random values.} %
and returns it by using a \texttt{resume} expression, and, as a result, the
expression in Lines 3--5 evaluates to 11.

A resumption expression \texttt{resume} $\ottnt{M}$ in an operation clause makes it
possible to return a value of $\ottnt{M}$ to the point where an effect operation was
invoked.
This behavior is realized by constructing a \emph{delimited continuation} from
the point of the effect invocation up to the \texttt{handle}--\texttt{with} expression that deals with
the effect and passing the value of $\ottnt{M}$ to the continuation.
We illustrate it by using the program above.
When the handled expression \texttt{\#choose(1,2) + \#choose(10,20)} is evaluated,
continuation $c \defeq$ \texttt{$ [\,] $ + \#choose(10,20)} is constructed.
% ,
% where $ [\,] $ is a hole that will be replaced with an argument of a resumption
% expression and we write $c[ \ottnt{M} ]$ for the term obtained by filling the hole
% of $[c]$ with $\ottnt{M}$.
%
Then, the body \texttt{resume (fst x)} of the operation clause is evaluated after
binding \texttt{x} to the invocation argument \texttt{(1,2)}.
Receiving the value \texttt{1} of \texttt{fst (1,2)}, the resumption expression
passes it to the continuation $c$ and $c[\texttt{1}] = \texttt{1 + \#choose(10,20)}$ is evaluated under the same handler.
Next, \texttt{choose} is invoked with argument \texttt{(10,20)}.
Similarly, continuation $c' \defeq \texttt{1 + $ [\,] $}$ is
constructed and the operation clause for \texttt{choose} is executed again.
Since \texttt{fst (10,20)} evaluates to \texttt{10},
$c'[\texttt{10}] = \texttt{1 + 10}$ is evaluated under the same handler.
Since the \texttt{return} clause returns what it receives, the entire 
expression evaluates to \texttt{11}.

Finally, we briefly review how an operation clause involving resumption expressions is
typechecked~\cite{Kammar/Lindley/Oury_2013_ICFP,Bauer/Pretnar_2015_JLAMP,Leijen_2017_POPL}.
Let us consider operation clause \texttt{op(x)$\; \rightarrow M$} for $\mathsf{op}$ of
type signature $A \hookrightarrow B$.
The typechecking is performed as follows.
First, argument \texttt{x} is assigned the domain type $\ottnt{A}$ of the signature as it
will be bound to an argument of an effect invocation.
Second, for resumption expression \texttt{resume} $\ottnt{M'}$ in $\ottnt{M}$, (1)
$\ottnt{M'}$ is required to have the codomain type $\ottnt{B}$ of the signature because its
value will be passed to the continuation as the result of the invocation and (2)
the resumption expression is assigned the same type as the return clause.
Third, the type of the body $\ottnt{M}$ has to be the same as that of the return
clause because the value of $\ottnt{M}$ is the result of the entire \texttt{handle}--\texttt{with}
expression.
For example, the above operation clause for \texttt{choose} is
typechecked as follows: first, argument \texttt{x} is assigned type
\texttt{int$\;\times\;$int}; second, it is checked whether the
argument \texttt{fst x} of the resumption expression has \texttt{int},
the codomain type of \texttt{choose}; third, it is checked whether the
body \texttt{resume (fst x)} of the clause has the same type as the return clause,
i.e., \texttt{int}.
If all the requirements are satisfied, the clause is well typed.

\subsection{Polymorphic algebraic effects and handlers}
\label{sec:poly_effects}
This section discusses motivation for polymorphism in algebraic effects and handlers.
There are two ways to introduce polymorphism: by \emph{parameterized effects}
and by \emph{polymorphic effects}.

The former is used to parameterize the declaration of an effect by types.
For example, one might declare:
\begin{lstlisting}[xleftmargin=1.3\parindent]
effect (/$\alpha$/) choose : (/$\alpha \times \alpha$/) (/$\hookrightarrow$/) (/$\alpha$/)
\end{lstlisting}
An invocation \texttt{\#choose} involves a parameterized effect of the form
$A\;\texttt{choose}$ (where $A$ denotes a type), 
according to the type of arguments: For example, \texttt{\#choose(true,false)} has the effect \texttt{bool choose}
and \texttt{\#choose(1,-1)} has \texttt{int choose}.  Handlers are required for each effect $A\;\texttt{choose}$.

The latter is used to give a polymorphic type to an effect.  For example,
one may declare
\begin{lstlisting}[xleftmargin=1.3\parindent]
effect choose : (/$\forall \alpha.$/) (/$\alpha \times \alpha$/) (/$\hookrightarrow$/) (/$\alpha$/)
\end{lstlisting}
In this case, the effect can be invoked with different types, but all
invocations have the same effect \texttt{choose}.  One can implement a
single operation clause that can handle all invocations of
\texttt{choose}, regardless of argument types.  Koka supports both
styles~\cite{Leijen_2017_POPL} (with the value restriction); we focus, however,
on the latter in this paper.  A type system for parameterized effects
lifting the value restriction is studied
by Kammar and Pretnar~\cite{Kammar/Pretnar_2017_JFP} (see \refsec{relwork} for comparison).
% because we expect that parameterized effects
% are sound even without restriction on let-bound expressions.

In what follows, we show a polymorphic version of the examples we have
seen, along with brief discussions on how polymorphic effects help with
reasoning about effect implementations.  Other practical examples of
polymorphic effects can be found in Leijen's work~\cite{Leijen_2017_POPL}.

\paragraph{Polymorphic exception.}
First, we extend the exception effect \texttt{fail} with polymorphism.
\begin{lstlisting}[numbers=left,xleftmargin=1.3\parindent]
effect (/\failp/) : (/$\forall \alpha.$/) unit (/$\hookrightarrow$/) (/$\alpha$/)

let (/\divp/) (x:int) : int =
  if x = 0 then (/$\graybox{\texttt{\#\failp()}}$/)
  else 100 / x
\end{lstlisting}
The polymorphic type signature of effect {\failp}, given in Line 1, means that
the codomain type $\alpha$ can be any.
Thus, we do not need to append the dummy value \texttt{-1} to the invocation
of {\failp} by instantiating the bound type variable $\alpha$ with \texttt{int}
(the shaded part).

\paragraph{Choice.}
Next, let us make \texttt{choose} polymorphic.
\begin{lstlisting}[numbers=left,xleftmargin=1.3\parindent]
effect (/\choosep/) : (/$\forall \alpha.$/) (/$\alpha \times \alpha$/) (/$\hookrightarrow$/) (/$\alpha$/)

let rec random_walk (x:int) : int =
  let b = #(/\choosep/)(true,false) in
  if b then random_walk (x + #(/\choosep/)(1,-1))
  else x

let f (s:int) =
  handle random_walk s with
    return x   (/$\rightarrow$/) x
    (/\choosep/) y (/$\rightarrow$/) if rand() < 0.0 then resume (fst y)
                                     else resume (snd y)
\end{lstlisting}
%
% Here, \texttt{let rec f x = $\ottnt{M}$} defines a recursive function \texttt{f}
% with argument \texttt{x} and the body $\ottnt{M}$.
%
The function \texttt{random\_walk} implements random walk;
it takes the current coordinate \texttt{x}, chooses whether it stops,
and, if it decides to continue, recursively calls itself with a new coordinate.  In the
definition, {\choosep} is used twice with different types:
\texttt{bool} and \texttt{int}.
Lines 11--12 give {\choosep} an interpretation, which calls
\texttt{rand} to obtain a random \texttt{float},\footnote{%
  One might implement \texttt{rand} as another effect operation.  }
and returns either the first or the second element of \texttt{y}.

Typechecking of operation clauses could be extended in a straightforward manner.
That is, an operation clause \texttt{op(x)$\;\rightarrow M$} for an effect
operation of signature $\forall \alpha. A \hookrightarrow B$ would be typechecked as
follows:
first, $\alpha$ is locally bound in the clause and $\mathit{x}$ is assigned type $\ottnt{A}$; second, an argument of a resumption expression must have type
$\ottnt{B}$ (which may contain type variable $\alpha$); third, $\ottnt{M}$ must
have the same type as that of the return clause (its type cannot contain $\alpha$ as $\alpha$ is local)
under the assumption that resumption expressions have the same type as the
return clause.
For example, let us consider typechecking of the above operation clause for
{\choosep}.
First, the typechecking algorithm allocates a local type variable $\alpha$ and
assigns type $\alpha \times \alpha$ to \texttt{y}.
The body has two resumption expressions, and it is checked whether the arguments
\texttt{fst y} and \texttt{snd y} have the codomain type $\alpha$ of the signature.
Finally, it is checked whether the body is typed at \texttt{int} assuming that
the resumption expressions have type \texttt{int}.
The operation clause meets all the requirements, and, therefore, it would be well typed.

An obvious advantage of polymorphic effects is reusability.  Without
polymorphism, one has to declare many versions of \texttt{choose} for
different types.

Another pleasant effect of polymorphic effects is that, thanks to
parametricity, inappropriate implementations for an effect operation
can be excluded.  For example, it is not possible for an
implementation of {\choosep} to resume with values other than the
first or second element of \texttt{y}.  In the monomorphic version,
however, it is possible to resume with any integer, as opposed to what
the name of the operation suggests.  A similar argument applies to
{\failp}; since the codomain type is $\alpha$, which does not appear
in the domain type, it is not possible to resume!  In other words, the
signature $\forall \alpha. \;\texttt{unit} \hookrightarrow \alpha$
enforces that no invocation of {\failp} will return.

% The polymorphism is not only making the code concise but also making it possible
% to reason about implementations for {\failp}.
% %
% To see it, let us consider typechecking of the handler below:
% %
% \lstinputlisting[numbers=left,xleftmargin=1.3\parindent]{source/poly_exception_impl1.ml}
% %
% Since the domain type of {\failp} is \texttt{unit}, the typechecking assigns
% \texttt{unit} to \texttt{y}.
% %
% {\failp} is also expected to return values of \emph{any} type, so arguments to
% \texttt{resume} has to have \emph{any} type.
% %
% To represent this requirement, the typechecking introduces a fresh type variable
% $\alpha$ and requires $\ottnt{M}$, which is an argument to \texttt{resume}, to have
% $\alpha$.
% %
% However, there exist no terms of a fresh type variable, which is ensured by
% parametricity~\cite{Reynolds_IFIP_1983}.
% %
% Thus, the handler above is rejected---actually, handlers for {\failp} are
% accepted only if they do not contain resumption expressions.
% %
% This fact ensures that invocation of {\failp} does not return the control to the
% caller, which is expected behavior as exceptions.

% Finally, we note that the monomorphic version may not meet this expectation.
% %
% For example, the following handler is acceptable
% %
% \lstinputlisting[numbers=left,xleftmargin=1.3\parindent]{source/mono_exception_impl1.ml}
% but then \texttt{div100 0} returns the dummy value \texttt{-1}, which may cause
% unexpected things.

\subsection{Problem in naive combination with let-polymorphism}

Although polymorphic effects and handlers provide an ability to abstract and
restrict effect implementations, one may easily expect that their
unrestricted use with naive \emph{let-polymorphism}, which allows any let-bound
expressions to be polymorphic, breaks type safety.  Indeed, it does.

We develop a counterexample, inspired by Harper and
Lillibridge~\cite{Harper/Lillibridge_1993_LSC}, below.
\begin{lstlisting}[xleftmargin=1.3\parindent]
effect get_id : (/$\forall \alpha.$/) unit (/$\hookrightarrow$/) (/$(\alpha \rightarrow \alpha)$/)

let f () : int =
  let g = #get_id() in (* g : (/$\forall \alpha.\,\alpha \rightarrow \alpha$/) *)
  if (g true) then ((g 0) + 1) else 2
\end{lstlisting}
The function \texttt{f} first binds \texttt{g} to the invocation result of
\texttt{op}.
The expression \texttt{\#get\_id()} is given type $\alpha \rightarrow \alpha$
and the naive let-polymorphism would assign type scheme
$\forall \alpha. \alpha \rightarrow \alpha$ to \texttt{g}, which makes
both \texttt{g true} and \texttt{g 0} (and thus the definition of
\texttt{f}) well typed.

An intended use of \texttt{f} is as follows:
\begin{lstlisting}[xleftmargin=1.3\parindent]
handle f () with
  return x (/$\rightarrow$/) x
  get_id y (/$\rightarrow$/) resume ((/$\lambda$/)z. z)
\end{lstlisting}
The operation clause for \texttt{get\_id} resumes with the identity
function \texttt{$\lambda$z.z}.  It would be well typed under the
typechecking procedure described in \refsec{poly_effects} and
it safely returns \texttt{1}.

However, the following strange expression
\begin{lstlisting}[xleftmargin=1.3\parindent]
handle f () with
  return x (/$\rightarrow$/) x
  get_id y (/$\rightarrow$/) resume ((/$\lambda$/)z1. (resume ((/$\lambda$/)z2. z1)); z1)
\end{lstlisting}
%
% The handler in Lines 7--9, especially the operation clause for \texttt{op},
% would be typechecked as follows.
% %
% First, the typechecking assigns a fresh type variable $\alpha$ for the operation
% clause and checks that the arguments to the resumptions have $\alpha \rightarrow
% \alpha$.
% %
% Thus, variables \texttt{z1} and \texttt{z2}, which are bound in the lambda
% abstractions given to \texttt{resume}, would be assigned $\alpha$, and it would
% be checked that the bodies of the lambda abstractions---\texttt{resume
% ($\lambda$z2.z1); z1} and \texttt{z1}---have $\alpha$, which would succeed
% since the type assigned to \texttt{z1} is $\alpha$.
% %
% Hence, the program above would be typechecked successfully.
will get stuck, although this expression would be well typed:
both \texttt{$\lambda$z1. $\cdots$ ;z1} and
\texttt{$\lambda$z2. z1} could be given type $\alpha\rightarrow\alpha$
by assigning both \texttt{z1} and \texttt{z2} type $\alpha$, which is
the type variable local to this clause.
Let us see how the evaluation gets stuck in detail.
When the handled expression \texttt{f ()} invokes effect \texttt{get\_id},
the following continuation will be constructed:
\[   c \defeq
  \texttt{let g = $ [\,] $ in if (g true) then ((g 0) + 1) else 2}\ .
\]
Next, the body of the operation clause \texttt{get\_id} is evaluated.
It immediately resumes and reduces to
\[
 c'[ \texttt{($\lambda$z1. $c' [$($\lambda$z2.z1)$]$; z1)} ]
\]
where
\[
 c' \defeq 
  \begin{array}{l}
 \texttt{handle } c \texttt{ with}\\
 \texttt{\quad return x $\rightarrow$ x} \\
  \texttt{\quad get\_id y $\rightarrow$ resume ($\lambda$z1. (resume ($\lambda$z2.z1)); z1)}\ ,
  \end{array}\]
which is the continuation $c$ under the same handler.
The evaluation proceeds as follows (here,
$k \defeq \texttt{$\lambda$z1. $c' [$($\lambda$z2.z1)$]$; z1}$):
\[\begin{array}{cl}
 &
 c'[ \texttt{($\lambda$z1. $c' [$($\lambda$z2.z1)$]$; z1)} ]
 \\ = &
 \texttt{handle let g = $k$ in if (g true) then ((g 0) + 1) else 2 with } \ldots
 \\ \longrightarrow &
 \texttt{handle if ($k$ true) then (($k$ 0) + 1) else 2 with } \ldots
 \\ \longrightarrow &
 \texttt{handle if $c' [$($\lambda$z2.true)$]$; true then (($k$ 0) + 1) else 2 with } \ldots
  \end{array}\]
Here, the hole in $c'$ is filled by function (\texttt{$\lambda$z2.true}),
which returns a Boolean value, \emph{though the hole is supposed to be filled
  by a function of $\forall \alpha.\,\alpha \rightarrow \alpha$}.
This weird gap triggers a run-time error:
\[\begin{array}{cl}
 &
 c' [\texttt{($\lambda$z2.true)}]
 \\[1ex] &
 \texttt{handle } \\ = & \quad
   \texttt{let g = $\lambda$z2.true in if (g true) then ((g 0) + 1) else 2} \\ &
 \texttt{with } \ldots
 \\[1ex] \longrightarrow^* &
 \texttt{handle if true then ((($\lambda$z2.true) 0) + 1) else 2 with } \ldots
 \\ \longrightarrow &
 \texttt{handle (($\lambda$z2.true) 0) + 1 with } \ldots
 \\ \longrightarrow &
 \texttt{handle true + 1 with } \ldots
  \end{array}\]
We stop here because \texttt{true + 1} cannot reduce.

\subsection{Our solution}

A standard approach to this problem is to restrict the form of
let-bound
expressions by some means such as the (relaxed) value
restriction~\cite{Tofte_1990_IC,Wright_1995_LSC,Garrigue_2004_FLOPS} or weak
polymorphism~\cite{Appel/MacQueen_1991_PLILP,Hoang/Mitchell/Viswanathan_1993_LICS}.
This approach amounts to restricting how effect operations can be \emph{used}.

In this paper, we seek for a complementary approach, which is to
restrict how effect operations can be \emph{implemented}.\footnote{%
  We compare our approach with the standard approaches in
  \refsec{relwork} in detail.}
  More concretely, we develop a type
system such that let-bound expressions are polymorphic as long as they
invoke only ``safe'' polymorphic effects and the notion of safe
polymorphic effects is formalized in terms of typing rules (for
handlers).

To see what are ``safe'' effects, let us examine the above
counterexample to type safety.
The crux of the counterexample is that
\begin{enumerate}
\item continuation $c$ uses \texttt{g} polymorphically, namely,
as $\texttt{bool} \rightarrow \texttt{bool}$
in \texttt{g true} and as $\texttt{int} \rightarrow \texttt{int}$ in \texttt{g 1};
\item $c$ is invoked twice; and
\item the use of \texttt{g} as
  $\texttt{bool} \rightarrow \texttt{bool}$ in the first invocation of
  $c$---where \texttt{g} is bound to
  \texttt{$\lambda$z1.$\cdots$; z1}---``alters'' the type of
  \texttt{$\lambda$z2. z1} (passed to \texttt{resume}) from
  $\alpha \rightarrow \alpha$ to $\alpha \rightarrow \texttt{bool}$,
  contradicting the second use of \texttt{g} as
  $\texttt{int} \rightarrow \texttt{int}$ in the second invocation of $c$.
\end{enumerate}
The last point is crucial---if \texttt{$\lambda$z2.z1} were, say,
$\lambda\texttt{z2}.\texttt{z2}$, there would be no influence from the
first invocation of $c$ and the evaluation would succeed.  The problem
we see here is that the naive type system mistakenly allows
\emph{interference} between the arguments to the two resumptions by
assuming that \texttt{z1} and \texttt{z2} share the same type.

Based on this observation, the typing rule for resumption is revised
to disallow interference between different resumptions by separating
their types: for each \texttt{resume $M$} in the operation clause for
$\texttt{op} : \forall \alpha_1 \cdots \alpha_n.A \hookrightarrow B$, $M$
has to have type $B'$ obtained by renaming all type variables
$\alpha_i$ in $B$ with \emph{fresh} type variables $\alpha_i'$.  In
the case of \texttt{get\_id}, the two resumptions should be called
with \(\beta \rightarrow \beta\) and \(\gamma \rightarrow \gamma\) for
fresh \(\beta\) and \(\gamma\); for the first \texttt{resume} to be well
typed, \texttt{z1} has to be of type \(\beta\), although it means that the
return type of \texttt{$\lambda$z2.z1} (given to the second
resumption) is \(\beta\), making the entire clause ill typed, as we
expect.  If a clause does not have interfering resumptions like
\begin{center}
\texttt{get\_id y $\rightarrow$ resume ($\lambda$z1.z1)}
\end{center}
or
\begin{center}
\texttt{get\_id y $\rightarrow$ resume ($\lambda$z1. (resume ($\lambda$z2.z2)); z1)},
\end{center}
it will be well typed.

\section{Surface language: {\surfacelang}}
\label{sec:surfacelang}

We define a lambda calculus {\surfacelang} that supports
let-polymorphism, polymorphic algebraic effects, and handlers without
interfering resumptions.
This section introduces the syntax and the type system of {\surfacelang}.
The semantics is given by a formal elaboration to intermediate calculus {\interlang},
which will be introduced in \refsec{interlang}.

\subsection{Syntax}

\begin{figure}[t]
 \[
  \begin{array}{l@{\quad}l@{\quad}r@{\quad}l}
   \textbf{Effect operations} & \mathsf{op} &&
   \textbf{Type variables} \quad \alpha, \beta, \gamma \\[.5ex]
   \textbf{Effects} & \epsilon & ::= & \text{sets of effect operations} \\[.5ex]
   \textbf{Base types} & \iota & ::= &  \mathsf{bool}  \mid  \mathsf{int}  \mid ... \\[.5ex]
   \textbf{Types} & \ottnt{A}, \ottnt{B}, \ottnt{C}, \ottnt{D} & ::= & \alpha \mid \iota \mid  \ottnt{A}   \rightarrow  \!  \epsilon  \;  \ottnt{B}  \\[.5ex]
   \textbf{Type schemes} & \sigma & ::= & \ottnt{A} \mid  \text{\unboldmath$\forall$}  \, \alpha  \ottsym{.}  \sigma \\[.5ex]
   \textbf{Constants} & \ottnt{c} & ::= &  \mathsf{true}  \mid  \mathsf{false}  \mid  0  \mid  \mathsf{+}  \mid ... \\[.5ex]
   \textbf{Terms} & \ottnt{M} & ::= &
    \mathit{x} \mid \ottnt{c} \mid  \lambda\!  \, \mathit{x}  \ottsym{.}  \ottnt{M} \mid \ottnt{M_{{\mathrm{1}}}} \, \ottnt{M_{{\mathrm{2}}}} \mid \mathsf{let} \, \mathit{x}  \ottsym{=}  \ottnt{M_{{\mathrm{1}}}} \,  \mathsf{in}  \, \ottnt{M_{{\mathrm{2}}}} \mid \\ &&&
     \textup{\texttt{\#}\relax}  \mathsf{op}   \ottsym{(}   \ottnt{M}   \ottsym{)}  \mid \mathsf{handle} \, \ottnt{M} \, \mathsf{with} \, \ottnt{H} \mid \mathsf{resume} \, \ottnt{M}
    \\[.5ex]
   \textbf{Handlers} & \ottnt{H} & ::= & \mathsf{return} \, \mathit{x}  \rightarrow  \ottnt{M} \mid \ottnt{H}  \ottsym{;}  \mathsf{op}  \ottsym{(}  \mathit{x}  \ottsym{)}  \rightarrow  \ottnt{M}
    \\[.5ex]
   \textbf{Typing contexts} & \Gamma & ::= &  \emptyset  \mid
    \Gamma  \ottsym{,}  \mathit{x} \,  \mathord{:}  \, \sigma \mid \Gamma  \ottsym{,}  \alpha
    \\[.5ex]
  \end{array}
 \]
 \caption{Syntax of {\surfacelang}.}
 \label{fig:surface-syntax}
\end{figure}

The syntax of {\surfacelang} is given in \reffig{surface-syntax}.
Effect operations are denoted by $\mathsf{op}$ and type variables by $\alpha$,
$\beta$, and $\gamma$.
An effect, denoted by $\epsilon$, is a finite set of effect operations.
We write $ \langle \rangle $ for the empty effect set.
A type, denoted by $\ottnt{A}$, $\ottnt{B}$, $\ottnt{C}$, and $\ottnt{D}$, is a type variable; a base type
$\iota$, which includes, e.g., $ \mathsf{bool} $ and $ \mathsf{int} $; or a function type
$ \ottnt{A}   \rightarrow  \!  \epsilon  \;  \ottnt{B} $, which is given to functions that take an argument of type $\ottnt{A}$
and compute a value of type $\ottnt{B}$ possibly with effect $\epsilon$.
A type scheme $\sigma$ is obtained by abstracting type variables.
Terms, denoted by $\ottnt{M}$, consist of variables; constants (including primitive
operations); lambda abstractions $ \lambda\!  \, \mathit{x}  \ottsym{.}  \ottnt{M}$, which bind $\mathit{x}$ in $\ottnt{M}$;
function applications; let-expressions $\mathsf{let} \, \mathit{x}  \ottsym{=}  \ottnt{M_{{\mathrm{1}}}} \,  \mathsf{in}  \, \ottnt{M_{{\mathrm{2}}}}$, which bind
$\mathit{x}$ in $\ottnt{M_{{\mathrm{2}}}}$; effect invocations $ \textup{\texttt{\#}\relax}  \mathsf{op}   \ottsym{(}   \ottnt{M}   \ottsym{)} $; $ \mathsf{handle} $--$ \mathsf{with} $ expressions
$\mathsf{handle} \, \ottnt{M} \, \mathsf{with} \, \ottnt{H}$; and resumption expressions $\mathsf{resume} \, \ottnt{M}$.
All type information in {\surfacelang} is implicit; thus the terms have no
type annotations.
A handler $\ottnt{H}$ has a single return clause $\mathsf{return} \, \mathit{x}  \rightarrow  \ottnt{M}$, where $\mathit{x}$ is
bound in $\ottnt{M}$, and zero or more operation clauses of the form $\mathsf{op}  \ottsym{(}  \mathit{x}  \ottsym{)}  \rightarrow  \ottnt{M}$, where $\mathit{x}$ is bound in $\ottnt{M}$.
A typing context $\Gamma$ binds a sequence of variable declarations
$\mathit{x} \,  \mathord{:}  \, \sigma$ and type variable declarations $\alpha$.

% where $ \algeffseqover{ \sigma } $ denotes a sequence of type schemes, and type variables.
% %
% A variable in a typing context can be associated with two or more type schemes.
% %
% This allows using an argument to an operation clause in the resumption
% expressions; we will explain this in detail in the next section.

We introduce the following notations used throughout this paper.
We write $ \text{\unboldmath$\forall$}  \,  \algeffseqoverindex{ \alpha }{ \ottmv{i}  \in  \text{\unboldmath$\mathit{I}$} }   \ottsym{.}  \ottnt{A}$ for $ \text{\unboldmath$\forall$}  \, \alpha_{{\mathrm{1}}}  \ottsym{.}  ...   \text{\unboldmath$\forall$}  \, \alpha_{\ottmv{n}}  \ottsym{.}  \ottnt{A}$
where $\text{\unboldmath$\mathit{I}$} = \{ 1, ..., n \}$.
% , and
% $ \algeffseqoverindex{  \text{\unboldmath$\forall$}  \,  \algeffseqoverindex{ \alpha_{\ottmv{i}} }{ \ottmv{i}  \in  \text{\unboldmath$\mathit{I}$} }   \ottsym{.}  \ottnt{A_{\ottmv{j}}} }{ \ottmv{j}  \in  \text{\unboldmath$\mathit{J}$} } $ for a sequence of types
% $ \text{\unboldmath$\forall$}  \,  \algeffseqoverindex{ \alpha_{\ottmv{i}} }{ \ottmv{i}  \in  \text{\unboldmath$\mathit{I}$} }   \ottsym{.}  \ottnt{A_{{\mathrm{1}}}}$, ..., $ \text{\unboldmath$\forall$}  \,  \algeffseqoverindex{ \alpha_{\ottmv{i}} }{ \ottmv{i}  \in  \text{\unboldmath$\mathit{I}$} }   \ottsym{.}  \ottnt{A_{\ottmv{m}}}$
% where $\text{\unboldmath$\mathit{J}$} = \{ 1, ..., m \}$.
%
We often omit indices ($\ottmv{i}$ and $\ottmv{j}$) and index sets ($\text{\unboldmath$\mathit{I}$}$ and $\text{\unboldmath$\mathit{J}$}$) if
they are not important: e.g., we often abbreviate $ \text{\unboldmath$\forall$}  \,  \algeffseqoverindex{ \alpha }{ \ottmv{i}  \in  \text{\unboldmath$\mathit{I}$} }   \ottsym{.}  \ottnt{A}$ to
$ \text{\unboldmath$\forall$}  \,  \algeffseqoverindex{ \alpha }{ \text{\unboldmath$\mathit{I}$} }   \ottsym{.}  \ottnt{A}$ or even to $ \text{\unboldmath$\forall$}  \,  \algeffseqover{ \alpha }   \ottsym{.}  \ottnt{A}$.
% and $ \algeffseqoverindex{  \text{\unboldmath$\forall$}  \,  \algeffseqoverindex{ \alpha_{\ottmv{i}} }{ \ottmv{i}  \in  \text{\unboldmath$\mathit{I}$} }   \ottsym{.}  \ottnt{A_{\ottmv{j}}} }{ \ottmv{j}  \in  \text{\unboldmath$\mathit{J}$} } $ to
% $ \algeffseqoverindex{  \text{\unboldmath$\forall$}  \,  \algeffseqoverindex{ \alpha }{ \text{\unboldmath$\mathit{I}$} }   \ottsym{.}  \ottnt{A} }{ \text{\unboldmath$\mathit{J}$} } $ and $ \algeffseqover{  \text{\unboldmath$\forall$}  \,  \algeffseqover{ \alpha }   \ottsym{.}  \ottnt{A} } $.
%
Similarly, we use a bold font for other sequences
($ \algeffseqoverindex{ \ottnt{A} }{ \ottmv{i}  \in  \text{\unboldmath$\mathit{I}$} } $ for a sequence of types, $ \algeffseqoverindex{ \ottnt{v} }{ \ottmv{i}  \in  \text{\unboldmath$\mathit{I}$} } $ for a sequence of values, etc.).
We sometimes write $\ottsym{\{}   \algeffseqover{ \alpha }   \ottsym{\}}$ to view the sequence $ \algeffseqover{ \alpha } $ as a set by ignoring the order.
Free type variables $ \mathit{ftv}  (  \sigma  ) $ in a type scheme $\sigma$ and type
substitution $ \ottnt{B}    [   \algeffseqover{ \ottnt{A} }   \ottsym{/}   \algeffseqover{ \alpha }   ]  $ of $ \algeffseqover{ \ottnt{A} } $ for type variables $ \algeffseqover{ \alpha } $ in
$\ottnt{B}$ are defined as usual (with the understanding that the omitted index sets for $ \algeffseqover{ \ottnt{A} } $ and $ \algeffseqover{ \alpha } $ are the same).

We suppose that each constant $\ottnt{c}$ is assigned
a first-order closed type $ \mathit{ty}  (  \ottnt{c}  ) $ of the form
$\iota_{{\mathrm{1}}}  \rightarrow \! \langle \rangle \; \cdots  \rightarrow \! \langle \rangle  \; \iota_{\ottmv{n}}$
and that each effect operation $\mathsf{op}$ is assigned
a signature of the form $  \text{\unboldmath$\forall$}     \algeffseqover{ \alpha }   .  \ottnt{A}  \hookrightarrow  \ottnt{B} $, which
means that an invocation of $\mathsf{op}$ with type instantiation $ \algeffseqover{ \ottnt{C} } $
takes an argument of $ \ottnt{A}    [   \algeffseqover{ \ottnt{C} }   \ottsym{/}   \algeffseqover{ \alpha }   ]  $ and returns a value of
$ \ottnt{B}    [   \algeffseqover{ \ottnt{C} }   \ottsym{/}   \algeffseqover{ \alpha }   ]  $.
We also assume that, for $\mathit{ty} \, \ottsym{(}  \mathsf{op}  \ottsym{)} \,  =  \,   \text{\unboldmath$\forall$}     \algeffseqover{ \alpha }   .  \ottnt{A}  \hookrightarrow  \ottnt{B} $, $ \mathit{ftv}  (  \ottnt{A}  )  \,  \subseteq  \, \ottsym{\{}   \algeffseqover{ \alpha }   \ottsym{\}}$ and
$ \mathit{ftv}  (  \ottnt{B}  )  \,  \subseteq  \, \ottsym{\{}   \algeffseqover{ \alpha }   \ottsym{\}}$.

\subsection{Type system}

The type system of {\surfacelang} consists of four judgments:
well-formedness of typing contexts $\vdash  \Gamma$;
well formedness of type schemes $\Gamma  \vdash  \sigma$;
term typing judgment $\Gamma  \ottsym{;}  \ottnt{R}  \vdash  \ottnt{M}  \ottsym{:}  \ottnt{A} \,  |  \, \epsilon$, which means that
$\ottnt{M}$ computes a value of $\ottnt{A}$ possibly with effect $\epsilon$
under typing context $\Gamma$ and resumption type $\ottnt{R}$ (discussed below); and
handler typing judgment $\Gamma  \ottsym{;}  \ottnt{R}  \vdash  \ottnt{H}  \ottsym{:}  \ottnt{A} \,  |  \, \epsilon  \Rightarrow  \ottnt{B} \,  |  \, \epsilon'$,
which means that $\ottnt{H}$ handles a computation that produces a value of $\ottnt{A}$
with effect $\epsilon$ and that the clauses in $\ottnt{H}$ compute a value of $\ottnt{B}$
possibly with effect $\epsilon'$ under $\Gamma$ and $\ottnt{R}$.

A resumption type $\ottnt{R}$ contains type information for resumption.
\begin{defn}[Resumption type]
 Resumption types in {\surfacelang}, denoted by $\ottnt{R}$, are defined as
 follows:
 \[\begin{array}{lll}
  \ottnt{R} & ::= &
    \mathsf{none}  \mid
   \ottsym{(}   \algeffseqover{ \alpha }   \ottsym{,}  \mathit{x} \,  \mathord{:}  \, \ottnt{A}  \ottsym{,}   \ottnt{B}   \rightarrow  \!  \epsilon  \;  \ottnt{C}   \ottsym{)} \\ &&
   \qquad\qquad \text{(if $ \mathit{ftv}  (  \ottnt{A}  )  \,  \mathbin{\cup}  \,  \mathit{ftv}  (  \ottnt{B}  )  \,  \subseteq  \, \ottsym{\{}   \algeffseqover{ \alpha }   \ottsym{\}}$ and $ \mathit{ftv}  (  \ottnt{C}  )  \,  \mathbin{\cap}  \, \ottsym{\{}   \algeffseqover{ \alpha }   \ottsym{\}} \,  =  \,  \emptyset $)}
   \end{array}\]
\end{defn}
If $\ottnt{M}$ is not a subterm of an operation clause, it is typechecked under
$\ottnt{R} =  \mathsf{none} $, which means that $\ottnt{M}$ cannot contain resumption
expressions.
Otherwise, suppose that $\ottnt{M}$ is a subterm of an operation clause
$\mathsf{op}  \ottsym{(}  \mathit{x}  \ottsym{)}  \rightarrow  \ottnt{M'}$ that handles effect $\mathsf{op}$ of signature $  \text{\unboldmath$\forall$}     \algeffseqover{ \alpha }   .  \ottnt{A}  \hookrightarrow  \ottnt{B} $
and computes a value of $\ottnt{C}$ possibly with effect $\epsilon$.
Then, $\ottnt{M}$ is typechecked under $\ottnt{R} = \ottsym{(}   \algeffseqover{ \alpha }   \ottsym{,}  \mathit{x} \,  \mathord{:}  \, \ottnt{A}  \ottsym{,}   \ottnt{B}   \rightarrow  \!  \epsilon  \;  \ottnt{C}   \ottsym{)}$, which
means that argument $\mathit{x}$ to the operation clause has type $\ottnt{A}$ and
that resumptions in $\ottnt{M}$ are effectful functions from $\ottnt{B}$ to $\ottnt{C}$ with
effect $\epsilon$.
Note that type variables $ \algeffseqover{ \alpha } $ occur free only in $\ottnt{A}$ and
$\ottnt{B}$ but not in $\ottnt{C}$.

\begin{figure}[t!]
\textbf{Well-formed rules for typing contexts} \\[1.5ex]
\framebox{$\vdash  \Gamma$}
\begin{center}
 $\ottdruleWFXXEmpty{}$ \hfil
 $\ottdruleWFXXVar{}$ \\[2ex]
 $\ottdruleWFXXTyVar{}$ \\[2.5ex]
\end{center}
\textbf{Typing rules} \\[1.5ex]
\framebox{$\Gamma  \ottsym{;}  \ottnt{R}  \vdash  \ottnt{M}  \ottsym{:}  \ottnt{A} \,  |  \, \epsilon$}
\begin{center}
 $\ottdruleTSXXVar{}$ \hfil
 $\ottdruleTSXXConst{}$ \\[2ex]
 $\ottdruleTSXXAbs{}$ \\[2ex]
 $\ottdruleTSXXApp{}$ \\[2ex]
 $\ottdruleTSXXLet{}$ \\[2ex]
 $\ottdruleTSXXWeak{}$ \\[2ex]
 $\ottdruleTSXXOp{}$ \\[2ex]
 $\ottdruleTSXXHandle{}$ \\[2ex]
 $\ottdruleTSXXResume{}$ \\[2.5ex]
\end{center}
\framebox{$\Gamma  \ottsym{;}  \ottnt{R}  \vdash  \ottnt{H}  \ottsym{:}  \ottnt{A} \,  |  \, \epsilon  \Rightarrow  \ottnt{B} \,  |  \, \epsilon'$}
\begin{center}
 $\ottdruleTHSXXReturn{}$ \\[2ex]
 $\ottdruleTHSXXOp{}$
\end{center}
 \caption{Typing rules.}
 \label{fig:surface-typing}
\end{figure}

\reffig{surface-typing} shows the inference rules of the judgments
(except for $\Gamma  \vdash  \sigma$, which is defined by:
$\Gamma  \vdash  \sigma$ if and only if all free type variables in
$\sigma$ are bound by $\Gamma$).
For a sequence of type schemes $ \algeffseqover{ \sigma } $, we write
$\Gamma  \vdash   \algeffseqover{ \sigma } $ if and only if every type scheme in $ \algeffseqover{ \sigma } $ is well formed under
$\Gamma$.

Well-formedness rules for typing contexts, shown at the top of
\reffig{surface-typing}, are standard.
A typing context is well formed if it is empty \WF{Empty} or a variable in the
typing context is associated with a type scheme that is well formed in the
remaining typing context \WF{Var} and a type variable in the typing context is
not declared \WF{TVar}.
For typing context $\Gamma$, $ \mathit{dom}  (  \Gamma  ) $ denotes the set of type and term
variables declared in $\Gamma$.

Typing rules for terms are given in the middle of \reffig{surface-typing}.
The first six rules are standard for the lambda calculus with let-polymorphism
and a type-and-effect system. %, except \TSrule{Var}.
If a variable $\mathit{x}$ is introduced by a let-expression and has type scheme
$ \text{\unboldmath$\forall$}  \,  \algeffseqover{ \alpha }   \ottsym{.}  \ottnt{A}$ in $\Gamma$, it is given type $ \ottnt{A}    [   \algeffseqover{ \ottnt{B} }   \ottsym{/}   \algeffseqover{ \alpha }   ]  $, 
obtained by instantiating type variables $ \algeffseqover{ \alpha } $ with well-formed types
$ \algeffseqover{ \ottnt{B} } $.
If $\mathit{x}$ is bound by other constructors (e.g., a lambda abstraction), $\mathit{x}$
is always bound to a monomorphic type and both $ \algeffseqover{ \alpha } $ and $ \algeffseqover{ \ottnt{B} } $ are
the empty sequence.
Note that \TSrule{Var} gives any effect $\epsilon$ to the typing judgment for
$\mathit{x}$.
In general, $\epsilon$ in judgment $\Gamma  \ottsym{;}  \ottnt{R}  \vdash  \ottnt{M}  \ottsym{:}  \ottnt{A} \,  |  \, \epsilon$ means that
the evaluation of $\ottnt{M}$ \emph{may} invoke effect operations in $\epsilon$.
Since a reference to a variable involves no effect, it is given any effect; for the
same reason, value constructors are also given any effect.
The rule \TSrule{Const} means that the type of a constant is given by
(meta-level) function $ \mathit{ty} $.
The typing rules for lambda abstractions and function applications are standard
in the lambda calculus equipped with a type-and-effect system.
The rule \TSrule{Abs} gives lambda abstraction $ \lambda\!  \, \mathit{x}  \ottsym{.}  \ottnt{M}$ function type $ \ottnt{A}   \rightarrow  \!  \epsilon'  \;  \ottnt{B} $ if $\ottnt{M}$ computes a value of $\ottnt{B}$ possibly with effect
$\epsilon'$ by using $\mathit{x}$ of type $\ottnt{A}$.
The rule \TSrule{App} requires that (1) the argument type of function part
$\ottnt{M_{{\mathrm{1}}}}$ be equivalent to the type of actual argument $\ottnt{M_{{\mathrm{2}}}}$ and (2) effect
$\epsilon'$ invoked by function $\ottnt{M_{{\mathrm{1}}}}$ be contained in the whole effect
$\epsilon$.
The rule \TSrule{Weak} allows weakening of effects.

The next two rules are mostly standard for algebraic effects and handlers.
The rule \TSrule{Op} is applied to effect invocations.
Since {\surfacelang} supports implicit polymorphism, an invocation $ \textup{\texttt{\#}\relax}  \mathsf{op}   \ottsym{(}   \ottnt{M}   \ottsym{)} $
of polymorphic effect $\mathsf{op}$ of signature $  \text{\unboldmath$\forall$}     \algeffseqover{ \alpha }   .  \ottnt{A}  \hookrightarrow  \ottnt{B} $ also accompanies
implicit type substitution of well-formed types $ \algeffseqover{ \ottnt{C} } $ for $ \algeffseqover{ \alpha } $.
Thus, the type of argument $\ottnt{M}$ has to be $ \ottnt{A}    [   \algeffseqover{ \ottnt{C} }   \ottsym{/}   \algeffseqover{ \alpha }   ]  $ and the result
of the invocation is given type $ \ottnt{B}    [   \algeffseqover{ \ottnt{C} }   \ottsym{/}   \algeffseqover{ \alpha }   ]  $.
In addition, effect $\epsilon$ contains $\mathsf{op}$.
The typeability of $ \mathsf{handle} $--$ \mathsf{with} $ expressions depends on the typing of handlers \TSrule{Handle},
which will be explained below shortly.

The last typing rule \TSrule{Resume} is the key to gaining type safety in this
work.
Suppose that we are given resumption type $\ottsym{(}   \algeffseqover{ \alpha }   \ottsym{,}  \mathit{x} \,  \mathord{:}  \, \ottnt{A}  \ottsym{,}   \ottnt{B}   \rightarrow  \!  \epsilon  \;  \ottnt{C}   \ottsym{)}$.
% given $ \mathsf{none} $, we cannot perform resumption.
%
Intuitively, $ \ottnt{B}   \rightarrow  \!  \epsilon  \;  \ottnt{C} $ is the type of the continuation for resumption
and, therefore, argument $\ottnt{M}$ to $ \mathsf{resume} $ is required to have type $\ottnt{B}$.
As we have discussed in \refsec{overview}, we avoid interference between
different resumptions by renaming $ \algeffseqover{ \alpha } $, the type parameters to the
effect operation, to fresh type variables $ \algeffseqover{ \beta } $, in typechecking $\ottnt{M}$.
Freshness of $ \algeffseqover{ \beta } $ will be ensured when well-formedness of typing
contexts $\Gamma_{{\mathrm{1}}}  \ottsym{,}  \Gamma_{{\mathrm{2}}}  \ottsym{,}   \algeffseqover{ \beta } , \ldots$ is checked at the leaves of the type
derivation.
The type variables $ \algeffseqover{ \alpha } $ in the type of $\mathit{x}$, the parameter to the
operation, are also renamed for $\mathit{x}$ to be useful in $\ottnt{M}$.
To see why this renaming is useful, let us consider an extension of
the calculus with pairs and typechecking
of an operation clause for $ \mathsf{choose}^{\forall} $ of signature
$  \text{\unboldmath$\forall$}    \alpha  .   \alpha  \times  \alpha   \hookrightarrow  \alpha $:
\[
  \mathsf{choose}^{\forall}   \ottsym{(}  \mathit{x}  \ottsym{)}  \rightarrow  \mathsf{resume} \, \ottsym{(}   \mathsf{fst}  \, \mathit{x}  \ottsym{)}
\]
Variable $\mathit{x}$ is assigned product type $ \alpha  \times  \alpha $ for fresh type variable
$\alpha$ and the body $\mathsf{resume} \, \ottsym{(}   \mathsf{fst}  \, \mathit{x}  \ottsym{)}$ is typechecked under the resumption type
$\ottsym{(}  \alpha  \ottsym{,}  \mathit{x} \,  \mathord{:}  \,  \alpha  \times  \alpha   \ottsym{,}   \alpha   \rightarrow  \!  \epsilon  \;  \ottnt{A}   \ottsym{)}$ for some $\epsilon$ and $\ottnt{A}$ (see the typing rules
for handlers for details).
To typecheck $\mathsf{resume} \, \ottsym{(}   \mathsf{fst}  \, \mathit{x}  \ottsym{)}$, the argument $ \mathsf{fst}  \, \mathit{x}$ is required
to have type $\beta$, freshly generated for this $ \mathsf{resume} $.
Without applying renaming also to $\mathit{x}$, the clause would not
typecheck.
Finally, \TSrule{Resume} also requires that (1) the typing context
contains $ \algeffseqover{ \alpha } $, which should have been declared at an application
of the typing rule for the operation clause that surrounds this
$ \mathsf{resume} $ and (2) effect $\epsilon$, which may be invoked by
resumption of a continuation, be contained in the whole effect
$\epsilon'$.
The binding $\mathit{x} \,  \mathord{:}  \, \ottnt{D}$ in the conclusion means that parameter $\mathit{x}$ to the
operation clause is declared outside the resumption expression.

The typing rules for handlers are standard~\cite{Kammar/Lindley/Oury_2013_ICFP,Bauer/Pretnar_2015_JLAMP,Leijen_2017_POPL}.
The rule \THSrule{Return} for a return clause $\mathsf{return} \, \mathit{x}  \rightarrow  \ottnt{M}$
checks that the body $\ottnt{M}$ is given a type under the assumption that
argument $\mathit{x}$ has type $\ottnt{A}$, which is the type of the handled
expression.
The effect $\epsilon$ stands for effects that are not handled by the
operation clauses that follow the return clause and it must be a
subset of the effect $\epsilon'$ that $\ottnt{M}$ may cause.\footnote{%
  Thus, handlers in {\surfacelang} are open~\cite{Kammar/Lindley/Oury_2013_ICFP} in the sense that
  a $ \mathsf{handle} $--$ \mathsf{with} $ expression does not have to handle
  \emph{all} effects caused by the handled expression.
}
A handler having operation clauses is typechecked by \THSrule{Op},
which checks that the body of the operation clause $\mathsf{op}  \ottsym{(}  \mathit{x}  \ottsym{)}  \rightarrow  \ottnt{M}$
for $\mathsf{op}$ of signature $  \text{\unboldmath$\forall$}     \algeffseqover{ \alpha }   .  \ottnt{C}  \hookrightarrow  \ottnt{D} $ is typed at the result
type $\ottnt{B}$, which is the same as the type of the return clause,
under the typing context extended with fresh assigned type variables
$ \algeffseqover{ \alpha } $ and argument $\mathit{x}$ of type $\ottnt{C}$, together with the
resumption type $\ottsym{(}   \algeffseqover{ \alpha }   \ottsym{,}  \mathit{x} \,  \mathord{:}  \, \ottnt{C}  \ottsym{,}   \ottnt{D}   \rightarrow  \!  \epsilon'  \;  \ottnt{B}   \ottsym{)}$.
The effect $ \epsilon    \mathbin{\uplus}    \ottsym{\{}  \mathsf{op}  \ottsym{\}} $ in the conclusion means that
the effect operation $\mathsf{op}$ is handled by this clause
and no other clauses (in the present handler) handle it.
Our semantics adopts deep handlers~\cite{Kammar/Lindley/Oury_2013_ICFP}, i.e., when a handled
expression invokes an effect operation, the continuation, which passed
to the operation clause, is wrapped by the same handler.
Thus, resumption may invoke the same effect $\epsilon'$ as the one
possibly invoked by the clauses of the handler, hence $ \ottnt{D}   \rightarrow  \!  \epsilon'  \;  \ottnt{B} $
in the resumption type.

Finally, we show how the type system rejects the counterexample given in
\refsec{overview}.
The problem is in the following operation clause.
\[
 \mathsf{op}  \ottsym{(}  \mathit{y}  \ottsym{)}  \rightarrow  \mathsf{resume} \,  \lambda\!  \, \mathit{z_{{\mathrm{1}}}}  \ottsym{.}  \ottsym{(}  \mathsf{resume} \,  \lambda\!  \, \mathit{z_{{\mathrm{2}}}}  \ottsym{.}  \mathit{z_{{\mathrm{1}}}}  \ottsym{)}  \ottsym{;}  \mathit{z_{{\mathrm{1}}}}
\]
where $\mathsf{op}$ has effect signature $  \text{\unboldmath$\forall$}    \alpha  .   \mathsf{unit}   \hookrightarrow  \ottsym{(}   \alpha   \rightarrow  \!   \langle \rangle   \;  \alpha   \ottsym{)} $.
This clause is typechecked under resumption type $\ottsym{(}  \alpha  \ottsym{,}  \mathit{y} \,  \mathord{:}  \,  \mathsf{unit}   \ottsym{,}   \alpha   \rightarrow  \!  \epsilon  \;  \alpha   \ottsym{)}$
for some $\epsilon$.
By \TSrule{Resume}, the two resumption expressions are assigned two different
type variables $\gamma_{{\mathrm{1}}}$ and $\gamma_{{\mathrm{2}}}$, and the arguments
$ \lambda\!  \, \mathit{z_{{\mathrm{1}}}}  \ottsym{.}  \ottsym{(}  \mathsf{resume} \,  \lambda\!  \, \mathit{z_{{\mathrm{2}}}}  \ottsym{.}  \mathit{z_{{\mathrm{1}}}}  \ottsym{)}  \ottsym{;}  \mathit{z_{{\mathrm{1}}}}$ and $ \lambda\!  \, \mathit{z_{{\mathrm{2}}}}  \ottsym{.}  \mathit{z_{{\mathrm{1}}}}$ are required to have
$ \gamma_{{\mathrm{1}}}   \rightarrow  \!  \epsilon  \;  \gamma_{{\mathrm{1}}} $ and $ \gamma_{{\mathrm{2}}}   \rightarrow  \!  \epsilon  \;  \gamma_{{\mathrm{2}}} $, respectively.
However, $ \lambda\!  \, \mathit{z_{{\mathrm{2}}}}  \ottsym{.}  \mathit{z_{{\mathrm{1}}}}$ cannot because $\mathit{z_{{\mathrm{1}}}}$ is associated with $\gamma_{{\mathrm{1}}}$ but
not with $\gamma_{{\mathrm{2}}}$.

\paragraph{Remark.}
The rule \TSrule{Resume} allows only the type of the argument to an
operation clause to be renamed.
Thus, other variables bound by, e.g., lambda abstractions and let-expressions
outside the resumption expression cannot be used as such a type.
As a result, more care may be required as to where to introduce a new variable.
For example, let us consider the following operation clause (which is a variant
of the example of $ \mathsf{choose}^{\forall} $ above).
\[
  \mathsf{choose}^{\forall}   \ottsym{(}  \mathit{x}  \ottsym{)}  \rightarrow  \mathsf{let} \, \mathit{y}  \ottsym{=}   \mathsf{fst}  \, \mathit{x} \,  \mathsf{in}  \, \mathsf{resume} \, \mathit{y}
\]
The variable $\mathit{x}$ is assigned $ \alpha  \times  \alpha $ first and the resumption
requires $\mathit{y}$ to be typed at fresh type variable $\beta$.
This clause would be rejected in the current type system because
$ \mathsf{fst}  \, \mathit{x}$ appears outside $ \mathsf{resume} $ and, therefore, $\mathit{y}$ is given type
$\alpha$, not $\beta$.
This inconvenience may be addressed by moving down the let-binding
in some cases: e.g., $\mathsf{resume} \, \ottsym{(}  \mathsf{let} \, \mathit{y}  \ottsym{=}   \mathsf{fst}  \, \mathit{x} \,  \mathsf{in}  \, \mathit{y}  \ottsym{)}$ is well
typed.

% typechecking the result of expanding pure
% expressions bound to variables---e.g., $\mathsf{let} \, \mathit{y}  \ottsym{=}   \mathsf{fst}  \, \mathit{x} \,  \mathsf{in}  \, \mathsf{resume} \, \mathit{y}$ could be
% translated to $\mathsf{resume} \, \ottsym{(}   \mathsf{fst}  \, \mathit{x}  \ottsym{)}$ safely---but it is left as future work to
% solve it completely.

\section{Intermediate language: {\interlang}}
\label{sec:interlang}

The semantics of {\surfacelang} is given by a formal elaboration to an
intermediate language {\interlang}, wherein type abstraction and type
application appear explicitly.
We define the syntax, operational semantics, and type system of {\interlang}
and the formal elaboration from {\surfacelang} to {\interlang}.
Finally, we show type safety of {\surfacelang} via type preservation of the
elaboration and type soundness of {\interlang}.

\subsection{Syntax}

\begin{figure}[t]
\[
  \begin{array}{l@{\quad}l@{\quad}r@{\quad}l}
   \textbf{Values} & \ottnt{v} & ::= & \ottnt{c} \mid  \lambda\!  \, \mathit{x}  \ottsym{.}  \ottnt{e} \\[.5ex]
   \textbf{Polymorphic values} & \ottnt{w} & ::= & \ottnt{v} \mid  \Lambda\!  \, \alpha  \ottsym{.}  \ottnt{w} \\[.5ex]
   \textbf{Terms} & \ottnt{e} & ::= &
     \mathit{x} \,  \algeffseqover{ \ottnt{A} }  \mid \ottnt{c} \mid  \lambda\!  \, \mathit{x}  \ottsym{.}  \ottnt{e} \mid \ottnt{e_{{\mathrm{1}}}} \, \ottnt{e_{{\mathrm{2}}}} \mid
     \mathsf{let} \, \mathit{x}  \ottsym{=}   \Lambda\!  \,  \algeffseqover{ \alpha }   \ottsym{.}  \ottnt{e_{{\mathrm{1}}}} \,  \mathsf{in}  \, \ottnt{e_{{\mathrm{2}}}} \mid \\ &&&
      \textup{\texttt{\#}\relax}  \mathsf{op}   \ottsym{(}    \algeffseqover{ \ottnt{A} }    \ottsym{,}   \ottnt{e}   \ottsym{)}  \mid
      \textup{\texttt{\#}\relax}  \mathsf{op}   \ottsym{(}    \algeffseqover{ \sigma }    \ottsym{,}   \ottnt{w}   \ottsym{,}   \ottnt{E}   \ottsym{)}  \mid
     \mathsf{handle} \, \ottnt{e} \, \mathsf{with} \, \ottnt{h} \mid \\ &&&
     \mathsf{resume} \,  \algeffseqover{ \alpha }  \, \mathit{x}  \ottsym{.}  \ottnt{e}
     \\[.5ex]
   \textbf{Handlers} & \ottnt{h} & ::= &
     \mathsf{return} \, \mathit{x}  \rightarrow  \ottnt{e} \mid
     \ottnt{h}  \ottsym{;}   \Lambda\!  \,  \algeffseqover{ \alpha }   \ottsym{.}  \mathsf{op}  \ottsym{(}  \mathit{x}  \ottsym{)}  \rightarrow  \ottnt{e} \\[.5ex]
   \textbf{Evaluation contexts} &
     \ottnt{E} ^{  \algeffseqoverindex{ \alpha }{ \text{\unboldmath$\mathit{I}$} }  }  & ::= &  [\,]  \ \text{(if $ \algeffseqoverindex{ \alpha }{ \text{\unboldmath$\mathit{I}$} }  \,  =  \,  \emptyset $)} \mid
                   \ottnt{E} ^{  \algeffseqoverindex{ \alpha }{ \text{\unboldmath$\mathit{I}$} }  }  \, \ottnt{e_{{\mathrm{2}}}} \mid \ottnt{v_{{\mathrm{1}}}} \,  \ottnt{E} ^{  \algeffseqoverindex{ \alpha }{ \text{\unboldmath$\mathit{I}$} }  }  \mid \\ &&&
                  \mathsf{let} \, \mathit{x}  \ottsym{=}   \Lambda\!  \,  \algeffseqoverindex{ \beta }{ \text{\unboldmath$\mathit{J_{{\mathrm{1}}}}$} }   \ottsym{.}   \ottnt{E} ^{  \algeffseqoverindex{ \gamma }{ \text{\unboldmath$\mathit{J_{{\mathrm{2}}}}$} }  }  \,  \mathsf{in}  \, \ottnt{e_{{\mathrm{2}}}} \ \text{(if $ \algeffseqoverindex{ \alpha }{ \text{\unboldmath$\mathit{I}$} }  \,  =  \,  \algeffseqoverindex{ \beta }{ \text{\unboldmath$\mathit{J_{{\mathrm{1}}}}$} }   \ottsym{,}   \algeffseqoverindex{ \gamma }{ \text{\unboldmath$\mathit{J_{{\mathrm{2}}}}$} } $)} \mid \\ &&&
                   \textup{\texttt{\#}\relax}  \mathsf{op}   \ottsym{(}    \algeffseqoverindex{ \ottnt{A} }{ \text{\unboldmath$\mathit{J}$} }    \ottsym{,}    \ottnt{E} ^{  \algeffseqoverindex{ \alpha }{ \text{\unboldmath$\mathit{I}$} }  }    \ottsym{)}  \mid \mathsf{handle} \,  \ottnt{E} ^{  \algeffseqoverindex{ \alpha }{ \text{\unboldmath$\mathit{I}$} }  }  \, \mathsf{with} \, \ottnt{h} \\[.5ex]
   % \textbf{Top-level typing contexts} &
   %  \Delta & ::= &  \emptyset  \mid \Delta  \ottsym{,}  \alpha
  \end{array}
\]
 \caption{Syntax of {\interlang}.}
 \label{fig:inter-syntax}
\end{figure}

The syntax of {\interlang} is shown in \reffig{inter-syntax}.
Values, denoted by $\ottnt{v}$, consist of constants and lambda abstractions.
Polymorphic values, denoted by $\ottnt{w}$, are values abstracted over types.
Terms, denoted by $\ottnt{e}$, and handlers, denoted by $\ottnt{h}$, are the same as those of {\surfacelang} except for the following three points.
First, type abstraction and type arguments are explicit in
{\interlang}: variables and effect invocations are accompanied by a sequence of
types and let-bound expressions, resumption expressions, and
operation clauses bind type variables.
Second, a new term constructor of the form $ \textup{\texttt{\#}\relax}  \mathsf{op}   \ottsym{(}    \algeffseqover{ \sigma }    \ottsym{,}   \ottnt{w}   \ottsym{,}   \ottnt{E}   \ottsym{)} $ is
added.  It represents an intermediate state in which an effect
invocation is capturing the continuation up to the closest handler for
$\mathsf{op}$.
Here, $\ottnt{E}$ is an evaluation context~\cite{Felleisen/Hieb_1992_TCS} and denotes a
continuation to be resumed by an operation clause handling $\mathsf{op}$.
In the operational semantics, an operation invocation $ \textup{\texttt{\#}\relax}  \mathsf{op}   \ottsym{(}    \algeffseqover{ \ottnt{A} }    \ottsym{,}   \ottnt{v}   \ottsym{)} $
is first transformed to $ \textup{\texttt{\#}\relax}  \mathsf{op}   \ottsym{(}    \algeffseqover{ \ottnt{A} }    \ottsym{,}   \ottnt{v}   \ottsym{,}    [\,]    \ottsym{)} $ (where $ [\,] $ denotes the
empty context or the identity continuation) and then it bubbles up
by  capturing its context and pushing it onto the third argument.
Note that $ \algeffseqover{ \sigma } $ and $\ottnt{w}$ of $ \textup{\texttt{\#}\relax}  \mathsf{op}   \ottsym{(}    \algeffseqover{ \sigma }    \ottsym{,}   \ottnt{w}   \ottsym{,}   \ottnt{E}   \ottsym{)} $ become
polymorphic when it bubbles up from the body of a type abstraction.
Third, each resumption expression $\mathsf{resume} \,  \algeffseqover{ \alpha }  \, \mathit{x}  \ottsym{.}  \ottnt{e}$ declares
distinct (type) variables $ \algeffseqover{ \alpha } $ and $\mathit{x}$ to denote the (type)
argument to an operation clause, whereas a single variable
declared at $\mathsf{op}  \ottsym{(}  \mathit{x}  \ottsym{)}  \rightarrow  \ottnt{M}$ and implicit type variables are used
for the same purpose in {\surfacelang}.
For example, the {\surfacelang} operation clause $ \mathsf{choose}^{\forall}   \ottsym{(}  \mathit{x}  \ottsym{)}  \rightarrow  \mathsf{resume} \, \ottsym{(}   \mathsf{fst}  \, \mathit{x}  \ottsym{)}$
is translated to $ \Lambda\!  \, \alpha  \ottsym{.}   \mathsf{choose}^{\forall}   \ottsym{(}  \mathit{x}  \ottsym{)}  \rightarrow  \mathsf{resume} \, \beta \, \mathit{y}  \ottsym{.}  \ottsym{(}   \mathsf{fst}  \, \mathit{y}  \ottsym{)}$.
This change simplifies the semantics.

Evaluation contexts, denoted by $ \ottnt{E} ^{  \algeffseqover{ \alpha }  } $, are standard for the lambda calculus
with call-by-value, left-to-right evaluation except for two points.
First, they contain the form $\mathsf{let} \, \mathit{x}  \ottsym{=}   \Lambda\!  \,  \algeffseqover{ \alpha }   \ottsym{.}   \ottnt{E} ^{  \algeffseqover{ \beta }  }  \,  \mathsf{in}  \, \ottnt{e_{{\mathrm{2}}}}$, which
allows the body of a type abstraction to be evaluated.
%
% While such evaluation is not always safe in the presence of
% polymorphic effects, we show that it is safe if resumptions in a
% handler do no interfere with each other.
%
Second, the metavariable $\ottnt{E}$ for evaluation contexts is indexed
by type variables $ \algeffseqover{ \alpha } $, meaning that the hole in the context
appears under type abstractions binding $ \algeffseqover{ \alpha } $.
For example, $\mathsf{let} \, \mathit{x}  \ottsym{=}   \Lambda\!  \, \alpha  \ottsym{.}  \mathsf{let} \, \mathit{y}  \ottsym{=}   \Lambda\!  \, \beta  \ottsym{.}   [\,]  \,  \mathsf{in}  \, \ottnt{e_{{\mathrm{2}}}} \,  \mathsf{in}  \, \ottnt{e_{{\mathrm{1}}}}$ is denoted
by $\ottnt{E}^{\alpha,\beta}$ and, more generally, $\mathsf{let} \, \mathit{x}  \ottsym{=}   \Lambda\!  \,  \algeffseqoverindex{ \beta }{ \text{\unboldmath$\mathit{J_{{\mathrm{1}}}}$} }   \ottsym{.}   \ottnt{E} ^{  \algeffseqoverindex{ \gamma }{ \text{\unboldmath$\mathit{J_{{\mathrm{2}}}}$} }  }  \,  \mathsf{in}  \, \ottnt{e}$ is denoted by
$\ottnt{E}^{ \algeffseqoverindex{ \beta }{ \text{\unboldmath$\mathit{J_{{\mathrm{1}}}}$} } , \algeffseqoverindex{ \gamma }{ \text{\unboldmath$\mathit{J_{{\mathrm{2}}}}$} } }$.  (Here, $ \algeffseqoverindex{ \beta }{ \text{\unboldmath$\mathit{J_{{\mathrm{1}}}}$} }   \ottsym{,}   \algeffseqoverindex{ \gamma }{ \text{\unboldmath$\mathit{J_{{\mathrm{2}}}}$} } $ stands for the
concatenation of the two sequences $ \algeffseqoverindex{ \beta }{ \text{\unboldmath$\mathit{J_{{\mathrm{1}}}}$} } $ and $ \algeffseqoverindex{ \gamma }{ \text{\unboldmath$\mathit{J_{{\mathrm{2}}}}$} } $.)  If
$ \algeffseqover{ \alpha } $ is not important, we simply write $\ottnt{E}$ for $ \ottnt{E} ^{  \algeffseqover{ \alpha }  } $.
We often use the term ``continuation'' to mean ``evaluation context,''
especially when it is expected to be resumed.

As usual, substitution $ \ottnt{e}    [  \ottnt{w}  \ottsym{/}  \mathit{x}  ]  $ of $\ottnt{w}$ for $\mathit{x}$ in $\ottnt{e}$ is defined
in a capture-avoiding manner.
Since variables come along with type arguments, the case for variables is
defined as follows:
\[
  \ottsym{(}  \mathit{x} \,  \algeffseqover{ \ottnt{A} }   \ottsym{)}    [   \Lambda\!  \,  \algeffseqover{ \alpha }   \ottsym{.}  \ottnt{v}  \ottsym{/}  \mathit{x}  ]   \defeq  \ottnt{v}    [   \algeffseqover{ \ottnt{A} }   \ottsym{/}   \algeffseqover{ \alpha }   ]  
\]
Application of substitution $ [   \Lambda\!  \,  \algeffseqoverindex{ \alpha }{ \text{\unboldmath$\mathit{I}$} }   \ottsym{.}  \ottnt{v}  \ottsym{/}  \mathit{x}  ] $ to $\mathit{x} \,  \algeffseqoverindex{ \ottnt{A} }{ \text{\unboldmath$\mathit{J}$} } $, where $\text{\unboldmath$\mathit{I}$}
\neq \text{\unboldmath$\mathit{J}$}$, is undefined.
We define free type variables $ \mathit{ftv}  (  \ottnt{e}  ) $ and $ \mathit{ftv}  (  \ottnt{E}  ) $ in $\ottnt{e}$ and
$\ottnt{E}$, respectively, as usual.

\subsection{Semantics}
\label{sec:inter-semantics}

\begin{figure}[t]
 \textbf{Reduction rules} \quad \framebox{$\ottnt{e_{{\mathrm{1}}}}  \rightsquigarrow  \ottnt{e_{{\mathrm{2}}}}$}
 \[\begin{array}{rcll@{\quad}rcll}
  \ottnt{c_{{\mathrm{1}}}} \, \ottnt{c_{{\mathrm{2}}}}               &  \rightsquigarrow  &  \zeta  (  \ottnt{c_{{\mathrm{1}}}}  ,  \ottnt{c_{{\mathrm{2}}}}  )  & \R{Const} &
  \ottsym{(}   \lambda\!  \, \mathit{x}  \ottsym{.}  \ottnt{e}  \ottsym{)} \, \ottnt{v}            &  \rightsquigarrow  &  \ottnt{e}    [  \ottnt{v}  \ottsym{/}  \mathit{x}  ]        & \R{Beta} \\[1.5ex]
  \multirow{2}{*}{\ensuremath{\mathsf{let} \, \mathit{x}  \ottsym{=}   \Lambda\!  \,  \algeffseqover{ \alpha }   \ottsym{.}  \ottnt{v} \,  \mathsf{in}  \, \ottnt{e}}} & \multirow{2}{*}{$ \rightsquigarrow $} & \multirow{2}{*}{$ \ottnt{e}    [   \Lambda\!  \,  \algeffseqover{ \alpha }   \ottsym{.}  \ottnt{v}  \ottsym{/}  \mathit{x}  ]  $} & \multirow{2}{*}{\R{Let}} &
  \mathsf{handle} \, \ottnt{v} \, \mathsf{with} \, \ottnt{h}     &  \rightsquigarrow  &  \ottnt{e}    [  \ottnt{v}  \ottsym{/}  \mathit{x}  ]   & \R{Return} \\
    &&&& \multicolumn{4}{r}{\text{(where $ \ottnt{h} ^\mathsf{return}  \,  =  \, \mathsf{return} \, \mathit{x}  \rightarrow  \ottnt{e}$)}} \\[1.5ex]
   \textup{\texttt{\#}\relax}  \mathsf{op}   \ottsym{(}    \algeffseqover{ \ottnt{A} }    \ottsym{,}   \ottnt{v}   \ottsym{)}            &  \rightsquigarrow  &  \textup{\texttt{\#}\relax}  \mathsf{op}   \ottsym{(}    \algeffseqover{ \ottnt{A} }    \ottsym{,}   \ottnt{v}   \ottsym{,}    [\,]    \ottsym{)}  & \R{Op} \\
   \end{array}\]
 \[\begin{array}{rcl@{\ \ }l}
   \textup{\texttt{\#}\relax}  \mathsf{op}   \ottsym{(}    \algeffseqover{ \sigma }    \ottsym{,}   \ottnt{w}   \ottsym{,}   \ottnt{E}   \ottsym{)}  \, \ottnt{e_{{\mathrm{2}}}}   &  \rightsquigarrow  &  \textup{\texttt{\#}\relax}  \mathsf{op}   \ottsym{(}    \algeffseqover{ \sigma }    \ottsym{,}   \ottnt{w}   \ottsym{,}   \ottnt{E} \, \ottnt{e_{{\mathrm{2}}}}   \ottsym{)}  & \R{OpApp1} \\[1ex]
  \ottnt{v_{{\mathrm{1}}}} \,  \textup{\texttt{\#}\relax}  \mathsf{op}   \ottsym{(}    \algeffseqover{ \sigma }    \ottsym{,}   \ottnt{w}   \ottsym{,}   \ottnt{E}   \ottsym{)}    &  \rightsquigarrow  &  \textup{\texttt{\#}\relax}  \mathsf{op}   \ottsym{(}    \algeffseqover{ \sigma }    \ottsym{,}   \ottnt{w}   \ottsym{,}   \ottnt{v_{{\mathrm{1}}}} \, \ottnt{E}   \ottsym{)}  & \R{OpApp2} \\[1ex]
   \textup{\texttt{\#}\relax}  \mathsf{op}'   \ottsym{(}    \algeffseqoverindex{ \ottnt{A} }{ \text{\unboldmath$\mathit{I}$} }    \ottsym{,}    \textup{\texttt{\#}\relax}  \mathsf{op}   \ottsym{(}    \algeffseqoverindex{ \sigma }{ \text{\unboldmath$\mathit{J}$} }    \ottsym{,}   \ottnt{w}   \ottsym{,}   \ottnt{E}   \ottsym{)}    \ottsym{)}  &  \rightsquigarrow  &  \textup{\texttt{\#}\relax}  \mathsf{op}   \ottsym{(}    \algeffseqoverindex{ \sigma }{ \text{\unboldmath$\mathit{J}$} }    \ottsym{,}   \ottnt{w}   \ottsym{,}    \textup{\texttt{\#}\relax}  \mathsf{op}'   \ottsym{(}    \algeffseqoverindex{ \ottnt{A} }{ \text{\unboldmath$\mathit{I}$} }    \ottsym{,}   \ottnt{E}   \ottsym{)}    \ottsym{)}  & \R{OpOp} \\[1ex]
  \mathsf{handle} \,  \textup{\texttt{\#}\relax}  \mathsf{op}   \ottsym{(}    \algeffseqover{ \sigma }    \ottsym{,}   \ottnt{w}   \ottsym{,}   \ottnt{E}   \ottsym{)}  \, \mathsf{with} \, \ottnt{h} &  \rightsquigarrow  &  \textup{\texttt{\#}\relax}  \mathsf{op}   \ottsym{(}    \algeffseqover{ \sigma }    \ottsym{,}   \ottnt{w}   \ottsym{,}   \mathsf{handle} \, \ottnt{E} \, \mathsf{with} \, \ottnt{h}   \ottsym{)}  & \multirow{2}{*}{\R{OpHandle}} \\
    && \multicolumn{1}{r@{\quad}}{\text{(where $\mathsf{op} \,  \not\in  \,  \mathit{ops}  (  \ottnt{h}  ) $)}} \\[1ex]
  \mathsf{let} \, \mathit{x}  \ottsym{=}   \Lambda\!  \,  \algeffseqoverindex{ \alpha }{ \text{\unboldmath$\mathit{I}$} }   \ottsym{.}   \textup{\texttt{\#}\relax}  \mathsf{op}   \ottsym{(}    \algeffseqoverindex{ \sigma }{ \text{\unboldmath$\mathit{J}$} }    \ottsym{,}   \ottnt{w}   \ottsym{,}   \ottnt{E}   \ottsym{)}  \,  \mathsf{in}  \, \ottnt{e_{{\mathrm{2}}}} &  \rightsquigarrow  & & \multirow{2}{*}{\R{OpLet}} \\
    \multicolumn{3}{r@{\quad}}{ \textup{\texttt{\#}\relax}  \mathsf{op}   \ottsym{(}    \algeffseqoverindex{  \text{\unboldmath$\forall$}  \,  \algeffseqoverindex{ \alpha }{ \text{\unboldmath$\mathit{I}$} }   \ottsym{.}  \sigma }{ \text{\unboldmath$\mathit{J}$} }    \ottsym{,}    \Lambda\!  \,  \algeffseqoverindex{ \alpha }{ \text{\unboldmath$\mathit{I}$} }   \ottsym{.}  \ottnt{w}   \ottsym{,}   \mathsf{let} \, \mathit{x}  \ottsym{=}   \Lambda\!  \,  \algeffseqoverindex{ \alpha }{ \text{\unboldmath$\mathit{I}$} }   \ottsym{.}  \ottnt{E} \,  \mathsf{in}  \, \ottnt{e_{{\mathrm{2}}}}   \ottsym{)} } \\[1.5ex]
  \mathsf{handle} \,  \textup{\texttt{\#}\relax}  \mathsf{op}   \ottsym{(}    \algeffseqoverindex{  \text{\unboldmath$\forall$}  \,  \algeffseqoverindex{ \beta }{ \text{\unboldmath$\mathit{J}$} }   \ottsym{.}  \ottnt{A} }{ \text{\unboldmath$\mathit{I}$} }    \ottsym{,}    \Lambda\!  \,  \algeffseqoverindex{ \beta }{ \text{\unboldmath$\mathit{J}$} }   \ottsym{.}  \ottnt{v}   \ottsym{,}    \ottnt{E} ^{  \algeffseqoverindex{ \beta }{ \text{\unboldmath$\mathit{J}$} }  }    \ottsym{)}  \, \mathsf{with} \, \ottnt{h} &  \rightsquigarrow  & \\ \multicolumn{3}{r@{\quad}}{
      \ottnt{e}    [  \mathsf{handle} \,  \ottnt{E} ^{  \algeffseqoverindex{ \beta }{ \text{\unboldmath$\mathit{J}$} }  }  \, \mathsf{with} \, \ottnt{h}  /  \mathsf{resume}  ]^{  \algeffseqoverindex{  \text{\unboldmath$\forall$}  \,  \algeffseqoverindex{ \beta }{ \text{\unboldmath$\mathit{J}$} }   \ottsym{.}  \ottnt{A} }{ \text{\unboldmath$\mathit{I}$} }  }_{  \Lambda\!  \,  \algeffseqoverindex{ \beta }{ \text{\unboldmath$\mathit{J}$} }   \ottsym{.}  \ottnt{v} }      [    \algeffseqoverindex{ \ottnt{A} }{ \text{\unboldmath$\mathit{I}$} }     [   \algeffseqover{  \bot  }   \ottsym{/}   \algeffseqoverindex{ \beta }{ \text{\unboldmath$\mathit{J}$} }   ]    \ottsym{/}   \algeffseqoverindex{ \alpha }{ \text{\unboldmath$\mathit{I}$} }   ]      [   \ottnt{v}    [   \algeffseqover{  \bot  }   \ottsym{/}   \algeffseqoverindex{ \beta }{ \text{\unboldmath$\mathit{J}$} }   ]    \ottsym{/}  \mathit{x}  ]  } & \R{Handle} \\
   \multicolumn{3}{r@{\quad}}{\text{(where $ \ottnt{h} ^{ \mathsf{op} }  \,  =  \,  \Lambda\!  \,  \algeffseqoverindex{ \alpha }{ \text{\unboldmath$\mathit{I}$} }   \ottsym{.}  \mathsf{op}  \ottsym{(}  \mathit{x}  \ottsym{)}  \rightarrow  \ottnt{e}$)}}
   \end{array}\]
 \textbf{Evaluation rules} \quad \framebox{$\ottnt{e_{{\mathrm{1}}}}  \longrightarrow  \ottnt{e_{{\mathrm{2}}}}$}
 \begin{center}
  $\ottdruleEXXEval{}$
 \end{center}
 \caption{Semantics of {\interlang}.}
 \label{fig:inter-semantics}
\end{figure}

The semantics of {\interlang} is given in the small-step style and consists of two
relations: the reduction relation $ \rightsquigarrow $, which is for basic computation, and
the evaluation relation $ \longrightarrow $, which is for top-level execution.
\reffig{inter-semantics} shows the rules for these relations.
In what follows, we write $ \ottnt{h} ^\mathsf{return} $ for the return clause of handler $\ottnt{h}$,
$ \mathit{ops}  (  \ottnt{h}  ) $ for the set of effect operations handled by $\ottnt{h}$, and
$ \ottnt{h} ^{ \mathsf{op} } $ for the operation clause for $\mathsf{op}$ in $\ottnt{h}$.

Most of the reduction rules are standard~\cite{Kammar/Lindley/Oury_2013_ICFP,Leijen_2017_POPL}.
A constant application $\ottnt{c_{{\mathrm{1}}}} \, \ottnt{c_{{\mathrm{2}}}}$ reduces to $ \zeta  (  \ottnt{c_{{\mathrm{1}}}}  ,  \ottnt{c_{{\mathrm{2}}}}  ) $ \R{Const}, where
function $ \zeta $ maps a pair of constants to another constant.
A function application $\ottsym{(}   \lambda\!  \, \mathit{x}  \ottsym{.}  \ottnt{e}  \ottsym{)} \, \ottnt{v}$ and
a let-expression $\mathsf{let} \, \mathit{x}  \ottsym{=}   \Lambda\!  \,  \algeffseqover{ \alpha }   \ottsym{.}  \ottnt{v} \,  \mathsf{in}  \, \ottnt{e}$ reduce to
$ \ottnt{e}    [  \ottnt{v}  \ottsym{/}  \mathit{x}  ]  $ \R{Beta} and
$ \ottnt{e}    [   \Lambda\!  \,  \algeffseqover{ \alpha }   \ottsym{.}  \ottnt{v}  \ottsym{/}  \mathit{x}  ]  $ \R{Let}, respectively.
If a handled expression is a value $\ottnt{v}$, the $ \mathsf{handle} $--$ \mathsf{with} $ expression
reduces to the body of the return clause where $\ottnt{v}$ is substituted
for the parameter $\mathit{x}$ \R{Return}.
An effect invocation $ \textup{\texttt{\#}\relax}  \mathsf{op}   \ottsym{(}    \algeffseqover{ \ottnt{A} }    \ottsym{,}   \ottnt{v}   \ottsym{)} $ reduces to $ \textup{\texttt{\#}\relax}  \mathsf{op}   \ottsym{(}    \algeffseqover{ \ottnt{A} }    \ottsym{,}   \ottnt{v}   \ottsym{,}    [\,]    \ottsym{)} $
with the identity continuation, as explained above \R{Op}; the process
of capturing its evaluation context is expressed by the rules
\R{OpApp1}, \R{OpApp2}, \R{OpOp}, \R{OpHandle}, and \R{OpLet}.
The rule \R{OpHandle} can be applied only if the handler $\ottnt{h}$ does \emph{not}
handle $\mathsf{op}$.
The rule \R{OpLet} is applied to a let-expression where $ \textup{\texttt{\#}\relax}  \mathsf{op}   \ottsym{(}    \algeffseqoverindex{ \sigma }{ \text{\unboldmath$\mathit{J}$} }    \ottsym{,}   \ottnt{w}   \ottsym{,}   \ottnt{E}   \ottsym{)} $ appears under
a type abstraction with bound type variables $ \algeffseqoverindex{ \alpha }{ \text{\unboldmath$\mathit{I}$} } $.
Since $ \algeffseqoverindex{ \sigma }{ \text{\unboldmath$\mathit{J}$} } $ and $\ottnt{w}$ may refer to $ \algeffseqoverindex{ \alpha }{ \text{\unboldmath$\mathit{I}$} } $, the reduction
result binds $ \algeffseqoverindex{ \alpha }{ \text{\unboldmath$\mathit{I}$} } $ in both $ \algeffseqoverindex{ \sigma }{ \text{\unboldmath$\mathit{J}$} } $ and $\ottnt{w}$.
We write $ \algeffseqoverindex{  \text{\unboldmath$\forall$}  \,  \algeffseqoverindex{ \alpha }{ \text{\unboldmath$\mathit{I}$} }   \ottsym{.}  \sigma }{ \text{\unboldmath$\mathit{J}$} } $ for a sequence $ \text{\unboldmath$\forall$}  \,  \algeffseqoverindex{ \alpha }{ \text{\unboldmath$\mathit{I}$} }   \ottsym{.}  \sigma_{j_1}$,
\ldots, $ \text{\unboldmath$\forall$}  \,  \algeffseqoverindex{ \alpha }{ \text{\unboldmath$\mathit{I}$} }   \ottsym{.}  \sigma_{j_n}$ of type schemes (where
$\text{\unboldmath$\mathit{J}$} = \{j_1,\ldots,j_n\}$).

The crux of the semantics is \R{Handle}: it is applied when
$ \textup{\texttt{\#}\relax}  \mathsf{op}   \ottsym{(}    \algeffseqoverindex{ \sigma }{ \text{\unboldmath$\mathit{I}$} }    \ottsym{,}   \ottnt{w}   \ottsym{,}   \ottnt{E}   \ottsym{)} $ reaches the handler $\ottnt{h}$
that handles $\mathsf{op}$.
Since the handled term $ \textup{\texttt{\#}\relax}  \mathsf{op}   \ottsym{(}    \algeffseqoverindex{ \sigma }{ \text{\unboldmath$\mathit{I}$} }    \ottsym{,}   \ottnt{w}   \ottsym{,}   \ottnt{E}   \ottsym{)} $ is constructed from an effect
invocation $ \textup{\texttt{\#}\relax}  \mathsf{op}   \ottsym{(}    \algeffseqoverindex{ \ottnt{A} }{ \text{\unboldmath$\mathit{I}$} }    \ottsym{,}   \ottnt{v}   \ottsym{)} $, if the captured continuation $\ottnt{E}$
binds type variables $ \algeffseqoverindex{ \beta }{ \text{\unboldmath$\mathit{J}$} } $, the same type variables $ \algeffseqoverindex{ \beta }{ \text{\unboldmath$\mathit{J}$} } $
should have been added to $ \algeffseqoverindex{ \ottnt{A} }{ \text{\unboldmath$\mathit{I}$} } $ and $\ottnt{v}$ along the capture.
Thus, the handled expression on the left-hand side of the rule takes the
form $ \textup{\texttt{\#}\relax}  \mathsf{op}   \ottsym{(}    \algeffseqoverindex{  \text{\unboldmath$\forall$}  \,  \algeffseqoverindex{ \beta }{ \text{\unboldmath$\mathit{J}$} }   \ottsym{.}  \ottnt{A} }{ \text{\unboldmath$\mathit{I}$} }    \ottsym{,}    \Lambda\!  \,  \algeffseqoverindex{ \beta }{ \text{\unboldmath$\mathit{J}$} }   \ottsym{.}  \ottnt{v}   \ottsym{,}    \ottnt{E} ^{  \algeffseqoverindex{ \beta }{ \text{\unboldmath$\mathit{J}$} }  }    \ottsym{)} $ (with the same type variables
$ \algeffseqoverindex{ \beta }{ \text{\unboldmath$\mathit{J}$} } $).

The right-hand side of \R{Handle} involves three types of substitution: continuation
substitution
$ [  \mathsf{handle} \,  \ottnt{E} ^{  \algeffseqoverindex{ \beta }{ \text{\unboldmath$\mathit{J}$} }  }  \, \mathsf{with} \, \ottnt{h}  /  \mathsf{resume}  ]^{  \algeffseqoverindex{  \text{\unboldmath$\forall$}  \,  \algeffseqoverindex{ \beta }{ \text{\unboldmath$\mathit{J}$} }   \ottsym{.}  \ottnt{A} }{ \text{\unboldmath$\mathit{I}$} }  }_{  \Lambda\!  \,  \algeffseqoverindex{ \beta }{ \text{\unboldmath$\mathit{J}$} }   \ottsym{.}  \ottnt{v} } $
for resumptions, type substitution for $ \algeffseqoverindex{ \alpha }{ \text{\unboldmath$\mathit{I}$} } $, and value
substitution for $\mathit{x}$.
We explain them one by one below.
In the following, let $ \ottnt{h} ^{ \mathsf{op} }  \,  =  \,  \Lambda\!  \,  \algeffseqoverindex{ \alpha }{ \text{\unboldmath$\mathit{I}$} }   \ottsym{.}  \mathsf{op}  \ottsym{(}  \mathit{x}  \ottsym{)}  \rightarrow  \ottnt{e}$
and $ \ottnt{E'} ^{  \algeffseqoverindex{ \beta }{ \text{\unboldmath$\mathit{J}$} }  }  \,  =  \, \mathsf{handle} \,  \ottnt{E} ^{  \algeffseqoverindex{ \beta }{ \text{\unboldmath$\mathit{J}$} }  }  \, \mathsf{with} \, \ottnt{h}$.
\paragraph{Continuation substitution.}
%
%  the continuation $ \ottnt{E} ^{  \algeffseqoverindex{ \beta }{ \text{\unboldmath$\mathit{J}$} }  } $ is handled
% by $\ottnt{h}$ because our calculus supports deep handlers~\cite{Kammar/Lindley/Oury_2013_ICFP}.
%
Let us start with a simple case where the sequence $ \algeffseqoverindex{ \beta }{ \text{\unboldmath$\mathit{J}$} } $ is empty.
Intuitively, continuation substitution $ [  \ottnt{E'}  /  \mathsf{resume}  ]^{  \algeffseqoverindex{ \ottnt{A} }{ \text{\unboldmath$\mathit{I}$} }  }_{ \ottnt{v} } $
replaces a resumption expression $\mathsf{resume} \,  \algeffseqoverindex{ \gamma }{ \text{\unboldmath$\mathit{I}$} }  \, \mathit{z}  \ottsym{.}  \ottnt{e'}$ in the body $\ottnt{e}$ with
$ \ottnt{E'}  [  \ottnt{v'}  ] $, where $\ottnt{v'}$ is the value of $\ottnt{e'}$, and
substitutes $ \algeffseqoverindex{ \ottnt{A} }{ \text{\unboldmath$\mathit{I}$} } $ and $\ottnt{v}$ (arguments to the invocation of $\mathsf{op}$)
for $ \algeffseqoverindex{ \gamma }{ \text{\unboldmath$\mathit{I}$} } $ and $\mathit{z}$, respectively.
Therefore, assuming $ \mathsf{resume} $ does not appear in $\ottnt{e'}$,
we define $ \ottsym{(}  \mathsf{resume} \,  \algeffseqoverindex{ \gamma }{ \text{\unboldmath$\mathit{I}$} }  \, \mathit{z}  \ottsym{.}  \ottnt{e'}  \ottsym{)}    [  \ottnt{E'}  /  \mathsf{resume}  ]^{  \algeffseqoverindex{ \ottnt{A} }{ \text{\unboldmath$\mathit{I}$} }  }_{ \ottnt{v} }  $ to be
$\mathsf{let} \, \mathit{y}  \ottsym{=}    \ottnt{e'}    [   \algeffseqoverindex{ \ottnt{A} }{ \text{\unboldmath$\mathit{I}$} }   \ottsym{/}   \algeffseqoverindex{ \gamma }{ \text{\unboldmath$\mathit{I}$} }   ]      [  \ottnt{v}  \ottsym{/}  \mathit{z}  ]   \,  \mathsf{in}  \,  \ottnt{E'}  [  \mathit{y}  ] $ (for fresh $\mathit{y}$).  Note
that the evaluation of $\ottnt{e'}$ takes place outside of $\ottnt{E}$ so that
an invocation of an effect in $\ottnt{e'}$ is \emph{not} handled by handlers in $\ottnt{E}$.
When $ \algeffseqoverindex{ \beta }{ \text{\unboldmath$\mathit{J}$} } $ is not empty,
\begin{eqnarray*}
  &&  \ottsym{(}  \mathsf{resume} \,  \algeffseqoverindex{ \gamma }{ \text{\unboldmath$\mathit{I}$} }  \, \mathit{z}  \ottsym{.}  \ottnt{e'}  \ottsym{)}    [   \ottnt{E} ^{  \algeffseqoverindex{ \beta }{ \text{\unboldmath$\mathit{J}$} }  }   /  \mathsf{resume}  ]^{  \algeffseqoverindex{  \text{\unboldmath$\forall$}  \,  \algeffseqoverindex{ \beta }{ \text{\unboldmath$\mathit{J}$} }   \ottsym{.}  \ottnt{A} }{ \text{\unboldmath$\mathit{I}$} }  }_{  \Lambda\!  \,  \algeffseqoverindex{ \beta }{ \text{\unboldmath$\mathit{J}$} }   \ottsym{.}  \ottnt{v} }   \defeq \\
  && \qquad \mathsf{let} \, \mathit{y}  \ottsym{=}   \graybox{  \Lambda\!   \,   \algeffseqoverindex{ \beta }{ \text{\unboldmath$\mathit{J}$} }  .} \,    \ottnt{e'}    [   \algeffseqoverindex{ \ottnt{A} }{ \text{\unboldmath$\mathit{I}$} }   \ottsym{/}   \algeffseqoverindex{ \gamma }{ \text{\unboldmath$\mathit{I}$} }   ]      [  \ottnt{v}  \ottsym{/}  \mathit{z}  ]    \,  \mathsf{in}  \,   \ottnt{E} ^{  \algeffseqoverindex{ \beta }{ \text{\unboldmath$\mathit{J}$} }  }   [  \mathit{y} \,  \graybox{  \algeffseqoverindex{ \beta }{ \text{\unboldmath$\mathit{J}$} }  }\,   ] \ .
\end{eqnarray*}
(The differences from the simple case are shaded.)
The idea is to bind $ \algeffseqoverindex{ \beta }{ \text{\unboldmath$\mathit{J}$} } $ that appear free in
$ \algeffseqoverindex{ \ottnt{A} }{ \text{\unboldmath$\mathit{I}$} } $ and $\ottnt{v}$ by type abstraction at $ \mathsf{let} $ and
to instantiate with the same variables at $\mathit{y} \,  \algeffseqoverindex{ \beta }{ \text{\unboldmath$\mathit{J}$} } $, where
$ \algeffseqoverindex{ \beta }{ \text{\unboldmath$\mathit{J}$} } $ are bound by type abstractions in $ \ottnt{E} ^{  \algeffseqoverindex{ \beta }{ \text{\unboldmath$\mathit{J}$} }  } $.

Continuation substitution is formally defined as follows:

\begin{defn}[Continuation substitution]
 Substitution of continuation $ \ottnt{E} ^{  \algeffseqoverindex{ \beta }{ \text{\unboldmath$\mathit{J}$} }  } $ for resumptions in $\ottnt{e}$, written
 $ \ottnt{e}    [   \ottnt{E} ^{  \algeffseqoverindex{ \beta }{ \text{\unboldmath$\mathit{J}$} }  }   /  \mathsf{resume}  ]^{  \algeffseqoverindex{  \text{\unboldmath$\forall$}  \,  \algeffseqoverindex{ \beta }{ \text{\unboldmath$\mathit{J}$} }   \ottsym{.}  \ottnt{A} }{ \text{\unboldmath$\mathit{I}$} }  }_{  \Lambda\!  \,  \algeffseqoverindex{ \beta }{ \text{\unboldmath$\mathit{J}$} }   \ottsym{.}  \ottnt{v} }  $, is defined in a capture-avoiding
 manner, as follows (we describe only the important cases):
 \[\begin{array}{rcl}
   \ottsym{(}  \mathsf{resume} \,  \algeffseqoverindex{ \gamma }{ \text{\unboldmath$\mathit{I}$} }  \, \mathit{z}  \ottsym{.}  \ottnt{e}  \ottsym{)}    [   \ottnt{E} ^{  \algeffseqoverindex{ \beta }{ \text{\unboldmath$\mathit{J}$} }  }   /  \mathsf{resume}  ]^{  \algeffseqoverindex{  \text{\unboldmath$\forall$}  \,  \algeffseqoverindex{ \beta }{ \text{\unboldmath$\mathit{J}$} }   \ottsym{.}  \ottnt{A} }{ \text{\unboldmath$\mathit{I}$} }  }_{  \Lambda\!  \,  \algeffseqoverindex{ \beta }{ \text{\unboldmath$\mathit{J}$} }   \ottsym{.}  \ottnt{v} }   &\defeq& \\[1ex]
    \multicolumn{3}{r}{
     \mathsf{let} \, \mathit{y}  \ottsym{=}   \Lambda\!  \,  \algeffseqoverindex{ \beta }{ \text{\unboldmath$\mathit{J}$} }   \ottsym{.}     \ottnt{e}    [   \ottnt{E} ^{  \algeffseqoverindex{ \beta }{ \text{\unboldmath$\mathit{J}$} }  }   /  \mathsf{resume}  ]^{  \algeffseqoverindex{  \text{\unboldmath$\forall$}  \,  \algeffseqoverindex{ \beta }{ \text{\unboldmath$\mathit{J}$} }   \ottsym{.}  \ottnt{A} }{ \text{\unboldmath$\mathit{I}$} }  }_{  \Lambda\!  \,  \algeffseqoverindex{ \beta }{ \text{\unboldmath$\mathit{J}$} }   \ottsym{.}  \ottnt{v} }      [   \algeffseqoverindex{ \ottnt{A} }{ \text{\unboldmath$\mathit{I}$} }   \ottsym{/}   \algeffseqoverindex{ \gamma }{ \text{\unboldmath$\mathit{I}$} }   ]      [  \ottnt{v}  \ottsym{/}  \mathit{z}  ]   \,  \mathsf{in}  \,   \ottnt{E} ^{  \algeffseqoverindex{ \beta }{ \text{\unboldmath$\mathit{J}$} }  }   [  \mathit{y} \,  \algeffseqoverindex{ \beta }{ \text{\unboldmath$\mathit{J}$} }   ] 
    } \\[1ex]
    \multicolumn{3}{r}{
     \text{(if $\ottsym{(}   \mathit{ftv}  (  \ottnt{e}  )  \,  \mathbin{\cup}  \,  \mathit{ftv}  (   \ottnt{E} ^{  \algeffseqoverindex{ \beta }{ \text{\unboldmath$\mathit{J}$} }  }   )   \ottsym{)} \,  \mathbin{\cap}  \, \ottsym{\{}   \algeffseqoverindex{ \beta }{ \text{\unboldmath$\mathit{J}$} }   \ottsym{\}} \,  =  \,  \emptyset $ and $\mathit{y}$ is fresh)}
    }
    \\[1ex]
      \ottsym{(}  \mathsf{return} \, \mathit{x}  \rightarrow  \ottnt{e}  \ottsym{)}    [  \ottnt{E}  /  \mathsf{resume}  ]^{  \algeffseqover{ \sigma }  }_{ \ottnt{w} }   &\defeq& \mathsf{return} \, \mathit{x}  \rightarrow   \ottnt{e}    [  \ottnt{E}  /  \mathsf{resume}  ]^{  \algeffseqover{ \sigma }  }_{ \ottnt{w} }   \\[1ex]
     \ottsym{(}  \ottnt{h'}  \ottsym{;}   \Lambda\!  \,  \algeffseqoverindex{ \gamma }{ \text{\unboldmath$\mathit{J}$} }   \ottsym{.}  \mathsf{op}  \ottsym{(}  \mathit{x}  \ottsym{)}  \rightarrow  \ottnt{e}  \ottsym{)}    [  \ottnt{E}  /  \mathsf{resume}  ]^{  \algeffseqoverindex{ \sigma }{ \text{\unboldmath$\mathit{I}$} }  }_{ \ottnt{w} }   &\defeq&  \ottnt{h'}    [  \ottnt{E}  /  \mathsf{resume}  ]^{  \algeffseqoverindex{ \sigma }{ \text{\unboldmath$\mathit{I}$} }  }_{ \ottnt{w} }    \ottsym{;}   \Lambda\!  \,  \algeffseqoverindex{ \gamma }{ \text{\unboldmath$\mathit{J}$} }   \ottsym{.}  \mathsf{op}  \ottsym{(}  \mathit{x}  \ottsym{)}  \rightarrow  \ottnt{e} \\
   \end{array}\]
\end{defn}
The second and third clauses (for a handler) mean that continuation
substitution is applied only to return clauses.

\paragraph{Type and value substitution.}
The type and value substitutions $  \algeffseqoverindex{ \ottnt{A} }{ \text{\unboldmath$\mathit{I}$} }     [   \algeffseqoverindex{  \bot  }{ \text{\unboldmath$\mathit{J}$} }   \ottsym{/}   \algeffseqoverindex{ \beta }{ \text{\unboldmath$\mathit{J}$} }   ]  $ and
$ \ottnt{v}    [   \algeffseqoverindex{  \bot  }{ \text{\unboldmath$\mathit{J}$} }   \ottsym{/}   \algeffseqoverindex{ \beta }{ \text{\unboldmath$\mathit{J}$} }   ]  $, respectively, in \R{Handle} are for (type)
parameters in $ \ottnt{h} ^{ \mathsf{op} }  \,  =  \,  \Lambda\!  \,  \algeffseqoverindex{ \alpha }{ \text{\unboldmath$\mathit{I}$} }   \ottsym{.}  \mathsf{op}  \ottsym{(}  \mathit{x}  \ottsym{)}  \rightarrow  \ottnt{e}$.  The basic idea is
to substitute $ \algeffseqoverindex{ \ottnt{A} }{ \text{\unboldmath$\mathit{I}$} } $ for $ \algeffseqoverindex{ \alpha }{ \text{\unboldmath$\mathit{I}$} } $ and $\ottnt{v}$ for $\mathit{x}$---similarly to continuation substitution.  We erase free type
variables $ \algeffseqoverindex{ \beta }{ \text{\unboldmath$\mathit{J}$} } $ in $ \algeffseqoverindex{ \ottnt{A} }{ \text{\unboldmath$\mathit{I}$} } $ and $\ottnt{v}$ by substituting the
designated base type $ \bot $ for all of them.  (We write
$  \algeffseqoverindex{ \ottnt{A} }{ \text{\unboldmath$\mathit{I}$} }     [   \algeffseqoverindex{  \bot  }{ \text{\unboldmath$\mathit{J}$} }   \ottsym{/}   \algeffseqoverindex{ \beta }{ \text{\unboldmath$\mathit{J}$} }   ]  $ and $ \ottnt{v}    [   \algeffseqoverindex{  \bot  }{ \text{\unboldmath$\mathit{J}$} }   \ottsym{/}   \algeffseqoverindex{ \beta }{ \text{\unboldmath$\mathit{J}$} }   ]  $ for
the types and value, respectively, after the erasure.)

The evaluation rule is ordinary:
Evaluation of a term proceeds by reducing a subterm
under an evaluation context.

\subsection{Type system}

\begin{figure}[t]
 \textbf{Typing rules} \\[1ex]
\framebox{$\Gamma  \ottsym{;}  \ottnt{r} \,   \vdash  \ottnt{e} \,   \ottsym{:}  \ottnt{A} \,  |  \, \epsilon$}
\begin{center}
 $\ottdruleTXXVar{}$ \hfil
 $\ottdruleTXXConst{}$ \\[1.5ex]
 $\ottdruleTXXAbs{}$ \\[1.5ex]
 $\ottdruleTXXApp{}$ \\[1.5ex]
 $\ottdruleTXXOp{}$ \\[1.5ex]
 $\ottdruleTXXOpCont{}$ \\[1.5ex]
 $\ottdruleTXXWeak{}$ \\[1.5ex]
 $\ottdruleTXXHandle{}$ \\[1.5ex]
 $\ottdruleTXXLet{}$ \\[1.5ex]
 $\ottdruleTXXResume{}$ \hfil
\end{center}
 \caption{Typing rules for terms in {\interlang}.}
 \label{fig:inter-typing1}
\end{figure}

\begin{figure}[t]
\framebox{$\Gamma  \ottsym{;}  \ottnt{r} \,   \vdash  \ottnt{h}  \ottsym{:}  \ottnt{A} \,  |  \, \epsilon  \Rightarrow  \ottnt{B} \,  |  \, \epsilon'$}
\begin{center}
 $\ottdruleTHXXReturn{}$ \\[1.5ex]
 $\ottdruleTHXXOp{}$
\end{center}
\framebox{$ \Gamma   \vdash   \ottnt{E}   \ottsym{:}    \sigma  \multimap  \ottnt{A}   \,  |  \, \epsilon $}
\begin{center}
 $\ottdruleTEXXHole{}$ \\[1.5ex]
 $\ottdruleTEXXAppOne{}$ \\[1.5ex]
 $\ottdruleTEXXAppTwo{}$ \\[1.5ex]
 $\ottdruleTEXXOp{}$ \\[1.5ex]
 $\ottdruleTEXXHandle{}$ \\[1.5ex]
 $\ottdruleTEXXWeak{}$ \\[1.5ex]
 $\ottdruleTEXXLet{}$ \\[1.5ex]
\end{center}
 \caption{Typing rules for handlers and continuations in {\interlang}.}
 \label{fig:inter-typing2}
\end{figure}

The type system of {\interlang} is similar to
that of {\surfacelang} and has five judgments:
well-formedness of typing contexts $\vdash  \Gamma$;
well formedness of type schemes $\Gamma  \vdash  \sigma$;
term typing judgment $\Gamma  \ottsym{;}  \ottnt{r} \,   \vdash  \ottnt{e} \,   \ottsym{:}  \ottnt{A} \,  |  \, \epsilon$;
handler typing judgment $\Gamma  \ottsym{;}  \ottnt{r} \,   \vdash  \ottnt{h}  \ottsym{:}  \ottnt{A} \,  |  \, \epsilon  \Rightarrow  \ottnt{B} \,  |  \, \epsilon'$; and
continuation typing judgment $ \Gamma   \vdash   \ottnt{E}   \ottsym{:}     \text{\unboldmath$\forall$}  \,  \algeffseqover{ \alpha }   \ottsym{.}  \ottnt{A}  \multimap  \ottnt{B}   \,  |  \, \epsilon $.
The first two are defined in the same way as those of {\surfacelang}.
The last judgment means
that a term obtained by filling the hole of $\ottnt{E}$ with a term having $\ottnt{A}$
under $\Gamma  \ottsym{,}   \algeffseqover{ \alpha } $ is typed at $\ottnt{B}$ under $\Gamma$ and possibly involves effect
$\epsilon$.
A resumption type $\ottnt{r}$ is similar to $\ottnt{R}$ but does not contain an
argument variable.
\begin{defn}[Resumption type]
 Resumption types in {\interlang}, denoted by $\ottnt{r}$, are defined as
 follows:
 \[\begin{array}{lll}
  \ottnt{r} & ::= &
    \mathsf{none}  \mid
   \ottsym{(}   \algeffseqover{ \alpha }   \ottsym{,}  \ottnt{A}  \ottsym{,}   \ottnt{B}   \rightarrow  \!  \epsilon  \;  \ottnt{C}   \ottsym{)} \\ &&
   \qquad\qquad \text{(if $ \mathit{ftv}  (  \ottnt{A}  )  \,  \mathbin{\cup}  \,  \mathit{ftv}  (  \ottnt{B}  )  \,  \subseteq  \, \ottsym{\{}   \algeffseqover{ \alpha }   \ottsym{\}}$ and $ \mathit{ftv}  (  \ottnt{C}  )  \,  \mathbin{\cap}  \, \ottsym{\{}   \algeffseqover{ \alpha }   \ottsym{\}} \,  =  \,  \emptyset $)}
   \end{array}\]
\end{defn}

The typing rules for terms, shown in \reffig{inter-typing1}, and handlers, shown
in the upper half of \reffig{inter-typing2}, are similar to those of {\surfacelang}
except for a new rule \T{OpCont}, which is applied to an effect invocation $ \textup{\texttt{\#}\relax}  \mathsf{op}   \ottsym{(}    \algeffseqoverindex{  \text{\unboldmath$\forall$}  \,  \algeffseqoverindex{ \beta }{ \text{\unboldmath$\mathit{J}$} }   \ottsym{.}  \ottnt{C} }{ \text{\unboldmath$\mathit{I}$} }    \ottsym{,}    \Lambda\!  \,  \algeffseqoverindex{ \beta }{ \text{\unboldmath$\mathit{J}$} }   \ottsym{.}  \ottnt{v}   \ottsym{,}    \ottnt{E} ^{  \algeffseqoverindex{ \beta }{ \text{\unboldmath$\mathit{J}$} }  }    \ottsym{)} $ with a continuation.
Let $\mathit{ty} \, \ottsym{(}  \mathsf{op}  \ottsym{)} \,  =  \,   \text{\unboldmath$\forall$}     \algeffseqoverindex{ \alpha }{ \text{\unboldmath$\mathit{I}$} }   .  \ottnt{A}  \hookrightarrow  \ottnt{B} $.
Since $\mathsf{op}$ should have been invoked with $ \algeffseqoverindex{ \ottnt{C} }{ \text{\unboldmath$\mathit{I}$} } $ and $\ottnt{v}$ under type
abstractions with bound type variables $ \algeffseqoverindex{ \beta }{ \text{\unboldmath$\mathit{J}$} } $, the argument $\ottnt{v}$ has type
$ \ottnt{A}    [   \algeffseqoverindex{ \ottnt{C} }{ \text{\unboldmath$\mathit{I}$} }   \ottsym{/}   \algeffseqoverindex{ \alpha }{ \text{\unboldmath$\mathit{I}$} }   ]  $ under the typing context extended with $ \algeffseqoverindex{ \beta }{ \text{\unboldmath$\mathit{J}$} } $.
Similarly, the hole of $ \ottnt{E} ^{  \algeffseqoverindex{ \beta }{ \text{\unboldmath$\mathit{J}$} }  } $ expects to be filled with the result of the
invocation, i.e., a value of $ \ottnt{B}    [   \algeffseqoverindex{ \ottnt{C} }{ \text{\unboldmath$\mathit{I}$} }   \ottsym{/}   \algeffseqoverindex{ \alpha }{ \text{\unboldmath$\mathit{I}$} }   ]  $.
Since the continuation denotes the context before the evaluation, its result type
matches with the type of the whole term.

The typing rules for continuations are shown in the lower half of
\reffig{inter-typing2}.
They are similar to the corresponding typing rules for terms except that a
subterm is replaced with a continuation.
In \TE{Let}, the continuation $\mathsf{let} \, \mathit{x}  \ottsym{=}   \Lambda\!  \,  \algeffseqover{ \alpha }   \ottsym{.}  \ottnt{E} \,  \mathsf{in}  \, \ottnt{e}$ has type
$  \text{\unboldmath$\forall$}  \,  \algeffseqover{ \alpha }   \ottsym{.}  \sigma  \multimap  \ottnt{B} $ because the hole of $\ottnt{E}$ appears
inside the scope of $ \algeffseqover{ \alpha } $.

\subsection{Elaboration}

This section defines the elaboration from {\surfacelang} to {\interlang}.
The important difference between the two languages from the viewpoint
of elaboration is that, whereas the parameter of an operation clause is
referred to by a single variable in {\surfacelang}, it is done by
one or more variables in {\interlang}.
Therefore, one variable in {\surfacelang} is represented by multiple
variables (required for each $ \mathsf{resume} $) in {\interlang}.
We use $\ottnt{S}$, a mapping from variables to variables, to
make the correspondence between variable names.
We write $ \ottnt{S}  \,\circ\, \{  \mathit{x}  \, {\mapsto} \,  \mathit{y}  \} $ for the same mapping as $\ottnt{S}$ except that
$\mathit{x}$ is mapped to $\mathit{y}$.

Elaboration is defined by two judgments: term elaboration judgment $ \Gamma ;  \ottnt{R}   \vdash   \ottnt{M}  :  \ottnt{A}  \,  |   \, \epsilon  \mathrel{\algefftransarrow{ \ottnt{S} } }  \ottnt{e} $, which denotes elaboration from a typing derivation of
judgment $\Gamma  \ottsym{;}  \ottnt{R}  \vdash  \ottnt{M}  \ottsym{:}  \ottnt{A} \,  |  \, \epsilon$ to $\ottnt{e}$ with $\ottnt{S}$, and handler
elaboration judgment $ \Gamma ;  \ottnt{R}   \vdash   \ottnt{H}  :  \ottnt{A}  \,  |   \, \epsilon  \, \Rightarrow   \ottnt{B}  \,  |   \, \epsilon'  \mathrel{\algefftransarrow{ \ottnt{S} } }  \ottnt{h} $, which denotes
elaboration from a typing derivation of judgment $\Gamma  \ottsym{;}  \ottnt{R}  \vdash  \ottnt{H}  \ottsym{:}  \ottnt{A} \,  |  \, \epsilon  \Rightarrow  \ottnt{B} \,  |  \, \epsilon'$ to $\ottnt{h}$ with $\ottnt{S}$.

\begin{figure}[t!]
 \textbf{Term elaboration rules} \quad
 \framebox{$ \Gamma ;  \ottnt{R}   \vdash   \ottnt{M}  :  \ottnt{A}  \,  |   \, \epsilon  \mathrel{\algefftransarrow{ \ottnt{S} } }  \ottnt{e} $}
 \begin{center}
  $\ottdruleElabXXVar{}$ \\[1.5ex]
  % $\ottdruleElabXXConst{}$ \\[3ex]
  $\ottdruleElabXXAbs{}$ \\[1.5ex]
  % $\ottdruleElabXXApp{}$ \\[3ex]
  % $\ottdruleElabXXOp{}$ \\[3ex]
  $\ottdruleElabXXHandle{}$ \\[1.5ex]
  $\ottdruleElabXXLet{}$ \\[1.5ex]
  $\ottdruleElabXXResume{}$ \\[1.5ex]
  % $\ottdruleElabXXWeak{}$ \\[3ex]
 \end{center}
 \textbf{Handler elaboration rules} \quad
 \framebox{$ \Gamma ;  \ottnt{R}   \vdash   \ottnt{H}  :  \ottnt{A}  \,  |   \, \epsilon  \, \Rightarrow   \ottnt{B}  \,  |   \, \epsilon'  \mathrel{\algefftransarrow{ \ottnt{S} } }  \ottnt{h} $}
 \begin{center}
  $\ottdruleElabHXXReturn{}$ \\[1.5ex]
  $\ottdruleElabHXXOp{}$
 \end{center}
 \caption{Elaboration rules (excerpt).}
 \label{fig:elaboration}
\end{figure}
Selected elaboration rules are shown in \reffig{elaboration}; the complete
set of the rules is found in the full version of the paper.
The elaboration rules are straightforward except for the use of $\ottnt{S}$.  A variable $\mathit{x}$ is
translated to $\ottnt{S}  \ottsym{(}  \mathit{x}  \ottsym{)}$ \Elab{Var} and, every time a new variable is
introduced, $\ottnt{S}$ is extended: see the rules other than \Elab{Var}
and \Elab{Handle}.

\subsection{Properties}
We show type safety of {\surfacelang}, i.e., a well-typed program in
{\surfacelang} does not get stuck, by proving (1) type preservation of the
elaboration from {\surfacelang} to {\interlang} and (2) type soundness of
{\interlang}.
Term $\ottnt{M}$ is a well-typed program of $\ottnt{A}$ if and only if
$ \emptyset   \ottsym{;}   \mathsf{none}   \vdash  \ottnt{M}  \ottsym{:}  \ottnt{A} \,  |  \,  \langle \rangle $.

The first can be shown easily.
We write $ \emptyset $ also for the identity mapping for variables.
\begin{restatable}[Elaboration is type-preserving]{thm}{thmElab}
 \label{thm:trans-preserving}
 If $\ottnt{M}$ is a well-typed program of $\ottnt{A}$,
 then $  \emptyset  ;   \mathsf{none}    \vdash   \ottnt{M}  :  \ottnt{A}  \,  |   \,  \langle \rangle   \mathrel{\algefftransarrow{  \emptyset  } }  \ottnt{e} $ and
 $ \emptyset   \ottsym{;}   \mathsf{none}  \,   \vdash  \ottnt{e} \,   \ottsym{:}  \ottnt{A} \,  |  \,  \langle \rangle $ for some $\ottnt{e}$.
\end{restatable}

We show the second---type soundness of {\interlang}---via progress and subject
reduction~\cite{Wright/Felleisen_1994_IC}.
We write $\Delta$ for a typing context that consists only of type variables.
Progress can be shown as usual.
\begin{restatable}[Progress]{lemm}{lemmProgress}
\label{lem:progress}
% \begin{lemmap}{Progress}{progress}
 If $\Delta  \ottsym{;}   \mathsf{none}  \,   \vdash  \ottnt{e} \,   \ottsym{:}  \ottnt{A} \,  |  \, \epsilon$, then
 (1) $\ottnt{e}  \longrightarrow  \ottnt{e'}$ for some $\ottnt{e'}$,
 (2) $\ottnt{e}$ is a value, or
 (3) $\ottnt{e} \,  =  \,  \textup{\texttt{\#}\relax}  \mathsf{op}   \ottsym{(}    \algeffseqover{ \sigma }    \ottsym{,}   \ottnt{w}   \ottsym{,}   \ottnt{E}   \ottsym{)} $ for some $\mathsf{op} \,  \in  \, \epsilon$, $ \algeffseqover{ \sigma } $, $\ottnt{w}$, and $\ottnt{E}$.
\end{restatable}
{\iffull
\begin{proof}
 By induction on the derivation of $\Delta  \ottsym{;}   \mathsf{none}  \,   \vdash  \ottnt{e} \,   \ottsym{:}  \ottnt{A} \,  |  \, \epsilon$.
\end{proof}
\fi}
A key lemma to show subject reduction is type preservation of continuation
substitution.
\begin{restatable}[Continuation substitution]{lemm}{lemmContSubst}
 \label{lem:cont-subst}
 Suppose that
 $\Gamma  \vdash   \algeffseqoverindex{  \text{\unboldmath$\forall$}  \,  \algeffseqoverindex{ \beta }{ \text{\unboldmath$\mathit{J}$} }   \ottsym{.}  \ottnt{C} }{ \text{\unboldmath$\mathit{I}$} } $ and
 $ \Gamma   \vdash    \ottnt{E} ^{  \algeffseqoverindex{ \beta }{ \text{\unboldmath$\mathit{J}$} }  }    \ottsym{:}     \text{\unboldmath$\forall$}  \,  \algeffseqoverindex{ \beta }{ \text{\unboldmath$\mathit{J}$} }   \ottsym{.}  \ottsym{(}   \ottnt{B}    [   \algeffseqoverindex{ \ottnt{C} }{ \text{\unboldmath$\mathit{I}$} }   \ottsym{/}   \algeffseqoverindex{ \alpha }{ \text{\unboldmath$\mathit{I}$} }   ]    \ottsym{)}  \multimap  \ottnt{D}   \,  |  \, \epsilon $ and
 $\Gamma  \ottsym{,}   \algeffseqoverindex{ \beta }{ \text{\unboldmath$\mathit{J}$} }   \vdash  \ottnt{v}  \ottsym{:}   \ottnt{A}    [   \algeffseqoverindex{ \ottnt{C} }{ \text{\unboldmath$\mathit{I}$} }   \ottsym{/}   \algeffseqoverindex{ \alpha }{ \text{\unboldmath$\mathit{I}$} }   ]  $.
 {\sloppy
 \begin{enumerate}
  \item If $\Gamma  \ottsym{;}  \ottsym{(}   \algeffseqoverindex{ \alpha }{ \text{\unboldmath$\mathit{I}$} }   \ottsym{,}  \ottnt{A}  \ottsym{,}   \ottnt{B}   \rightarrow  \!  \epsilon  \;  \ottnt{D}   \ottsym{)} \,   \vdash  \ottnt{e} \,   \ottsym{:}  \ottnt{D'} \,  |  \, \epsilon'$, then
        $\Gamma  \ottsym{;}   \mathsf{none}  \,   \vdash   \ottnt{e}    [   \ottnt{E} ^{  \algeffseqoverindex{ \beta }{ \text{\unboldmath$\mathit{J}$} }  }   /  \mathsf{resume}  ]^{  \algeffseqoverindex{  \text{\unboldmath$\forall$}  \,  \algeffseqoverindex{ \beta }{ \text{\unboldmath$\mathit{J}$} }   \ottsym{.}  \ottnt{C} }{ \text{\unboldmath$\mathit{I}$} }  }_{  \Lambda\!  \,  \algeffseqoverindex{ \beta }{ \text{\unboldmath$\mathit{J}$} }   \ottsym{.}  \ottnt{v} }   \,   \ottsym{:}  \ottnt{D'} \,  |  \, \epsilon'$.

  \item If $\Gamma  \ottsym{;}  \ottsym{(}   \algeffseqoverindex{ \alpha }{ \text{\unboldmath$\mathit{I}$} }   \ottsym{,}  \ottnt{A}  \ottsym{,}   \ottnt{B}   \rightarrow  \!  \epsilon  \;  \ottnt{D}   \ottsym{)} \,   \vdash  \ottnt{h}  \ottsym{:}  \ottnt{D_{{\mathrm{1}}}} \,  |  \, \epsilon_{{\mathrm{1}}}  \Rightarrow  \ottnt{D_{{\mathrm{2}}}} \,  |  \, \epsilon_{{\mathrm{2}}}$, then
        $\Gamma  \ottsym{;}   \mathsf{none}  \,   \vdash   \ottnt{h}    [   \ottnt{E} ^{  \algeffseqoverindex{ \beta }{ \text{\unboldmath$\mathit{J}$} }  }   /  \mathsf{resume}  ]^{  \algeffseqoverindex{  \text{\unboldmath$\forall$}  \,  \algeffseqoverindex{ \beta }{ \text{\unboldmath$\mathit{J}$} }   \ottsym{.}  \ottnt{C} }{ \text{\unboldmath$\mathit{I}$} }  }_{  \Lambda\!  \,  \algeffseqoverindex{ \beta }{ \text{\unboldmath$\mathit{J}$} }   \ottsym{.}  \ottnt{v} }    \ottsym{:}  \ottnt{D_{{\mathrm{1}}}} \,  |  \, \epsilon_{{\mathrm{1}}}  \Rightarrow  \ottnt{D_{{\mathrm{2}}}} \,  |  \, \epsilon_{{\mathrm{2}}}$.
 \end{enumerate}
 }
\end{restatable}
{\iffull
\begin{proof}
 By mutual induction on the typing derivations.
\end{proof}
\fi}
Using the continuation substitution lemma as well as other lemmas, we show
subject reduction.
\begin{restatable}[Subject reduction]{lemm}{lemmSubjectRed}
 \label{lem:subject-red}
 \begin{enumerate}
  \item If $\Delta  \ottsym{;}   \mathsf{none}  \,   \vdash  \ottnt{e_{{\mathrm{1}}}} \,   \ottsym{:}  \ottnt{A} \,  |  \, \epsilon$ and $\ottnt{e_{{\mathrm{1}}}}  \rightsquigarrow  \ottnt{e_{{\mathrm{2}}}}$,
        then $\Delta  \ottsym{;}   \mathsf{none}  \,   \vdash  \ottnt{e_{{\mathrm{2}}}} \,   \ottsym{:}  \ottnt{A} \,  |  \, \epsilon$.
  \item If $\Delta  \ottsym{;}   \mathsf{none}  \,   \vdash  \ottnt{e_{{\mathrm{1}}}} \,   \ottsym{:}  \ottnt{A} \,  |  \, \epsilon$ and $\ottnt{e_{{\mathrm{1}}}}  \longrightarrow  \ottnt{e_{{\mathrm{2}}}}$,
        then $\Delta  \ottsym{;}   \mathsf{none}  \,   \vdash  \ottnt{e_{{\mathrm{2}}}} \,   \ottsym{:}  \ottnt{A} \,  |  \, \epsilon$.
 \end{enumerate}
\end{restatable}
{\iffull
\begin{proof}
 We can show each item by induction on the typing derivation.
\end{proof}
\fi}
We write $\ottnt{e}  \centernot\longrightarrow$ if and only if $\ottnt{e}$ cannot evaluate further.
Moreover, $ \longrightarrow^{*} $ denotes the reflexive and transitive closure of the evaluation
relation $ \longrightarrow $.
\begin{restatable}[Type soundness of {\interlang}]{thm}{thmTypeSoundness}
 \label{thm:type-sound}
 If $\Delta  \ottsym{;}   \mathsf{none}  \,   \vdash  \ottnt{e} \,   \ottsym{:}  \ottnt{A} \,  |  \, \epsilon$ and $\ottnt{e}  \longrightarrow^{*}  \ottnt{e'}$ and
 $\ottnt{e'}  \centernot\longrightarrow$, then (1) $\ottnt{e'}$ is a value or (2) $\ottnt{e'} \,  =  \,  \textup{\texttt{\#}\relax}  \mathsf{op}   \ottsym{(}    \algeffseqover{ \sigma }    \ottsym{,}   \ottnt{w}   \ottsym{,}   \ottnt{E}   \ottsym{)} $
 for some $\mathsf{op} \,  \in  \, \epsilon$, $ \algeffseqover{ \sigma } $, $\ottnt{w}$, and $\ottnt{E}$.
\end{restatable}

Now, type safety of {\surfacelang} is obtained as a corollary
of Theorems~\ref{thm:trans-preserving} and \ref{thm:type-sound}.
\begin{corollary}[Type safety of {\surfacelang}]
 If $\ottnt{M}$ is a well-typed program of $\ottnt{A}$,
 there exists some $\ottnt{e}$ such that
 $  \emptyset  ;   \mathsf{none}    \vdash   \ottnt{M}  :  \ottnt{A}  \,  |   \,  \langle \rangle   \mathrel{\algefftransarrow{  \emptyset  } }  \ottnt{e} $ and $\ottnt{e}$ does not get stuck.
\end{corollary}

\section{Related work}
\label{sec:relwork}

\subsection{Polymorphic effects and let-polymorphism}
Many researchers have attacked the problem of combining effects---not necessarily
algebraic---and let-polymorphism so far~\cite{Tofte_1990_IC,Leroy/Weis_1991_POPL,Appel/MacQueen_1991_PLILP,Hoang/Mitchell/Viswanathan_1993_LICS,Wright_1995_LSC,Garrigue_2004_FLOPS,Asai/Kameyama_2007_APLAS,Kammar/Pretnar_2017_JFP}.
In particular, most of them have focused on ML-style polymorphic references.
The algebraic effect handlers dealt with in this paper seem to be unable to
implement general ML-style references---i.e., give an appropriate implementation
to a set of effect operations \texttt{new} with the signature $  \text{\unboldmath$\forall$}    \alpha  .  \alpha  \hookrightarrow   \alpha  \, \mathsf{ref}  $,
\texttt{get} with $  \text{\unboldmath$\forall$}    \alpha  .   \alpha  \, \mathsf{ref}   \hookrightarrow  \alpha $, and \texttt{put} with $  \text{\unboldmath$\forall$}    \alpha  .    \alpha  \times  \alpha   \, \mathsf{ref}   \hookrightarrow   \mathsf{unit}  $ for abstract datatype $ \alpha  \, \mathsf{ref} $---even without the restriction on
handlers because each operation clause in a handler assigns type variables
locally and it is impossible to share such type variables between operation
clauses.\footnote{One possible approach to dealing with ML-style references is to
extend algebraic effects and handlers so that a handler for \emph{parameterized}
effects can be connected with dynamic
resources~\cite{Bauer/Pretnar_2015_JLAMP}.}
Nevertheless, their approaches would be applicable to algebraic effects and
handlers.

A common idea in the literature is to restrict the form of
expressions bound by polymorphic let.
Thus, they are complementary to our approach in that they restrict how effect
operations are used whereas we restrict how effect operations are implemented.

% Our approach, which restricts how effect operations are implemented, is
% complementary to the researches in literature, which restrict how they are used.
%
% A common idea of the researches on literature is to restrict the form of
% expressions bound by polymorphic let. \TS{?}

Value restriction~\cite{Tofte_1990_IC,Wright_1995_LSC}, a standard way adopted
in ML-like languages, restricts polymorphic let-bound expressions to syntactic values.
Garrigue~\cite{Garrigue_2004_FLOPS} relaxes the value restriction so that, if a
let-bound expression is not a syntactic value, type variables that appear only
at positive positions in the type of the expression can be generalized.
Although the (relaxed) value restriction is a quite clear criterion that indicates what
let-bound expressions can be polymorphic safely and it even accepts interfering
handlers, it is too restrictive in some cases.
We give an example for such a case below.
\begin{lstlisting}[xleftmargin=1.3\parindent]
effect (/\choosep/) : (/$\forall \alpha.$/) (/$\alpha \times \alpha$/) (/$\hookrightarrow$/) (/$\alpha$/)

let f1 () =
  let g = #(/\choosep/)(fst, snd) in
  if g (true,false) then g (-1,1) else g (1,-1)
\end{lstlisting}
In the definition of function \texttt{f1}, variable \texttt{g} is used
polymorphically.
Execution of this function under an appropriate handler would succeed, and 
in fact our calculus accepts it.
By contrast, the (relaxed) value restriction rejects it because the
let-bound expression \texttt{\#{\choosep}(fst,snd)} is not a syntactic
value and the type variable appear in both positive and negative
positions, and so \texttt{g} is assigned a monomorphic type.
A workaround for this problem is to make a function wrapper that calls either of
\texttt{fst} or \texttt{snd} depending on the Boolean value chosen by
\texttt{\choosep}:
\begin{lstlisting}[xleftmargin=1.3\parindent]
let f2 () =
  let b = #(/\choosep/)(true,false) in
  let g = (/$\lambda$/)x. if b then (fst x) else (snd x) in
  if g (true,false) then g (-1,1) else g (1,-1)
\end{lstlisting}
However, this workaround makes the program complicated and incurs
additional run-time cost for the branching and an extra call to the wrapper
function.

Asai and Kameyama~\cite{Asai/Kameyama_2007_APLAS} study a combination of
let-polymorphism with delimited control operators
shift/reset~\cite{Danvy/Filinski_1990_LFP}.
They allow a let-bound expression to be polymorphic if it invokes no control
operation.
Thus, the function \texttt{f1} above would be rejected in their approach.

Another research line to restrict the use of effects is to allow only type
variables unrelated to effect invocations to be generalized.
Tofte~\cite{Tofte_1990_IC} distinguishes between applicative type variables, which
cannot be used for effect invocations, and imperative ones, which can be used,
and proposes a type system that enforces restrictions that (1) type variables of
imperative operations can be instantiated only with types wherein all type
variables are imperative and (2) if a let-bound expression is not a syntactic
value, only applicative type variables can be generalized.
Leroy and Weis~\cite{Leroy/Weis_1991_POPL} allow generalization only of type
variables that do not appear in a parameter type to the reference type in the
type of a let-expression.
To detect the hidden use of references, their type system gives a term not only a
type but also the types of free variables used in the term.
Standard ML of New Jersey (before ML97) adopted weak
polymorphism~\cite{Appel/MacQueen_1991_PLILP}, which was later formalized and
investigated deeply by Hoang et al.~\cite{Hoang/Mitchell/Viswanathan_1993_LICS}.
Weak polymorphism equips a type variable with the number of function calls
after which a value of a type containing the type variable will be passed to an
imperative operation.
The type system ensures that type variables with positive numbers are not
related to imperative constructs, and so such type variables can be generalized
safely.
In this line of research, the function \texttt{f1} above would not typecheck
because generalized type variables are used to instantiate those of the effect
signature, although it could be rewritten to an acceptable one by taking care not
to involve type variables in effect invocation.
\begin{lstlisting}[xleftmargin=1.3\parindent]
let f3 () =
  let g = if #(/\choosep/)(true,false) then fst then snd in
  if g (true,false) then g (-1,1) else g (1,-1)
\end{lstlisting}

More recently, Kammar and Pretnar~\cite{Kammar/Pretnar_2017_JFP} show that
\emph{parameterized} algebraic effects and handlers do not need the value
restriction \emph{if} the type variables used in an effect invocation are not
generalized.
Thus, as the other work that restricts generalized type variables, their
approach would reject function \texttt{f1} but would accept \texttt{f3}.

\subsection{Algebraic effects and handlers}
Algebraic effects~\cite{Plotkin/Power_2003_ACS} are a way to represent the
denotation of an effect by giving a set of operations and an equational theory
that capture their properties.
Algebraic effect handlers, introduced by Plotkin and
Pretnar~\cite{Plotkin/Pretnar_2009_ESOP}, make it possible to provide user-defined
effects.
Algebraic effect handlers have been gaining popularity owing to their flexibility and
have been made available as
libraries~\cite{Kammar/Lindley/Oury_2013_ICFP,Wu/Schrijvers/Hinze_2014_Haskell,Kiselyov/Ishii_2015_Haskell}
or as primitive features of languages, such as Eff~\cite{Bauer/Pretnar_2015_JLAMP},
Koka~\cite{Leijen_2017_POPL}, Frank~\cite{Lindley/MacBrid/McLaughlin_2017_POPL},
and Multicore OCaml~\cite{multicoreOCaml_2017_TFP}.
In these languages, let-bound expressions that can be polymorphic are restricted
to values or pure expressions.

Recently, Forster et al.~\cite{Forster/Kammar/Lindley/Pretnar_2017_ICFP}
investigate the relationships between algebraic effect handlers and other
mechanisms for user-defined effects---delimited control
shift0~\cite{Materzok/Biernack_2012_APLAS} and monadic
reflection~\cite{Filinski_1994_POPL,Filinski_2010_POPL}---conjecturing that
there would be no type-preserving translation from a language with delimited control or
monadic reflection to one with algebraic effect handlers.
It would be an interesting direction to export our idea to delimited control and
monadic reflection.

\section{Conclusion}
\label{sec:conclusion}

There has been a long history of collaboration between effects and let-polymor\-phism.
This work focuses on polymorphic algebraic effects and handlers, wherein the type
signature of an effect operation can be polymorphic and an operation clause has
a type binder, and shows that a naive combination of polymorphic effects and let-polymorphism breaks
type safety.
Our novel observation to address this problem is that any let-bound expression
can be polymorphic safely if resumptions from a handler do not interfere with each
other.
We formalized this idea by developing a type system that requires the argument of
each resumption expression to have a type obtained by renaming the type variables
assigned in the operation clause to those assigned in the resumption.
We have proven that a well-typed program in our type system does not get stuck via
elaboration to an intermediate language wherein type information appears
explicitly.

There are many directions for future work.
The first is to address the problem, described at the end of
\refsec{surfacelang}, that renaming the type variables assigned in an operation
clause to those assigned in a resumption expression is allowed for the argument
of the clause but not for variables bound by lambda abstractions and
let-expressions outside the resumption expression.
Second, we are interested in incorporating other features from the literature
on algebraic effect handlers, such as dynamic
resources~\cite{Bauer/Pretnar_2015_JLAMP} and parameterized algebraic effects,
and restriction techniques that have been developed for type-safe imperative
programming with let-polymorphism such as (relaxed) value
restriction~\cite{Tofte_1990_IC,Wright_1995_LSC,Garrigue_2004_FLOPS}.
For example, we would like to develop a type system that enforces the
non-interfering restriction only to handlers implementing effect operations
invoked in polymorphic computation.
We also expect that it is possible to determine whether implementations of an effect
operation have no interfering resumption from the type signature of the
operation, as relaxed value restriction makes it possible to find safely
generalizable type variables from the type of a let-bound
expression~\cite{Garrigue_2004_FLOPS}.
Finally, we are also interested in implementing our idea for a language with
effect handlers such as Koka~\cite{Leijen_2017_POPL} and in applying the idea of
analyzing handlers to other settings such as dependent typing.

\section*{Acknowledgments}
We would like to thank the anonymous reviewers for their valuable comments.
This work was supported in part by ERATO HASUO Metamathematics for Systems
Design Project (No.\ JPMJER1603), JST (Sekiyama), and JSPS KAKENHI Grant
Number JP15H05706 (Igarashi).

\bibliographystyle{splncs04}
\bibliography{main}

{\iffull

\clearpage
\section*{Errata}
The following typographical errors in the ESOP'19 paper are fixed (the page
numbers are those of the ESOP'19).

\begin{description}
 \item[Page 363, Definition 1]
            ``$\ottnt{R} ::=  \mathsf{none}  \mid \ottsym{(}   \algeffseqover{ \alpha }   \ottsym{,}  \ottnt{A}  \ottsym{,}   \ottnt{B}   \rightarrow  \!  \epsilon  \;  \ottnt{C}   \ottsym{)}$'' is corrected to ``$\ottnt{R} ::=  \mathsf{none}  \mid \ottsym{(}   \algeffseqover{ \alpha }   \ottsym{,}  \mathit{x} \,  \mathord{:}  \, \ottnt{A}  \ottsym{,}   \ottnt{B}   \rightarrow  \!  \epsilon  \;  \ottnt{C}   \ottsym{)}$''.
 \item[Page 367]
            ``We define the syntax, operational semantics, and type system of {\surfacelang}'' is corrected to ``We define ... of {\interlang}''.
 \item[page 370]
            ``The basic idea is to substitute $ \algeffseqoverindex{ \ottnt{A} }{ \text{\unboldmath$\mathit{I}$} } $ for $ \algeffseqoverindex{ \beta }{ \text{\unboldmath$\mathit{I}$} } $'' is corrected to ``The basic idea is ... for $ \algeffseqoverindex{ \alpha }{ \text{\unboldmath$\mathit{I}$} } $''.
\end{description}

\appendix

\section{Definition}

\subsection{Surface language}

\subsubsection{Syntax}

\[
  \begin{array}{rrl}
   \multicolumn{3}{l}{\textbf{Effect operations}} \\
    \mathsf{op} \\[.5ex]
   \multicolumn{3}{l}{\textbf{Effects}} \\
    \epsilon & ::= & \text{sets of effect operations} \\[.5ex]
   \multicolumn{3}{l}{\textbf{Base types}} \\
    \iota & ::= &  \mathsf{bool}  \mid  \mathsf{int}  \mid  \bot  \mid ... \\[.5ex]
   \multicolumn{3}{l}{\textbf{Type variables}} \\
    \alpha, \beta, \gamma \\[.5ex]
   \multicolumn{3}{l}{\textbf{Types}} \\
    \ottnt{A}, \ottnt{B}, \ottnt{C}, \ottnt{D} & ::= & \alpha \mid \iota \mid  \ottnt{A}   \rightarrow  \!  \epsilon  \;  \ottnt{B}  \\[.5ex]
   \multicolumn{3}{l}{\textbf{Type schemes}} \\
    \sigma & ::= & \ottnt{A} \mid  \text{\unboldmath$\forall$}  \, \alpha  \ottsym{.}  \sigma \\[.5ex]
   \multicolumn{3}{l}{\textbf{Constants}} \\
    \ottnt{c} & ::= &  \mathsf{true}  \mid  \mathsf{false}  \mid  0  \mid  \mathsf{+}  \mid ... \\[.5ex]
   \multicolumn{3}{l}{\textbf{Terms}} \\
   \ottnt{M} & ::= & \mathit{x} \mid \ottnt{c} \mid  \lambda\!  \, \mathit{x}  \ottsym{.}  \ottnt{M} \mid \ottnt{M_{{\mathrm{1}}}} \, \ottnt{M_{{\mathrm{2}}}} \mid
    \mathsf{let} \, \mathit{x}  \ottsym{=}  \ottnt{M_{{\mathrm{1}}}} \,  \mathsf{in}  \, \ottnt{M_{{\mathrm{2}}}} \mid \\ &&
     \textup{\texttt{\#}\relax}  \mathsf{op}   \ottsym{(}   \ottnt{M}   \ottsym{)}  \mid
    \mathsf{handle} \, \ottnt{M} \, \mathsf{with} \, \ottnt{H} \mid
    \mathsf{resume} \, \ottnt{M}
    \\[.5ex]
   \multicolumn{3}{l}{\textbf{Handlers}} \\
    \ottnt{H} & ::= & \mathsf{return} \, \mathit{x}  \rightarrow  \ottnt{M} \mid \ottnt{H}  \ottsym{;}  \mathsf{op}  \ottsym{(}  \mathit{x}  \ottsym{)}  \rightarrow  \ottnt{M}
    \\[.5ex]
   \multicolumn{3}{l}{\textbf{Typing contexts}} \\
   \Gamma & ::= &  \emptyset  \mid
    \Gamma  \ottsym{,}  \mathit{x} \,  \mathord{:}  \, \sigma \mid \Gamma  \ottsym{,}  \alpha
    \\[.5ex]
  \end{array}
  \]

\begin{conv}
We write $ \text{\unboldmath$\forall$}  \,  \algeffseqoverindex{ \alpha }{ \ottmv{i}  \in  \text{\unboldmath$\mathit{I}$} }   \ottsym{.}  \ottnt{A}$ for $ \text{\unboldmath$\forall$}  \, \alpha_{{\mathrm{1}}}  \ottsym{.}  ...   \text{\unboldmath$\forall$}  \, \alpha_{\ottmv{n}}  \ottsym{.}  \ottnt{A}$
where $\text{\unboldmath$\mathit{I}$} = \{ 1, ..., n \}$.
We often omit indices ($\ottmv{i}$ and $\ottmv{j}$) and index sets ($\text{\unboldmath$\mathit{I}$}$ and $\text{\unboldmath$\mathit{J}$}$) if
they are not important: for example, we often abbreviate $ \text{\unboldmath$\forall$}  \,  \algeffseqoverindex{ \alpha }{ \ottmv{i}  \in  \text{\unboldmath$\mathit{I}$} }   \ottsym{.}  \ottnt{A}$ to
$ \text{\unboldmath$\forall$}  \,  \algeffseqoverindex{ \alpha }{ \text{\unboldmath$\mathit{I}$} }   \ottsym{.}  \ottnt{A}$ or even $ \text{\unboldmath$\forall$}  \,  \algeffseqover{ \alpha }   \ottsym{.}  \ottnt{A}$.
Similarly, we use a bold font for other sequences ($ \algeffseqoverindex{ \ottnt{A} }{ \ottmv{i}  \in  \text{\unboldmath$\mathit{I}$} } $ for a
sequence of types, $ \algeffseqoverindex{ \ottnt{v} }{ \ottmv{i}  \in  \text{\unboldmath$\mathit{I}$} } $ for a sequence of values, and so on).
We sometimes write $\ottsym{\{}   \algeffseqover{ \alpha }   \ottsym{\}}$ to view the sequence $ \algeffseqover{ \alpha } $ as a set by
ignoring the order.
We also write $ \algeffseqoverindex{  \text{\unboldmath$\forall$}  \,  \algeffseqoverindex{ \alpha }{ \text{\unboldmath$\mathit{I}$} }   \ottsym{.}  \sigma }{ \text{\unboldmath$\mathit{J}$} } $ for a sequence $ \text{\unboldmath$\forall$}  \,  \algeffseqoverindex{ \alpha }{ \text{\unboldmath$\mathit{I}$} }   \ottsym{.}  \sigma_{j_1}$, \ldots,
$ \text{\unboldmath$\forall$}  \,  \algeffseqoverindex{ \alpha }{ \text{\unboldmath$\mathit{I}$} }   \ottsym{.}  \sigma_{j_n}$ of type schemes (where $\text{\unboldmath$\mathit{J}$} = \{j_1,\ldots,j_n\}$).
\end{conv}

\begin{defn}[Domain of typing contexts]
 We define $ \mathit{dom}  (  \Gamma  ) $ as follows.
 \[\begin{array}{lll}
   \mathit{dom}  (  \Gamma  \ottsym{,}  \mathit{x} \,  \mathord{:}  \, \sigma  )  &\defeq&  \mathit{dom}  (  \Gamma  )  \,  \mathbin{\cup}  \, \ottsym{\{}  \mathit{x}  \ottsym{\}} \\
     \mathit{dom}  (  \Gamma  \ottsym{,}  \alpha  )   &\defeq&  \mathit{dom}  (  \Gamma  )  \,  \mathbin{\cup}  \, \ottsym{\{}  \alpha  \ottsym{\}} \\
   \end{array}\]
\end{defn}

\begin{defn}[Free type variables and type substitution in type schemes]
 Free type variables $ \mathit{ftv}  (  \sigma  ) $ in a type scheme $\sigma$ and type
 substitution $ \ottnt{B}    [   \algeffseqover{ \ottnt{A} }   \ottsym{/}   \algeffseqover{ \alpha }   ]  $ of $ \algeffseqover{ \ottnt{A} } $ for type variables $ \algeffseqover{ \alpha } $ in
 $\ottnt{B}$ are defined as usual.
\end{defn}

\begin{assum}
 We suppose that each constant $\ottnt{c}$ is assigned a first-order closed type
 $ \mathit{ty}  (  \ottnt{c}  ) $ of the form $\iota_{{\mathrm{1}}}  \rightarrow \! \langle \rangle \; \cdots  \rightarrow \! \langle \rangle  \;
 \iota_{\ottmv{n}}$ and that each effect operation $\mathsf{op}$ is assigned a signature of the
 form $  \text{\unboldmath$\forall$}     \algeffseqover{ \alpha }   .  \ottnt{A}  \hookrightarrow  \ottnt{B} $.
 We also assume that, for $\mathit{ty} \, \ottsym{(}  \mathsf{op}  \ottsym{)} \,  =  \,   \text{\unboldmath$\forall$}     \algeffseqover{ \alpha }   .  \ottnt{A}  \hookrightarrow  \ottnt{B} $, $ \mathit{ftv}  (  \ottnt{A}  )  \,  \subseteq  \, \ottsym{\{}   \algeffseqover{ \alpha }   \ottsym{\}}$ and
 $ \mathit{ftv}  (  \ottnt{B}  )  \,  \subseteq  \, \ottsym{\{}   \algeffseqover{ \alpha }   \ottsym{\}}$.
 
 Suppose that, for any $\iota$, there is a set $ K_{ \iota } $ of constants of
 $\iota$.
 For any constant $\ottnt{c}$, $ \mathit{ty}  (  \ottnt{c}  )  \,  =  \, \iota$ if and only if $\ottnt{c} \,  \in  \,  K_{ \iota } $.
 The function $ \zeta $ gives a denotation to pairs of constants.
 In particular, for any constants $\ottnt{c_{{\mathrm{1}}}}$ and $\ottnt{c_{{\mathrm{2}}}}$:
 (1) $ \zeta  (  \ottnt{c_{{\mathrm{1}}}}  ,  \ottnt{c_{{\mathrm{2}}}}  ) $ is defined if and only if
 $ \mathit{ty}  (  \ottnt{c_{{\mathrm{1}}}}  )  \,  =  \,  \iota   \rightarrow  \!   \langle \rangle   \;  \ottnt{A} $ and $ \mathit{ty}  (  \ottnt{c_{{\mathrm{2}}}}  )  \,  =  \, \iota$ for some $\ottnt{A}$; and
 (2) if $ \zeta  (  \ottnt{c_{{\mathrm{1}}}}  ,  \ottnt{c_{{\mathrm{2}}}}  ) $ is defined, $ \zeta  (  \ottnt{c_{{\mathrm{1}}}}  ,  \ottnt{c_{{\mathrm{2}}}}  ) $ is a constant and
 $ \mathit{ty}  (   \zeta  (  \ottnt{c_{{\mathrm{1}}}}  ,  \ottnt{c_{{\mathrm{2}}}}  )   )  \,  =  \, \ottnt{A}$ where $ \mathit{ty}  (  \ottnt{c_{{\mathrm{1}}}}  )  \,  =  \,  \iota   \rightarrow  \!   \langle \rangle   \;  \ottnt{A} $.
\end{assum}

\subsubsection{Typing}

\begin{defn}[Resumption type]
 We define resumption type $\ottnt{R}$ as follows.
 \[\begin{array}{lll}
  \ottnt{R}        & ::= &
    \mathsf{none}  \mid
   \ottsym{(}   \algeffseqover{ \alpha }   \ottsym{,}  \mathit{x} \,  \mathord{:}  \, \ottnt{A}  \ottsym{,}   \ottnt{B}   \rightarrow  \!  \epsilon  \;  \ottnt{C}   \ottsym{)} \\ &&
   \qquad\qquad \text{(if $ \mathit{ftv}  (  \ottnt{A}  )  \,  \mathbin{\cup}  \,  \mathit{ftv}  (  \ottnt{B}  )  \,  \subseteq  \, \ottsym{\{}   \algeffseqover{ \alpha }   \ottsym{\}}$ and $ \mathit{ftv}  (  \ottnt{C}  )  \,  \mathbin{\cap}  \, \ottsym{\{}   \algeffseqover{ \alpha }   \ottsym{\}} \,  =  \,  \emptyset $)}
   \end{array}\]
\end{defn}

\begin{defn}[Type scheme well-formedness]
 We write $\Gamma  \vdash  \sigma$ if and only if $ \mathit{ftv}  (  \sigma  )  \,  \subseteq  \,  \mathit{dom}  (  \Gamma  ) $.
\end{defn}

\begin{figure}[t!]
\noindent
\textbf{Well-formed rules for typing contexts} \\[1ex]
\framebox{$\vdash  \Gamma$}
\begin{center}
 $\ottdruleWFXXEmpty{}$ \hfil
 $\ottdruleWFXXVar{}$ \hfil
 $\ottdruleWFXXTyVar{}$ \\[3ex]
\end{center}
\textbf{Typing rules} \\[1ex]
\framebox{$\Gamma  \ottsym{;}  \ottnt{R}  \vdash  \ottnt{M}  \ottsym{:}  \ottnt{A} \,  |  \, \epsilon$}
\begin{center}
 $\ottdruleTSXXVar{}$ \hfil
 $\ottdruleTSXXConst{}$ \\[3ex]
 $\ottdruleTSXXAbs{}$ \hfil
 $\ottdruleTSXXApp{}$ \\[3ex]
 $\ottdruleTSXXOp{}$ \\[3ex]
 $\ottdruleTSXXLet{}$ \\[3ex]
 $\ottdruleTSXXWeak{}$ \hfil
 $\ottdruleTSXXHandle{}$ \\[3ex]
 $\ottdruleTSXXResume{}$ \hfil
\end{center}
\framebox{$\Gamma  \ottsym{;}  \ottnt{R}  \vdash  \ottnt{H}  \ottsym{:}  \ottnt{A} \,  |  \, \epsilon  \Rightarrow  \ottnt{B} \,  |  \, \epsilon'$}
\begin{center}
 $\ottdruleTHSXXReturn{}$ \\[3ex]
 $\ottdruleTHSXXOp{}$
\end{center}
 \caption{Typing rules in {\surfacelang}.}
 \label{fig:app-surface-typing}
\end{figure}

\begin{defn}
 Judgments $\vdash  \Gamma$ and $\Gamma  \ottsym{;}  \ottnt{R}  \vdash  \ottnt{M}  \ottsym{:}  \ottnt{A} \,  |  \, \epsilon$ and
 $\Gamma  \ottsym{;}  \ottnt{R}  \vdash  \ottnt{H}  \ottsym{:}  \ottnt{A} \,  |  \, \epsilon  \Rightarrow  \ottnt{B} \,  |  \, \epsilon'$ are the least relations
 satisfying the rules in \reffig{app-surface-typing}.

 Term $\ottnt{M}$ is a well-typed program of $\ottnt{A}$ if and only if
 $ \emptyset   \ottsym{;}   \mathsf{none}   \vdash  \ottnt{M}  \ottsym{:}  \ottnt{A} \,  |  \,  \langle \rangle $.
\end{defn}

\subsection{Intermediate language}
\subsubsection{Syntax}

\[
  \begin{array}{rrl}
   \multicolumn{3}{l}{\textbf{Values}} \\
    \ottnt{v} & ::= & \ottnt{c} \mid  \lambda\!  \, \mathit{x}  \ottsym{.}  \ottnt{e} \\[.5ex]
   \multicolumn{3}{l}{\textbf{Polymorphic values}} \\
    \ottnt{w} & ::= & \ottnt{v} \mid  \Lambda\!  \, \alpha  \ottsym{.}  \ottnt{w} \\[.5ex]
   \multicolumn{3}{l}{\textbf{Terms}} \\
    \ottnt{e} & ::= & \mathit{x} \,  \algeffseqover{ \ottnt{A} }  \mid \ottnt{c} \mid  \lambda\!  \, \mathit{x}  \ottsym{.}  \ottnt{e} \mid \ottnt{e_{{\mathrm{1}}}} \, \ottnt{e_{{\mathrm{2}}}} \mid
                  \mathsf{let} \, \mathit{x}  \ottsym{=}   \Lambda\!  \,  \algeffseqover{ \alpha }   \ottsym{.}  \ottnt{e_{{\mathrm{1}}}} \,  \mathsf{in}  \, \ottnt{e_{{\mathrm{2}}}} \mid
                   \textup{\texttt{\#}\relax}  \mathsf{op}   \ottsym{(}    \algeffseqover{ \ottnt{A} }    \ottsym{,}   \ottnt{e}   \ottsym{)}  \mid
                   \textup{\texttt{\#}\relax}  \mathsf{op}   \ottsym{(}    \algeffseqover{ \sigma }    \ottsym{,}   \ottnt{w}   \ottsym{,}   \ottnt{E}   \ottsym{)}  \mid \\ &&
                  \mathsf{handle} \, \ottnt{e} \, \mathsf{with} \, \ottnt{h} \mid
                  \mathsf{resume} \,  \algeffseqover{ \alpha }  \, \mathit{x}  \ottsym{.}  \ottnt{e}
                  \\[.5ex]
   \multicolumn{3}{l}{\textbf{Handlers}} \\
    \ottnt{h} & ::= & \mathsf{return} \, \mathit{x}  \rightarrow  \ottnt{e} \mid
                  \ottnt{h}  \ottsym{;}   \Lambda\!  \,  \algeffseqover{ \alpha }   \ottsym{.}  \mathsf{op}  \ottsym{(}  \mathit{x}  \ottsym{)}  \rightarrow  \ottnt{e} \\[.5ex]
   \multicolumn{3}{l}{\textbf{Evaluation contexts}} \\
     \ottnt{E} ^{  \algeffseqoverindex{ \alpha }{ \text{\unboldmath$\mathit{I}$} }  }  & ::= &  [\,]  \ \text{(if $ \algeffseqoverindex{ \alpha }{ \text{\unboldmath$\mathit{I}$} }  \,  =  \,  \emptyset $)} \mid
                      \ottnt{E} ^{  \algeffseqoverindex{ \alpha }{ \text{\unboldmath$\mathit{I}$} }  }  \, \ottnt{e_{{\mathrm{2}}}} \mid \ottnt{v_{{\mathrm{1}}}} \,  \ottnt{E} ^{  \algeffseqoverindex{ \alpha }{ \text{\unboldmath$\mathit{I}$} }  }  \mid
                     \mathsf{let} \, \mathit{x}  \ottsym{=}   \Lambda\!  \,  \algeffseqoverindex{ \beta }{ \text{\unboldmath$\mathit{J_{{\mathrm{1}}}}$} }   \ottsym{.}   \ottnt{E} ^{  \algeffseqoverindex{ \gamma }{ \text{\unboldmath$\mathit{J_{{\mathrm{2}}}}$} }  }  \,  \mathsf{in}  \, \ottnt{e_{{\mathrm{2}}}}
                       \ \text{(if $ \algeffseqoverindex{ \alpha }{ \text{\unboldmath$\mathit{I}$} }  \,  =  \,  \algeffseqoverindex{ \beta }{ \text{\unboldmath$\mathit{J_{{\mathrm{1}}}}$} }   \ottsym{,}   \algeffseqoverindex{ \gamma }{ \text{\unboldmath$\mathit{J_{{\mathrm{2}}}}$} } $)} \mid \\ &&
                      \textup{\texttt{\#}\relax}  \mathsf{op}   \ottsym{(}    \algeffseqoverindex{ \ottnt{A} }{ \text{\unboldmath$\mathit{J}$} }    \ottsym{,}    \ottnt{E} ^{  \algeffseqoverindex{ \alpha }{ \text{\unboldmath$\mathit{I}$} }  }    \ottsym{)}  \mid \mathsf{handle} \,  \ottnt{E} ^{  \algeffseqoverindex{ \alpha }{ \text{\unboldmath$\mathit{I}$} }  }  \, \mathsf{with} \, \ottnt{h}
                     \\[.5ex]
   \multicolumn{3}{l}{\textbf{Top-level typing contexts}} \\
    \Delta & ::= &  \emptyset  \mid
                    \Delta  \ottsym{,}  \alpha
  \end{array}
\]

\begin{conv}
 We write $\ottnt{E}$ for $ \ottnt{E} ^{  \algeffseqover{ \alpha }  } $ if $ \algeffseqover{ \alpha } $ is not important.
\end{conv}

\begin{defn}[Free type variables]
 We write $ \mathit{ftv}  (  \ottnt{e}  ) $ and $ \mathit{ftv}  (  \ottnt{E}  ) $ for sets of type variables that occur
 free in $\ottnt{e}$ and $\ottnt{E}$, respectively.
 The notion of free type variables is defined as usual.
\end{defn}

\begin{defn}[Substitution]
 Substitution $ \ottnt{e}    [   \algeffseqover{ \ottnt{A} }   \ottsym{/}   \algeffseqover{ \alpha }   ]  $ of $ \algeffseqover{ \ottnt{A} } $ for $ \algeffseqover{ \alpha } $ in $\ottnt{e}$ is defined
 in a capture-avoiding manner as usual.
 Substitution $ \ottnt{e}    [  \ottnt{w}  \ottsym{/}  \mathit{x}  ]  $ of polymorphic value $\ottnt{w}$ for variable $\mathit{x}$ in
 $\ottnt{e}$ is also defined in a standard capture-avoiding manner: in particular,
 \[\begin{array}{rcl}
   \ottsym{(}  \mathit{x} \,  \algeffseqover{ \ottnt{A} }   \ottsym{)}    [   \Lambda\!  \,  \algeffseqover{ \alpha }   \ottsym{.}  \ottnt{v}  \ottsym{/}  \mathit{x}  ]   &\defeq&  \ottnt{v}    [   \algeffseqover{ \ottnt{A} }   \ottsym{/}   \algeffseqover{ \alpha }   ]  
   \end{array}\]

 Substitution $ \ottnt{e}    [   \ottnt{E} ^{  \algeffseqoverindex{ \beta }{ \text{\unboldmath$\mathit{J}$} }  }   /  \mathsf{resume}  ]^{  \algeffseqoverindex{  \text{\unboldmath$\forall$}  \,  \algeffseqoverindex{ \beta }{ \text{\unboldmath$\mathit{J}$} }   \ottsym{.}  \ottnt{A} }{ \text{\unboldmath$\mathit{I}$} }  }_{  \Lambda\!  \,  \algeffseqoverindex{ \beta }{ \text{\unboldmath$\mathit{J}$} }   \ottsym{.}  \ottnt{v} }  $ of continuation
 $ \ottnt{E} ^{  \algeffseqoverindex{ \beta }{ \text{\unboldmath$\mathit{J}$} }  } $ for resumptions in $\ottnt{e}$ is defined in a capture-avoiding manner,
 as follows (we describe only important cases).
 \[\begin{array}{rcl}
   \ottsym{(}  \mathsf{resume} \,  \algeffseqoverindex{ \alpha }{ \text{\unboldmath$\mathit{I}$} }  \, \mathit{x}  \ottsym{.}  \ottnt{e}  \ottsym{)}    [   \ottnt{E} ^{  \algeffseqoverindex{ \beta }{ \text{\unboldmath$\mathit{J}$} }  }   /  \mathsf{resume}  ]^{  \algeffseqoverindex{  \text{\unboldmath$\forall$}  \,  \algeffseqoverindex{ \beta }{ \text{\unboldmath$\mathit{J}$} }   \ottsym{.}  \ottnt{A} }{ \text{\unboldmath$\mathit{I}$} }  }_{  \Lambda\!  \,  \algeffseqoverindex{ \beta }{ \text{\unboldmath$\mathit{J}$} }   \ottsym{.}  \ottnt{v} }   &\defeq \\
   \multicolumn{3}{r}{
   \mathsf{let} \, \mathit{y}  \ottsym{=}   \Lambda\!  \,  \algeffseqoverindex{ \beta }{ \text{\unboldmath$\mathit{J}$} }   \ottsym{.}     \ottnt{e}    [   \ottnt{E} ^{  \algeffseqoverindex{ \beta }{ \text{\unboldmath$\mathit{J}$} }  }   /  \mathsf{resume}  ]^{  \algeffseqoverindex{  \text{\unboldmath$\forall$}  \,  \algeffseqoverindex{ \beta }{ \text{\unboldmath$\mathit{J}$} }   \ottsym{.}  \ottnt{A} }{ \text{\unboldmath$\mathit{I}$} }  }_{  \Lambda\!  \,  \algeffseqoverindex{ \beta }{ \text{\unboldmath$\mathit{J}$} }   \ottsym{.}  \ottnt{v} }      [   \algeffseqoverindex{ \ottnt{A} }{ \text{\unboldmath$\mathit{I}$} }   \ottsym{/}   \algeffseqoverindex{ \alpha }{ \text{\unboldmath$\mathit{I}$} }   ]      [  \ottnt{v}  \ottsym{/}  \mathit{x}  ]   \,  \mathsf{in}  \,   \ottnt{E} ^{  \algeffseqoverindex{ \beta }{ \text{\unboldmath$\mathit{J}$} }  }   [  \mathit{y} \,  \algeffseqoverindex{ \beta }{ \text{\unboldmath$\mathit{J}$} }   ] } \\
    \multicolumn{3}{r}{\text{(if $\ottsym{(}   \mathit{ftv}  (  \ottnt{e}  )  \,  \mathbin{\cup}  \,  \mathit{ftv}  (   \ottnt{E} ^{  \algeffseqoverindex{ \beta }{ \text{\unboldmath$\mathit{J}$} }  }   )   \ottsym{)} \,  \mathbin{\cap}  \, \ottsym{\{}   \algeffseqoverindex{ \beta }{ \text{\unboldmath$\mathit{J}$} }   \ottsym{\}} \,  =  \,  \emptyset $ and $\mathit{y}$ is fresh)}} \\[.5ex]
      \ottsym{(}  \mathsf{return} \, \mathit{x}  \rightarrow  \ottnt{e}  \ottsym{)}    [  \ottnt{E}  /  \mathsf{resume}  ]^{  \algeffseqover{ \sigma }  }_{ \ottnt{w} }   &\defeq& \mathsf{return} \, \mathit{x}  \rightarrow   \ottnt{e}    [  \ottnt{E}  /  \mathsf{resume}  ]^{  \algeffseqover{ \sigma }  }_{ \ottnt{w} }   \\
      \ottsym{(}  \ottnt{h'}  \ottsym{;}   \Lambda\!  \,  \algeffseqoverindex{ \alpha }{ \text{\unboldmath$\mathit{J}$} }   \ottsym{.}  \mathsf{op}  \ottsym{(}  \mathit{x}  \ottsym{)}  \rightarrow  \ottnt{e}  \ottsym{)}    [  \ottnt{E}  /  \mathsf{resume}  ]^{  \algeffseqoverindex{ \sigma }{ \text{\unboldmath$\mathit{I}$} }  }_{ \ottnt{w} }   &\defeq&  \ottnt{h'}    [  \ottnt{E}  /  \mathsf{resume}  ]^{  \algeffseqoverindex{ \sigma }{ \text{\unboldmath$\mathit{I}$} }  }_{ \ottnt{w} }    \ottsym{;}   \Lambda\!  \,  \algeffseqoverindex{ \alpha }{ \text{\unboldmath$\mathit{J}$} }   \ottsym{.}  \mathsf{op}  \ottsym{(}  \mathit{x}  \ottsym{)}  \rightarrow  \ottnt{e} \\
   \end{array}\]
\end{defn}

\begin{defn}[Resumption type]
 We define resumption type $\ottnt{r}$ as follows.
 \[\begin{array}{lll}
  \ottnt{r}        & ::= &  \mathsf{none}  \mid \ottsym{(}   \algeffseqover{ \alpha }   \ottsym{,}  \ottnt{A}  \ottsym{,}   \ottnt{B}   \rightarrow  \!  \epsilon  \;  \ottnt{C}   \ottsym{)}
   \quad \text{(if $ \mathit{ftv}  (  \ottnt{A}  )  \,  \mathbin{\cup}  \,  \mathit{ftv}  (  \ottnt{B}  )  \,  \subseteq  \, \ottsym{\{}   \algeffseqover{ \alpha }   \ottsym{\}}$ and $ \mathit{ftv}  (  \ottnt{C}  )  \,  \mathbin{\cap}  \, \ottsym{\{}   \algeffseqover{ \alpha }   \ottsym{\}} \,  =  \,  \emptyset $)}
   \end{array}
 \]
 We also define a set of type variables captured by a resume type:
 \[\begin{array}{lll}
   \mathit{tyvars}  (   \mathsf{none}   )  &\defeq&  \emptyset  \\
   \mathit{tyvars}  (  \ottsym{(}   \algeffseqover{ \alpha }   \ottsym{,}  \ottnt{A}  \ottsym{,}   \ottnt{B}   \rightarrow  \!  \epsilon  \;  \ottnt{C}   \ottsym{)}  )  &\defeq& \ottsym{\{}   \algeffseqover{ \alpha }   \ottsym{\}}
   \end{array}\]
 \end{defn}

\newpage

\subsubsection{Semantics}

\begin{figure}[t!]
 \textbf{Reduction rules} \quad \framebox{$\ottnt{e_{{\mathrm{1}}}}  \rightsquigarrow  \ottnt{e_{{\mathrm{2}}}}$}
 \[\begin{array}{rcll@{\quad}rcll}
  \ottnt{c_{{\mathrm{1}}}} \, \ottnt{c_{{\mathrm{2}}}}               &  \rightsquigarrow  &  \zeta  (  \ottnt{c_{{\mathrm{1}}}}  ,  \ottnt{c_{{\mathrm{2}}}}  )  & \R{Const} &
  \ottsym{(}   \lambda\!  \, \mathit{x}  \ottsym{.}  \ottnt{e}  \ottsym{)} \, \ottnt{v}            &  \rightsquigarrow  &  \ottnt{e}    [  \ottnt{v}  \ottsym{/}  \mathit{x}  ]        & \R{Beta} \\[1.5ex]
  \multirow{2}{*}{\ensuremath{\mathsf{let} \, \mathit{x}  \ottsym{=}   \Lambda\!  \,  \algeffseqover{ \alpha }   \ottsym{.}  \ottnt{v} \,  \mathsf{in}  \, \ottnt{e}}} & \multirow{2}{*}{$ \rightsquigarrow $} & \multirow{2}{*}{$ \ottnt{e}    [   \Lambda\!  \,  \algeffseqover{ \alpha }   \ottsym{.}  \ottnt{v}  \ottsym{/}  \mathit{x}  ]  $} & \multirow{2}{*}{\R{Let}} &
  \mathsf{handle} \, \ottnt{v} \, \mathsf{with} \, \ottnt{h}     &  \rightsquigarrow  &  \ottnt{e}    [  \ottnt{v}  \ottsym{/}  \mathit{x}  ]   & \R{Return} \\
    &&&& \multicolumn{4}{r}{\text{(where $ \ottnt{h} ^\mathsf{return}  \,  =  \, \mathsf{return} \, \mathit{x}  \rightarrow  \ottnt{e}$)}} \\[1.5ex]
   \textup{\texttt{\#}\relax}  \mathsf{op}   \ottsym{(}    \algeffseqover{ \ottnt{A} }    \ottsym{,}   \ottnt{v}   \ottsym{)}            &  \rightsquigarrow  &  \textup{\texttt{\#}\relax}  \mathsf{op}   \ottsym{(}    \algeffseqover{ \ottnt{A} }    \ottsym{,}   \ottnt{v}   \ottsym{,}    [\,]    \ottsym{)}  & \R{Op} \\
   \end{array}\]
 \[\begin{array}{rcl@{\quad}l}
   \textup{\texttt{\#}\relax}  \mathsf{op}   \ottsym{(}    \algeffseqover{ \sigma }    \ottsym{,}   \ottnt{w}   \ottsym{,}   \ottnt{E}   \ottsym{)}  \, \ottnt{e_{{\mathrm{2}}}}   &  \rightsquigarrow  &  \textup{\texttt{\#}\relax}  \mathsf{op}   \ottsym{(}    \algeffseqover{ \sigma }    \ottsym{,}   \ottnt{w}   \ottsym{,}   \ottnt{E} \, \ottnt{e_{{\mathrm{2}}}}   \ottsym{)}  & \R{OpApp1} \\[1ex]
  \ottnt{v_{{\mathrm{1}}}} \,  \textup{\texttt{\#}\relax}  \mathsf{op}   \ottsym{(}    \algeffseqover{ \sigma }    \ottsym{,}   \ottnt{w}   \ottsym{,}   \ottnt{E}   \ottsym{)}    &  \rightsquigarrow  &  \textup{\texttt{\#}\relax}  \mathsf{op}   \ottsym{(}    \algeffseqover{ \sigma }    \ottsym{,}   \ottnt{w}   \ottsym{,}   \ottnt{v_{{\mathrm{1}}}} \, \ottnt{E}   \ottsym{)}  & \R{OpApp2} \\[1ex]
   \textup{\texttt{\#}\relax}  \mathsf{op}'   \ottsym{(}    \algeffseqoverindex{ \ottnt{A} }{ \text{\unboldmath$\mathit{I}$} }    \ottsym{,}    \textup{\texttt{\#}\relax}  \mathsf{op}   \ottsym{(}    \algeffseqoverindex{ \sigma }{ \text{\unboldmath$\mathit{J}$} }    \ottsym{,}   \ottnt{w}   \ottsym{,}   \ottnt{E}   \ottsym{)}    \ottsym{)}  &  \rightsquigarrow  &  \textup{\texttt{\#}\relax}  \mathsf{op}   \ottsym{(}    \algeffseqoverindex{ \sigma }{ \text{\unboldmath$\mathit{J}$} }    \ottsym{,}   \ottnt{w}   \ottsym{,}    \textup{\texttt{\#}\relax}  \mathsf{op}'   \ottsym{(}    \algeffseqoverindex{ \ottnt{A} }{ \text{\unboldmath$\mathit{I}$} }    \ottsym{,}   \ottnt{E}   \ottsym{)}    \ottsym{)}  & \R{OpOp} \\[1ex]
  \mathsf{handle} \,  \textup{\texttt{\#}\relax}  \mathsf{op}   \ottsym{(}    \algeffseqover{ \sigma }    \ottsym{,}   \ottnt{w}   \ottsym{,}   \ottnt{E}   \ottsym{)}  \, \mathsf{with} \, \ottnt{h} &  \rightsquigarrow  &  \textup{\texttt{\#}\relax}  \mathsf{op}   \ottsym{(}    \algeffseqover{ \sigma }    \ottsym{,}   \ottnt{w}   \ottsym{,}   \mathsf{handle} \, \ottnt{E} \, \mathsf{with} \, \ottnt{h}   \ottsym{)}  & \multirow{2}{*}{\R{OpHandle}} \\
    && \multicolumn{1}{r@{\quad}}{\text{(where $\mathsf{op} \,  \not\in  \,  \mathit{ops}  (  \ottnt{h}  ) $)}} \\[1ex]
  \mathsf{let} \, \mathit{x}  \ottsym{=}   \Lambda\!  \,  \algeffseqoverindex{ \alpha }{ \text{\unboldmath$\mathit{I}$} }   \ottsym{.}   \textup{\texttt{\#}\relax}  \mathsf{op}   \ottsym{(}    \algeffseqoverindex{ \sigma }{ \text{\unboldmath$\mathit{J}$} }    \ottsym{,}   \ottnt{w}   \ottsym{,}   \ottnt{E}   \ottsym{)}  \,  \mathsf{in}  \, \ottnt{e_{{\mathrm{2}}}} &  \rightsquigarrow  & & \multirow{2}{*}{\R{OpLet}} \\
    \multicolumn{3}{r@{\quad}}{ \textup{\texttt{\#}\relax}  \mathsf{op}   \ottsym{(}    \algeffseqoverindex{  \text{\unboldmath$\forall$}  \,  \algeffseqoverindex{ \alpha }{ \text{\unboldmath$\mathit{I}$} }   \ottsym{.}  \sigma }{ \text{\unboldmath$\mathit{J}$} }    \ottsym{,}    \Lambda\!  \,  \algeffseqoverindex{ \alpha }{ \text{\unboldmath$\mathit{I}$} }   \ottsym{.}  \ottnt{w}   \ottsym{,}   \mathsf{let} \, \mathit{x}  \ottsym{=}   \Lambda\!  \,  \algeffseqoverindex{ \alpha }{ \text{\unboldmath$\mathit{I}$} }   \ottsym{.}  \ottnt{E} \,  \mathsf{in}  \, \ottnt{e_{{\mathrm{2}}}}   \ottsym{)} } \\[1.5ex]
  \mathsf{handle} \,  \textup{\texttt{\#}\relax}  \mathsf{op}   \ottsym{(}    \algeffseqoverindex{  \text{\unboldmath$\forall$}  \,  \algeffseqoverindex{ \beta }{ \text{\unboldmath$\mathit{J}$} }   \ottsym{.}  \ottnt{A} }{ \text{\unboldmath$\mathit{I}$} }    \ottsym{,}    \Lambda\!  \,  \algeffseqoverindex{ \beta }{ \text{\unboldmath$\mathit{J}$} }   \ottsym{.}  \ottnt{v}   \ottsym{,}    \ottnt{E} ^{  \algeffseqoverindex{ \beta }{ \text{\unboldmath$\mathit{J}$} }  }    \ottsym{)}  \, \mathsf{with} \, \ottnt{h} &  \rightsquigarrow  & \\ \multicolumn{3}{r@{\quad}}{
      \ottnt{e}    [  \mathsf{handle} \,  \ottnt{E} ^{  \algeffseqoverindex{ \beta }{ \text{\unboldmath$\mathit{J}$} }  }  \, \mathsf{with} \, \ottnt{h}  /  \mathsf{resume}  ]^{  \algeffseqoverindex{  \text{\unboldmath$\forall$}  \,  \algeffseqoverindex{ \beta }{ \text{\unboldmath$\mathit{J}$} }   \ottsym{.}  \ottnt{A} }{ \text{\unboldmath$\mathit{I}$} }  }_{  \Lambda\!  \,  \algeffseqoverindex{ \beta }{ \text{\unboldmath$\mathit{J}$} }   \ottsym{.}  \ottnt{v} }      [    \algeffseqoverindex{ \ottnt{A} }{ \text{\unboldmath$\mathit{I}$} }     [   \algeffseqover{  \bot  }   \ottsym{/}   \algeffseqoverindex{ \beta }{ \text{\unboldmath$\mathit{J}$} }   ]    \ottsym{/}   \algeffseqoverindex{ \alpha }{ \text{\unboldmath$\mathit{I}$} }   ]      [   \ottnt{v}    [   \algeffseqover{  \bot  }   \ottsym{/}   \algeffseqoverindex{ \beta }{ \text{\unboldmath$\mathit{J}$} }   ]    \ottsym{/}  \mathit{x}  ]  } & \R{Handle} \\
   \multicolumn{3}{r@{\quad}}{\text{(where $ \ottnt{h} ^{ \mathsf{op} }  \,  =  \,  \Lambda\!  \,  \algeffseqoverindex{ \alpha }{ \text{\unboldmath$\mathit{I}$} }   \ottsym{.}  \mathsf{op}  \ottsym{(}  \mathit{x}  \ottsym{)}  \rightarrow  \ottnt{e}$)}}
   \end{array}\]
 \textbf{Evaluation rules} \quad \framebox{$\ottnt{e_{{\mathrm{1}}}}  \longrightarrow  \ottnt{e_{{\mathrm{2}}}}$}
 \begin{center}
  $\ottdruleEXXEval{}$
 \end{center}
 % \textbf{Reduction rules} \quad \framebox{$\ottnt{e_{{\mathrm{1}}}}  \rightsquigarrow  \ottnt{e_{{\mathrm{2}}}}$}
 % \begin{center}
 %  $\ottdruleRXXConst{}$ \hfil
 %  $\ottdruleRXXBeta{}$ \hfil
 %  $\ottdruleRXXLet{}$ \hfil
 %  $\ottdruleRXXReturn{}$ \hfil
 %  $\ottdruleRXXOp{}$ \hfil
 %  $\ottdruleRXXOpAppOne{}$ \\[3ex]
 %  $\ottdruleRXXOpAppTwo{}$ \hfil
 %  $\ottdruleRXXOpOp{}$ \\[3ex]
 %  $\ottdruleRXXOpHandle{}$ \\[3ex]
 %  $\ottdruleRXXOpLet{}$ \\[3ex]
 %  $\ottdruleRXXHandle{}$ \\[3ex]
 % \end{center}
 % %
 % \textbf{Evaluation rules} \quad \framebox{$\ottnt{e_{{\mathrm{1}}}}  \longrightarrow  \ottnt{e_{{\mathrm{2}}}}$}
 % \begin{center}
 %  $\ottdruleEXXEval{}$
 % \end{center}
 \caption{Semantics of {\interlang}.}
 \label{fig:app-inter-semantics}
\end{figure}

\begin{defn}
 Relations $ \longrightarrow $ and $ \rightsquigarrow $ are the smallest relations
 satisfying the rules in \reffig{app-inter-semantics}.
\end{defn}
 
\begin{defn}[Multi-step evaluation]
 Binary relation $ \longrightarrow^{*} $ over terms is the reflexive and transitive closure of $ \longrightarrow $.
\end{defn}

\begin{defn}[Nonreducible terms]
 We write $\ottnt{e}  \centernot\longrightarrow$ if there no terms $\ottnt{e'}$ such that $\ottnt{e}  \longrightarrow  \ottnt{e'}$.
\end{defn}

\newpage

\subsubsection{Typing}

\begin{figure}[t!]
\textbf{Typing rules} \\[1ex]
\framebox{$\Gamma  \ottsym{;}  \ottnt{r} \,   \vdash  \ottnt{e} \,   \ottsym{:}  \ottnt{A} \,  |  \, \epsilon$}
\begin{center}
 $\ottdruleTXXVar{}$ \hfil
 $\ottdruleTXXConst{}$ \\[3ex]
 $\ottdruleTXXAbs{}$ \hfil
 $\ottdruleTXXApp{}$ \\[3ex]
 $\ottdruleTXXOp{}$ \\[3ex]
 $\ottdruleTXXOpCont{}$ \\[3ex]
 $\ottdruleTXXWeak{}$ \hfil
 $\ottdruleTXXHandle{}$ \\[3ex]
 $\ottdruleTXXLet{}$ \\[3ex]
 $\ottdruleTXXResume{}$ \hfil
\end{center}
 \caption{Typing rules for terms in {\interlang}.}
 \label{fig:app-inter-typing1}
\end{figure}
\begin{figure}[t!]
\framebox{$\Gamma  \ottsym{;}  \ottnt{r} \,   \vdash  \ottnt{h}  \ottsym{:}  \ottnt{A} \,  |  \, \epsilon  \Rightarrow  \ottnt{B} \,  |  \, \epsilon'$}
\begin{center}
 $\ottdruleTHXXReturn{}$ \\[3ex]
 $\ottdruleTHXXOp{}$
\end{center}
\framebox{$ \Gamma   \vdash   \ottnt{E}   \ottsym{:}    \sigma  \multimap  \ottnt{A}   \,  |  \, \epsilon $}
\begin{center}
 $\ottdruleTEXXHole{}$ \\[3ex]
 $\ottdruleTEXXAppOne{}$ \\[3ex]
 $\ottdruleTEXXAppTwo{}$ \\[3ex]
 $\ottdruleTEXXOp{}$ \\[3ex]
 $\ottdruleTEXXHandle{}$ \\[3ex]
 $\ottdruleTEXXWeak{}$ \\[3ex]
 $\ottdruleTEXXLet{}$ \\[3ex]
\end{center}
 \caption{Typing rules for handlers and continuations in {\interlang}.}
 \label{fig:app-inter-typing2}
\end{figure}

\begin{defn}
 Judgments $\Gamma  \ottsym{;}  \ottnt{r} \,   \vdash  \ottnt{e} \,   \ottsym{:}  \ottnt{A} \,  |  \, \epsilon$ and $\Gamma  \ottsym{;}  \ottnt{r} \,   \vdash  \ottnt{h}  \ottsym{:}  \ottnt{A} \,  |  \, \epsilon  \Rightarrow  \ottnt{B} \,  |  \, \epsilon'$
 and $ \Gamma   \vdash   \ottnt{E}   \ottsym{:}    \sigma  \multimap  \ottnt{A}   \,  |  \, \epsilon $ are the smallest relations satisfying
 the rules in Figures~\ref{fig:app-inter-typing1} and \ref{fig:app-inter-typing2}.
\end{defn}

\begin{conv}[Typing judgments of values]
 We write $\Gamma  \vdash  \ottnt{v}  \ottsym{:}  \ottnt{A}$ if $\Gamma  \ottsym{;}   \mathsf{none}  \,   \vdash  \ottnt{v} \,   \ottsym{:}  \ottnt{A} \,  |  \, \epsilon$; effect $\epsilon$
 given to value $\ottnt{v}$ can be any (validated by \reflem{val-any-eff}).
\end{conv}

\subsubsection{Elaboration}

\begin{figure}[t!]
 \textbf{Term elaboration rules} \quad \framebox{$ \Gamma ;  \ottnt{R}   \vdash   \ottnt{M}  :  \ottnt{A}  \,  |   \, \epsilon  \mathrel{\algefftransarrow{ \ottnt{S} } }  \ottnt{e} $}
 \begin{center}
  $\ottdruleElabXXVar{}$ \hfil
  $\ottdruleElabXXConst{}$ \\[3ex]
  $\ottdruleElabXXAbs{}$ \\[3ex]
  $\ottdruleElabXXApp{}$ \\[3ex]
  $\ottdruleElabXXOp{}$ \\[3ex]
  $\ottdruleElabXXHandle{}$ \\[3ex]
  $\ottdruleElabXXLet{}$ \\[3ex]
  $\ottdruleElabXXResume{}$ \\[3ex]
  $\ottdruleElabXXWeak{}$ \\[3ex]
 \end{center}
 \textbf{Handler elaboration rules} \quad \framebox{$ \Gamma ;  \ottnt{R}   \vdash   \ottnt{H}  :  \ottnt{A}  \,  |   \, \epsilon  \, \Rightarrow   \ottnt{B}  \,  |   \, \epsilon'  \mathrel{\algefftransarrow{ \ottnt{S} } }  \ottnt{h} $}
 \begin{center}
  $\ottdruleElabHXXReturn{}$ \\[3ex]
  $\ottdruleElabHXXOp{}$ \\[3ex]
 \end{center}
 \caption{Elaboration rules.}
 \label{fig:app-elaboration}
\end{figure}

\begin{defn}
 Relation $ \Gamma ;  \ottnt{R}   \vdash   \ottnt{M}  :  \ottnt{A}  \,  |   \, \epsilon  \mathrel{\algefftransarrow{ \ottnt{S} } }  \ottnt{e} $ is the smallest relation satisfying
 the rules in \reffig{app-elaboration}.
\end{defn}
\section{Proofs}

\subsection{Type soundness of {\interlang}}

%%% Weakening

\begin{lemma}{weakening-typing-context}
 Suppose that $\vdash  \Gamma_{{\mathrm{1}}}  \ottsym{,}  \Gamma_{{\mathrm{2}}}$ and $\vdash  \Gamma_{{\mathrm{1}}}  \ottsym{,}  \Gamma_{{\mathrm{3}}}$ and
 $ \mathit{dom}  (  \Gamma_{{\mathrm{2}}}  )  \,  \mathbin{\cap}  \,  \mathit{dom}  (  \Gamma_{{\mathrm{3}}}  )  \,  =  \,  \emptyset $.
 Then, $\vdash  \Gamma_{{\mathrm{1}}}  \ottsym{,}  \Gamma_{{\mathrm{2}}}  \ottsym{,}  \Gamma_{{\mathrm{3}}}$.
\end{lemma}
\begin{proof}
 Straightforward by induction on $\Gamma_{{\mathrm{3}}}$.
\end{proof}

\begin{lemmap}{Weakening}{weakening}
 Suppose that $\vdash  \Gamma_{{\mathrm{1}}}  \ottsym{,}  \Gamma_{{\mathrm{2}}}$.
 Let $\Gamma_{{\mathrm{3}}}$ be a typing context such that
 $ \mathit{dom}  (  \Gamma_{{\mathrm{2}}}  )  \,  \mathbin{\cap}  \,  \mathit{dom}  (  \Gamma_{{\mathrm{3}}}  )  \,  =  \,  \emptyset $.
 \begin{enumerate}
  \item \label{lem:weakening:term}
        If $\Gamma_{{\mathrm{1}}}  \ottsym{,}  \Gamma_{{\mathrm{3}}}  \ottsym{;}  \ottnt{r} \,   \vdash  \ottnt{e} \,   \ottsym{:}  \ottnt{A} \,  |  \, \epsilon$,
        then $\Gamma_{{\mathrm{1}}}  \ottsym{,}  \Gamma_{{\mathrm{2}}}  \ottsym{,}  \Gamma_{{\mathrm{3}}}  \ottsym{;}  \ottnt{r} \,   \vdash  \ottnt{e} \,   \ottsym{:}  \ottnt{A} \,  |  \, \epsilon$.
  \item If $\Gamma_{{\mathrm{1}}}  \ottsym{,}  \Gamma_{{\mathrm{3}}}  \ottsym{;}  \ottnt{r} \,   \vdash  \ottnt{h}  \ottsym{:}  \ottnt{A} \,  |  \, \epsilon  \Rightarrow  \ottnt{B} \,  |  \, \epsilon'$,
        then $\Gamma_{{\mathrm{1}}}  \ottsym{,}  \Gamma_{{\mathrm{2}}}  \ottsym{,}  \Gamma_{{\mathrm{3}}}  \ottsym{;}  \ottnt{r} \,   \vdash  \ottnt{h}  \ottsym{:}  \ottnt{A} \,  |  \, \epsilon  \Rightarrow  \ottnt{B} \,  |  \, \epsilon'$.
  \item \label{lem:weakening:ectx}
        If $ \Gamma_{{\mathrm{1}}}  \ottsym{,}  \Gamma_{{\mathrm{3}}}   \vdash   \ottnt{E}   \ottsym{:}    \ottnt{A}  \multimap  \ottnt{B}   \,  |  \, \epsilon $,
        then $ \Gamma_{{\mathrm{1}}}  \ottsym{,}  \Gamma_{{\mathrm{2}}}  \ottsym{,}  \Gamma_{{\mathrm{3}}}   \vdash   \ottnt{E}   \ottsym{:}    \ottnt{A}  \multimap  \ottnt{B}   \,  |  \, \epsilon $.
 \end{enumerate}
\end{lemmap}
\begin{proof}
 By mutual induction on the typing derivations.
 We mention only the interesting cases.
 \begin{caseanalysis}
  \case \T{Var} and \T{Const}: By \reflem{weakening-typing-context}.

  \case \T{Abs}:
  We are given $\Gamma_{{\mathrm{1}}}  \ottsym{,}  \Gamma_{{\mathrm{3}}}  \ottsym{;}  \ottnt{r} \,   \vdash   \lambda\!  \, \mathit{x}  \ottsym{.}  \ottnt{e'} \,   \ottsym{:}   \ottnt{A'}   \rightarrow  \!  \epsilon'  \;  \ottnt{B'}  \,  |  \, \epsilon$ and,
  by inversion, $\Gamma_{{\mathrm{1}}}  \ottsym{,}  \Gamma_{{\mathrm{3}}}  \ottsym{,}  \mathit{x} \,  \mathord{:}  \, \ottnt{A'}  \ottsym{;}  \ottnt{r} \,   \vdash  \ottnt{e'} \,   \ottsym{:}  \ottnt{B'} \,  |  \, \epsilon'$.
  With loss of generality, we can suppose that
  $\mathit{x} \,  \not\in  \,  \mathit{dom}  (  \Gamma_{{\mathrm{1}}}  \ottsym{,}  \Gamma_{{\mathrm{2}}}  \ottsym{,}  \Gamma_{{\mathrm{3}}}  ) $.
  By the IH, $\Gamma_{{\mathrm{1}}}  \ottsym{,}  \Gamma_{{\mathrm{2}}}  \ottsym{,}  \Gamma_{{\mathrm{3}}}  \ottsym{,}  \mathit{x} \,  \mathord{:}  \, \ottnt{A'}  \ottsym{;}  \ottnt{r} \,   \vdash  \ottnt{e'} \,   \ottsym{:}  \ottnt{B'} \,  |  \, \epsilon'$.
  Thus, by \T{Abs},
  $\Gamma_{{\mathrm{1}}}  \ottsym{,}  \Gamma_{{\mathrm{2}}}  \ottsym{,}  \Gamma_{{\mathrm{3}}}  \ottsym{;}  \ottnt{r} \,   \vdash   \lambda\!  \, \mathit{x}  \ottsym{.}  \ottnt{e'} \,   \ottsym{:}   \ottnt{A'}   \rightarrow  \!  \epsilon'  \;  \ottnt{B'}  \,  |  \, \epsilon$.
 \end{caseanalysis}
\end{proof}

%%% Type substitution

\begin{lemma}{ty-subst-typing-context}
 If $\vdash  \Gamma_{{\mathrm{1}}}  \ottsym{,}  \alpha  \ottsym{,}  \Gamma_{{\mathrm{2}}}$ and $\Gamma_{{\mathrm{1}}}  \vdash  \ottnt{A}$, then $\vdash  \Gamma_{{\mathrm{1}}}  \ottsym{,}  \Gamma_{{\mathrm{2}}} \,  [  \ottnt{A}  \ottsym{/}  \alpha  ] $.
\end{lemma}
\begin{proof}
 By induction on $\Gamma_{{\mathrm{2}}}$.
\end{proof}

\begin{lemma}{typing-context-wf}
 If $\Gamma  \ottsym{;}  \ottnt{r} \,   \vdash  \ottnt{e} \,   \ottsym{:}  \ottnt{A} \,  |  \, \epsilon$, then $\vdash  \Gamma$.
\end{lemma}
\begin{proof}
 By induction on the derivation of $\Gamma  \ottsym{;}  \ottnt{r} \,   \vdash  \ottnt{e} \,   \ottsym{:}  \ottnt{A} \,  |  \, \epsilon$.
\end{proof}

\begin{lemmap}{Type substitution}{ty-subst}
 Suppose that $\Gamma_{{\mathrm{1}}}  \vdash  \ottnt{A}$ and $\alpha \,  \not\in  \,  \mathit{tyvars}  (  \ottnt{r}  ) $ and
 $ \mathit{tyvars}  (  \ottnt{r}  )  \,  \subseteq  \,  \mathit{dom}  (  \Gamma_{{\mathrm{2}}}  ) $.
 \begin{enumerate}
  \item \label{lem:ty-subst:term}
        If $\Gamma_{{\mathrm{1}}}  \ottsym{,}  \alpha  \ottsym{,}  \Gamma_{{\mathrm{2}}}  \ottsym{;}  \ottnt{r} \,   \vdash  \ottnt{e} \,   \ottsym{:}  \ottnt{B} \,  |  \, \epsilon$,
        then $\Gamma_{{\mathrm{1}}}  \ottsym{,}  \Gamma_{{\mathrm{2}}} \,  [  \ottnt{A}  \ottsym{/}  \alpha  ]   \ottsym{;}   \ottnt{r}    [  \ottnt{A}  \ottsym{/}  \alpha  ]   \,   \vdash   \ottnt{e}    [  \ottnt{A}  \ottsym{/}  \alpha  ]   \,   \ottsym{:}   \ottnt{B}    [  \ottnt{A}  \ottsym{/}  \alpha  ]   \,  |  \, \epsilon$.

  \item If $\Gamma_{{\mathrm{1}}}  \ottsym{,}  \alpha  \ottsym{,}  \Gamma_{{\mathrm{2}}}  \ottsym{;}  \ottnt{r} \,   \vdash  \ottnt{h}  \ottsym{:}  \ottnt{B} \,  |  \, \epsilon  \Rightarrow  \ottnt{C} \,  |  \, \epsilon'$,
        then $\Gamma_{{\mathrm{1}}}  \ottsym{,}  \Gamma_{{\mathrm{2}}} \,  [  \ottnt{A}  \ottsym{/}  \alpha  ]   \ottsym{;}   \ottnt{r}    [  \ottnt{A}  \ottsym{/}  \alpha  ]   \,   \vdash   \ottnt{h}    [  \ottnt{A}  \ottsym{/}  \alpha  ]    \ottsym{:}   \ottnt{B}    [  \ottnt{A}  \ottsym{/}  \alpha  ]   \,  |  \, \epsilon  \Rightarrow   \ottnt{C}    [  \ottnt{A}  \ottsym{/}  \alpha  ]   \,  |  \, \epsilon'$.

  \item \label{lem:ty-subst:ectx}
        If $ \Gamma_{{\mathrm{1}}}  \ottsym{,}  \alpha  \ottsym{,}  \Gamma_{{\mathrm{2}}}   \vdash   \ottnt{E}   \ottsym{:}    \ottnt{B}  \multimap  \ottnt{C}   \,  |  \, \epsilon $,
        then $ \Gamma_{{\mathrm{1}}}  \ottsym{,}  \Gamma_{{\mathrm{2}}} \,  [  \ottnt{A}  \ottsym{/}  \alpha  ]    \vdash    \ottnt{E}    [  \ottnt{A}  \ottsym{/}  \alpha  ]     \ottsym{:}     \ottnt{B}    [  \ottnt{A}  \ottsym{/}  \alpha  ]    \multimap   \ottnt{C}    [  \ottnt{A}  \ottsym{/}  \alpha  ]     \,  |  \, \epsilon $.
 \end{enumerate}
\end{lemmap}
\begin{proof}
 By mutual induction on the typing derivations.
 We mention only the interesting cases.
 \begin{caseanalysis}
  \case \T{Var} and \T{Const}: By \reflem{ty-subst-typing-context}.

  \case \T{Resume}:
  We are given
  $\Gamma_{{\mathrm{1}}}  \ottsym{,}  \alpha  \ottsym{,}  \Gamma_{{\mathrm{2}}}  \ottsym{;}  \ottsym{(}   \algeffseqover{ \beta }   \ottsym{,}  \ottnt{A'}  \ottsym{,}   \ottnt{B'}   \rightarrow  \!  \epsilon'  \;  \ottnt{B}   \ottsym{)} \,   \vdash  \mathsf{resume} \,  \algeffseqover{ \gamma }  \, \mathit{x}  \ottsym{.}  \ottnt{e'} \,   \ottsym{:}  \ottnt{B} \,  |  \, \epsilon$
  and, by inversion,
  \begin{itemize}
   \item $ \algeffseqover{ \beta }  \,  \in  \, \Gamma_{{\mathrm{1}}}  \ottsym{,}  \alpha  \ottsym{,}  \Gamma_{{\mathrm{2}}}$,
   \item $\Gamma_{{\mathrm{1}}}  \ottsym{,}  \alpha  \ottsym{,}  \Gamma_{{\mathrm{2}}}  \ottsym{,}   \algeffseqover{ \gamma }   \ottsym{,}  \mathit{x} \,  \mathord{:}  \, \ottnt{A'} \,  [   \algeffseqover{ \gamma }   \ottsym{/}   \algeffseqover{ \beta }   ]   \ottsym{;}  \ottsym{(}   \algeffseqover{ \beta }   \ottsym{,}  \ottnt{A'}  \ottsym{,}   \ottnt{B'}   \rightarrow  \!  \epsilon'  \;  \ottnt{B}   \ottsym{)} \,   \vdash  \ottnt{e'} \,   \ottsym{:}   \ottnt{B'}    [   \algeffseqover{ \gamma }   \ottsym{/}   \algeffseqover{ \beta }   ]   \,  |  \, \epsilon'$, and
   \item $\epsilon' \,  \subseteq  \, \epsilon$.
  \end{itemize}
  Without loss of generality, we can suppose that $\alpha \,  \not\in  \,  \algeffseqover{ \gamma } $.
  By the IH,
  \[
  \Gamma_{{\mathrm{1}}}  \ottsym{,}  \Gamma_{{\mathrm{2}}} \,  [  \ottnt{A}  \ottsym{/}  \alpha  ]   \ottsym{,}   \algeffseqover{ \gamma }   \ottsym{,}  \mathit{x} \,  \mathord{:}  \, \ottnt{A'} \,  [   \algeffseqover{ \gamma }   \ottsym{/}   \algeffseqover{ \beta }   ]   \ottsym{;}  \ottsym{(}   \algeffseqover{ \beta }   \ottsym{,}  \ottnt{A'}  \ottsym{,}   \ottnt{B'}   \rightarrow  \!  \epsilon'  \;   \ottnt{B}    [  \ottnt{A}  \ottsym{/}  \alpha  ]     \ottsym{)} \,   \vdash   \ottnt{e'}    [  \ottnt{A}  \ottsym{/}  \alpha  ]   \,   \ottsym{:}   \ottnt{B'}    [   \algeffseqover{ \gamma }   \ottsym{/}   \algeffseqover{ \beta }   ]   \,  |  \, \epsilon'.
  \]
  Note that
  (1) $ \ottnt{A'}    [  \ottnt{A}  \ottsym{/}  \alpha  ]   \,  =  \, \ottnt{A'}$, $ \ottnt{B'}    [  \ottnt{A}  \ottsym{/}  \alpha  ]   \,  =  \, \ottnt{B'}$, and
      $  \ottnt{B'}    [   \algeffseqover{ \gamma }   \ottsym{/}   \algeffseqover{ \beta }   ]      [  \ottnt{A}  \ottsym{/}  \alpha  ]   \,  =  \,  \ottnt{B'}    [   \algeffseqover{ \gamma }   \ottsym{/}   \algeffseqover{ \beta }   ]  $ by the requirement
      to resume types and
  (2) $ \mathit{ftv}  (   \ottnt{B}    [  \ottnt{A}  \ottsym{/}  \alpha  ]    )  \,  \mathbin{\cap}  \,  \algeffseqover{ \beta }  \,  =  \,  \emptyset $
      since $\Gamma_{{\mathrm{1}}}  \vdash  \ottnt{A}$ and
            $ \algeffseqover{ \beta }  =  \mathit{tyvars}  (  \ottsym{(}   \algeffseqover{ \beta }   \ottsym{,}  \ottnt{A'}  \ottsym{,}   \ottnt{B'}   \rightarrow  \!  \epsilon'  \;  \ottnt{B}   \ottsym{)}  )  \,  \subseteq  \,  \mathit{dom}  (  \Gamma_{{\mathrm{2}}}  ) $ and
            $ \mathit{dom}  (  \Gamma_{{\mathrm{1}}}  )  \,  \mathbin{\cap}  \,  \mathit{dom}  (  \Gamma_{{\mathrm{2}}}  )  \,  =  \,  \emptyset $ by \reflem{typing-context-wf}.
  By \T{Resume}, we finish.
 \end{caseanalysis}
\end{proof}

%%% Value substitution

\begin{lemmap}{Values can be given any effects}{val-any-eff}
 If $\Gamma  \ottsym{;}  \ottnt{r} \,   \vdash  \ottnt{v} \,   \ottsym{:}  \ottnt{A} \,  |  \, \epsilon$, then $\Gamma  \ottsym{;}  \ottnt{r} \,   \vdash  \ottnt{v} \,   \ottsym{:}  \ottnt{A} \,  |  \, \epsilon'$
 for any $\epsilon'$.
\end{lemmap}
\begin{proof}
 Straightforward by induction on the derivation of $\Gamma  \ottsym{;}  \ottnt{r} \,   \vdash  \ottnt{v} \,   \ottsym{:}  \ottnt{A} \,  |  \, \epsilon$.
\end{proof}

\begin{lemma}{resume-type-activate}
 \begin{enumerate}
  \item If $\Gamma  \ottsym{;}   \mathsf{none}  \,   \vdash  \ottnt{e} \,   \ottsym{:}  \ottnt{A} \,  |  \, \epsilon$, then $\Gamma  \ottsym{;}  \ottnt{r} \,   \vdash  \ottnt{e} \,   \ottsym{:}  \ottnt{A} \,  |  \, \epsilon$ for any $\ottnt{r}$.
  \item If $\Gamma  \ottsym{;}   \mathsf{none}  \,   \vdash  \ottnt{h}  \ottsym{:}  \ottnt{A} \,  |  \, \epsilon  \Rightarrow  \ottnt{B} \,  |  \, \epsilon'$, then $\Gamma  \ottsym{;}  \ottnt{r} \,   \vdash  \ottnt{h}  \ottsym{:}  \ottnt{A} \,  |  \, \epsilon  \Rightarrow  \ottnt{B} \,  |  \, \epsilon'$ for any $\ottnt{r}$.
 \end{enumerate}
\end{lemma}
\begin{proof}
 Straightforward by mutual induction on the typing derivations.
\end{proof}

\begin{lemma}{strengthening-typing-context}
 If $\vdash  \Gamma_{{\mathrm{1}}}  \ottsym{,}  \mathit{x} \,  \mathord{:}  \, \sigma  \ottsym{,}  \Gamma_{{\mathrm{2}}}$, then $\vdash  \Gamma_{{\mathrm{1}}}  \ottsym{,}  \Gamma_{{\mathrm{2}}}$.
\end{lemma}
\begin{proof}
 Straightforward by induction on $\Gamma_{{\mathrm{2}}}$.
\end{proof}

\begin{lemmap}{Value substitution}{val-subst}
 Suppose that $\Gamma_{{\mathrm{1}}}  \ottsym{,}   \algeffseqoverindex{ \alpha }{ \text{\unboldmath$\mathit{I}$} }   \vdash  \ottnt{v}  \ottsym{:}  \ottnt{A}$.
 \begin{enumerate}
  \item \label{lem:val-subst:term}
        If $\Gamma_{{\mathrm{1}}}  \ottsym{,}  \mathit{x} \,  \mathord{:}  \,  \text{\unboldmath$\forall$}  \,  \algeffseqoverindex{ \alpha }{ \text{\unboldmath$\mathit{I}$} }   \ottsym{.}  \ottnt{A}  \ottsym{,}  \Gamma_{{\mathrm{2}}}  \ottsym{;}  \ottnt{r} \,   \vdash  \ottnt{e} \,   \ottsym{:}  \ottnt{B} \,  |  \, \epsilon$,
        then $\Gamma_{{\mathrm{1}}}  \ottsym{,}  \Gamma_{{\mathrm{2}}}  \ottsym{;}  \ottnt{r} \,   \vdash   \ottnt{e}    [   \Lambda\!  \,  \algeffseqoverindex{ \alpha }{ \text{\unboldmath$\mathit{I}$} }   \ottsym{.}  \ottnt{v}  \ottsym{/}  \mathit{x}  ]   \,   \ottsym{:}  \ottnt{B} \,  |  \, \epsilon$.

  \item If $\Gamma_{{\mathrm{1}}}  \ottsym{,}  \mathit{x} \,  \mathord{:}  \,  \text{\unboldmath$\forall$}  \,  \algeffseqoverindex{ \alpha }{ \text{\unboldmath$\mathit{I}$} }   \ottsym{.}  \ottnt{A}  \ottsym{,}  \Gamma_{{\mathrm{2}}}  \ottsym{;}  \ottnt{r} \,   \vdash  \ottnt{h}  \ottsym{:}  \ottnt{B} \,  |  \, \epsilon  \Rightarrow  \ottnt{C} \,  |  \, \epsilon'$,
        then $\Gamma_{{\mathrm{1}}}  \ottsym{,}  \Gamma_{{\mathrm{2}}}  \ottsym{;}  \ottnt{r} \,   \vdash   \ottnt{h}    [   \Lambda\!  \,  \algeffseqoverindex{ \alpha }{ \text{\unboldmath$\mathit{I}$} }   \ottsym{.}  \ottnt{v}  \ottsym{/}  \mathit{x}  ]    \ottsym{:}  \ottnt{B} \,  |  \, \epsilon  \Rightarrow  \ottnt{C} \,  |  \, \epsilon'$.

  \item If $ \Gamma_{{\mathrm{1}}}  \ottsym{,}  \mathit{x} \,  \mathord{:}  \,  \text{\unboldmath$\forall$}  \,  \algeffseqoverindex{ \alpha }{ \text{\unboldmath$\mathit{I}$} }   \ottsym{.}  \ottnt{A}  \ottsym{,}  \Gamma_{{\mathrm{2}}}   \vdash   \ottnt{E}   \ottsym{:}    \ottnt{B}  \multimap  \ottnt{C}   \,  |  \, \epsilon $,
        then $ \Gamma_{{\mathrm{1}}}  \ottsym{,}  \Gamma_{{\mathrm{2}}}   \vdash    \ottnt{E}    [   \Lambda\!  \,  \algeffseqoverindex{ \alpha }{ \text{\unboldmath$\mathit{I}$} }   \ottsym{.}  \ottnt{v}  \ottsym{/}  \mathit{x}  ]     \ottsym{:}    \ottnt{B}  \multimap  \ottnt{C}   \,  |  \, \epsilon $.
 \end{enumerate}
\end{lemmap}
\begin{proof}
 By mutual induction on the typing derivations.
 We mention only the interesting cases.
 \begin{caseanalysis}
  \case \T{Var}:
  We are given $\Gamma_{{\mathrm{1}}}  \ottsym{,}  \mathit{x} \,  \mathord{:}  \,  \text{\unboldmath$\forall$}  \,  \algeffseqoverindex{ \alpha }{ \text{\unboldmath$\mathit{I}$} }   \ottsym{.}  \ottnt{A}  \ottsym{,}  \Gamma_{{\mathrm{2}}}  \ottsym{;}  \ottnt{r} \,   \vdash  \mathit{y} \,  \algeffseqoverindex{ \ottnt{C} }{ \text{\unboldmath$\mathit{J}$} }  \,   \ottsym{:}   \ottnt{D}    [   \algeffseqoverindex{ \ottnt{C} }{ \text{\unboldmath$\mathit{J}$} }   \ottsym{/}   \algeffseqoverindex{ \beta }{ \text{\unboldmath$\mathit{J}$} }   ]   \,  |  \, \epsilon$ and, by inversion,
  \begin{itemize}
   \item $\vdash  \Gamma_{{\mathrm{1}}}  \ottsym{,}  \mathit{x} \,  \mathord{:}  \,  \text{\unboldmath$\forall$}  \,  \algeffseqoverindex{ \alpha }{ \text{\unboldmath$\mathit{I}$} }   \ottsym{.}  \ottnt{A}  \ottsym{,}  \Gamma_{{\mathrm{2}}}$,
   \item $\mathit{y} \,  \mathord{:}  \,  \text{\unboldmath$\forall$}  \,  \algeffseqoverindex{ \beta }{ \text{\unboldmath$\mathit{J}$} }   \ottsym{.}  \ottnt{D} \,  \in  \, \Gamma_{{\mathrm{1}}}  \ottsym{,}  \mathit{x} \,  \mathord{:}  \,  \text{\unboldmath$\forall$}  \,  \algeffseqoverindex{ \alpha }{ \text{\unboldmath$\mathit{I}$} }   \ottsym{.}  \ottnt{A}  \ottsym{,}  \Gamma_{{\mathrm{2}}}$, and
   \item $\Gamma_{{\mathrm{1}}}  \ottsym{,}  \mathit{x} \,  \mathord{:}  \,  \text{\unboldmath$\forall$}  \,  \algeffseqoverindex{ \alpha }{ \text{\unboldmath$\mathit{I}$} }   \ottsym{.}  \ottnt{A}  \ottsym{,}  \Gamma_{{\mathrm{2}}}  \vdash   \algeffseqoverindex{ \ottnt{C} }{ \text{\unboldmath$\mathit{J}$} } $.
  \end{itemize}
  By \reflem{strengthening-typing-context}, $\vdash  \Gamma_{{\mathrm{1}}}  \ottsym{,}  \Gamma_{{\mathrm{2}}}$.

  If $\mathit{x} \,  \not=  \, \mathit{y}$, then the conclusion is obvious by \T{Var}.
  Otherwise, if $\mathit{x} \,  =  \, \mathit{y}$, then $ \text{\unboldmath$\forall$}  \,  \algeffseqoverindex{ \alpha }{ \text{\unboldmath$\mathit{I}$} }   \ottsym{.}  \ottnt{A} \,  =  \,  \text{\unboldmath$\forall$}  \,  \algeffseqoverindex{ \beta }{ \text{\unboldmath$\mathit{J}$} }   \ottsym{.}  \ottnt{D}$ and so we have to show
  that
  \[
   \Gamma_{{\mathrm{1}}}  \ottsym{,}  \Gamma_{{\mathrm{2}}}  \ottsym{;}  \ottnt{r} \,   \vdash   \ottnt{v}    [   \algeffseqoverindex{ \ottnt{C} }{ \text{\unboldmath$\mathit{I}$} }   \ottsym{/}   \algeffseqoverindex{ \alpha }{ \text{\unboldmath$\mathit{I}$} }   ]   \,   \ottsym{:}   \ottnt{A}    [   \algeffseqoverindex{ \ottnt{C} }{ \text{\unboldmath$\mathit{I}$} }   \ottsym{/}   \algeffseqoverindex{ \alpha }{ \text{\unboldmath$\mathit{I}$} }   ]   \,  |  \, \epsilon.
  \]
  Since $\Gamma_{{\mathrm{1}}}  \ottsym{,}   \algeffseqoverindex{ \alpha }{ \text{\unboldmath$\mathit{I}$} }   \vdash  \ottnt{v}  \ottsym{:}  \ottnt{A}$ (i.e., $\Gamma_{{\mathrm{1}}}  \ottsym{,}   \algeffseqoverindex{ \alpha }{ \text{\unboldmath$\mathit{I}$} }   \ottsym{;}   \mathsf{none}  \,   \vdash  \ottnt{v} \,   \ottsym{:}  \ottnt{A} \,  |  \, \epsilon'$ for some $\epsilon'$),
  we have it
  by Lemmas~\ref{lem:ty-subst}, \ref{lem:weakening}, \ref{lem:val-any-eff}, and \ref{lem:resume-type-activate}.

  \case \T{Const}: By \reflem{strengthening-typing-context}.

 \end{caseanalysis}
\end{proof}

%%% Continuation substitution

\begin{lemma}{ectx-placement}
 If $ \Gamma   \vdash    \ottnt{E} ^{  \algeffseqoverindex{ \alpha }{ \text{\unboldmath$\mathit{J}$} }  }    \ottsym{:}     \text{\unboldmath$\forall$}  \,  \algeffseqoverindex{ \alpha }{ \text{\unboldmath$\mathit{J}$} }   \ottsym{.}  \ottnt{A}  \multimap  \ottnt{B}   \,  |  \, \epsilon $ and $\Gamma  \ottsym{,}   \algeffseqoverindex{ \alpha }{ \text{\unboldmath$\mathit{J}$} }   \ottsym{;}   \mathsf{none}  \,   \vdash  \ottnt{e} \,   \ottsym{:}  \ottnt{A} \,  |  \,  \langle \rangle $,
 then $\Gamma  \ottsym{;}   \mathsf{none}  \,   \vdash    \ottnt{E} ^{  \algeffseqoverindex{ \alpha }{ \text{\unboldmath$\mathit{J}$} }  }   [  \ottnt{e}  ]  \,   \ottsym{:}  \ottnt{B} \,  |  \, \epsilon$.
\end{lemma}
\begin{proof}
 By induction on the derivation of $ \Gamma   \vdash    \ottnt{E} ^{  \algeffseqoverindex{ \alpha }{ \text{\unboldmath$\mathit{J}$} }  }    \ottsym{:}     \text{\unboldmath$\forall$}  \,  \algeffseqoverindex{ \alpha }{ \text{\unboldmath$\mathit{J}$} }   \ottsym{.}  \ottnt{A}  \multimap  \ottnt{B}   \,  |  \, \epsilon $.
 \begin{caseanalysis}
  \case \TE{Hole}: By \T{Weak}.
  \case \TE{App1} and \TE{App2}: By the IH, \T{Weak}, and \T{App}.
  \case \TE{Op}: By the IH and \T{Op}.
  \case \TE{Handle}: By the IH and \T{Handle}.
  \case \TE{Weak}: By the IH and \T{Weak}.
  \case \TE{Let}: By the IH and \T{Let}.
 \end{caseanalysis}
\end{proof}

\ifrestate
\lemmContSubst*
\else
\begin{lemmap}{Continuation substitution}{cont-subst}
 Suppose that
 $ \Gamma   \vdash    \ottnt{E} ^{  \algeffseqoverindex{ \beta }{ \text{\unboldmath$\mathit{J}$} }  }    \ottsym{:}     \text{\unboldmath$\forall$}  \,  \algeffseqoverindex{ \beta }{ \text{\unboldmath$\mathit{J}$} }   \ottsym{.}  \ottsym{(}   \ottnt{B}    [   \algeffseqoverindex{ \ottnt{C} }{ \text{\unboldmath$\mathit{I}$} }   \ottsym{/}   \algeffseqoverindex{ \alpha }{ \text{\unboldmath$\mathit{I}$} }   ]    \ottsym{)}  \multimap  \ottnt{D}   \,  |  \, \epsilon $ and
 $\Gamma  \ottsym{,}   \algeffseqoverindex{ \beta }{ \text{\unboldmath$\mathit{J}$} }   \vdash  \ottnt{v}  \ottsym{:}   \ottnt{A}    [   \algeffseqoverindex{ \ottnt{C} }{ \text{\unboldmath$\mathit{I}$} }   \ottsym{/}   \algeffseqoverindex{ \alpha }{ \text{\unboldmath$\mathit{I}$} }   ]  $ and
 $\Gamma  \vdash   \algeffseqoverindex{  \text{\unboldmath$\forall$}  \,  \algeffseqoverindex{ \beta }{ \text{\unboldmath$\mathit{J}$} }   \ottsym{.}  \ottnt{C} }{ \text{\unboldmath$\mathit{I}$} } $.
 \begin{enumerate}
  \item If $\Gamma  \ottsym{;}  \ottsym{(}   \algeffseqoverindex{ \alpha }{ \text{\unboldmath$\mathit{I}$} }   \ottsym{,}  \ottnt{A}  \ottsym{,}   \ottnt{B}   \rightarrow  \!  \epsilon  \;  \ottnt{D}   \ottsym{)} \,   \vdash  \ottnt{e} \,   \ottsym{:}  \ottnt{D'} \,  |  \, \epsilon'$, then
        $\Gamma  \ottsym{;}   \mathsf{none}  \,   \vdash   \ottnt{e}    [   \ottnt{E} ^{  \algeffseqoverindex{ \beta }{ \text{\unboldmath$\mathit{J}$} }  }   /  \mathsf{resume}  ]^{  \algeffseqoverindex{  \text{\unboldmath$\forall$}  \,  \algeffseqoverindex{ \beta }{ \text{\unboldmath$\mathit{J}$} }   \ottsym{.}  \ottnt{C} }{ \text{\unboldmath$\mathit{I}$} }  }_{  \Lambda\!  \,  \algeffseqoverindex{ \beta }{ \text{\unboldmath$\mathit{J}$} }   \ottsym{.}  \ottnt{v} }   \,   \ottsym{:}  \ottnt{D'} \,  |  \, \epsilon'$.

  \item If $\Gamma  \ottsym{;}  \ottsym{(}   \algeffseqoverindex{ \alpha }{ \text{\unboldmath$\mathit{I}$} }   \ottsym{,}  \ottnt{A}  \ottsym{,}   \ottnt{B}   \rightarrow  \!  \epsilon  \;  \ottnt{D}   \ottsym{)} \,   \vdash  \ottnt{h}  \ottsym{:}  \ottnt{D_{{\mathrm{1}}}} \,  |  \, \epsilon_{{\mathrm{1}}}  \Rightarrow  \ottnt{D_{{\mathrm{2}}}} \,  |  \, \epsilon_{{\mathrm{2}}}$, then
        $\Gamma  \ottsym{;}   \mathsf{none}  \,   \vdash   \ottnt{h}    [   \ottnt{E} ^{  \algeffseqoverindex{ \beta }{ \text{\unboldmath$\mathit{J}$} }  }   /  \mathsf{resume}  ]^{  \algeffseqoverindex{  \text{\unboldmath$\forall$}  \,  \algeffseqoverindex{ \beta }{ \text{\unboldmath$\mathit{J}$} }   \ottsym{.}  \ottnt{C} }{ \text{\unboldmath$\mathit{I}$} }  }_{  \Lambda\!  \,  \algeffseqoverindex{ \beta }{ \text{\unboldmath$\mathit{J}$} }   \ottsym{.}  \ottnt{v} }    \ottsym{:}  \ottnt{D_{{\mathrm{1}}}} \,  |  \, \epsilon_{{\mathrm{1}}}  \Rightarrow  \ottnt{D_{{\mathrm{2}}}} \,  |  \, \epsilon_{{\mathrm{2}}}$.
 \end{enumerate}
\end{lemmap}
\fi
\begin{proof}
 By mutual induction on the typing derivations.
 \begin{enumerate}
  \item By case analysis on the typing rule applied last.
       \begin{caseanalysis}
        \case \T{Var} and \T{Const}: Obvious.
        \case \T{Abs}, \T{App}, \T{Op}, \T{Weak}, \T{Handle}, and \T{Let}: By the IH(s) with (if necessary) weakening (Lemmas~\ref{lem:typing-context-wf} and \ref{lem:weakening}).
        \case \T{OpCont}: By the IH; note that
        \[
           \textup{\texttt{\#}\relax}  \mathsf{op}   \ottsym{(}    \algeffseqover{ \sigma }    \ottsym{,}   \ottnt{w}   \ottsym{,}   \ottnt{E'}   \ottsym{)}     [   \ottnt{E} ^{  \algeffseqoverindex{ \beta }{ \text{\unboldmath$\mathit{J}$} }  }   /  \mathsf{resume}  ]^{  \algeffseqoverindex{  \text{\unboldmath$\forall$}  \,  \algeffseqoverindex{ \beta }{ \text{\unboldmath$\mathit{J}$} }   \ottsym{.}  \ottnt{C} }{ \text{\unboldmath$\mathit{I}$} }  }_{  \Lambda\!  \,  \algeffseqoverindex{ \beta }{ \text{\unboldmath$\mathit{J}$} }   \ottsym{.}  \ottnt{v} }   \,  =  \,  \textup{\texttt{\#}\relax}  \mathsf{op}   \ottsym{(}    \algeffseqover{ \sigma }    \ottsym{,}    \ottnt{w}    [   \ottnt{E} ^{  \algeffseqoverindex{ \beta }{ \text{\unboldmath$\mathit{J}$} }  }   /  \mathsf{resume}  ]^{  \algeffseqoverindex{  \text{\unboldmath$\forall$}  \,  \algeffseqoverindex{ \beta }{ \text{\unboldmath$\mathit{J}$} }   \ottsym{.}  \ottnt{C} }{ \text{\unboldmath$\mathit{I}$} }  }_{  \Lambda\!  \,  \algeffseqoverindex{ \beta }{ \text{\unboldmath$\mathit{J}$} }   \ottsym{.}  \ottnt{v} }     \ottsym{,}   \ottnt{E'}   \ottsym{)} .
        \]

        \case \T{Resume}:
        We are given
        $\Gamma  \ottsym{;}  \ottsym{(}   \algeffseqoverindex{ \alpha }{ \text{\unboldmath$\mathit{I}$} }   \ottsym{,}  \ottnt{A}  \ottsym{,}   \ottnt{B}   \rightarrow  \!  \epsilon  \;  \ottnt{D}   \ottsym{)} \,   \vdash  \mathsf{resume} \,  \algeffseqoverindex{ \gamma }{ \text{\unboldmath$\mathit{I}$} }  \, \mathit{x}  \ottsym{.}  \ottnt{e'} \,   \ottsym{:}  \ottnt{D} \,  |  \, \epsilon'$ and,
        by inversion,
        \begin{itemize}
         \item $ \algeffseqoverindex{ \alpha }{ \text{\unboldmath$\mathit{I}$} }  \,  \in  \, \Gamma$,
         \item $\Gamma  \ottsym{,}   \algeffseqoverindex{ \gamma }{ \text{\unboldmath$\mathit{I}$} }   \ottsym{,}  \mathit{x} \,  \mathord{:}  \, \ottnt{A} \,  [   \algeffseqoverindex{ \gamma }{ \text{\unboldmath$\mathit{I}$} }   \ottsym{/}   \algeffseqoverindex{ \alpha }{ \text{\unboldmath$\mathit{I}$} }   ]   \ottsym{;}  \ottsym{(}   \algeffseqoverindex{ \alpha }{ \text{\unboldmath$\mathit{I}$} }   \ottsym{,}  \ottnt{A}  \ottsym{,}   \ottnt{B}   \rightarrow  \!  \epsilon  \;  \ottnt{D}   \ottsym{)} \,   \vdash  \ottnt{e'} \,   \ottsym{:}   \ottnt{B}    [   \algeffseqoverindex{ \gamma }{ \text{\unboldmath$\mathit{I}$} }   \ottsym{/}   \algeffseqoverindex{ \alpha }{ \text{\unboldmath$\mathit{I}$} }   ]   \,  |  \, \epsilon'$ and
         \item $\epsilon \,  \subseteq  \, \epsilon'$.
        \end{itemize}
        Without loss of generality, we can suppose that each type variable
        of $ \algeffseqoverindex{ \gamma }{ \text{\unboldmath$\mathit{I}$} } $ is distinct from $ \algeffseqoverindex{ \beta }{ \text{\unboldmath$\mathit{J}$} } $.
        Thus, by weakening (Lemmas~\ref{lem:typing-context-wf} and \ref{lem:weakening}) and the IH,
        \[
         \Gamma  \ottsym{,}   \algeffseqoverindex{ \gamma }{ \text{\unboldmath$\mathit{I}$} }   \ottsym{,}  \mathit{x} \,  \mathord{:}  \, \ottnt{A} \,  [   \algeffseqoverindex{ \gamma }{ \text{\unboldmath$\mathit{I}$} }   \ottsym{/}   \algeffseqoverindex{ \alpha }{ \text{\unboldmath$\mathit{I}$} }   ]   \ottsym{;}   \mathsf{none}  \,   \vdash   \ottnt{e'}    [   \ottnt{E} ^{  \algeffseqoverindex{ \beta }{ \text{\unboldmath$\mathit{J}$} }  }   /  \mathsf{resume}  ]^{  \algeffseqoverindex{  \text{\unboldmath$\forall$}  \,  \algeffseqoverindex{ \beta }{ \text{\unboldmath$\mathit{J}$} }   \ottsym{.}  \ottnt{C} }{ \text{\unboldmath$\mathit{I}$} }  }_{  \Lambda\!  \,  \algeffseqoverindex{ \beta }{ \text{\unboldmath$\mathit{J}$} }   \ottsym{.}  \ottnt{v} }   \,   \ottsym{:}   \ottnt{B}    [   \algeffseqoverindex{ \gamma }{ \text{\unboldmath$\mathit{I}$} }   \ottsym{/}   \algeffseqoverindex{ \alpha }{ \text{\unboldmath$\mathit{I}$} }   ]   \,  |  \, \epsilon'.
        \]
        By \reflem{weakening},
        \[
         \Gamma  \ottsym{,}   \algeffseqoverindex{ \beta }{ \text{\unboldmath$\mathit{J}$} }   \ottsym{,}   \algeffseqoverindex{ \gamma }{ \text{\unboldmath$\mathit{I}$} }   \ottsym{,}  \mathit{x} \,  \mathord{:}  \, \ottnt{A} \,  [   \algeffseqoverindex{ \gamma }{ \text{\unboldmath$\mathit{I}$} }   \ottsym{/}   \algeffseqoverindex{ \alpha }{ \text{\unboldmath$\mathit{I}$} }   ]   \ottsym{;}   \mathsf{none}  \,   \vdash   \ottnt{e'}    [   \ottnt{E} ^{  \algeffseqoverindex{ \beta }{ \text{\unboldmath$\mathit{J}$} }  }   /  \mathsf{resume}  ]^{  \algeffseqoverindex{  \text{\unboldmath$\forall$}  \,  \algeffseqoverindex{ \beta }{ \text{\unboldmath$\mathit{J}$} }   \ottsym{.}  \ottnt{C} }{ \text{\unboldmath$\mathit{I}$} }  }_{  \Lambda\!  \,  \algeffseqoverindex{ \beta }{ \text{\unboldmath$\mathit{J}$} }   \ottsym{.}  \ottnt{v} }   \,   \ottsym{:}   \ottnt{B}    [   \algeffseqoverindex{ \gamma }{ \text{\unboldmath$\mathit{I}$} }   \ottsym{/}   \algeffseqoverindex{ \alpha }{ \text{\unboldmath$\mathit{I}$} }   ]   \,  |  \, \epsilon'.
        \]
        Since $\Gamma  \ottsym{,}   \algeffseqoverindex{ \beta }{ \text{\unboldmath$\mathit{J}$} }   \vdash   \algeffseqoverindex{ \ottnt{C} }{ \text{\unboldmath$\mathit{I}$} } $, we have
        \[
         \Gamma  \ottsym{,}   \algeffseqoverindex{ \beta }{ \text{\unboldmath$\mathit{J}$} }   \ottsym{,}  \mathit{x} \,  \mathord{:}  \, \ottnt{A} \,  [   \algeffseqoverindex{ \ottnt{C} }{ \text{\unboldmath$\mathit{I}$} }   \ottsym{/}   \algeffseqoverindex{ \alpha }{ \text{\unboldmath$\mathit{I}$} }   ]   \ottsym{;}   \mathsf{none}  \,   \vdash    \ottnt{e'}    [   \ottnt{E} ^{  \algeffseqoverindex{ \beta }{ \text{\unboldmath$\mathit{J}$} }  }   /  \mathsf{resume}  ]^{  \algeffseqoverindex{  \text{\unboldmath$\forall$}  \,  \algeffseqoverindex{ \beta }{ \text{\unboldmath$\mathit{J}$} }   \ottsym{.}  \ottnt{C} }{ \text{\unboldmath$\mathit{I}$} }  }_{  \Lambda\!  \,  \algeffseqoverindex{ \beta }{ \text{\unboldmath$\mathit{J}$} }   \ottsym{.}  \ottnt{v} }      [   \algeffseqoverindex{ \ottnt{C} }{ \text{\unboldmath$\mathit{I}$} }   \ottsym{/}   \algeffseqoverindex{ \gamma }{ \text{\unboldmath$\mathit{I}$} }   ]   \,   \ottsym{:}   \ottnt{B}    [   \algeffseqoverindex{ \ottnt{C} }{ \text{\unboldmath$\mathit{I}$} }   \ottsym{/}   \algeffseqoverindex{ \alpha }{ \text{\unboldmath$\mathit{I}$} }   ]   \,  |  \, \epsilon'
        \]
        by \reflem{ty-subst}.
        Since $\Gamma  \ottsym{,}   \algeffseqoverindex{ \beta }{ \text{\unboldmath$\mathit{J}$} }   \vdash  \ottnt{v}  \ottsym{:}   \ottnt{A}    [   \algeffseqoverindex{ \ottnt{C} }{ \text{\unboldmath$\mathit{I}$} }   \ottsym{/}   \algeffseqoverindex{ \alpha }{ \text{\unboldmath$\mathit{I}$} }   ]  $,
        we have
        \begin{equation}
         \Gamma  \ottsym{,}   \algeffseqoverindex{ \beta }{ \text{\unboldmath$\mathit{J}$} }   \ottsym{;}   \mathsf{none}  \,   \vdash     \ottnt{e'}    [   \ottnt{E} ^{  \algeffseqoverindex{ \beta }{ \text{\unboldmath$\mathit{J}$} }  }   /  \mathsf{resume}  ]^{  \algeffseqoverindex{  \text{\unboldmath$\forall$}  \,  \algeffseqoverindex{ \beta }{ \text{\unboldmath$\mathit{J}$} }   \ottsym{.}  \ottnt{C} }{ \text{\unboldmath$\mathit{I}$} }  }_{  \Lambda\!  \,  \algeffseqoverindex{ \beta }{ \text{\unboldmath$\mathit{J}$} }   \ottsym{.}  \ottnt{v} }      [   \algeffseqoverindex{ \ottnt{C} }{ \text{\unboldmath$\mathit{I}$} }   \ottsym{/}   \algeffseqoverindex{ \gamma }{ \text{\unboldmath$\mathit{I}$} }   ]      [  \ottnt{v}  \ottsym{/}  \mathit{x}  ]   \,   \ottsym{:}   \ottnt{B}    [   \algeffseqoverindex{ \ottnt{C} }{ \text{\unboldmath$\mathit{I}$} }   \ottsym{/}   \algeffseqoverindex{ \alpha }{ \text{\unboldmath$\mathit{I}$} }   ]   \,  |  \, \epsilon'
         \label{eqn:cont-subst:one}
        \end{equation}
        by \reflem{val-subst}.

        Since $ \Gamma   \vdash    \ottnt{E} ^{  \algeffseqoverindex{ \beta }{ \text{\unboldmath$\mathit{J}$} }  }    \ottsym{:}     \text{\unboldmath$\forall$}  \,  \algeffseqoverindex{ \beta }{ \text{\unboldmath$\mathit{J}$} }   \ottsym{.}  \ottsym{(}   \ottnt{B}    [   \algeffseqoverindex{ \ottnt{C} }{ \text{\unboldmath$\mathit{I}$} }   \ottsym{/}   \algeffseqoverindex{ \alpha }{ \text{\unboldmath$\mathit{I}$} }   ]    \ottsym{)}  \multimap  \ottnt{D}   \,  |  \, \epsilon $, we have
        \[
          \Gamma  \ottsym{,}  \mathit{y} \,  \mathord{:}  \,  \text{\unboldmath$\forall$}  \,  \algeffseqoverindex{ \beta }{ \text{\unboldmath$\mathit{J}$} }   \ottsym{.}  \ottnt{B} \,  [   \algeffseqoverindex{ \ottnt{C} }{ \text{\unboldmath$\mathit{I}$} }   \ottsym{/}   \algeffseqoverindex{ \alpha }{ \text{\unboldmath$\mathit{I}$} }   ]    \vdash    \ottnt{E} ^{  \algeffseqoverindex{ \beta }{ \text{\unboldmath$\mathit{J}$} }  }    \ottsym{:}     \text{\unboldmath$\forall$}  \,  \algeffseqoverindex{ \beta }{ \text{\unboldmath$\mathit{J}$} }   \ottsym{.}  \ottsym{(}   \ottnt{B}    [   \algeffseqoverindex{ \ottnt{C} }{ \text{\unboldmath$\mathit{I}$} }   \ottsym{/}   \algeffseqoverindex{ \alpha }{ \text{\unboldmath$\mathit{I}$} }   ]    \ottsym{)}  \multimap  \ottnt{D}   \,  |  \, \epsilon 
        \]
        for some fresh variable $\mathit{y}$ by \reflem{weakening}.
        Since $\Gamma  \ottsym{,}  \mathit{y} \,  \mathord{:}  \,  \text{\unboldmath$\forall$}  \,  \algeffseqoverindex{ \beta }{ \text{\unboldmath$\mathit{J}$} }   \ottsym{.}  \ottnt{B} \,  [   \algeffseqoverindex{ \ottnt{C} }{ \text{\unboldmath$\mathit{I}$} }   \ottsym{/}   \algeffseqoverindex{ \alpha }{ \text{\unboldmath$\mathit{I}$} }   ]   \ottsym{,}   \algeffseqoverindex{ \beta }{ \text{\unboldmath$\mathit{J}$} }   \ottsym{;}   \mathsf{none}  \,   \vdash  \mathit{y} \,  \algeffseqoverindex{ \beta }{ \text{\unboldmath$\mathit{J}$} }  \,   \ottsym{:}   \ottnt{B}    [   \algeffseqoverindex{ \ottnt{C} }{ \text{\unboldmath$\mathit{I}$} }   \ottsym{/}   \algeffseqoverindex{ \alpha }{ \text{\unboldmath$\mathit{I}$} }   ]   \,  |  \,  \langle \rangle $
        by \T{Var},
        we have
        \begin{equation}
         \Gamma  \ottsym{,}  \mathit{y} \,  \mathord{:}  \,  \text{\unboldmath$\forall$}  \,  \algeffseqoverindex{ \beta }{ \text{\unboldmath$\mathit{J}$} }   \ottsym{.}  \ottnt{B} \,  [   \algeffseqoverindex{ \ottnt{C} }{ \text{\unboldmath$\mathit{I}$} }   \ottsym{/}   \algeffseqoverindex{ \alpha }{ \text{\unboldmath$\mathit{I}$} }   ]   \ottsym{;}   \mathsf{none}  \,   \vdash    \ottnt{E} ^{  \algeffseqoverindex{ \beta }{ \text{\unboldmath$\mathit{J}$} }  }   [  \mathit{y} \,  \algeffseqoverindex{ \beta }{ \text{\unboldmath$\mathit{J}$} }   ]  \,   \ottsym{:}  \ottnt{D} \,  |  \, \epsilon'
         \label{eqn:cont-subst:two}
        \end{equation}
        by \reflem{ectx-placement} and \T{Weak}.

        By (\ref{eqn:cont-subst:one}), (\ref{eqn:cont-subst:two}), and \T{Let},
        \[
         \Gamma  \ottsym{;}   \mathsf{none}  \,   \vdash  \mathsf{let} \, \mathit{y}  \ottsym{=}   \Lambda\!  \,  \algeffseqoverindex{ \beta }{ \text{\unboldmath$\mathit{J}$} }   \ottsym{.}     \ottnt{e'}    [   \ottnt{E} ^{  \algeffseqoverindex{ \beta }{ \text{\unboldmath$\mathit{J}$} }  }   /  \mathsf{resume}  ]^{  \algeffseqoverindex{  \text{\unboldmath$\forall$}  \,  \algeffseqoverindex{ \beta }{ \text{\unboldmath$\mathit{J}$} }   \ottsym{.}  \ottnt{C} }{ \text{\unboldmath$\mathit{I}$} }  }_{  \Lambda\!  \,  \algeffseqoverindex{ \beta }{ \text{\unboldmath$\mathit{J}$} }   \ottsym{.}  \ottnt{v} }      [   \algeffseqoverindex{ \ottnt{C} }{ \text{\unboldmath$\mathit{I}$} }   \ottsym{/}   \algeffseqoverindex{ \gamma }{ \text{\unboldmath$\mathit{I}$} }   ]      [  \ottnt{v}  \ottsym{/}  \mathit{x}  ]   \,  \mathsf{in}  \,   \ottnt{E} ^{  \algeffseqoverindex{ \beta }{ \text{\unboldmath$\mathit{J}$} }  }   [  \mathit{y} \,  \algeffseqoverindex{ \beta }{ \text{\unboldmath$\mathit{J}$} }   ]  \,   \ottsym{:}  \ottnt{D} \,  |  \, \epsilon',
        \]
        which is what we have to show by definition of substitution for $ \mathsf{resume} $.
       \end{caseanalysis}

  \item By case analysis on the typing rule applied last.
        \begin{caseanalysis}
         \case \THrule{Return}: By the IH.
         \case \THrule{Op}: By the IH; note that
         \[\begin{array}{ll}
          &  \ottsym{(}  \ottnt{h'}  \ottsym{;}   \Lambda\!  \,  \algeffseqoverindex{ \gamma }{ \text{\unboldmath$\mathit{I'}$} }   \ottsym{.}  \mathsf{op}  \ottsym{(}  \mathit{x}  \ottsym{)}  \rightarrow  \ottnt{e'}  \ottsym{)}    [   \ottnt{E} ^{  \algeffseqoverindex{ \beta }{ \text{\unboldmath$\mathit{J}$} }  }   /  \mathsf{resume}  ]^{  \algeffseqoverindex{  \text{\unboldmath$\forall$}  \,  \algeffseqoverindex{ \beta }{ \text{\unboldmath$\mathit{J}$} }   \ottsym{.}  \ottnt{C} }{ \text{\unboldmath$\mathit{I}$} }  }_{  \Lambda\!  \,  \algeffseqoverindex{ \beta }{ \text{\unboldmath$\mathit{J}$} }   \ottsym{.}  \ottnt{v} }   \\
        = &  \ottnt{h'}    [   \ottnt{E} ^{  \algeffseqoverindex{ \beta }{ \text{\unboldmath$\mathit{J}$} }  }   /  \mathsf{resume}  ]^{  \algeffseqoverindex{  \text{\unboldmath$\forall$}  \,  \algeffseqoverindex{ \beta }{ \text{\unboldmath$\mathit{J}$} }   \ottsym{.}  \ottnt{C} }{ \text{\unboldmath$\mathit{I}$} }  }_{  \Lambda\!  \,  \algeffseqoverindex{ \beta }{ \text{\unboldmath$\mathit{J}$} }   \ottsym{.}  \ottnt{v} }    \ottsym{;}   \Lambda\!  \,  \algeffseqoverindex{ \gamma }{ \text{\unboldmath$\mathit{I'}$} }   \ottsym{.}  \mathsf{op}  \ottsym{(}  \mathit{x}  \ottsym{)}  \rightarrow  \ottnt{e'}.
         \end{array}\]
        \end{caseanalysis}
 \end{enumerate}
\end{proof}

%%% Value inversion

\begin{lemmap}{Constant inversion}{value-inversion-constant}
 If $\Gamma  \ottsym{;}  \ottnt{r} \,   \vdash  \ottnt{c} \,   \ottsym{:}  \ottnt{A} \,  |  \, \epsilon$, then $ \mathit{ty}  (  \ottnt{c}  )  \,  =  \, \ottnt{A}$.
\end{lemmap}
\begin{proof}
 Straightforward by induction on the derivation of $\Gamma  \ottsym{;}  \ottnt{r} \,   \vdash  \ottnt{c} \,   \ottsym{:}  \ottnt{A} \,  |  \, \epsilon$.
\end{proof}

\begin{lemmap}{Abstraction inversion}{value-inversion-abs}
 If $\Gamma  \ottsym{;}  \ottnt{r} \,   \vdash   \lambda\!  \, \mathit{x}  \ottsym{.}  \ottnt{e} \,   \ottsym{:}   \ottnt{A}   \rightarrow  \!  \epsilon'  \;  \ottnt{B}  \,  |  \, \epsilon$, then
 $\Gamma  \ottsym{,}  \mathit{x} \,  \mathord{:}  \, \ottnt{A}  \ottsym{;}  \ottnt{r} \,   \vdash  \ottnt{e} \,   \ottsym{:}  \ottnt{B} \,  |  \, \epsilon'$.
\end{lemmap}
\begin{proof}
 Straightforward by induction on the derivation of $\Gamma  \ottsym{;}  \ottnt{r} \,   \vdash   \lambda\!  \, \mathit{x}  \ottsym{.}  \ottnt{e} \,   \ottsym{:}   \ottnt{A}   \rightarrow  \!  \epsilon'  \;  \ottnt{B}  \,  |  \, \epsilon$.
\end{proof}

\begin{lemmap}{Continuation inversion}{cont-inversion}
 If $\Gamma  \ottsym{;}  \ottnt{r} \,   \vdash   \textup{\texttt{\#}\relax}  \mathsf{op}   \ottsym{(}    \algeffseqoverindex{ \sigma }{ \text{\unboldmath$\mathit{I}$} }    \ottsym{,}   \ottnt{w}   \ottsym{,}   \ottnt{E}   \ottsym{)}  \,   \ottsym{:}  \ottnt{D} \,  |  \, \epsilon$, then
 \begin{itemize}
  \item $ \algeffseqoverindex{ \sigma }{ \text{\unboldmath$\mathit{I}$} }  \,  =  \,  \algeffseqoverindex{  \text{\unboldmath$\forall$}  \,  \algeffseqoverindex{ \beta }{ \text{\unboldmath$\mathit{J}$} }   \ottsym{.}  \ottnt{C} }{ \text{\unboldmath$\mathit{I}$} } $,
  \item $\ottnt{w} \,  =  \,  \Lambda\!  \,  \algeffseqoverindex{ \beta }{ \text{\unboldmath$\mathit{J}$} }   \ottsym{.}  \ottnt{v}$,
  \item $\ottnt{E}$ captures $ \algeffseqoverindex{ \beta }{ \text{\unboldmath$\mathit{J}$} } $ at the hole,
  \item $\epsilon' \,  \subseteq  \, \epsilon$,
  \item $\mathit{ty} \, \ottsym{(}  \mathsf{op}  \ottsym{)} \,  =  \,   \text{\unboldmath$\forall$}     \algeffseqoverindex{ \alpha }{ \text{\unboldmath$\mathit{I}$} }   .  \ottnt{A}  \hookrightarrow  \ottnt{B} $,
  \item $\mathsf{op} \,  \in  \, \epsilon'$,
  \item $\Gamma  \vdash   \algeffseqoverindex{  \text{\unboldmath$\forall$}  \,  \algeffseqoverindex{ \beta }{ \text{\unboldmath$\mathit{J}$} }   \ottsym{.}  \ottnt{C} }{ \text{\unboldmath$\mathit{I}$} } $,
  \item $\Gamma  \ottsym{,}   \algeffseqoverindex{ \beta }{ \text{\unboldmath$\mathit{J}$} }   \ottsym{;}  \ottnt{r} \,   \vdash  \ottnt{v} \,   \ottsym{:}   \ottnt{A}    [   \algeffseqoverindex{ \ottnt{C} }{ \text{\unboldmath$\mathit{I}$} }   \ottsym{/}   \algeffseqoverindex{ \alpha }{ \text{\unboldmath$\mathit{I}$} }   ]   \,  |  \, \epsilon'$, and
  \item $ \Gamma   \vdash   \ottnt{E}   \ottsym{:}     \text{\unboldmath$\forall$}  \,  \algeffseqoverindex{ \beta }{ \text{\unboldmath$\mathit{J}$} }   \ottsym{.}  \ottsym{(}   \ottnt{B}    [   \algeffseqoverindex{ \ottnt{C} }{ \text{\unboldmath$\mathit{I}$} }   \ottsym{/}   \algeffseqoverindex{ \alpha }{ \text{\unboldmath$\mathit{I}$} }   ]    \ottsym{)}  \multimap  \ottnt{D}   \,  |  \, \epsilon' $
 \end{itemize}
 for some $ \algeffseqoverindex{ \alpha }{ \text{\unboldmath$\mathit{I}$} } $, $ \algeffseqoverindex{ \beta }{ \text{\unboldmath$\mathit{J}$} } $, $ \algeffseqoverindex{ \ottnt{C} }{ \text{\unboldmath$\mathit{I}$} } $, $\ottnt{A}$, $\ottnt{B}$, $\ottnt{v}$, and $\epsilon'$.
\end{lemmap}
\begin{proof}
 Straightforward by induction on the derivation of $\Gamma  \ottsym{;}  \ottnt{r} \,   \vdash   \textup{\texttt{\#}\relax}  \mathsf{op}   \ottsym{(}    \algeffseqoverindex{ \sigma }{ \text{\unboldmath$\mathit{I}$} }    \ottsym{,}   \ottnt{w}   \ottsym{,}   \ottnt{E}   \ottsym{)}  \,   \ottsym{:}  \ottnt{D} \,  |  \, \epsilon$.
\end{proof}

\begin{lemmap}{Handler inversion}{handler-inversion}
 Suppose that $\Gamma  \ottsym{;}  \ottnt{r} \,   \vdash  \ottnt{h}  \ottsym{:}  \ottnt{A} \,  |  \, \epsilon  \Rightarrow  \ottnt{B} \,  |  \, \epsilon'$.
 \begin{enumerate}
  \item If $ \ottnt{h} ^\mathsf{return}  \,  =  \, \mathsf{return} \, \mathit{x}  \rightarrow  \ottnt{e}$, then
       $\Gamma  \ottsym{,}  \mathit{x} \,  \mathord{:}  \, \ottnt{A}  \ottsym{;}  \ottnt{r} \,   \vdash  \ottnt{e} \,   \ottsym{:}  \ottnt{B} \,  |  \, \epsilon'$ for some $\mathit{x}$ and $\ottnt{e}$.
  \item For any $\mathsf{op} \,  \in  \,  \mathit{ops}  (  \ottnt{h}  ) $,
        \begin{itemize}
         \item $ \ottnt{h} ^{ \mathsf{op} }  \,  =  \,  \Lambda\!  \,  \algeffseqoverindex{ \alpha }{ \text{\unboldmath$\mathit{I}$} }   \ottsym{.}  \mathsf{op}  \ottsym{(}  \mathit{x}  \ottsym{)}  \rightarrow  \ottnt{e}$,
         \item $\mathit{ty} \, \ottsym{(}  \mathsf{op}  \ottsym{)} \,  =  \,   \text{\unboldmath$\forall$}     \algeffseqoverindex{ \alpha }{ \text{\unboldmath$\mathit{I}$} }   .  \ottnt{C}  \hookrightarrow  \ottnt{D} $, and
         \item $\Gamma  \ottsym{,}   \algeffseqoverindex{ \alpha }{ \text{\unboldmath$\mathit{I}$} }   \ottsym{,}  \mathit{x} \,  \mathord{:}  \, \ottnt{C}  \ottsym{;}  \ottsym{(}   \algeffseqoverindex{ \alpha }{ \text{\unboldmath$\mathit{I}$} }   \ottsym{,}  \ottnt{C}  \ottsym{,}   \ottnt{D}   \rightarrow  \!  \epsilon'  \;  \ottnt{B}   \ottsym{)} \,   \vdash  \ottnt{e} \,   \ottsym{:}  \ottnt{B} \,  |  \, \epsilon'$
        \end{itemize}
        for some $ \algeffseqoverindex{ \alpha }{ \text{\unboldmath$\mathit{I}$} } $, $\mathit{x}$, $\ottnt{e}$, $\ottnt{C}$, and $\ottnt{D}$.
 \end{enumerate}
\end{lemmap}
\begin{proof}
 Straightforward by induction on the derivation of
 $\Gamma  \ottsym{;}  \ottnt{r} \,   \vdash  \ottnt{h}  \ottsym{:}  \ottnt{A} \,  |  \, \epsilon  \Rightarrow  \ottnt{B} \,  |  \, \epsilon'$.
\end{proof}

%%% Progress

\begin{lemmap}{Canonical forms}{canonical-forms}
 If $\Gamma  \ottsym{;}  \ottnt{r} \,   \vdash  \ottnt{v} \,   \ottsym{:}  \iota \,  |  \, \epsilon$, then $\ottnt{v} \,  =  \, \ottnt{c}$.
\end{lemmap}
\begin{proof}
 Straightforward by induction on the derivation.
\end{proof}

\begin{lemma}{handler-op-inheritance}
 If $\Gamma  \ottsym{;}  \ottnt{r} \,   \vdash  \ottnt{h}  \ottsym{:}  \ottnt{A} \,  |  \, \epsilon  \Rightarrow  \ottnt{B} \,  |  \, \epsilon'$ and $\mathsf{op} \,  \in  \, \epsilon$ and $\mathsf{op} \,  \not\in  \,  \mathit{ops}  (  \ottnt{h}  ) $,
 then $\mathsf{op} \,  \in  \, \epsilon'$
\end{lemma}
\begin{proof}
 Straightforward by induction on the derivation of $\Gamma  \ottsym{;}  \ottnt{r} \,   \vdash  \ottnt{h}  \ottsym{:}  \ottnt{A} \,  |  \, \epsilon  \Rightarrow  \ottnt{B} \,  |  \, \epsilon'$.
\end{proof}

\ifrestate
\lemmProgress*
\else
\begin{lemmap}{Progress}{progress}
 If $\Delta  \ottsym{;}   \mathsf{none}  \,   \vdash  \ottnt{e} \,   \ottsym{:}  \ottnt{A} \,  |  \, \epsilon$, then
 (1) $\ottnt{e}  \longrightarrow  \ottnt{e'}$ for some $\ottnt{e'}$,
 (2) $\ottnt{e}$ is a value, or
 (3) $\ottnt{e} \,  =  \,  \textup{\texttt{\#}\relax}  \mathsf{op}   \ottsym{(}    \algeffseqover{ \sigma }    \ottsym{,}   \ottnt{w}   \ottsym{,}   \ottnt{E}   \ottsym{)} $ for some $\mathsf{op} \,  \in  \, \epsilon$, $ \algeffseqover{ \sigma } $, $\ottnt{w}$, and $\ottnt{E}$.
\end{lemmap}
\fi
\begin{proof}
 By induction on the derivation of $\Delta  \ottsym{;}   \mathsf{none}  \,   \vdash  \ottnt{e} \,   \ottsym{:}  \ottnt{A} \,  |  \, \epsilon$.
 \begin{caseanalysis}
  \case \T{Var}: Contradictory.
  \case \T{Const}: Obvious.
  \case \T{Abs}: Obvious.
  \case \T{App}:
  We are given
  \begin{itemize}
   \item $\ottnt{e} \,  =  \, \ottnt{e_{{\mathrm{1}}}} \, \ottnt{e_{{\mathrm{2}}}}$,
   \item $\Delta  \ottsym{;}   \mathsf{none}  \,   \vdash  \ottnt{e_{{\mathrm{1}}}} \, \ottnt{e_{{\mathrm{2}}}} \,   \ottsym{:}  \ottnt{A} \,  |  \, \epsilon$,
   \item $\Delta  \ottsym{;}   \mathsf{none}  \,   \vdash  \ottnt{e_{{\mathrm{1}}}} \,   \ottsym{:}   \ottnt{B}   \rightarrow  \!  \epsilon'  \;  \ottnt{A}  \,  |  \, \epsilon$,
   \item $\Delta  \ottsym{;}   \mathsf{none}  \,   \vdash  \ottnt{e_{{\mathrm{2}}}} \,   \ottsym{:}  \ottnt{B} \,  |  \, \epsilon$, and
   \item $\epsilon' \,  \subseteq  \, \epsilon$.
  \end{itemize}
  By case analysis on the behavior of $\ottnt{e_{{\mathrm{1}}}}$.
  We have three cases to be considered by the IH.
  \begin{caseanalysis}
   \case $\ottnt{e_{{\mathrm{1}}}}  \longrightarrow  \ottnt{e'_{{\mathrm{1}}}}$ for some $\ottnt{e'_{{\mathrm{1}}}}$: We have $\ottnt{e}  \longrightarrow  \ottnt{e'_{{\mathrm{1}}}} \, \ottnt{e_{{\mathrm{2}}}}$.
   \case $\ottnt{e_{{\mathrm{1}}}} \,  =  \,  \textup{\texttt{\#}\relax}  \mathsf{op}   \ottsym{(}    \algeffseqover{ \sigma }    \ottsym{,}   \ottnt{w}   \ottsym{,}   \ottnt{E}   \ottsym{)} $ for some $\mathsf{op} \,  \in  \, \epsilon$, $ \algeffseqover{ \sigma } $, $\ottnt{w}$, and $\ottnt{E}$:
     By \R{OpApp1} and \E{Eval}.
   \case $\ottnt{e_{{\mathrm{1}}}} \,  =  \, \ottnt{v_{{\mathrm{1}}}}$ for some $\ottnt{v_{{\mathrm{1}}}}$:
   By case analysis on the behavior of $\ottnt{e_{{\mathrm{2}}}}$ with the IH.
   \begin{caseanalysis}
    \case $\ottnt{e_{{\mathrm{2}}}}  \longrightarrow  \ottnt{e'_{{\mathrm{2}}}}$ for some $\ottnt{e'_{{\mathrm{2}}}}$: We have $\ottnt{e}  \longrightarrow  \ottnt{v_{{\mathrm{1}}}} \, \ottnt{e'_{{\mathrm{2}}}}$.
    \case $\ottnt{e_{{\mathrm{2}}}} \,  =  \,  \textup{\texttt{\#}\relax}  \mathsf{op}   \ottsym{(}    \algeffseqover{ \sigma }    \ottsym{,}   \ottnt{w}   \ottsym{,}   \ottnt{E}   \ottsym{)} $ for some $\mathsf{op} \,  \in  \, \epsilon$, $ \algeffseqover{ \sigma } $, $\ottnt{w}$, and $\ottnt{E}$: By \R{OpApp2} and \E{Eval}.
    \case $\ottnt{e_{{\mathrm{2}}}} \,  =  \, \ottnt{v_{{\mathrm{2}}}}$ for some $\ottnt{v_{{\mathrm{2}}}}$:
    If $\ottnt{v_{{\mathrm{1}}}} \,  =  \, \ottnt{c_{{\mathrm{1}}}}$ for some $\ottnt{c_{{\mathrm{1}}}}$,
    then $\ottnt{B} \,  =  \, \iota$ and $ \mathit{ty}  (  \ottnt{c_{{\mathrm{1}}}}  )  \,  =  \,  \iota   \rightarrow  \!   \langle \rangle   \;  \ottnt{A} $ and
    $\epsilon' \,  =  \,  \langle \rangle $ by \reflem{value-inversion-constant}.
    Since $\Delta  \ottsym{;}   \mathsf{none}  \,   \vdash  \ottnt{v_{{\mathrm{2}}}} \,   \ottsym{:}  \iota \,  |  \, \epsilon$,
    there exists some $\ottnt{c_{{\mathrm{2}}}}$ such that $\ottnt{v_{{\mathrm{2}}}} \,  =  \, \ottnt{c_{{\mathrm{2}}}}$ and $ \mathit{ty}  (  \ottnt{c_{{\mathrm{2}}}}  )  \,  =  \, \iota$.
    By the assumption about constants, $ \zeta  (  \ottnt{c_{{\mathrm{1}}}}  ,  \ottnt{c_{{\mathrm{2}}}}  ) $ is defined and
    $ \zeta  (  \ottnt{c_{{\mathrm{1}}}}  ,  \ottnt{c_{{\mathrm{2}}}}  ) $ is a constant and $ \mathit{ty}  (   \zeta  (  \ottnt{c_{{\mathrm{1}}}}  ,  \ottnt{c_{{\mathrm{2}}}}  )   )  \,  =  \, \ottnt{A}$.
    Thus, $\ottnt{e} = \ottnt{c_{{\mathrm{1}}}} \, \ottnt{c_{{\mathrm{2}}}}  \longrightarrow   \zeta  (  \ottnt{c_{{\mathrm{1}}}}  ,  \ottnt{c_{{\mathrm{2}}}}  ) $ by \R{Const}/\E{Eval}.

    If $\ottnt{v_{{\mathrm{1}}}} \,  =  \,  \lambda\!  \, \mathit{x}  \ottsym{.}  \ottnt{e'}$ for some $\mathit{x}$ and $\ottnt{e'}$, then
    $\ottnt{e} = \ottsym{(}   \lambda\!  \, \mathit{x}  \ottsym{.}  \ottnt{e'}  \ottsym{)} \, \ottnt{v_{{\mathrm{2}}}}  \longrightarrow   \ottnt{e'}    [  \ottnt{v_{{\mathrm{2}}}}  \ottsym{/}  \mathit{x}  ]  $ by \R{Beta}/\E{Eval}.
    Note that substitution of $\ottnt{v_{{\mathrm{2}}}}$ for $\mathit{x}$ in $\ottnt{e'}$ is defined
    since $\Delta  \ottsym{,}  \mathit{x} \,  \mathord{:}  \, \ottnt{B}  \ottsym{;}   \mathsf{none}  \,   \vdash  \ottnt{e'} \,   \ottsym{:}  \ottnt{A} \,  |  \, \epsilon'$ by \reflem{value-inversion-abs}.

   \end{caseanalysis}
  \end{caseanalysis}

  \case \T{Op}:
  We are given
  \begin{itemize}
   \item $\ottnt{e} \,  =  \,  \textup{\texttt{\#}\relax}  \mathsf{op}   \ottsym{(}    \algeffseqoverindex{ \ottnt{C} }{ \text{\unboldmath$\mathit{I}$} }    \ottsym{,}   \ottnt{e'}   \ottsym{)} $,
   \item $\ottnt{A} \,  =  \,  \ottnt{B'}    [   \algeffseqoverindex{ \ottnt{C} }{ \text{\unboldmath$\mathit{I}$} }   \ottsym{/}   \algeffseqoverindex{ \alpha }{ \text{\unboldmath$\mathit{I}$} }   ]  $,
   \item $\Delta  \ottsym{;}   \mathsf{none}  \,   \vdash   \textup{\texttt{\#}\relax}  \mathsf{op}   \ottsym{(}    \algeffseqoverindex{ \ottnt{C} }{ \text{\unboldmath$\mathit{I}$} }    \ottsym{,}   \ottnt{e'}   \ottsym{)}  \,   \ottsym{:}   \ottnt{B'}    [   \algeffseqoverindex{ \ottnt{C} }{ \text{\unboldmath$\mathit{I}$} }   \ottsym{/}   \algeffseqoverindex{ \alpha }{ \text{\unboldmath$\mathit{I}$} }   ]   \,  |  \, \epsilon$
   \item $\mathit{ty} \, \ottsym{(}  \mathsf{op}  \ottsym{)} \,  =  \,   \text{\unboldmath$\forall$}     \algeffseqoverindex{ \alpha }{ \text{\unboldmath$\mathit{I}$} }   .  \ottnt{A'}  \hookrightarrow  \ottnt{B'} $,
   \item $\mathsf{op} \,  \in  \, \epsilon$,
   \item $\Delta  \ottsym{;}   \mathsf{none}  \,   \vdash  \ottnt{e'} \,   \ottsym{:}   \ottnt{A'}    [   \algeffseqoverindex{ \ottnt{C} }{ \text{\unboldmath$\mathit{I}$} }   \ottsym{/}   \algeffseqoverindex{ \alpha }{ \text{\unboldmath$\mathit{I}$} }   ]   \,  |  \, \epsilon$, and
   \item $\Delta  \vdash   \algeffseqoverindex{ \ottnt{C} }{ \text{\unboldmath$\mathit{I}$} } $.
  \end{itemize}
  By case analysis on the behavior of $\ottnt{e'}$ with the IH.
  \begin{caseanalysis}
   \case $\ottnt{e'}  \longrightarrow  \ottnt{e''}$ for some $\ottnt{e''}$: We have $\ottnt{e}  \longrightarrow   \textup{\texttt{\#}\relax}  \mathsf{op}   \ottsym{(}    \algeffseqoverindex{ \ottnt{C} }{ \text{\unboldmath$\mathit{I}$} }    \ottsym{,}   \ottnt{e''}   \ottsym{)} $.
   \case $\ottnt{e'} \,  =  \,  \textup{\texttt{\#}\relax}  \mathsf{op}'   \ottsym{(}    \algeffseqoverindex{ \sigma }{ \text{\unboldmath$\mathit{I'}$} }    \ottsym{,}   \ottnt{w}   \ottsym{,}   \ottnt{E}   \ottsym{)} $ for some $\mathsf{op}' \,  \in  \, \epsilon$, $ \algeffseqoverindex{ \sigma }{ \text{\unboldmath$\mathit{I'}$} } $, $\ottnt{w}$, and $\ottnt{E}$: By \R{OpOp} and \E{Eval}.
   \case $\ottnt{e'} \,  =  \, \ottnt{v}$ for some $\ottnt{v}$: By \R{Op}/\E{Eval}.
  \end{caseanalysis}

  \case \T{OpCont}: Obvious.
  \case \T{Weak}: By the IH.
  \case \T{Handle}:
  We are given
  \begin{itemize}
   \item $\ottnt{e} \,  =  \, \mathsf{handle} \, \ottnt{e'} \, \mathsf{with} \, \ottnt{h}$,
   \item $\Delta  \ottsym{;}   \mathsf{none}  \,   \vdash  \mathsf{handle} \, \ottnt{e'} \, \mathsf{with} \, \ottnt{h} \,   \ottsym{:}  \ottnt{A} \,  |  \, \epsilon$,
   \item $\Delta  \ottsym{;}   \mathsf{none}  \,   \vdash  \ottnt{e'} \,   \ottsym{:}  \ottnt{B} \,  |  \, \epsilon'$, and
   \item $\Delta  \ottsym{;}   \mathsf{none}  \,   \vdash  \ottnt{h}  \ottsym{:}  \ottnt{B} \,  |  \, \epsilon'  \Rightarrow  \ottnt{A} \,  |  \, \epsilon$.
  \end{itemize}
  By case analysis on the behavior of $\ottnt{e'}$ with the IH.
  \begin{caseanalysis}
   \case $\ottnt{e'}  \longrightarrow  \ottnt{e''}$ for some $\ottnt{e''}$: We have $\ottnt{e}  \longrightarrow  \mathsf{handle} \, \ottnt{e''} \, \mathsf{with} \, \ottnt{h}$.
   \case $\ottnt{e'} \,  =  \,  \textup{\texttt{\#}\relax}  \mathsf{op}   \ottsym{(}    \algeffseqover{ \sigma }    \ottsym{,}   \ottnt{w}   \ottsym{,}   \ottnt{E}   \ottsym{)} $ for some $\mathsf{op} \,  \in  \, \epsilon'$, $ \algeffseqover{ \sigma } $, $\ottnt{w}$, and $\ottnt{E}$:
   If $\mathsf{op} \,  \in  \,  \mathit{ops}  (  \ottnt{h}  ) $, then we finish by \reflem{cont-inversion} and \R{Handle}/\E{Eval}.
   Otherwise, if $\mathsf{op} \,  \not\in  \,  \mathit{ops}  (  \ottnt{h}  ) $, we have $\ottnt{e} = \mathsf{handle} \,  \textup{\texttt{\#}\relax}  \mathsf{op}   \ottsym{(}    \algeffseqover{ \sigma }    \ottsym{,}   \ottnt{w}   \ottsym{,}   \ottnt{E}   \ottsym{)}  \, \mathsf{with} \, \ottnt{h}  \longrightarrow   \textup{\texttt{\#}\relax}  \mathsf{op}   \ottsym{(}    \algeffseqover{ \sigma }    \ottsym{,}   \ottnt{w}   \ottsym{,}   \mathsf{handle} \, \ottnt{E} \, \mathsf{with} \, \ottnt{h}   \ottsym{)} $
   by \R{OpHandle}/\E{Eval}.
   Note that $\mathsf{op} \,  \in  \, \epsilon$ by \reflem{handler-op-inheritance}.

   \case $\ottnt{e'} \,  =  \, \ottnt{v}$ for some $\ottnt{v}$: By \R{Return}/\E{Eval}.
  \end{caseanalysis}

  \case \T{Resume}: Contradictory.
  \case \T{Let}:
  We are given
  \begin{itemize}
   \item $\ottnt{e} \,  =  \, \mathsf{let} \, \mathit{x}  \ottsym{=}   \Lambda\!  \,  \algeffseqoverindex{ \alpha }{ \text{\unboldmath$\mathit{I}$} }   \ottsym{.}  \ottnt{e_{{\mathrm{1}}}} \,  \mathsf{in}  \, \ottnt{e_{{\mathrm{2}}}}$,
   \item $\Delta  \ottsym{,}   \algeffseqoverindex{ \alpha }{ \text{\unboldmath$\mathit{I}$} }   \ottsym{;}   \mathsf{none}  \,   \vdash  \ottnt{e_{{\mathrm{1}}}} \,   \ottsym{:}  \ottnt{B} \,  |  \, \epsilon$, and
   \item $\Delta  \ottsym{,}  \mathit{x} \,  \mathord{:}  \,  \text{\unboldmath$\forall$}  \,  \algeffseqoverindex{ \alpha }{ \text{\unboldmath$\mathit{I}$} }   \ottsym{.}  \ottnt{B}  \ottsym{;}   \mathsf{none}  \,   \vdash  \ottnt{e_{{\mathrm{2}}}} \,   \ottsym{:}  \ottnt{A} \,  |  \, \epsilon$.
  \end{itemize}
  By case analysis on the behavior of $\ottnt{e_{{\mathrm{1}}}}$.
  \begin{caseanalysis}
   \case $\ottnt{e_{{\mathrm{1}}}}  \longrightarrow  \ottnt{e'_{{\mathrm{1}}}}$ for some $\ottnt{e'_{{\mathrm{1}}}}$:
   \case $\ottnt{e_{{\mathrm{1}}}} \,  =  \,  \textup{\texttt{\#}\relax}  \mathsf{op}   \ottsym{(}    \algeffseqoverindex{ \sigma }{ \text{\unboldmath$\mathit{J}$} }    \ottsym{,}   \ottnt{w}   \ottsym{,}   \ottnt{E}   \ottsym{)} $ for some $\mathsf{op} \,  \in  \, \epsilon$, $ \algeffseqoverindex{ \sigma }{ \text{\unboldmath$\mathit{J}$} } $, $\ottnt{w}$, and $\ottnt{E}$:
   By \R{OpLet} and \E{Eval}.

   \case $\ottnt{e_{{\mathrm{1}}}} \,  =  \, \ottnt{v_{{\mathrm{1}}}}$ for some $\ottnt{v_{{\mathrm{1}}}}$: By \R{Let}/\E{Eval}.
   Note that substitution of $ \Lambda\!  \,  \algeffseqoverindex{ \alpha }{ \text{\unboldmath$\mathit{I}$} }   \ottsym{.}  \ottnt{v}$ for $\mathit{x}$ in $\ottnt{e_{{\mathrm{2}}}}$ is defined
   since $\Delta  \ottsym{,}  \mathit{x} \,  \mathord{:}  \,  \text{\unboldmath$\forall$}  \,  \algeffseqoverindex{ \alpha }{ \text{\unboldmath$\mathit{I}$} }   \ottsym{.}  \ottnt{B}  \ottsym{;}   \mathsf{none}  \,   \vdash  \ottnt{e_{{\mathrm{2}}}} \,   \ottsym{:}  \ottnt{A} \,  |  \, \epsilon$.
  \end{caseanalysis}
 \end{caseanalysis}
\end{proof}

%%% Subject reduction

\begin{lemma}{type-wf}
 \begin{enumerate}
  \item If $\Gamma  \ottsym{;}   \mathsf{none}  \,   \vdash  \ottnt{e} \,   \ottsym{:}  \ottnt{A} \,  |  \, \epsilon$, then $\Gamma  \vdash  \ottnt{A}$.
  \item If $\Gamma  \ottsym{;}   \mathsf{none}  \,   \vdash  \ottnt{h}  \ottsym{:}  \ottnt{A} \,  |  \, \epsilon  \Rightarrow  \ottnt{B} \,  |  \, \epsilon'$, then $\Gamma  \vdash  \ottnt{B}$.
  \item If $ \Gamma   \vdash   \ottnt{E}   \ottsym{:}    \sigma  \multimap  \ottnt{A}   \,  |  \, \epsilon $ and $\Gamma  \vdash  \sigma$, then $\Gamma  \vdash  \ottnt{A}$.
 \end{enumerate}
\end{lemma}
\begin{proof}
 Straightforward by mutual induction on the typing derivations with \reflem{typing-context-wf}.
\end{proof}

\ifrestate
\lemmSubjectRed*
\else
\begin{lemmap}{Subject reduction}{subject-red}
 \begin{enumerate}
  \item If $\Delta  \ottsym{;}   \mathsf{none}  \,   \vdash  \ottnt{e_{{\mathrm{1}}}} \,   \ottsym{:}  \ottnt{A} \,  |  \, \epsilon$ and $\ottnt{e_{{\mathrm{1}}}}  \rightsquigarrow  \ottnt{e_{{\mathrm{2}}}}$,
        then $\Delta  \ottsym{;}   \mathsf{none}  \,   \vdash  \ottnt{e_{{\mathrm{2}}}} \,   \ottsym{:}  \ottnt{A} \,  |  \, \epsilon$.
  \item If $\Delta  \ottsym{;}   \mathsf{none}  \,   \vdash  \ottnt{e_{{\mathrm{1}}}} \,   \ottsym{:}  \ottnt{A} \,  |  \, \epsilon$ and $\ottnt{e_{{\mathrm{1}}}}  \longrightarrow  \ottnt{e_{{\mathrm{2}}}}$,
        then $\Delta  \ottsym{;}   \mathsf{none}  \,   \vdash  \ottnt{e_{{\mathrm{2}}}} \,   \ottsym{:}  \ottnt{A} \,  |  \, \epsilon$.
 \end{enumerate}
\end{lemmap}
\fi
\begin{proof}
 \begin{enumerate}
  \item Suppose that $\Delta  \ottsym{;}   \mathsf{none}  \,   \vdash  \ottnt{e_{{\mathrm{1}}}} \,   \ottsym{:}  \ottnt{A} \,  |  \, \epsilon$ and $\ottnt{e_{{\mathrm{1}}}}  \rightsquigarrow  \ottnt{e_{{\mathrm{2}}}}$.
        By induction on $\Delta  \ottsym{;}   \mathsf{none}  \,   \vdash  \ottnt{e_{{\mathrm{1}}}} \,   \ottsym{:}  \ottnt{A} \,  |  \, \epsilon$.
        \begin{caseanalysis}
         \case \T{Var}: Contradictory.

         \case \T{Const}: Contradictory;
         no reduction rules can be applied to constants.

         \case \T{Abs}: Contradictory;
         no reduction rules can be applied to lambda abstractions.

         \case \T{App}: We have four reduction rules which can be applied to
         function applications.
         \begin{caseanalysis}
          \case \R{Const}:
          We are given
          \begin{itemize}
           \item $\ottnt{e_{{\mathrm{1}}}} \,  =  \, \ottnt{c_{{\mathrm{1}}}} \, \ottnt{c_{{\mathrm{2}}}}$,
           \item $\ottnt{e_{{\mathrm{2}}}} \,  =  \,  \zeta  (  \ottnt{c_{{\mathrm{1}}}}  ,  \ottnt{c_{{\mathrm{2}}}}  ) $,
           \item $\Delta  \ottsym{;}   \mathsf{none}  \,   \vdash  \ottnt{c_{{\mathrm{1}}}} \, \ottnt{c_{{\mathrm{2}}}} \,   \ottsym{:}  \ottnt{A} \,  |  \, \epsilon$,
           \item $\Delta  \ottsym{;}   \mathsf{none}  \,   \vdash  \ottnt{c_{{\mathrm{1}}}} \,   \ottsym{:}   \ottnt{B}   \rightarrow  \!  \epsilon'  \;  \ottnt{A}  \,  |  \, \epsilon$,
           \item $\Delta  \ottsym{;}   \mathsf{none}  \,   \vdash  \ottnt{c_{{\mathrm{2}}}} \,   \ottsym{:}  \ottnt{B} \,  |  \, \epsilon$, and
           \item $\epsilon' \,  \subseteq  \, \epsilon$.
          \end{itemize}
          By \reflem{value-inversion-constant} and the assumption about constants, we have
          $\ottnt{B} =  \mathit{ty}  (  \ottnt{c_{{\mathrm{2}}}}  )  = \iota$ and
          $ \mathit{ty}  (  \ottnt{c_{{\mathrm{1}}}}  )  \,  =  \,  \iota   \rightarrow  \!   \langle \rangle   \;  \ottnt{A} $ and
          $\epsilon' \,  =  \,  \langle \rangle $
          for some $\iota$.
          By the assumption about $ \zeta $,
          $ \zeta  (  \ottnt{c_{{\mathrm{1}}}}  ,  \ottnt{c_{{\mathrm{2}}}}  ) $ is a constant and $ \mathit{ty}  (   \zeta  (  \ottnt{c_{{\mathrm{1}}}}  ,  \ottnt{c_{{\mathrm{2}}}}  )   )  \,  =  \, \ottnt{A}$.
          Thus, by \T{Const}, $\Delta  \ottsym{;}   \mathsf{none}  \,   \vdash   \zeta  (  \ottnt{c_{{\mathrm{1}}}}  ,  \ottnt{c_{{\mathrm{2}}}}  )  \,   \ottsym{:}  \ottnt{A} \,  |  \, \epsilon$;
          note that $\vdash  \Delta$ by \reflem{typing-context-wf}.

          \case \R{Beta}:
          We are given
          \begin{itemize}
           \item $\ottnt{e_{{\mathrm{1}}}} \,  =  \, \ottsym{(}   \lambda\!  \, \mathit{x}  \ottsym{.}  \ottnt{e}  \ottsym{)} \, \ottnt{v}$,
           \item $\ottnt{e_{{\mathrm{2}}}} \,  =  \,  \ottnt{e}    [  \ottnt{v}  \ottsym{/}  \mathit{x}  ]  $,
           \item $\Delta  \ottsym{;}   \mathsf{none}  \,   \vdash  \ottsym{(}   \lambda\!  \, \mathit{x}  \ottsym{.}  \ottnt{e}  \ottsym{)} \, \ottnt{v} \,   \ottsym{:}  \ottnt{A} \,  |  \, \epsilon$,
           \item $\Delta  \ottsym{;}   \mathsf{none}  \,   \vdash   \lambda\!  \, \mathit{x}  \ottsym{.}  \ottnt{e} \,   \ottsym{:}   \ottnt{B}   \rightarrow  \!  \epsilon'  \;  \ottnt{A}  \,  |  \, \epsilon$,
           \item $\Delta  \ottsym{;}   \mathsf{none}  \,   \vdash  \ottnt{v} \,   \ottsym{:}  \ottnt{B} \,  |  \, \epsilon$, and
           \item $\epsilon' \,  \subseteq  \, \epsilon$.
          \end{itemize}
          By \reflem{value-inversion-abs},
          $\Delta  \ottsym{,}  \mathit{x} \,  \mathord{:}  \, \ottnt{B}  \ottsym{;}   \mathsf{none}  \,   \vdash  \ottnt{e} \,   \ottsym{:}  \ottnt{A} \,  |  \, \epsilon'$.
          By \T{Weak}, $\Delta  \ottsym{,}  \mathit{x} \,  \mathord{:}  \, \ottnt{B}  \ottsym{;}   \mathsf{none}  \,   \vdash  \ottnt{e} \,   \ottsym{:}  \ottnt{A} \,  |  \, \epsilon$.
          By \reflem{val-subst} (\ref{lem:val-subst:term}),
          $\Delta  \ottsym{;}   \mathsf{none}  \,   \vdash   \ottnt{e}    [  \ottnt{v}  \ottsym{/}  \mathit{x}  ]   \,   \ottsym{:}  \ottnt{A} \,  |  \, \epsilon$.

          \case \R{OpApp1}:
          By \reflem{cont-inversion}, we are given
          \begin{itemize}
           \item $\ottnt{e_{{\mathrm{1}}}} \,  =  \,  \textup{\texttt{\#}\relax}  \mathsf{op}   \ottsym{(}    \algeffseqoverindex{  \text{\unboldmath$\forall$}  \,  \algeffseqoverindex{ \beta }{ \text{\unboldmath$\mathit{J}$} }   \ottsym{.}  \ottnt{C} }{ \text{\unboldmath$\mathit{I}$} }    \ottsym{,}    \Lambda\!  \,  \algeffseqoverindex{ \beta }{ \text{\unboldmath$\mathit{J}$} }   \ottsym{.}  \ottnt{v}   \ottsym{,}    \ottnt{E} ^{  \algeffseqoverindex{ \beta }{ \text{\unboldmath$\mathit{J}$} }  }    \ottsym{)}  \, \ottnt{e'_{{\mathrm{2}}}}$,
           \item $\ottnt{e_{{\mathrm{2}}}} \,  =  \,  \textup{\texttt{\#}\relax}  \mathsf{op}   \ottsym{(}    \algeffseqoverindex{  \text{\unboldmath$\forall$}  \,  \algeffseqoverindex{ \beta }{ \text{\unboldmath$\mathit{J}$} }   \ottsym{.}  \ottnt{C} }{ \text{\unboldmath$\mathit{I}$} }    \ottsym{,}    \Lambda\!  \,  \algeffseqoverindex{ \beta }{ \text{\unboldmath$\mathit{J}$} }   \ottsym{.}  \ottnt{v}   \ottsym{,}    \ottnt{E} ^{  \algeffseqover{ \beta }  }  \, \ottnt{e'_{{\mathrm{2}}}}   \ottsym{)} $,
           \item $\Delta  \ottsym{;}   \mathsf{none}  \,   \vdash   \textup{\texttt{\#}\relax}  \mathsf{op}   \ottsym{(}    \algeffseqoverindex{  \text{\unboldmath$\forall$}  \,  \algeffseqoverindex{ \beta }{ \text{\unboldmath$\mathit{J}$} }   \ottsym{.}  \ottnt{C} }{ \text{\unboldmath$\mathit{I}$} }    \ottsym{,}    \Lambda\!  \,  \algeffseqoverindex{ \beta }{ \text{\unboldmath$\mathit{J}$} }   \ottsym{.}  \ottnt{v}   \ottsym{,}    \ottnt{E} ^{  \algeffseqoverindex{ \beta }{ \text{\unboldmath$\mathit{J}$} }  }    \ottsym{)}  \, \ottnt{e'_{{\mathrm{2}}}} \,   \ottsym{:}  \ottnt{A} \,  |  \, \epsilon$,
           \item $\Delta  \ottsym{;}   \mathsf{none}  \,   \vdash   \textup{\texttt{\#}\relax}  \mathsf{op}   \ottsym{(}    \algeffseqoverindex{  \text{\unboldmath$\forall$}  \,  \algeffseqoverindex{ \beta }{ \text{\unboldmath$\mathit{J}$} }   \ottsym{.}  \ottnt{C} }{ \text{\unboldmath$\mathit{I}$} }    \ottsym{,}    \Lambda\!  \,  \algeffseqoverindex{ \beta }{ \text{\unboldmath$\mathit{J}$} }   \ottsym{.}  \ottnt{v}   \ottsym{,}    \ottnt{E} ^{  \algeffseqoverindex{ \beta }{ \text{\unboldmath$\mathit{J}$} }  }    \ottsym{)}  \,   \ottsym{:}   \ottnt{B}   \rightarrow  \!  \epsilon'  \;  \ottnt{A}  \,  |  \, \epsilon$,
           \item $\Delta  \ottsym{;}   \mathsf{none}  \,   \vdash  \ottnt{e'_{{\mathrm{2}}}} \,   \ottsym{:}  \ottnt{B} \,  |  \, \epsilon$,
           \item $\epsilon' \,  \subseteq  \, \epsilon$,
           \item $\mathit{ty} \, \ottsym{(}  \mathsf{op}  \ottsym{)} \,  =  \,   \text{\unboldmath$\forall$}     \algeffseqoverindex{ \alpha }{ \text{\unboldmath$\mathit{I}$} }   .  \ottnt{A'}  \hookrightarrow  \ottnt{B'} $
           \item $\epsilon'' \,  \subseteq  \, \epsilon$,
           \item $\mathsf{op} \,  \in  \, \epsilon''$,
           \item $\Delta  \vdash   \algeffseqoverindex{  \text{\unboldmath$\forall$}  \,  \algeffseqoverindex{ \beta }{ \text{\unboldmath$\mathit{J}$} }   \ottsym{.}  \ottnt{C} }{ \text{\unboldmath$\mathit{I}$} } $,
           \item $\Delta  \ottsym{,}   \algeffseqoverindex{ \beta }{ \text{\unboldmath$\mathit{J}$} }   \ottsym{;}   \mathsf{none}  \,   \vdash  \ottnt{v} \,   \ottsym{:}   \ottnt{A'}    [   \algeffseqoverindex{ \ottnt{C} }{ \text{\unboldmath$\mathit{I}$} }   \ottsym{/}   \algeffseqoverindex{ \alpha }{ \text{\unboldmath$\mathit{I}$} }   ]   \,  |  \, \epsilon''$, and
           \item $ \Delta   \vdash    \ottnt{E} ^{  \algeffseqoverindex{ \beta }{ \text{\unboldmath$\mathit{J}$} }  }    \ottsym{:}     \text{\unboldmath$\forall$}  \,  \algeffseqoverindex{ \beta }{ \text{\unboldmath$\mathit{J}$} }   \ottsym{.}  \ottsym{(}   \ottnt{B'}    [   \algeffseqoverindex{ \ottnt{C} }{ \text{\unboldmath$\mathit{I}$} }   \ottsym{/}   \algeffseqoverindex{ \alpha }{ \text{\unboldmath$\mathit{I}$} }   ]    \ottsym{)}  \multimap   \ottnt{B}   \rightarrow  \!  \epsilon'  \;  \ottnt{A}    \,  |  \, \epsilon'' $
          \end{itemize}
          for some $\ottnt{e'_{{\mathrm{2}}}}$, $\epsilon'$, $ \algeffseqoverindex{ \alpha }{ \text{\unboldmath$\mathit{I}$} } $, $ \algeffseqoverindex{ \beta }{ \text{\unboldmath$\mathit{J}$} } $, $\ottnt{B}$, $ \algeffseqoverindex{ \ottnt{C} }{ \text{\unboldmath$\mathit{I}$} } $, $\ottnt{A'}$, $\ottnt{B'}$, and $\epsilon''$.

          Since $\epsilon' \,  \subseteq  \, \epsilon$ and $\epsilon'' \,  \subseteq  \, \epsilon$, we have
          \[
            \Delta   \vdash    \ottnt{E} ^{  \algeffseqoverindex{ \beta }{ \text{\unboldmath$\mathit{J}$} }  }  \, \ottnt{e'_{{\mathrm{2}}}}   \ottsym{:}     \text{\unboldmath$\forall$}  \,  \algeffseqoverindex{ \beta }{ \text{\unboldmath$\mathit{J}$} }   \ottsym{.}  \ottsym{(}   \ottnt{B'}    [   \algeffseqoverindex{ \ottnt{C} }{ \text{\unboldmath$\mathit{I}$} }   \ottsym{/}   \algeffseqoverindex{ \alpha }{ \text{\unboldmath$\mathit{I}$} }   ]    \ottsym{)}  \multimap  \ottnt{A}   \,  |  \, \epsilon 
          \]
          by \TE{Weak} and \TE{App1}.
          By \T{Weak},
          \[
           \Delta  \ottsym{,}   \algeffseqoverindex{ \beta }{ \text{\unboldmath$\mathit{J}$} }   \ottsym{;}   \mathsf{none}  \,   \vdash  \ottnt{v} \,   \ottsym{:}   \ottnt{A'}    [   \algeffseqoverindex{ \ottnt{C} }{ \text{\unboldmath$\mathit{I}$} }   \ottsym{/}   \algeffseqoverindex{ \alpha }{ \text{\unboldmath$\mathit{I}$} }   ]   \,  |  \, \epsilon.
          \]
          Thus, by \T{OpCont}, we have the conclusion.

          \case \R{OpApp2}: Similar to the case of \R{OpApp1}.
         \end{caseanalysis}

         \case \T{Op}:
         We have two reduction rules which can be applied to effect invocation.
         \begin{caseanalysis}
          \case \R{Op}:
          We are given
          \begin{itemize}
           \item $\ottnt{e_{{\mathrm{1}}}} \,  =  \,  \textup{\texttt{\#}\relax}  \mathsf{op}   \ottsym{(}    \algeffseqoverindex{ \ottnt{C} }{ \text{\unboldmath$\mathit{I}$} }    \ottsym{,}   \ottnt{v}   \ottsym{)} $,
           \item $\ottnt{e_{{\mathrm{2}}}} \,  =  \,  \textup{\texttt{\#}\relax}  \mathsf{op}   \ottsym{(}    \algeffseqoverindex{ \ottnt{C} }{ \text{\unboldmath$\mathit{I}$} }    \ottsym{,}   \ottnt{v}   \ottsym{,}    [\,]    \ottsym{)} $,
           \item $\ottnt{A} \,  =  \,  \ottnt{B'}    [   \algeffseqoverindex{ \ottnt{C} }{ \text{\unboldmath$\mathit{I}$} }   \ottsym{/}   \algeffseqoverindex{ \alpha }{ \text{\unboldmath$\mathit{I}$} }   ]  $,
           \item $\Delta  \ottsym{;}   \mathsf{none}  \,   \vdash   \textup{\texttt{\#}\relax}  \mathsf{op}   \ottsym{(}    \algeffseqoverindex{ \ottnt{C} }{ \text{\unboldmath$\mathit{I}$} }    \ottsym{,}   \ottnt{v}   \ottsym{)}  \,   \ottsym{:}   \ottnt{B'}    [   \algeffseqoverindex{ \ottnt{C} }{ \text{\unboldmath$\mathit{I}$} }   \ottsym{/}   \algeffseqoverindex{ \alpha }{ \text{\unboldmath$\mathit{I}$} }   ]   \,  |  \, \epsilon$,
           \item $\mathit{ty} \, \ottsym{(}  \mathsf{op}  \ottsym{)} \,  =  \,   \text{\unboldmath$\forall$}     \algeffseqoverindex{ \alpha }{ \text{\unboldmath$\mathit{I}$} }   .  \ottnt{A'}  \hookrightarrow  \ottnt{B'} $,
           \item $\mathsf{op} \,  \in  \, \epsilon$,
           \item $\Delta  \ottsym{;}   \mathsf{none}  \,   \vdash  \ottnt{v} \,   \ottsym{:}   \ottnt{A'}    [   \algeffseqoverindex{ \ottnt{C} }{ \text{\unboldmath$\mathit{I}$} }   \ottsym{/}   \algeffseqoverindex{ \alpha }{ \text{\unboldmath$\mathit{I}$} }   ]   \,  |  \, \epsilon$, and
           \item $\Delta  \vdash   \algeffseqoverindex{ \ottnt{C} }{ \text{\unboldmath$\mathit{I}$} } $.
          \end{itemize}
          By \TE{Hole} and \T{OpCont}, we have the conclusion.

          \case \R{OpOp}:
          By \reflem{cont-inversion}, we are given
          \begin{itemize}
           \item $\ottnt{e_{{\mathrm{1}}}} \,  =  \,  \textup{\texttt{\#}\relax}  \mathsf{op}'   \ottsym{(}    \algeffseqoverindex{ \ottnt{C'} }{ \text{\unboldmath$\mathit{I'}$} }    \ottsym{,}    \textup{\texttt{\#}\relax}  \mathsf{op}   \ottsym{(}    \algeffseqoverindex{  \text{\unboldmath$\forall$}  \,  \algeffseqoverindex{ \beta }{ \text{\unboldmath$\mathit{J}$} }   \ottsym{.}  \ottnt{C} }{ \text{\unboldmath$\mathit{I}$} }    \ottsym{,}    \Lambda\!  \,  \algeffseqoverindex{ \beta }{ \text{\unboldmath$\mathit{J}$} }   \ottsym{.}  \ottnt{v}   \ottsym{,}    \ottnt{E} ^{  \algeffseqoverindex{ \beta }{ \text{\unboldmath$\mathit{J}$} }  }    \ottsym{)}    \ottsym{)} $,
           \item $\ottnt{e_{{\mathrm{2}}}} \,  =  \,  \textup{\texttt{\#}\relax}  \mathsf{op}   \ottsym{(}    \algeffseqoverindex{  \text{\unboldmath$\forall$}  \,  \algeffseqoverindex{ \beta }{ \text{\unboldmath$\mathit{J}$} }   \ottsym{.}  \ottnt{C} }{ \text{\unboldmath$\mathit{I}$} }    \ottsym{,}    \Lambda\!  \,  \algeffseqoverindex{ \beta }{ \text{\unboldmath$\mathit{J}$} }   \ottsym{.}  \ottnt{v}   \ottsym{,}    \textup{\texttt{\#}\relax}  \mathsf{op}'   \ottsym{(}    \algeffseqoverindex{ \ottnt{C'} }{ \text{\unboldmath$\mathit{I'}$} }    \ottsym{,}    \ottnt{E} ^{  \algeffseqoverindex{ \beta }{ \text{\unboldmath$\mathit{J}$} }  }    \ottsym{)}    \ottsym{)} $,
           \item $\ottnt{A} \,  =  \,  \ottnt{B'}    [   \algeffseqoverindex{ \ottnt{C'} }{ \text{\unboldmath$\mathit{I'}$} }   \ottsym{/}   \algeffseqoverindex{ \gamma }{ \text{\unboldmath$\mathit{I'}$} }   ]  $,
           \item $\Delta  \ottsym{;}   \mathsf{none}  \,   \vdash   \textup{\texttt{\#}\relax}  \mathsf{op}'   \ottsym{(}    \algeffseqoverindex{ \ottnt{C'} }{ \text{\unboldmath$\mathit{I'}$} }    \ottsym{,}    \textup{\texttt{\#}\relax}  \mathsf{op}   \ottsym{(}    \algeffseqoverindex{  \text{\unboldmath$\forall$}  \,  \algeffseqoverindex{ \beta }{ \text{\unboldmath$\mathit{J}$} }   \ottsym{.}  \ottnt{C} }{ \text{\unboldmath$\mathit{I}$} }    \ottsym{,}    \Lambda\!  \,  \algeffseqoverindex{ \beta }{ \text{\unboldmath$\mathit{J}$} }   \ottsym{.}  \ottnt{v}   \ottsym{,}    \ottnt{E} ^{  \algeffseqoverindex{ \beta }{ \text{\unboldmath$\mathit{J}$} }  }    \ottsym{)}    \ottsym{)}  \,   \ottsym{:}   \ottnt{B'}    [   \algeffseqoverindex{ \ottnt{C'} }{ \text{\unboldmath$\mathit{I'}$} }   \ottsym{/}   \algeffseqoverindex{ \gamma }{ \text{\unboldmath$\mathit{I'}$} }   ]   \,  |  \, \epsilon$
           \item $\mathit{ty} \, \ottsym{(}  \mathsf{op}'  \ottsym{)} \,  =  \,   \text{\unboldmath$\forall$}     \algeffseqoverindex{ \gamma }{ \text{\unboldmath$\mathit{I'}$} }   .  \ottnt{A'}  \hookrightarrow  \ottnt{B'} $,
           \item $\mathsf{op}' \,  \in  \, \epsilon$,
           \item $\Delta  \vdash   \algeffseqoverindex{ \ottnt{C'} }{ \text{\unboldmath$\mathit{I'}$} } $,
           \item $\Delta  \ottsym{;}   \mathsf{none}  \,   \vdash   \textup{\texttt{\#}\relax}  \mathsf{op}   \ottsym{(}    \algeffseqoverindex{  \text{\unboldmath$\forall$}  \,  \algeffseqoverindex{ \beta }{ \text{\unboldmath$\mathit{J}$} }   \ottsym{.}  \ottnt{C} }{ \text{\unboldmath$\mathit{I}$} }    \ottsym{,}    \Lambda\!  \,  \algeffseqoverindex{ \beta }{ \text{\unboldmath$\mathit{J}$} }   \ottsym{.}  \ottnt{v}   \ottsym{,}    \ottnt{E} ^{  \algeffseqoverindex{ \beta }{ \text{\unboldmath$\mathit{J}$} }  }    \ottsym{)}  \,   \ottsym{:}   \ottnt{A'}    [   \algeffseqoverindex{ \ottnt{C'} }{ \text{\unboldmath$\mathit{I'}$} }   \ottsym{/}   \algeffseqoverindex{ \gamma }{ \text{\unboldmath$\mathit{I'}$} }   ]   \,  |  \, \epsilon$,
           \item $\mathit{ty} \, \ottsym{(}  \mathsf{op}  \ottsym{)} \,  =  \,   \text{\unboldmath$\forall$}     \algeffseqoverindex{ \alpha }{ \text{\unboldmath$\mathit{I}$} }   .  \ottnt{A''}  \hookrightarrow  \ottnt{B''} $,
           \item $\mathsf{op} \,  \in  \, \epsilon'$,
           \item $\epsilon' \,  \subseteq  \, \epsilon$,
           \item $\Delta  \vdash   \algeffseqoverindex{  \text{\unboldmath$\forall$}  \,  \algeffseqoverindex{ \beta }{ \text{\unboldmath$\mathit{J}$} }   \ottsym{.}  \ottnt{C} }{ \text{\unboldmath$\mathit{I}$} } $,
           \item $\Delta  \ottsym{,}   \algeffseqoverindex{ \beta }{ \text{\unboldmath$\mathit{J}$} }   \ottsym{;}   \mathsf{none}  \,   \vdash  \ottnt{v} \,   \ottsym{:}   \ottnt{A''}    [   \algeffseqoverindex{ \ottnt{C} }{ \text{\unboldmath$\mathit{I}$} }   \ottsym{/}   \algeffseqoverindex{ \alpha }{ \text{\unboldmath$\mathit{I}$} }   ]   \,  |  \, \epsilon'$, and
           \item $ \Delta   \vdash    \ottnt{E} ^{  \algeffseqoverindex{ \beta }{ \text{\unboldmath$\mathit{J}$} }  }    \ottsym{:}     \text{\unboldmath$\forall$}  \,  \algeffseqoverindex{ \beta }{ \text{\unboldmath$\mathit{J}$} }   \ottsym{.}   \ottnt{B''}    [   \algeffseqoverindex{ \ottnt{C} }{ \text{\unboldmath$\mathit{I}$} }   \ottsym{/}   \algeffseqoverindex{ \alpha }{ \text{\unboldmath$\mathit{I}$} }   ]    \multimap   \ottnt{A'}    [   \algeffseqoverindex{ \ottnt{C'} }{ \text{\unboldmath$\mathit{I'}$} }   \ottsym{/}   \algeffseqoverindex{ \gamma }{ \text{\unboldmath$\mathit{I'}$} }   ]     \,  |  \, \epsilon' $.
          \end{itemize}

          By \TE{Weak} and \TE{Op},
          \[
            \Delta   \vdash    \textup{\texttt{\#}\relax}  \mathsf{op}   \ottsym{(}    \algeffseqoverindex{ \ottnt{C'} }{ \text{\unboldmath$\mathit{I'}$} }    \ottsym{,}    \ottnt{E} ^{  \algeffseqoverindex{ \beta }{ \text{\unboldmath$\mathit{J}$} }  }    \ottsym{)}    \ottsym{:}     \text{\unboldmath$\forall$}  \,  \algeffseqoverindex{ \beta }{ \text{\unboldmath$\mathit{J}$} }   \ottsym{.}   \ottnt{B''}    [   \algeffseqoverindex{ \ottnt{C} }{ \text{\unboldmath$\mathit{I}$} }   \ottsym{/}   \algeffseqoverindex{ \alpha }{ \text{\unboldmath$\mathit{I}$} }   ]    \multimap   \ottnt{B'}    [   \algeffseqoverindex{ \ottnt{C'} }{ \text{\unboldmath$\mathit{I'}$} }   \ottsym{/}   \algeffseqoverindex{ \gamma }{ \text{\unboldmath$\mathit{I'}$} }   ]     \,  |  \, \epsilon .
          \]
          By \T{Weak},
          \[
           \Delta  \ottsym{,}   \algeffseqoverindex{ \beta }{ \text{\unboldmath$\mathit{J}$} }   \ottsym{;}   \mathsf{none}  \,   \vdash  \ottnt{v} \,   \ottsym{:}   \ottnt{A''}    [   \algeffseqoverindex{ \ottnt{C} }{ \text{\unboldmath$\mathit{I}$} }   \ottsym{/}   \algeffseqoverindex{ \alpha }{ \text{\unboldmath$\mathit{I}$} }   ]   \,  |  \, \epsilon.
          \]
          By \T{OpCont}, we have the conclusion.

         \end{caseanalysis}

         \case \T{OpCont}: Contradictory.
         \case \T{Weak}: By the IH and \T{Weak}.

         \case \T{Handle}:
         We have three reduction rules which can be applied to handler expressions.
         \begin{caseanalysis}
         \case \R{Return}:
         We are given
         \begin{itemize}
          \item $\ottnt{e_{{\mathrm{1}}}} \,  =  \, \mathsf{handle} \, \ottnt{v} \, \mathsf{with} \, \ottnt{h}$,
          \item $ \ottnt{h} ^\mathsf{return}  \,  =  \, \mathsf{return} \, \mathit{x}  \rightarrow  \ottnt{e}$,
          \item $\ottnt{e_{{\mathrm{2}}}} \,  =  \,  \ottnt{e}    [  \ottnt{v}  \ottsym{/}  \mathit{x}  ]  $,
          \item $\Delta  \ottsym{;}   \mathsf{none}  \,   \vdash  \mathsf{handle} \, \ottnt{v} \, \mathsf{with} \, \ottnt{h} \,   \ottsym{:}  \ottnt{A} \,  |  \, \epsilon$,
          \item $\Delta  \ottsym{;}   \mathsf{none}  \,   \vdash  \ottnt{v} \,   \ottsym{:}  \ottnt{B} \,  |  \, \epsilon'$, and
          \item $\Delta  \ottsym{;}   \mathsf{none}  \,   \vdash  \ottnt{h}  \ottsym{:}  \ottnt{B} \,  |  \, \epsilon'  \Rightarrow  \ottnt{A} \,  |  \, \epsilon$.
         \end{itemize}
         By \reflem{handler-inversion},
         $\Gamma  \ottsym{,}  \mathit{x} \,  \mathord{:}  \, \ottnt{B}  \ottsym{;}   \mathsf{none}  \,   \vdash  \ottnt{e} \,   \ottsym{:}  \ottnt{A} \,  |  \, \epsilon$.
         By \reflem{val-subst}, we finish.

         \case \R{Handle}:
         By Lemmas~\ref{lem:handler-inversion} and \ref{lem:cont-inversion},
         we are given
         \begin{itemize}
          \item $\ottnt{e_{{\mathrm{1}}}} \,  =  \, \mathsf{handle} \,  \textup{\texttt{\#}\relax}  \mathsf{op}   \ottsym{(}    \algeffseqoverindex{  \text{\unboldmath$\forall$}  \,  \algeffseqoverindex{ \beta }{ \text{\unboldmath$\mathit{J}$} }   \ottsym{.}  \ottnt{C} }{ \text{\unboldmath$\mathit{I}$} }    \ottsym{,}    \Lambda\!  \,  \algeffseqoverindex{ \beta }{ \text{\unboldmath$\mathit{J}$} }   \ottsym{.}  \ottnt{v}   \ottsym{,}    \ottnt{E} ^{  \algeffseqoverindex{ \beta }{ \text{\unboldmath$\mathit{J}$} }  }    \ottsym{)}  \, \mathsf{with} \, \ottnt{h}$,
          \item $ \ottnt{h} ^{ \mathsf{op} }  \,  =  \,  \Lambda\!  \,  \algeffseqoverindex{ \alpha }{ \text{\unboldmath$\mathit{I}$} }   \ottsym{.}  \mathsf{op}  \ottsym{(}  \mathit{x}  \ottsym{)}  \rightarrow  \ottnt{e}$,
          \item $\ottnt{e_{{\mathrm{2}}}} \,  =  \,    \ottnt{e}    [  \mathsf{handle} \,  \ottnt{E} ^{  \algeffseqoverindex{ \beta }{ \text{\unboldmath$\mathit{J}$} }  }  \, \mathsf{with} \, \ottnt{h}  /  \mathsf{resume}  ]^{  \algeffseqoverindex{  \text{\unboldmath$\forall$}  \,  \algeffseqoverindex{ \beta }{ \text{\unboldmath$\mathit{J}$} }   \ottsym{.}  \ottnt{C} }{ \text{\unboldmath$\mathit{I}$} }  }_{  \Lambda\!  \,  \algeffseqoverindex{ \beta }{ \text{\unboldmath$\mathit{J}$} }   \ottsym{.}  \ottnt{v} }      [   \algeffseqoverindex{  \ottnt{C}    [   \algeffseqoverindex{  \bot  }{ \text{\unboldmath$\mathit{J}$} }   \ottsym{/}   \algeffseqoverindex{ \beta }{ \text{\unboldmath$\mathit{J}$} }   ]   }{ \text{\unboldmath$\mathit{I}$} }   \ottsym{/}   \algeffseqoverindex{ \alpha }{ \text{\unboldmath$\mathit{I}$} }   ]      [   \ottnt{v}    [   \algeffseqoverindex{  \bot  }{ \text{\unboldmath$\mathit{J}$} }   \ottsym{/}   \algeffseqoverindex{ \beta }{ \text{\unboldmath$\mathit{J}$} }   ]    \ottsym{/}  \mathit{x}  ]  $,
          \item $\Delta  \ottsym{;}   \mathsf{none}  \,   \vdash  \mathsf{handle} \,  \textup{\texttt{\#}\relax}  \mathsf{op}   \ottsym{(}    \algeffseqoverindex{  \text{\unboldmath$\forall$}  \,  \algeffseqoverindex{ \beta }{ \text{\unboldmath$\mathit{J}$} }   \ottsym{.}  \ottnt{C} }{ \text{\unboldmath$\mathit{I}$} }    \ottsym{,}    \Lambda\!  \,  \algeffseqoverindex{ \beta }{ \text{\unboldmath$\mathit{J}$} }   \ottsym{.}  \ottnt{v}   \ottsym{,}    \ottnt{E} ^{  \algeffseqoverindex{ \beta }{ \text{\unboldmath$\mathit{J}$} }  }    \ottsym{)}  \, \mathsf{with} \, \ottnt{h} \,   \ottsym{:}  \ottnt{A} \,  |  \, \epsilon$,
          \item $\Delta  \ottsym{;}   \mathsf{none}  \,   \vdash  \ottnt{h}  \ottsym{:}  \ottnt{B} \,  |  \, \epsilon'  \Rightarrow  \ottnt{A} \,  |  \, \epsilon$,
          \item $\Delta  \ottsym{;}   \mathsf{none}  \,   \vdash   \textup{\texttt{\#}\relax}  \mathsf{op}   \ottsym{(}    \algeffseqoverindex{  \text{\unboldmath$\forall$}  \,  \algeffseqoverindex{ \beta }{ \text{\unboldmath$\mathit{J}$} }   \ottsym{.}  \ottnt{C} }{ \text{\unboldmath$\mathit{I}$} }    \ottsym{,}    \Lambda\!  \,  \algeffseqoverindex{ \beta }{ \text{\unboldmath$\mathit{J}$} }   \ottsym{.}  \ottnt{v}   \ottsym{,}    \ottnt{E} ^{  \algeffseqoverindex{ \beta }{ \text{\unboldmath$\mathit{J}$} }  }    \ottsym{)}  \,   \ottsym{:}  \ottnt{B} \,  |  \, \epsilon'$,
          \item $\mathit{ty} \, \ottsym{(}  \mathsf{op}  \ottsym{)} \,  =  \,   \text{\unboldmath$\forall$}     \algeffseqoverindex{ \alpha }{ \text{\unboldmath$\mathit{I}$} }   .  \ottnt{A'}  \hookrightarrow  \ottnt{B'} $,
          \item $\Delta  \ottsym{,}   \algeffseqoverindex{ \alpha }{ \text{\unboldmath$\mathit{I}$} }   \ottsym{,}  \mathit{x} \,  \mathord{:}  \, \ottnt{A'}  \ottsym{;}  \ottsym{(}   \algeffseqoverindex{ \alpha }{ \text{\unboldmath$\mathit{I}$} }   \ottsym{,}  \ottnt{A'}  \ottsym{,}   \ottnt{B'}   \rightarrow  \!  \epsilon  \;  \ottnt{A}   \ottsym{)} \,   \vdash  \ottnt{e} \,   \ottsym{:}  \ottnt{A} \,  |  \, \epsilon$,
          \item $\mathsf{op} \,  \in  \, \epsilon''$,
          \item $\epsilon'' \,  \subseteq  \, \epsilon'$,
          \item $\Delta  \vdash   \algeffseqoverindex{  \text{\unboldmath$\forall$}  \,  \algeffseqoverindex{ \beta }{ \text{\unboldmath$\mathit{J}$} }   \ottsym{.}  \ottnt{C} }{ \text{\unboldmath$\mathit{I}$} } $,
          \item $\Delta  \ottsym{,}   \algeffseqoverindex{ \beta }{ \text{\unboldmath$\mathit{J}$} }   \ottsym{;}   \mathsf{none}  \,   \vdash  \ottnt{v} \,   \ottsym{:}   \ottnt{A'}    [   \algeffseqoverindex{ \ottnt{C} }{ \text{\unboldmath$\mathit{I}$} }   \ottsym{/}   \algeffseqoverindex{ \alpha }{ \text{\unboldmath$\mathit{I}$} }   ]   \,  |  \, \epsilon''$, and
          \item $ \Delta   \vdash    \ottnt{E} ^{  \algeffseqoverindex{ \beta }{ \text{\unboldmath$\mathit{J}$} }  }    \ottsym{:}     \text{\unboldmath$\forall$}  \,  \algeffseqoverindex{ \beta }{ \text{\unboldmath$\mathit{J}$} }   \ottsym{.}   \ottnt{B'}    [   \algeffseqoverindex{ \ottnt{C} }{ \text{\unboldmath$\mathit{I}$} }   \ottsym{/}   \algeffseqoverindex{ \alpha }{ \text{\unboldmath$\mathit{I}$} }   ]    \multimap  \ottnt{B}   \,  |  \, \epsilon'' $.
         \end{itemize}
         Since $ \Delta   \vdash    \ottnt{E} ^{  \algeffseqoverindex{ \beta }{ \text{\unboldmath$\mathit{J}$} }  }    \ottsym{:}     \text{\unboldmath$\forall$}  \,  \algeffseqoverindex{ \beta }{ \text{\unboldmath$\mathit{J}$} }   \ottsym{.}   \ottnt{B'}    [   \algeffseqoverindex{ \ottnt{C} }{ \text{\unboldmath$\mathit{I}$} }   \ottsym{/}   \algeffseqoverindex{ \alpha }{ \text{\unboldmath$\mathit{I}$} }   ]    \multimap  \ottnt{B}   \,  |  \, \epsilon'' $ and
         $\Delta  \ottsym{;}   \mathsf{none}  \,   \vdash  \ottnt{h}  \ottsym{:}  \ottnt{B} \,  |  \, \epsilon'  \Rightarrow  \ottnt{A} \,  |  \, \epsilon$,
         we have
         \begin{equation}
           \Delta  \ottsym{,}   \algeffseqoverindex{ \alpha }{ \text{\unboldmath$\mathit{I}$} }   \ottsym{,}  \mathit{x} \,  \mathord{:}  \, \ottnt{A'}   \vdash   \mathsf{handle} \,  \ottnt{E} ^{  \algeffseqoverindex{ \beta }{ \text{\unboldmath$\mathit{J}$} }  }  \, \mathsf{with} \, \ottnt{h}   \ottsym{:}     \text{\unboldmath$\forall$}  \,  \algeffseqoverindex{ \beta }{ \text{\unboldmath$\mathit{J}$} }   \ottsym{.}   \ottnt{B'}    [   \algeffseqoverindex{ \ottnt{C} }{ \text{\unboldmath$\mathit{I}$} }   \ottsym{/}   \algeffseqoverindex{ \alpha }{ \text{\unboldmath$\mathit{I}$} }   ]    \multimap  \ottnt{A}   \,  |  \, \epsilon 
          \label{eqn:subject-red:handle:one}
         \end{equation}
         by \TE{Weak}, \TE{Handle}, and weakening (Lemmas~\ref{lem:typing-context-wf} and \ref{lem:weakening}).
         Since $\Delta  \ottsym{,}   \algeffseqoverindex{ \beta }{ \text{\unboldmath$\mathit{J}$} }   \ottsym{;}   \mathsf{none}  \,   \vdash  \ottnt{v} \,   \ottsym{:}   \ottnt{A'}    [   \algeffseqoverindex{ \ottnt{C} }{ \text{\unboldmath$\mathit{I}$} }   \ottsym{/}   \algeffseqoverindex{ \alpha }{ \text{\unboldmath$\mathit{I}$} }   ]   \,  |  \, \epsilon''$, we have
         \begin{equation}
          \Delta  \ottsym{,}   \algeffseqoverindex{ \alpha }{ \text{\unboldmath$\mathit{I}$} }   \ottsym{,}  \mathit{x} \,  \mathord{:}  \, \ottnt{A'}  \ottsym{,}   \algeffseqoverindex{ \beta }{ \text{\unboldmath$\mathit{J}$} }   \vdash  \ottnt{v}  \ottsym{:}   \ottnt{A'}    [   \algeffseqoverindex{ \ottnt{C} }{ \text{\unboldmath$\mathit{I}$} }   \ottsym{/}   \algeffseqoverindex{ \alpha }{ \text{\unboldmath$\mathit{I}$} }   ]  
          \label{eqn:subject-red:handle:two}
         \end{equation}
         by weakening.
         Since $\Delta  \vdash   \algeffseqoverindex{  \text{\unboldmath$\forall$}  \,  \algeffseqoverindex{ \beta }{ \text{\unboldmath$\mathit{J}$} }   \ottsym{.}  \ottnt{C} }{ \text{\unboldmath$\mathit{I}$} } $, we have
         \begin{equation}
          \Delta  \ottsym{,}   \algeffseqoverindex{ \alpha }{ \text{\unboldmath$\mathit{I}$} }   \ottsym{,}  \mathit{x} \,  \mathord{:}  \, \ottnt{A'}  \vdash   \algeffseqoverindex{  \text{\unboldmath$\forall$}  \,  \algeffseqoverindex{ \beta }{ \text{\unboldmath$\mathit{J}$} }   \ottsym{.}  \ottnt{C} }{ \text{\unboldmath$\mathit{I}$} } .
          \label{eqn:subject-red:handle:three}
         \end{equation}
         Since $\Delta  \ottsym{,}   \algeffseqoverindex{ \alpha }{ \text{\unboldmath$\mathit{I}$} }   \ottsym{,}  \mathit{x} \,  \mathord{:}  \, \ottnt{A'}  \ottsym{;}  \ottsym{(}   \algeffseqoverindex{ \alpha }{ \text{\unboldmath$\mathit{I}$} }   \ottsym{,}  \ottnt{A'}  \ottsym{,}   \ottnt{B'}   \rightarrow  \!  \epsilon  \;  \ottnt{A}   \ottsym{)} \,   \vdash  \ottnt{e} \,   \ottsym{:}  \ottnt{A} \,  |  \, \epsilon$,
         we have
         \[
          \Delta  \ottsym{,}   \algeffseqoverindex{ \alpha }{ \text{\unboldmath$\mathit{I}$} }   \ottsym{,}  \mathit{x} \,  \mathord{:}  \, \ottnt{A'}  \ottsym{;}   \mathsf{none}  \,   \vdash   \ottnt{e}    [  \mathsf{handle} \,  \ottnt{E} ^{  \algeffseqoverindex{ \beta }{ \text{\unboldmath$\mathit{J}$} }  }  \, \mathsf{with} \, \ottnt{h}  /  \mathsf{resume}  ]^{  \algeffseqoverindex{  \text{\unboldmath$\forall$}  \,  \algeffseqoverindex{ \beta }{ \text{\unboldmath$\mathit{J}$} }   \ottsym{.}  \ottnt{C} }{ \text{\unboldmath$\mathit{I}$} }  }_{  \Lambda\!  \,  \algeffseqoverindex{ \beta }{ \text{\unboldmath$\mathit{J}$} }   \ottsym{.}  \ottnt{v} }   \,   \ottsym{:}  \ottnt{A} \,  |  \, \epsilon
         \]
         by \reflem{cont-subst} with
         (\ref{eqn:subject-red:handle:one}),
         (\ref{eqn:subject-red:handle:two}), and
         (\ref{eqn:subject-red:handle:three}).
         Since $\Delta  \vdash   \algeffseqoverindex{  \text{\unboldmath$\forall$}  \,  \algeffseqoverindex{ \beta }{ \text{\unboldmath$\mathit{J}$} }   \ottsym{.}  \ottnt{C} }{ \text{\unboldmath$\mathit{I}$} } $,
         we have $\Delta  \vdash   \algeffseqoverindex{  \ottnt{C}    [   \algeffseqoverindex{  \bot  }{ \text{\unboldmath$\mathit{J}$} }   \ottsym{/}   \algeffseqoverindex{ \beta }{ \text{\unboldmath$\mathit{J}$} }   ]   }{ \text{\unboldmath$\mathit{I}$} } $.
         Thus,
         \[\begin{array}{l}
          \Delta  \ottsym{,}  \mathit{x} \,  \mathord{:}  \, \ottnt{A'} \,  [   \algeffseqoverindex{  \ottnt{C}    [   \algeffseqoverindex{  \bot  }{ \text{\unboldmath$\mathit{J}$} }   \ottsym{/}   \algeffseqoverindex{ \beta }{ \text{\unboldmath$\mathit{J}$} }   ]   }{ \text{\unboldmath$\mathit{I}$} }   \ottsym{/}   \algeffseqoverindex{ \alpha }{ \text{\unboldmath$\mathit{I}$} }   ]   \ottsym{;}   \mathsf{none}  \,   \\     \quad    \vdash    \ottnt{e}    [  \mathsf{handle} \,  \ottnt{E} ^{  \algeffseqoverindex{ \beta }{ \text{\unboldmath$\mathit{J}$} }  }  \, \mathsf{with} \, \ottnt{h}  /  \mathsf{resume}  ]^{  \algeffseqoverindex{  \text{\unboldmath$\forall$}  \,  \algeffseqoverindex{ \beta }{ \text{\unboldmath$\mathit{J}$} }   \ottsym{.}  \ottnt{C} }{ \text{\unboldmath$\mathit{I}$} }  }_{  \Lambda\!  \,  \algeffseqoverindex{ \beta }{ \text{\unboldmath$\mathit{J}$} }   \ottsym{.}  \ottnt{v} }      [   \algeffseqoverindex{  \ottnt{C}    [   \algeffseqoverindex{  \bot  }{ \text{\unboldmath$\mathit{J}$} }   \ottsym{/}   \algeffseqoverindex{ \beta }{ \text{\unboldmath$\mathit{J}$} }   ]   }{ \text{\unboldmath$\mathit{I}$} }   \ottsym{/}   \algeffseqoverindex{ \alpha }{ \text{\unboldmath$\mathit{I}$} }   ]   \,   \ottsym{:}  \ottnt{A} \,  |  \, \epsilon
         \end{array}\]
         by \reflem{ty-subst},
         where note that $ \ottnt{A}    [   \algeffseqoverindex{  \ottnt{C}    [   \algeffseqoverindex{  \bot  }{ \text{\unboldmath$\mathit{J}$} }   \ottsym{/}   \algeffseqoverindex{ \beta }{ \text{\unboldmath$\mathit{J}$} }   ]   }{ \text{\unboldmath$\mathit{I}$} }   \ottsym{/}   \algeffseqoverindex{ \alpha }{ \text{\unboldmath$\mathit{I}$} }   ]   \,  =  \, \ottnt{A}$ because $\Delta  \vdash  \ottnt{A}$ by \reflem{type-wf}.
         Since $\Delta  \ottsym{,}   \algeffseqoverindex{ \beta }{ \text{\unboldmath$\mathit{J}$} }   \ottsym{;}   \mathsf{none}  \,   \vdash  \ottnt{v} \,   \ottsym{:}   \ottnt{A'}    [   \algeffseqoverindex{ \ottnt{C} }{ \text{\unboldmath$\mathit{I}$} }   \ottsym{/}   \algeffseqoverindex{ \alpha }{ \text{\unboldmath$\mathit{I}$} }   ]   \,  |  \, \epsilon''$, we have
         \[
          \Delta  \vdash   \ottnt{v}    [   \algeffseqoverindex{  \bot  }{ \text{\unboldmath$\mathit{J}$} }   \ottsym{/}   \algeffseqoverindex{ \beta }{ \text{\unboldmath$\mathit{J}$} }   ]    \ottsym{:}   \ottnt{A'}    [   \algeffseqoverindex{  \ottnt{C}    [   \algeffseqoverindex{  \bot  }{ \text{\unboldmath$\mathit{J}$} }   \ottsym{/}   \algeffseqoverindex{ \beta }{ \text{\unboldmath$\mathit{J}$} }   ]   }{ \text{\unboldmath$\mathit{I}$} }   \ottsym{/}   \algeffseqoverindex{ \alpha }{ \text{\unboldmath$\mathit{I}$} }   ]  
         \]
         by \reflem{ty-subst}, where note that $ \algeffseqoverindex{ \beta }{ \text{\unboldmath$\mathit{J}$} } $ do not occur free in $\ottnt{A'}$
         since $\ottnt{A'}$ is the argument type of $\mathsf{op}$.
         By \reflem{val-subst},
         \[\begin{array}{l}
          \Delta  \ottsym{;}   \mathsf{none}  \,   \\     \quad    \vdash     \ottnt{e}    [  \mathsf{handle} \,  \ottnt{E} ^{  \algeffseqoverindex{ \beta }{ \text{\unboldmath$\mathit{J}$} }  }  \, \mathsf{with} \, \ottnt{h}  /  \mathsf{resume}  ]^{  \algeffseqoverindex{  \text{\unboldmath$\forall$}  \,  \algeffseqoverindex{ \beta }{ \text{\unboldmath$\mathit{J}$} }   \ottsym{.}  \ottnt{C} }{ \text{\unboldmath$\mathit{I}$} }  }_{  \Lambda\!  \,  \algeffseqoverindex{ \beta }{ \text{\unboldmath$\mathit{J}$} }   \ottsym{.}  \ottnt{v} }      [   \algeffseqoverindex{  \ottnt{C}    [   \algeffseqoverindex{  \bot  }{ \text{\unboldmath$\mathit{J}$} }   \ottsym{/}   \algeffseqoverindex{ \beta }{ \text{\unboldmath$\mathit{J}$} }   ]   }{ \text{\unboldmath$\mathit{I}$} }   \ottsym{/}   \algeffseqoverindex{ \alpha }{ \text{\unboldmath$\mathit{I}$} }   ]   \, \ottnt{v}    [   \algeffseqoverindex{  \bot  }{ \text{\unboldmath$\mathit{J}$} }   \ottsym{/}   \algeffseqoverindex{ \beta }{ \text{\unboldmath$\mathit{J}$} }   ]   \,   \ottsym{:}  \ottnt{A} \,  |  \, \epsilon,
         \end{array}\]
         which is what we have to show.

         \case \R{OpHandle}:
         By \reflem{cont-inversion}, we are given
         \begin{itemize}
          \item $\ottnt{e_{{\mathrm{1}}}} \,  =  \, \mathsf{handle} \,  \textup{\texttt{\#}\relax}  \mathsf{op}   \ottsym{(}    \algeffseqoverindex{  \text{\unboldmath$\forall$}  \,  \algeffseqoverindex{ \beta }{ \text{\unboldmath$\mathit{J}$} }   \ottsym{.}  \ottnt{C} }{ \text{\unboldmath$\mathit{I}$} }    \ottsym{,}    \Lambda\!  \,  \algeffseqoverindex{ \beta }{ \text{\unboldmath$\mathit{J}$} }   \ottsym{.}  \ottnt{v}   \ottsym{,}    \ottnt{E} ^{  \algeffseqoverindex{ \beta }{ \text{\unboldmath$\mathit{J}$} }  }    \ottsym{)}  \, \mathsf{with} \, \ottnt{h}$,
          \item $\ottnt{e_{{\mathrm{2}}}} \,  =  \,  \textup{\texttt{\#}\relax}  \mathsf{op}   \ottsym{(}    \algeffseqoverindex{  \text{\unboldmath$\forall$}  \,  \algeffseqoverindex{ \beta }{ \text{\unboldmath$\mathit{J}$} }   \ottsym{.}  \ottnt{C} }{ \text{\unboldmath$\mathit{I}$} }    \ottsym{,}    \Lambda\!  \,  \algeffseqoverindex{ \beta }{ \text{\unboldmath$\mathit{J}$} }   \ottsym{.}  \ottnt{v}   \ottsym{,}   \mathsf{handle} \,  \ottnt{E} ^{  \algeffseqoverindex{ \beta }{ \text{\unboldmath$\mathit{J}$} }  }  \, \mathsf{with} \, \ottnt{h}   \ottsym{)} $,
          \item $\mathsf{op} \,  \not\in  \,  \mathit{ops}  (  \ottnt{h}  ) $,
          \item $\Delta  \ottsym{;}   \mathsf{none}  \,   \vdash  \mathsf{handle} \,  \textup{\texttt{\#}\relax}  \mathsf{op}   \ottsym{(}    \algeffseqoverindex{  \text{\unboldmath$\forall$}  \,  \algeffseqoverindex{ \beta }{ \text{\unboldmath$\mathit{J}$} }   \ottsym{.}  \ottnt{C} }{ \text{\unboldmath$\mathit{I}$} }    \ottsym{,}    \Lambda\!  \,  \algeffseqoverindex{ \beta }{ \text{\unboldmath$\mathit{J}$} }   \ottsym{.}  \ottnt{v}   \ottsym{,}    \ottnt{E} ^{  \algeffseqoverindex{ \beta }{ \text{\unboldmath$\mathit{J}$} }  }    \ottsym{)}  \, \mathsf{with} \, \ottnt{h} \,   \ottsym{:}  \ottnt{A} \,  |  \, \epsilon$,
          \item $\Delta  \ottsym{;}   \mathsf{none}  \,   \vdash   \textup{\texttt{\#}\relax}  \mathsf{op}   \ottsym{(}    \algeffseqoverindex{  \text{\unboldmath$\forall$}  \,  \algeffseqoverindex{ \beta }{ \text{\unboldmath$\mathit{J}$} }   \ottsym{.}  \ottnt{C} }{ \text{\unboldmath$\mathit{I}$} }    \ottsym{,}    \Lambda\!  \,  \algeffseqoverindex{ \beta }{ \text{\unboldmath$\mathit{J}$} }   \ottsym{.}  \ottnt{v}   \ottsym{,}    \ottnt{E} ^{  \algeffseqoverindex{ \beta }{ \text{\unboldmath$\mathit{J}$} }  }    \ottsym{)}  \,   \ottsym{:}  \ottnt{B} \,  |  \, \epsilon'$,
          \item $\Delta  \ottsym{;}   \mathsf{none}  \,   \vdash  \ottnt{h}  \ottsym{:}  \ottnt{B} \,  |  \, \epsilon'  \Rightarrow  \ottnt{A} \,  |  \, \epsilon$,
          \item $\epsilon'' \,  \subseteq  \, \epsilon'$,
          \item $\mathit{ty} \, \ottsym{(}  \mathsf{op}  \ottsym{)} \,  =  \,   \text{\unboldmath$\forall$}     \algeffseqoverindex{ \alpha }{ \text{\unboldmath$\mathit{I}$} }   .  \ottnt{A'}  \hookrightarrow  \ottnt{B'} $,
          \item $\mathsf{op} \,  \in  \, \epsilon''$,
          \item $\Delta  \vdash   \algeffseqoverindex{  \text{\unboldmath$\forall$}  \,  \algeffseqoverindex{ \beta }{ \text{\unboldmath$\mathit{J}$} }   \ottsym{.}  \ottnt{C} }{ \text{\unboldmath$\mathit{I}$} } $,
          \item $\Delta  \ottsym{,}   \algeffseqoverindex{ \beta }{ \text{\unboldmath$\mathit{J}$} }   \ottsym{;}   \mathsf{none}  \,   \vdash  \ottnt{v} \,   \ottsym{:}   \ottnt{A'}    [   \algeffseqoverindex{ \ottnt{C} }{ \text{\unboldmath$\mathit{I}$} }   \ottsym{/}   \algeffseqoverindex{ \alpha }{ \text{\unboldmath$\mathit{I}$} }   ]   \,  |  \, \epsilon''$, and
          \item $ \Delta   \vdash    \ottnt{E} ^{  \algeffseqoverindex{ \beta }{ \text{\unboldmath$\mathit{J}$} }  }    \ottsym{:}     \text{\unboldmath$\forall$}  \,  \algeffseqoverindex{ \beta }{ \text{\unboldmath$\mathit{J}$} }   \ottsym{.}  \ottsym{(}   \ottnt{B'}    [   \algeffseqoverindex{ \ottnt{C} }{ \text{\unboldmath$\mathit{I}$} }   \ottsym{/}   \algeffseqoverindex{ \alpha }{ \text{\unboldmath$\mathit{I}$} }   ]    \ottsym{)}  \multimap  \ottnt{B}   \,  |  \, \epsilon'' $.
         \end{itemize}
         By \TE{Weak} and \TE{Handle},
         \[
           \Delta   \vdash   \mathsf{handle} \,  \ottnt{E} ^{  \algeffseqoverindex{ \beta }{ \text{\unboldmath$\mathit{J}$} }  }  \, \mathsf{with} \, \ottnt{h}   \ottsym{:}     \text{\unboldmath$\forall$}  \,  \algeffseqoverindex{ \beta }{ \text{\unboldmath$\mathit{J}$} }   \ottsym{.}  \ottsym{(}   \ottnt{B'}    [   \algeffseqoverindex{ \ottnt{C} }{ \text{\unboldmath$\mathit{I}$} }   \ottsym{/}   \algeffseqoverindex{ \alpha }{ \text{\unboldmath$\mathit{I}$} }   ]    \ottsym{)}  \multimap  \ottnt{A}   \,  |  \, \epsilon .
         \]
         Since $\Delta  \ottsym{,}   \algeffseqoverindex{ \beta }{ \text{\unboldmath$\mathit{J}$} }   \ottsym{;}   \mathsf{none}  \,   \vdash  \ottnt{v} \,   \ottsym{:}   \ottnt{A'}    [   \algeffseqoverindex{ \ottnt{C} }{ \text{\unboldmath$\mathit{I}$} }   \ottsym{/}   \algeffseqoverindex{ \alpha }{ \text{\unboldmath$\mathit{I}$} }   ]   \,  |  \, \epsilon''$, we have
         \[
          \Delta  \ottsym{,}   \algeffseqoverindex{ \beta }{ \text{\unboldmath$\mathit{J}$} }   \ottsym{;}   \mathsf{none}  \,   \vdash  \ottnt{v} \,   \ottsym{:}   \ottnt{A'}    [   \algeffseqoverindex{ \ottnt{C} }{ \text{\unboldmath$\mathit{I}$} }   \ottsym{/}   \algeffseqoverindex{ \alpha }{ \text{\unboldmath$\mathit{I}$} }   ]   \,  |  \, \epsilon
         \]
         by \reflem{val-any-eff}.
         Since $\mathsf{op} \,  \in  \, \epsilon'' \subseteq \epsilon'$ and $\Delta  \ottsym{;}   \mathsf{none}  \,   \vdash  \ottnt{h}  \ottsym{:}  \ottnt{B} \,  |  \, \epsilon'  \Rightarrow  \ottnt{A} \,  |  \, \epsilon$ and $\mathsf{op} \,  \not\in  \,  \mathit{ops}  (  \ottnt{h}  ) $,
         we have $\mathsf{op} \,  \in  \, \epsilon$ by \reflem{handler-op-inheritance}.
         Thus, we finish by \T{OpCont}.
         \end{caseanalysis}

         \case \T{Resume}: Contradictory.
         \case \T{Let}:
         We have two reduction rules which can be applied to let expressions.
         \begin{caseanalysis}
          \case \R{Let}:
          We are given
          \begin{itemize}
           \item $\ottnt{e_{{\mathrm{1}}}} \,  =  \, \mathsf{let} \, \mathit{x}  \ottsym{=}   \Lambda\!  \,  \algeffseqover{ \alpha }   \ottsym{.}  \ottnt{v} \,  \mathsf{in}  \, \ottnt{e}$,
           \item $\ottnt{e_{{\mathrm{2}}}} \,  =  \,  \ottnt{e}    [   \Lambda\!  \,  \algeffseqover{ \alpha }   \ottsym{.}  \ottnt{v}  \ottsym{/}  \mathit{x}  ]  $,
           \item $\Delta  \ottsym{;}   \mathsf{none}  \,   \vdash  \mathsf{let} \, \mathit{x}  \ottsym{=}   \Lambda\!  \,  \algeffseqover{ \alpha }   \ottsym{.}  \ottnt{v} \,  \mathsf{in}  \, \ottnt{e} \,   \ottsym{:}  \ottnt{A} \,  |  \, \epsilon$,
           \item $\Delta  \ottsym{,}   \algeffseqover{ \alpha }   \ottsym{;}   \mathsf{none}  \,   \vdash  \ottnt{v} \,   \ottsym{:}  \ottnt{B} \,  |  \, \epsilon$, and
           \item $\Delta  \ottsym{,}  \mathit{x} \,  \mathord{:}  \,  \text{\unboldmath$\forall$}  \,  \algeffseqover{ \alpha }   \ottsym{.}  \ottnt{B}  \ottsym{;}   \mathsf{none}  \,   \vdash  \ottnt{e} \,   \ottsym{:}  \ottnt{A} \,  |  \, \epsilon$.
          \end{itemize}
          We have the conclusion by \reflem{val-subst}.

          \case \R{OpLet}:
          By \reflem{cont-inversion}, we are given
          \begin{itemize}
           \item $\ottnt{e_{{\mathrm{1}}}} \,  =  \, \mathsf{let} \, \mathit{x}  \ottsym{=}   \Lambda\!  \,  \algeffseqoverindex{ \alpha }{ \text{\unboldmath$\mathit{I}$} }   \ottsym{.}   \textup{\texttt{\#}\relax}  \mathsf{op}   \ottsym{(}    \algeffseqoverindex{  \text{\unboldmath$\forall$}  \,  \algeffseqoverindex{ \beta }{ \text{\unboldmath$\mathit{J}$} }   \ottsym{.}  \ottnt{C} }{ \text{\unboldmath$\mathit{I'}$} }    \ottsym{,}    \Lambda\!  \,  \algeffseqoverindex{ \beta }{ \text{\unboldmath$\mathit{J}$} }   \ottsym{.}  \ottnt{v}   \ottsym{,}    \ottnt{E} ^{  \algeffseqoverindex{ \beta }{ \text{\unboldmath$\mathit{J}$} }  }    \ottsym{)}  \,  \mathsf{in}  \, \ottnt{e}$,
           \item $\ottnt{e_{{\mathrm{2}}}} \,  =  \,  \textup{\texttt{\#}\relax}  \mathsf{op}   \ottsym{(}    \algeffseqoverindex{  \text{\unboldmath$\forall$}  \,  \algeffseqoverindex{ \alpha }{ \text{\unboldmath$\mathit{I}$} }   \ottsym{.}   \text{\unboldmath$\forall$}  \,  \algeffseqoverindex{ \beta }{ \text{\unboldmath$\mathit{J}$} }   \ottsym{.}  \ottnt{C} }{ \text{\unboldmath$\mathit{I'}$} }    \ottsym{,}    \Lambda\!  \,  \algeffseqoverindex{ \alpha }{ \text{\unboldmath$\mathit{I}$} }   \ottsym{.}   \Lambda\!  \,  \algeffseqoverindex{ \beta }{ \text{\unboldmath$\mathit{J}$} }   \ottsym{.}  \ottnt{v}   \ottsym{,}   \mathsf{let} \, \mathit{x}  \ottsym{=}   \Lambda\!  \,  \algeffseqoverindex{ \alpha }{ \text{\unboldmath$\mathit{I}$} }   \ottsym{.}   \ottnt{E} ^{  \algeffseqoverindex{ \beta }{ \text{\unboldmath$\mathit{J}$} }  }  \,  \mathsf{in}  \, \ottnt{e}   \ottsym{)} $,
           \item $\Delta  \ottsym{;}   \mathsf{none}  \,   \vdash  \mathsf{let} \, \mathit{x}  \ottsym{=}   \Lambda\!  \,  \algeffseqoverindex{ \alpha }{ \text{\unboldmath$\mathit{I}$} }   \ottsym{.}   \textup{\texttt{\#}\relax}  \mathsf{op}   \ottsym{(}    \algeffseqoverindex{  \text{\unboldmath$\forall$}  \,  \algeffseqoverindex{ \beta }{ \text{\unboldmath$\mathit{J}$} }   \ottsym{.}  \ottnt{C} }{ \text{\unboldmath$\mathit{I'}$} }    \ottsym{,}    \Lambda\!  \,  \algeffseqoverindex{ \beta }{ \text{\unboldmath$\mathit{J}$} }   \ottsym{.}  \ottnt{v}   \ottsym{,}    \ottnt{E} ^{  \algeffseqoverindex{ \beta }{ \text{\unboldmath$\mathit{J}$} }  }    \ottsym{)}  \,  \mathsf{in}  \, \ottnt{e} \,   \ottsym{:}  \ottnt{A} \,  |  \, \epsilon$,
           \item $\Delta  \ottsym{,}   \algeffseqoverindex{ \alpha }{ \text{\unboldmath$\mathit{I}$} }   \ottsym{;}   \mathsf{none}  \,   \vdash   \textup{\texttt{\#}\relax}  \mathsf{op}   \ottsym{(}    \algeffseqoverindex{  \text{\unboldmath$\forall$}  \,  \algeffseqoverindex{ \beta }{ \text{\unboldmath$\mathit{J}$} }   \ottsym{.}  \ottnt{C} }{ \text{\unboldmath$\mathit{I'}$} }    \ottsym{,}    \Lambda\!  \,  \algeffseqoverindex{ \beta }{ \text{\unboldmath$\mathit{J}$} }   \ottsym{.}  \ottnt{v}   \ottsym{,}    \ottnt{E} ^{  \algeffseqoverindex{ \beta }{ \text{\unboldmath$\mathit{J}$} }  }    \ottsym{)}  \,   \ottsym{:}  \ottnt{B} \,  |  \, \epsilon$,
           \item $\Delta  \ottsym{,}  \mathit{x} \,  \mathord{:}  \,  \text{\unboldmath$\forall$}  \,  \algeffseqoverindex{ \alpha }{ \text{\unboldmath$\mathit{I}$} }   \ottsym{.}  \ottnt{B}  \ottsym{;}   \mathsf{none}  \,   \vdash  \ottnt{e} \,   \ottsym{:}  \ottnt{A} \,  |  \, \epsilon$,
           \item $\epsilon' \,  \subseteq  \, \epsilon$,
           \item $\mathit{ty} \, \ottsym{(}  \mathsf{op}  \ottsym{)} \,  =  \,   \text{\unboldmath$\forall$}     \algeffseqoverindex{ \gamma }{ \text{\unboldmath$\mathit{I'}$} }   .  \ottnt{A'}  \hookrightarrow  \ottnt{B'} $,
           \item $\mathsf{op} \,  \in  \, \epsilon'$,
           \item $\Delta  \ottsym{,}   \algeffseqoverindex{ \alpha }{ \text{\unboldmath$\mathit{I}$} }   \vdash   \algeffseqoverindex{  \text{\unboldmath$\forall$}  \,  \algeffseqoverindex{ \beta }{ \text{\unboldmath$\mathit{J}$} }   \ottsym{.}  \ottnt{C} }{ \text{\unboldmath$\mathit{I'}$} } $,
           \item $\Delta  \ottsym{,}   \algeffseqoverindex{ \alpha }{ \text{\unboldmath$\mathit{I}$} }   \ottsym{,}   \algeffseqoverindex{ \beta }{ \text{\unboldmath$\mathit{J}$} }   \ottsym{;}   \mathsf{none}  \,   \vdash  \ottnt{v} \,   \ottsym{:}   \ottnt{A'}    [   \algeffseqoverindex{ \ottnt{C} }{ \text{\unboldmath$\mathit{I'}$} }   \ottsym{/}   \algeffseqoverindex{ \gamma }{ \text{\unboldmath$\mathit{I'}$} }   ]   \,  |  \, \epsilon'$, and
           \item $ \Delta  \ottsym{,}   \algeffseqoverindex{ \alpha }{ \text{\unboldmath$\mathit{I}$} }    \vdash    \ottnt{E} ^{  \algeffseqoverindex{ \beta }{ \text{\unboldmath$\mathit{J}$} }  }    \ottsym{:}     \text{\unboldmath$\forall$}  \,  \algeffseqoverindex{ \beta }{ \text{\unboldmath$\mathit{J}$} }   \ottsym{.}  \ottsym{(}   \ottnt{B'}    [   \algeffseqoverindex{ \ottnt{C} }{ \text{\unboldmath$\mathit{I'}$} }   \ottsym{/}   \algeffseqoverindex{ \gamma }{ \text{\unboldmath$\mathit{I'}$} }   ]    \ottsym{)}  \multimap  \ottnt{B}   \,  |  \, \epsilon' $.
          \end{itemize}
          By \TE{Weak} and \TE{Let},
          \[
            \Delta   \vdash   \mathsf{let} \, \mathit{x}  \ottsym{=}   \Lambda\!  \,  \algeffseqoverindex{ \alpha }{ \text{\unboldmath$\mathit{I}$} }   \ottsym{.}   \ottnt{E} ^{  \algeffseqoverindex{ \beta }{ \text{\unboldmath$\mathit{J}$} }  }  \,  \mathsf{in}  \, \ottnt{e}   \ottsym{:}     \text{\unboldmath$\forall$}  \,  \algeffseqoverindex{ \alpha }{ \text{\unboldmath$\mathit{I}$} }   \ottsym{.}   \text{\unboldmath$\forall$}  \,  \algeffseqoverindex{ \beta }{ \text{\unboldmath$\mathit{J}$} }   \ottsym{.}  \ottsym{(}   \ottnt{B'}    [   \algeffseqoverindex{ \ottnt{C} }{ \text{\unboldmath$\mathit{I'}$} }   \ottsym{/}   \algeffseqoverindex{ \gamma }{ \text{\unboldmath$\mathit{I'}$} }   ]    \ottsym{)}  \multimap  \ottnt{A}   \,  |  \, \epsilon .
          \]
          Since $\Delta  \ottsym{,}   \algeffseqoverindex{ \alpha }{ \text{\unboldmath$\mathit{I}$} }   \vdash   \algeffseqoverindex{  \text{\unboldmath$\forall$}  \,  \algeffseqoverindex{ \beta }{ \text{\unboldmath$\mathit{J}$} }   \ottsym{.}  \ottnt{C} }{ \text{\unboldmath$\mathit{I'}$} } $,
          \[
           \Delta  \vdash   \algeffseqoverindex{  \text{\unboldmath$\forall$}  \,  \algeffseqoverindex{ \alpha }{ \text{\unboldmath$\mathit{I}$} }   \ottsym{.}   \text{\unboldmath$\forall$}  \,  \algeffseqoverindex{ \beta }{ \text{\unboldmath$\mathit{J}$} }   \ottsym{.}  \ottnt{C} }{ \text{\unboldmath$\mathit{I'}$} } .
          \]
          Thus, by \T{OpCont}, we have the conclusion.
         \end{caseanalysis}
        \end{caseanalysis}

  \item Suppose that $\Delta  \ottsym{;}   \mathsf{none}  \,   \vdash  \ottnt{e_{{\mathrm{1}}}} \,   \ottsym{:}  \ottnt{A} \,  |  \, \epsilon$ and $\ottnt{e_{{\mathrm{1}}}}  \longrightarrow  \ottnt{e_{{\mathrm{2}}}}$.
        By definition, there exists some $\ottnt{E}$, $\ottnt{e'_{{\mathrm{1}}}}$, and $\ottnt{e'_{{\mathrm{2}}}}$ such that
        $\ottnt{e_{{\mathrm{1}}}} \,  =  \,  \ottnt{E}  [  \ottnt{e'_{{\mathrm{1}}}}  ] $, $\ottnt{e_{{\mathrm{2}}}} \,  =  \,  \ottnt{E}  [  \ottnt{e'_{{\mathrm{2}}}}  ] $, and $\ottnt{e'_{{\mathrm{1}}}}  \rightsquigarrow  \ottnt{e'_{{\mathrm{2}}}}$.
        By induction on the derivation of $\Delta  \ottsym{;}   \mathsf{none}  \,   \vdash   \ottnt{E}  [  \ottnt{e'_{{\mathrm{1}}}}  ]  \,   \ottsym{:}  \ottnt{A} \,  |  \, \epsilon$.
        If $\ottnt{E} \,  =  \,  [\,] $, then we have the conclusion by the first case.
        In what follows, we suppose that $\ottnt{E} \,  \not=  \,  [\,] $.
        By case analysis on the typing rule applied last to derive
        $\Delta  \ottsym{;}   \mathsf{none}  \,   \vdash   \ottnt{E}  [  \ottnt{e'_{{\mathrm{1}}}}  ]  \,   \ottsym{:}  \ottnt{A} \,  |  \, \epsilon$.
        \begin{caseanalysis}
         \case \T{Var}, \T{Const}, \T{Abs}, \T{OpCont}, and \T{Resume}:
         It is contradictory because $\ottnt{E} \,  =  \,  [\,] $.

         \case \T{App}:
         By case analysis on $\ottnt{E}$.
         \begin{caseanalysis}
          \case $\ottnt{E} \,  =  \, \ottnt{E'} \, \ottnt{e}$:
          We are given
          \begin{itemize}
           \item $\Delta  \ottsym{;}   \mathsf{none}  \,   \vdash   \ottnt{E'}  [  \ottnt{e'_{{\mathrm{1}}}}  ]  \,   \ottsym{:}   \ottnt{B}   \rightarrow  \!  \epsilon'  \;  \ottnt{A}  \,  |  \, \epsilon$,
           \item $\Delta  \ottsym{;}   \mathsf{none}  \,   \vdash  \ottnt{e} \,   \ottsym{:}  \ottnt{B} \,  |  \, \epsilon$, and
           \item $\epsilon' \,  \subseteq  \, \epsilon$
          \end{itemize}
          for some $\ottnt{B}$ and $\epsilon'$.
          By the IH, $\Delta  \ottsym{;}   \mathsf{none}  \,   \vdash   \ottnt{E'}  [  \ottnt{e'_{{\mathrm{2}}}}  ]  \,   \ottsym{:}   \ottnt{B}   \rightarrow  \!  \epsilon'  \;  \ottnt{A}  \,  |  \, \epsilon$.
          By \T{App}, we finish.

          \case $\ottnt{E} \,  =  \, \ottnt{v} \, \ottnt{E'}$: By the IH.
         \end{caseanalysis}

         \case \T{Op}: By the IH.
         \case \T{Weak}: By the IH.
         \case \T{Handle}: By the IH.
         \case \T{Let}: By the IH.
        \end{caseanalysis}
 \end{enumerate}
\end{proof}

\ifrestate
\thmTypeSoundness*
\else
\begin{theorem}[Type Soundness of {\interlang}]{type-sound}
 If $\Delta  \ottsym{;}   \mathsf{none}  \,   \vdash  \ottnt{e} \,   \ottsym{:}  \ottnt{A} \,  |  \, \epsilon$ and $\ottnt{e}  \longrightarrow^{*}  \ottnt{e'}$ and
 $\ottnt{e'}  \centernot\longrightarrow$, then (1) $\ottnt{e'}$ is a value or (2) $\ottnt{e'} \,  =  \,  \textup{\texttt{\#}\relax}  \mathsf{op}   \ottsym{(}    \algeffseqover{ \sigma }    \ottsym{,}   \ottnt{w}   \ottsym{,}   \ottnt{E}   \ottsym{)} $
 for some $\mathsf{op} \,  \in  \, \epsilon$, $ \algeffseqover{ \sigma } $, $\ottnt{w}$, and $\ottnt{E}$.
\end{theorem}
\fi
\begin{proof}
 By \reflem{subject-red}, $\Delta  \ottsym{;}   \mathsf{none}  \,   \vdash  \ottnt{e'} \,   \ottsym{:}  \ottnt{A} \,  |  \, \epsilon$.
 We have the conclusion by \reflem{progress}.
\end{proof}

\subsection{Elaboration is type-preserving}

\begin{defn}
 Elaboration $ \Gamma  \mathrel{\algefftransarrow{ \ottnt{S} } }  \Gamma' $ of $\Gamma$ to $\Gamma'$ with $\ottnt{S}$
 is the least relation that satisfies the following rules.
 \begin{center}
  $\ottdruleElabGXXEmpty{}$ \hfil
  $\ottdruleElabGXXVar{}$ \\[1ex]
  $\ottdruleElabGXXTyVar{}$
 \end{center}
\end{defn}

\begin{defn}
 Elaboration $ \ottnt{R}  \mathrel{\algefftransarrow{} }  \ottnt{r} $ of $\ottnt{R}$ to $\ottnt{r}$ is defined as follows.
 \[\begin{array}{c}
    \mathsf{none}   \mathrel{\algefftransarrow{} }   \mathsf{none}   \qquad
   \ottsym{(}   \algeffseqover{ \alpha }   \ottsym{,}  \mathit{x} \,  \mathord{:}  \, \ottnt{A}  \ottsym{,}   \ottnt{B}   \rightarrow  \!  \epsilon  \;  \ottnt{C}   \ottsym{)}  \mathrel{\algefftransarrow{} }  \ottsym{(}   \algeffseqover{ \alpha }   \ottsym{,}  \ottnt{A}  \ottsym{,}   \ottnt{B}   \rightarrow  \!  \epsilon  \;  \ottnt{C}   \ottsym{)} 
   \end{array}\]
\end{defn}

\begin{lemma}{trans-var-preserving}
 If $ \Gamma  \mathrel{\algefftransarrow{ \ottnt{S} } }  \Gamma' $, then, for any $\mathit{x} \,  \mathord{:}  \, \sigma \,  \in  \, \Gamma$,
 $\ottnt{S}  \ottsym{(}  \mathit{x}  \ottsym{)} \,  \mathord{:}  \, \sigma \,  \in  \, \Gamma'$.
\end{lemma}
\begin{proof}
 By induction on the derivation of $ \Gamma  \mathrel{\algefftransarrow{ \ottnt{S} } }  \Gamma' $.
\end{proof}

\begin{lemma}{trans-tyvar-preserving}
 If $ \Gamma  \mathrel{\algefftransarrow{ \ottnt{S} } }  \Gamma' $, then, for any $\alpha$,
 $\alpha \,  \in  \, \Gamma$ if and only if $\alpha \,  \in  \, \Gamma'$.
\end{lemma}
\begin{proof}
 By induction on the derivation of $ \Gamma  \mathrel{\algefftransarrow{ \ottnt{S} } }  \Gamma' $.
\end{proof}

\begin{lemma}{surface-typing-context-wf}
 \begin{enumerate}
  \item If $\Gamma  \ottsym{;}  \ottnt{R}  \vdash  \ottnt{M}  \ottsym{:}  \ottnt{A} \,  |  \, \epsilon$, then $\vdash  \Gamma$.
  \item If $\Gamma  \ottsym{;}  \ottnt{R}  \vdash  \ottnt{H}  \ottsym{:}  \ottnt{A} \,  |  \, \epsilon  \Rightarrow  \ottnt{B} \,  |  \, \epsilon'$, then $\vdash  \Gamma$.
 \end{enumerate}
 \end{lemma}
 \begin{proof}
 Straightforward by induction on the typing derivations.
\end{proof}

\begin{lemma}{trans-var-strengthening}
 If $ \Gamma_{{\mathrm{1}}}  \ottsym{,}  \mathit{x} \,  \mathord{:}  \, \sigma  \ottsym{,}  \Gamma_{{\mathrm{2}}}  \mathrel{\algefftransarrow{ \ottnt{S} } }  \Gamma' $,
 then $\Gamma' \,  =  \, \Gamma'_{{\mathrm{1}}}  \ottsym{,}  \ottnt{S}  \ottsym{(}  \mathit{x}  \ottsym{)} \,  \mathord{:}  \, \sigma  \ottsym{,}  \Gamma'_{{\mathrm{2}}}$ and
 $ \Gamma_{{\mathrm{1}}}  \ottsym{,}  \Gamma_{{\mathrm{2}}}  \mathrel{\algefftransarrow{ \ottnt{S} } }  \Gamma'_{{\mathrm{1}}}  \ottsym{,}  \Gamma'_{{\mathrm{2}}} $ for some $\Gamma'_{{\mathrm{1}}}$ and $\Gamma'_{{\mathrm{2}}}$.
\end{lemma}
\begin{proof}
 Straightforward by induction on $\Gamma_{{\mathrm{2}}}$.
\end{proof}

\begin{lemma}{trans-preserving}
 Suppose that $ \ottnt{R}  \mathrel{\algefftransarrow{} }  \ottnt{r} $ and $ \Gamma  \mathrel{\algefftransarrow{ \ottnt{S} } }  \Gamma' $ and $\vdash  \Gamma'$.
 \begin{enumerate}
  \item If $\Gamma  \ottsym{;}  \ottnt{R}  \vdash  \ottnt{M}  \ottsym{:}  \ottnt{A} \,  |  \, \epsilon$, then
        there exists some $\ottnt{e}$ such that 
        $ \Gamma ;  \ottnt{R}   \vdash   \ottnt{M}  :  \ottnt{A}  \,  |   \, \epsilon  \mathrel{\algefftransarrow{ \ottnt{S} } }  \ottnt{e} $ and
        $\Gamma'  \ottsym{;}  \ottnt{r} \,   \vdash  \ottnt{e} \,   \ottsym{:}  \ottnt{A} \,  |  \, \epsilon$.
  \item If $\Gamma  \ottsym{;}  \ottnt{R}  \vdash  \ottnt{H}  \ottsym{:}  \ottnt{A} \,  |  \, \epsilon  \Rightarrow  \ottnt{B} \,  |  \, \epsilon'$, then
        there exists some $\ottnt{h}$ such that 
        $ \Gamma ;  \ottnt{R}   \vdash   \ottnt{H}  :  \ottnt{A}  \,  |   \, \epsilon  \, \Rightarrow   \ottnt{B}  \,  |   \, \epsilon'  \mathrel{\algefftransarrow{ \ottnt{S} } }  \ottnt{h} $ and
        $\Gamma'  \ottsym{;}  \ottnt{r} \,   \vdash  \ottnt{h}  \ottsym{:}  \ottnt{A} \,  |  \, \epsilon  \Rightarrow  \ottnt{B} \,  |  \, \epsilon'$.
 \end{enumerate}
\end{lemma}
\begin{proof}
 By mutual induction on the typing derivations.

 \begin{enumerate}
  \item By case analysis on the typing rule applied last.
        \begin{caseanalysis}
         \case \TSrule{Var}:
         We are given $\Gamma  \ottsym{;}  \ottnt{R}  \vdash  \mathit{x}  \ottsym{:}   \ottnt{B}    [   \algeffseqover{ \ottnt{C} }   \ottsym{/}   \algeffseqover{ \alpha }   ]   \,  |  \, \epsilon$ and, by inversion,
         $\vdash  \Gamma$ and $\mathit{x} \,  \mathord{:}  \,  \text{\unboldmath$\forall$}  \,  \algeffseqover{ \alpha }   \ottsym{.}  \ottnt{B} \,  \in  \, \Gamma$ and $\Gamma  \vdash   \algeffseqover{ \ottnt{C} } $.
         By \reflem{trans-var-preserving}, $\ottnt{S}  \ottsym{(}  \mathit{x}  \ottsym{)} \,  \mathord{:}  \,  \text{\unboldmath$\forall$}  \,  \algeffseqover{ \alpha }   \ottsym{.}  \ottnt{B} \,  \in  \, \Gamma'$.
         Thus, $\ottnt{S}  \ottsym{(}  \mathit{x}  \ottsym{)}$ is defined, so
         \[
           \Gamma ;  \ottnt{R}   \vdash   \mathit{x}  :   \ottnt{B}    [   \algeffseqover{ \ottnt{C} }   \ottsym{/}   \algeffseqover{ \alpha }   ]    \,  |   \, \epsilon  \mathrel{\algefftransarrow{ \ottnt{S} } }  \ottnt{S}  \ottsym{(}  \mathit{x}  \ottsym{)} \,  \algeffseqover{ \ottnt{C} }  
         \]
         by \Elab{Var}.
         By \reflem{trans-tyvar-preserving},
         $\Gamma'  \vdash   \algeffseqover{ \ottnt{C} } $.
         By \T{Var}, we finish.

         \case \TSrule{Const}: By \Elab{Const} and \T{Const}.
         \case \TSrule{Abs}:
         We are given $\Gamma  \ottsym{;}  \ottnt{R}  \vdash   \lambda\!  \, \mathit{x}  \ottsym{.}  \ottnt{M'}  \ottsym{:}   \ottnt{B}   \rightarrow  \!  \epsilon'  \;  \ottnt{C}  \,  |  \, \epsilon$ and, by inversion,
         $\Gamma  \ottsym{,}  \mathit{x} \,  \mathord{:}  \, \ottnt{B}  \ottsym{;}  \ottnt{R}  \vdash  \ottnt{M'}  \ottsym{:}  \ottnt{C} \,  |  \, \epsilon'$.
         Without loss of generality, we can suppose that $\mathit{x}$ does not occur
         in $\ottnt{S}$ and $\Gamma'$.
         Since $ \Gamma  \mathrel{\algefftransarrow{ \ottnt{S} } }  \Gamma' $, we have $ \Gamma  \ottsym{,}  \mathit{x} \,  \mathord{:}  \, \ottnt{B}  \mathrel{\algefftransarrow{  \ottnt{S}  \,\circ\, \{  \mathit{x}  \, {\mapsto} \,  \mathit{x}  \}  } }  \Gamma  \ottsym{,}  \mathit{x} \,  \mathord{:}  \, \ottnt{B} $
         by \ElabG{Var}.
         By \reflem{surface-typing-context-wf}, $\vdash  \Gamma  \ottsym{,}  \mathit{x} \,  \mathord{:}  \, \ottnt{B}$.
         Thus, $\Gamma  \vdash  \ottnt{B}$.
         By \reflem{trans-tyvar-preserving}, $\Gamma'  \vdash  \ottnt{B}$.
         Thus, by \WF{Var}, $\vdash  \Gamma'  \ottsym{,}  \mathit{x} \,  \mathord{:}  \, \ottnt{B}$.
         By the IH, $ \Gamma  \ottsym{,}  \mathit{x} \,  \mathord{:}  \, \ottnt{B} ;  \ottnt{R}   \vdash   \ottnt{M'}  :  \ottnt{C}  \,  |   \, \epsilon'  \mathrel{\algefftransarrow{  \ottnt{S}  \,\circ\, \{  \mathit{x}  \, {\mapsto} \,  \mathit{x}  \}  } }  \ottnt{e'} $
         for some $\ottnt{e'}$ such that
         $\Gamma'  \ottsym{,}  \mathit{x} \,  \mathord{:}  \, \ottnt{B}  \ottsym{;}  \ottnt{r} \,   \vdash  \ottnt{e'} \,   \ottsym{:}  \ottnt{C} \,  |  \, \epsilon'$.
         By \Elab{Abs} and \T{Abs}, we finish.

         \case \TSrule{App}: By the IHs, \Elab{App}, and \T{App}.
         \case \TSrule{Op}: By the IH, \Elab{Op}, and \T{Op} with \reflem{trans-tyvar-preserving}.
         \case \TSrule{Let}: Similar to \TSrule{Abs}.
         \case \TSrule{Weak}: By the IH, \Elab{Weak}, and \T{Weak}.
         \case \TSrule{Handle}: By the IH, \Elab{Handle}, and \T{Handle}.
         \case \TSrule{Resume}:
         We are given
         $\Gamma_{{\mathrm{1}}}  \ottsym{,}  \mathit{x} \,  \mathord{:}  \, \ottnt{D}  \ottsym{,}  \Gamma_{{\mathrm{2}}}  \ottsym{;}  \ottsym{(}   \algeffseqover{ \alpha }   \ottsym{,}  \mathit{x} \,  \mathord{:}  \, \ottnt{B}  \ottsym{,}   \ottnt{C}   \rightarrow  \!  \epsilon'  \;  \ottnt{A}   \ottsym{)}  \vdash  \mathsf{resume} \, \ottnt{M'}  \ottsym{:}  \ottnt{A} \,  |  \, \epsilon$
         and, by inversion,
         \begin{itemize}
          \item $\vdash  \Gamma_{{\mathrm{1}}}  \ottsym{,}  \mathit{x} \,  \mathord{:}  \, \ottnt{D}  \ottsym{,}  \Gamma_{{\mathrm{2}}}$,
          \item $ \algeffseqover{ \alpha }  \,  \in  \, \Gamma_{{\mathrm{1}}}$,
          \item $\epsilon' \,  \subseteq  \, \epsilon$, and
          \item $\Gamma_{{\mathrm{1}}}  \ottsym{,}  \Gamma_{{\mathrm{2}}}  \ottsym{,}   \algeffseqover{ \beta }   \ottsym{,}  \mathit{x} \,  \mathord{:}  \, \ottnt{B} \,  [   \algeffseqover{ \beta }   \ottsym{/}   \algeffseqover{ \alpha }   ]   \ottsym{;}  \ottsym{(}   \algeffseqover{ \alpha }   \ottsym{,}  \mathit{x} \,  \mathord{:}  \, \ottnt{B}  \ottsym{,}   \ottnt{C}   \rightarrow  \!  \epsilon'  \;  \ottnt{A}   \ottsym{)}  \vdash  \ottnt{M'}  \ottsym{:}   \ottnt{C}    [   \algeffseqover{ \beta }   \ottsym{/}   \algeffseqover{ \alpha }   ]   \,  |  \, \epsilon$.
         \end{itemize}
         Let $\mathit{y}$ be a fresh variable.
         Since $ \Gamma_{{\mathrm{1}}}  \ottsym{,}  \mathit{x} \,  \mathord{:}  \, \ottnt{D}  \ottsym{,}  \Gamma_{{\mathrm{2}}}  \mathrel{\algefftransarrow{ \ottnt{S} } }  \Gamma' $, there exist some $\Gamma'_{{\mathrm{1}}}$ and $\Gamma'_{{\mathrm{2}}}$
         such that
         $\Gamma' \,  =  \, \Gamma'_{{\mathrm{1}}}  \ottsym{,}  \ottnt{S}  \ottsym{(}  \mathit{x}  \ottsym{)} \,  \mathord{:}  \, \ottnt{D}  \ottsym{,}  \Gamma'_{{\mathrm{2}}}$ and $ \Gamma_{{\mathrm{1}}}  \ottsym{,}  \Gamma_{{\mathrm{2}}}  \mathrel{\algefftransarrow{ \ottnt{S} } }  \Gamma'_{{\mathrm{1}}}  \ottsym{,}  \Gamma'_{{\mathrm{2}}} $
         by \reflem{trans-var-strengthening}.
         Since $\vdash  \Gamma_{{\mathrm{1}}}  \ottsym{,}  \mathit{x} \,  \mathord{:}  \, \ottnt{D}  \ottsym{,}  \Gamma_{{\mathrm{2}}}$, $\mathit{x} \,  \not\in  \,  \mathit{dom}  (  \Gamma_{{\mathrm{1}}}  \ottsym{,}  \Gamma_{{\mathrm{2}}}  ) $.
         Thus, $ \Gamma_{{\mathrm{1}}}  \ottsym{,}  \Gamma_{{\mathrm{2}}}  \mathrel{\algefftransarrow{  \ottnt{S}  \,\circ\, \{  \mathit{x}  \, {\mapsto} \,  \mathit{y}  \}  } }  \Gamma'_{{\mathrm{1}}}  \ottsym{,}  \Gamma'_{{\mathrm{2}}} $.
         By \ElabG{Var} and \ElabG{TyVar},
         $ \Gamma_{{\mathrm{1}}}  \ottsym{,}  \Gamma_{{\mathrm{2}}}  \ottsym{,}   \algeffseqover{ \beta }   \ottsym{,}  \mathit{x} \,  \mathord{:}  \, \ottnt{B} \,  [   \algeffseqover{ \beta }   \ottsym{/}   \algeffseqover{ \alpha }   ]   \mathrel{\algefftransarrow{  \ottnt{S}  \,\circ\, \{  \mathit{x}  \, {\mapsto} \,  \mathit{y}  \}  } }  \Gamma'_{{\mathrm{1}}}  \ottsym{,}  \Gamma'_{{\mathrm{2}}}  \ottsym{,}   \algeffseqover{ \beta }   \ottsym{,}  \mathit{y} \,  \mathord{:}  \, \ottnt{B} \,  [   \algeffseqover{ \beta }   \ottsym{/}   \algeffseqover{ \alpha }   ]  $.
         Since $\vdash  \Gamma'_{{\mathrm{1}}}  \ottsym{,}  \ottnt{S}  \ottsym{(}  \mathit{x}  \ottsym{)} \,  \mathord{:}  \, \ottnt{D}  \ottsym{,}  \Gamma'_{{\mathrm{2}}}$, we have $\vdash  \Gamma'_{{\mathrm{1}}}  \ottsym{,}  \Gamma'_{{\mathrm{2}}}  \ottsym{,}   \algeffseqover{ \beta }   \ottsym{,}  \mathit{y} \,  \mathord{:}  \, \ottnt{B} \,  [   \algeffseqover{ \beta }   \ottsym{/}   \algeffseqover{ \alpha }   ] $.
         by Lemmas~\ref{lem:strengthening-typing-context},
         \ref{lem:surface-typing-context-wf},
         \ref{lem:trans-tyvar-preserving}, and
         \ref{lem:weakening-typing-context}.
         Thus, by the IH,
         \[
           \Gamma_{{\mathrm{1}}}  \ottsym{,}  \Gamma_{{\mathrm{2}}}  \ottsym{,}   \algeffseqover{ \beta }   \ottsym{,}  \mathit{x} \,  \mathord{:}  \, \ottnt{B} \,  [   \algeffseqover{ \beta }   \ottsym{/}   \algeffseqover{ \alpha }   ]  ;  \ottsym{(}   \algeffseqover{ \alpha }   \ottsym{,}  \mathit{x} \,  \mathord{:}  \, \ottnt{B}  \ottsym{,}   \ottnt{C}   \rightarrow  \!  \epsilon'  \;  \ottnt{A}   \ottsym{)}   \vdash   \ottnt{M'}  :   \ottnt{C}    [   \algeffseqover{ \beta }   \ottsym{/}   \algeffseqover{ \alpha }   ]    \,  |   \, \epsilon  \mathrel{\algefftransarrow{  \ottnt{S}  \,\circ\, \{  \mathit{x}  \, {\mapsto} \,  \mathit{y}  \}  } }  \ottnt{e'} 
         \]
         for some $\ottnt{e'}$ such that
         $\Gamma'_{{\mathrm{1}}}  \ottsym{,}  \Gamma'_{{\mathrm{2}}}  \ottsym{,}   \algeffseqover{ \beta }   \ottsym{,}  \mathit{y} \,  \mathord{:}  \, \ottnt{B} \,  [   \algeffseqover{ \beta }   \ottsym{/}   \algeffseqover{ \alpha }   ]   \ottsym{;}  \ottnt{r} \,   \vdash  \ottnt{e'} \,   \ottsym{:}   \ottnt{C}    [   \algeffseqover{ \beta }   \ottsym{/}   \algeffseqover{ \alpha }   ]   \,  |  \, \epsilon$.
         By applying \Elab{Resume},
         \[
           \Gamma_{{\mathrm{1}}}  \ottsym{,}  \mathit{x} \,  \mathord{:}  \, \ottnt{D}  \ottsym{,}  \Gamma_{{\mathrm{2}}} ;  \ottsym{(}   \algeffseqover{ \alpha }   \ottsym{,}  \mathit{x} \,  \mathord{:}  \, \ottnt{B}  \ottsym{,}   \ottnt{C}   \rightarrow  \!  \epsilon'  \;  \ottnt{A}   \ottsym{)}   \vdash   \mathsf{resume} \, \ottnt{M'}  :  \ottnt{A}  \,  |   \, \epsilon  \mathrel{\algefftransarrow{ \ottnt{S} } }  \mathsf{resume} \,  \algeffseqover{ \beta }  \, \mathit{y}  \ottsym{.}  \ottnt{e'} .
         \]
         Since $\Gamma'_{{\mathrm{1}}}  \ottsym{,}  \ottnt{S}  \ottsym{(}  \mathit{x}  \ottsym{)} \,  \mathord{:}  \, \ottnt{D}  \ottsym{,}  \Gamma'_{{\mathrm{2}}}  \ottsym{,}   \algeffseqover{ \beta }   \ottsym{,}  \mathit{y} \,  \mathord{:}  \, \ottnt{B} \,  [   \algeffseqover{ \beta }   \ottsym{/}   \algeffseqover{ \alpha }   ]   \ottsym{;}  \ottnt{r} \,   \vdash  \ottnt{e'} \,   \ottsym{:}   \ottnt{C}    [   \algeffseqover{ \beta }   \ottsym{/}   \algeffseqover{ \alpha }   ]   \,  |  \, \epsilon$
         by \reflem{weakening} and
         $ \algeffseqover{ \alpha }  \,  \in  \, \Gamma'_{{\mathrm{1}}}  \ottsym{,}  \ottnt{S}  \ottsym{(}  \mathit{x}  \ottsym{)} \,  \mathord{:}  \, \ottnt{D}  \ottsym{,}  \Gamma'_{{\mathrm{2}}}$ by \reflem{trans-tyvar-preserving},
         we have
         \[
          \Gamma'_{{\mathrm{1}}}  \ottsym{,}  \ottnt{S}  \ottsym{(}  \mathit{x}  \ottsym{)} \,  \mathord{:}  \, \ottnt{D}  \ottsym{,}  \Gamma'_{{\mathrm{2}}}  \ottsym{;}  \ottnt{r} \,   \vdash  \mathsf{resume} \,  \algeffseqover{ \beta }  \, \mathit{y}  \ottsym{.}  \ottnt{e'} \,   \ottsym{:}  \ottnt{A} \,  |  \, \epsilon
         \]
         by \T{Resume}.
         
        \end{caseanalysis}
  \item By case analysis on the typing rule applied last.
        \begin{caseanalysis}
         \case \THSrule{Return}: Similar to \TSrule{Abs}.
         \case \THSrule{Op}: Similar to \TSrule{Abs}.
        \end{caseanalysis}
 \end{enumerate}
\end{proof}

\ifrestate
\thmElab*
\else
\begin{theorem}[Elaboration is type-preserving]{trans-preserving}
 If $\ottnt{M}$ is a well-typed program of $\ottnt{A}$,
 then $  \emptyset  ;   \mathsf{none}    \vdash   \ottnt{M}  :  \ottnt{A}  \,  |   \,  \langle \rangle   \mathrel{\algefftransarrow{  \emptyset  } }  \ottnt{e} $ and
 $ \emptyset   \ottsym{;}   \mathsf{none}  \,   \vdash  \ottnt{e} \,   \ottsym{:}  \ottnt{A} \,  |  \,  \langle \rangle $ for some $\ottnt{e}$.
\end{theorem}
\fi
\begin{proof}
 By \reflem{trans-preserving}.
\end{proof}

\fi}

\end{document}